\documentclass[11pt]{article}
\usepackage[utf8]{inputenc}
\usepackage{fullpage}
\usepackage{geometry}
\usepackage{mathtools}
\usepackage{blkarray}
\usepackage{amsmath}
\usepackage{tikz-cd}
\usepackage{graphicx}
\usetikzlibrary{decorations.pathmorphing}
\usepackage{amssymb}
\usepackage{enumitem}
\usepackage{mathrsfs}
\usepackage[colorlinks=true,linkcolor=blue,citecolor=blue,urlcolor=blue]{hyperref}
\usepackage{comment}
\usepackage{authblk}
\usepackage{float}
\usepackage{bigints}
\usepackage{relsize}
\usepackage{physics}
\usepackage{booktabs}
\usepackage{bbm}
\usetikzlibrary{matrix,positioning}
\usepackage[math]{cellspace}

\usepackage{xcolor}
\definecolor{toclink}{rgb}{0,0,0}

\usetikzlibrary{decorations.markings,decorations.pathreplacing}
\usepackage{amsthm}
\usepackage{bm}

\DeclarePairedDelimiterX{\infdivx}[2]{(}{)}{%
  #1\;\delimsize\|\;#2%
}

\newcommand{\rem}[1]{\hyperref[rem:#1]{Remark~\ref*{rem:#1}}}
\newcommand{\thm}[1]{\hyperref[thm:#1]{Theorem~\ref*{thm:#1}}}
\newcommand{\asm}[1]{\hyperref[asm:#1]{Assumption~\ref*{asm:#1}}}
\newcommand{\cor}[1]{\hyperref[cor:#1]{Corollary~\ref*{cor:#1}}}
\newcommand{\defn}[1]{\hyperref[def:#1]{Definition~\ref*{def:#1}}}
\newcommand{\lem}[1]{\hyperref[lem:#1]{Lemma~\ref*{lem:#1}}}
\newcommand{\prop}[1]{\hyperref[prop:#1]{Proposition~\ref*{prop:#1}}}
\newcommand{\fig}[1]{\hyperref[fig:#1]{Figure~\ref*{fig:#1}}}
\newcommand{\tab}[1]{\hyperref[tab:#1]{Table~\ref*{tab:#1}}}
\newcommand{\algo}[1]{\hyperref[algo:#1]{Algorithm~\ref*{algo:#1}}}
\renewcommand{\sec}[1]{\hyperref[sec:#1]{Section~\ref*{sec:#1}}}
\newcommand{\append}[1]{\hyperref[append:#1]{Appendix~\ref*{append:#1}}}

\newcommand{\mrm}[1]{\mathrm{#1}}
\DeclareMathOperator{\res}{res}

\newtheorem{theorem}{Theorem}[section]

\newtheorem{corollary}[theorem]{Corollary}
\newtheorem{lemma}[theorem]{Lemma}
\newtheorem{definition}[theorem]{Definition}

\newtheorem{remark}[theorem]{Remark}
\newtheorem{claim}[theorem]{Claim}

\theoremstyle{definition}
\newtheorem{example}[theorem]{Example}

\usepackage{cleveref}
\usepackage{url}

\DeclareMathOperator{\poly}{poly}

\newcommand{\R}{\mathbb{R}}

\newcommand{\cC}{\mathcal{C}}
\newcommand{\cD}{\mathcal{D}}
\newcommand{\cE}{\mathcal{E}}
\newcommand{\cF}{\mathcal{F}}

\newcommand{\cH}{\mathcal{H}}

\newcommand{\cL}{\mathcal{L}}

\newcommand{\cO}{\mathcal{O}}
\newcommand{\cR}{\mathcal{R}}

\newcommand{\cQ}{\mathcal{Q}}

\newcommand{\ff}{\mathbb{F}}
\newcommand{\zz}{\mathbb{Z}}
\newcommand{\RR}{\mathbb{R}}

\def\>{\rangle}
\def\<{\langle}

\def\id{\mathrm{id}}
\DeclareMathOperator{\Tot}{Tot}
\newcommand{\loc}{\mathrm{loc}}
\renewcommand{\mid}{\mathrm{mid}}
\renewcommand{\int}{\mathrm{int}}

\newcommand{\dec}{f}

\newcommand{\comB}{Q}

\def\be#1\ee{\begin{equation}#1\end{equation}}
\def\ba#1\ea{\begin{align}#1\end{align}}
\def\bas#1\eas{\begin{align*}#1\end{align*}}

\newcommand{\wt}{\widetilde}

\newcommand{\red}[1]{{\color{red} #1}}

\newcommand{\blue}[1]{{\color{blue} #1}}

\usepackage{etoolbox}
\makeatletter
\patchcmd{\enddefinition}{\@endpefalse}{}{}{}
\patchcmd{\endtheorem}{\@endpefalse}{}{}{}
\patchcmd{\endremark}{\@endpefalse}{}{}{}
\patchcmd{\endproof}{\@endpefalse}{}{}{}
\patchcmd{\endlemma}{\@endpefalse}{}{}{}
\patchcmd{\endproposition}{\@endpefalse}{}{}{}
\patchcmd{\endcorollary}{\@endpefalse}{}{}{}
\makeatother

\title{A passive self-correcting quantum memory in three dimensions}
\author[1]{Shankar Balasubramanian}

\author[1]{Margarita Davydova}

\author[2,3]{Ting-Chun Lin}

\affil[1]{\small Walter Burke Institute for Theoretical Physics and Institute for Quantum Information and Matter,
  California Institute of Technology, Pasadena, CA 91125, USA}

\affil[2]{\small Department of Physics, University of California San Diego, La Jolla, CA 92093, USA}

\affil[3]{\small Hon Hai Research Institute, Taipei, Taiwan}

\date{\vspace{-3em}}

\begin{document}

\maketitle

\begin{abstract}
We construct a 3D Pauli stabilizer Hamiltonian whose ground state space can encode a qubit for exponential time when coupled to a bath at non-zero temperature.  Our construction recursively applies a sequence of transformations to a seed Hamiltonian that increases the memory lifetime of the encoded qubit while maintaining geometric locality in $\mathbb{R}^3$.
\end{abstract}

\noindent\rule{\textwidth}{0.2pt}
\vspace{-25 pt}
\setcounter{tocdepth}{2}
{
  \hypersetup{linkcolor=toclink}
  \tableofcontents
}

\noindent\rule{\textwidth}{0.2pt}

\vspace{5 pt}

\section{Introduction}
A self-correcting quantum memory is a quantum many-body system that can store quantum information in its ground state subspace when interacting with a bath at non-zero temperature~\cite{dennis2002topological, brown2016quantum}.  In this system, quantum information is stored passively and in equilibrium, without the need for active error correction or other external sources of power.   Technologically, this could have significant implications for building energy-efficient quantum hard drives.  Theoretically, such a system provides an example of topological order at non-zero temperature and contributes to the ongoing program of classifying gapped phases of quantum matter.  However, all known examples of self-correcting quantum memories require four or more spatial dimensions, and whether self-correction is possible in three dimensions has remained a long-standing open question.  In this paper, we settle this question by providing a construction in three dimensions.

For a family of Hamiltonians $\{H_k, k \in \mathbb{N}\}$ with system size $n_k$, a standard way to diagnose self-correction is as follows.  We encode a qubit in $H_k$ and evolve it under local thermalization dynamics (such as generalized Glauber dynamics) at inverse temperature $\beta$. Define the memory lifetime $t_{\mrm{mem}}$ as the maximum time for which the encoded information can be reliably recovered.  If there exists a finite $\beta_c$ such that $t_{\mrm{mem}} \gtrsim \exp(n^{\eta}_k)$ for some $\eta > 0$ and for all $\beta > \beta_c$, we say that $\{H_k\}$ is self-correcting\footnote{A weaker but reasonable definition is $t_{\mrm{mem}} \to \infty$ as $n \to \infty$ for $\beta > \beta_c$.  We use the stronger definition because it captures the idea that a small system can keep a qubit safe for an incredibly long time, easily longer than the age of the universe.}.  In 4D, the (2,2) toric code satisfies $t_{\mrm{mem}} \sim \exp(n^{1/4})$ below a critical temperature~\cite{dennis2002topological, alicki2010thermal}.  In contrast, 2D stabilizer codes have $t_{\mrm{mem}} = O(1)$ for any finite $\beta$~\cite{alicki2009thermalization, bravyi2009no}. This is due to the existence of string-like operators creating point-like excitations that thermally proliferate.  This intuition extends to any 2D or 3D code described by a topological quantum field theory (TQFT)~\cite{yoshida2011feasibility}.

Thus, we must look beyond TQFT fixed point states for a 3D self-correcting quantum memory.  In 2011, Haah proposed the cubic code~\cite{haah2011local}, whose excitations are created by operators supported on a fractal subset of the lattice.  The energy barrier of this code scales like $\Theta(\log n)$, implying the best-case possibility of $t_{\mrm{mem}} \sim \exp(\Theta(\log n))$.  Unfortunately, subsequent numerical studies showed that $t_{\mrm{mem}} = O(1)$~\cite{bravyi2013quantum, siva2017topological}.  This rules out self-correction in the cubic code.  More generally,~\cite{haah2013commuting} proved that all 3D translation-invariant stabilizer codes have an $O(\log n)$ energy barrier and likely have $t_{\mrm{mem}} = O(1)$ due to the existence of string or fractal generators.  This suggests that breaking translation invariance is essential for achieving self-correction.
The first code that achieved a polynomial energy barrier is the welded toric code~\cite{michnicki20143d}, which is not translation-invariant. Despite an energy barrier of $O(n^{2/9})$, it was argued that $t_{\mrm{mem}} = O(1)$~\cite{siva2017topological}.  This is because the welded code is built from 3D toric codes of size $O(n^{2/3})$ and inherits the undesirable thermal properties of these codes. Recent developments in codes with optimal parameters in three dimensions~\cite{portnoy2023local, lin2023geometrically, williamson2023layer} achieve an energy barrier $O(n^{1/3})$, but these examples are not self-correcting for similar reasons~\cite{gu2025layer, williamson2025partial, baspin2025free}.

One of the most interesting ideas was proposed by Brell~\cite{brell2016proposal}.  Observing that the classical Ising model on a Sierpi\'nski carpet has an ordered phase~\cite{vezzani2003spontaneous}, he suggested taking a homological product of 2D surface codes supported on a Sierpi\'nski carpet and its dual.  This results in a 4D surface code on a manifold with many punctures, which one might hope could retain the self-correction properties of the 4D surface code.  Furthermore, the fractal dimension of the manifold can be made as small as $2 + \epsilon < 3$, suggesting the possibility of a local embedding in $\mathbb{R}^3$. No such procedure was provided in Brell's paper, but we believe that the ideas in the current paper can resolve this question.  Unfortunately, Brell's proposal has other issues.  First, it is likely that $t_{\mrm{mem}} \sim O(1)$ due to the existence of point-like excitations introduced by the punctures, as discussed in Ref.~\cite{lin2024proposals}. It is also likely that there is \emph{no} phase transition at non-zero temperature\footnote{The statistical mechanics duality argument in Ref.~\cite{brell2016proposal} does not correctly account for the extensive number of metachecks whose weight scales with the system size.  As a result, the dual model is ordered at all temperatures, corresponding to a zero-temperature transition in the original model.}.
Thus, the construction is \emph{not} self-correcting, contrary to the claim made in the paper. Additionally, the large number of small-distance logical operators makes the perturbative stability of the proposal unclear.

In this paper, we provide a three-dimensional construction that we prove is self-correcting.   Our code is CSS stabilizer and is obtained by alternating between applying two procedures. One procedure increases the energy cost of excitations associated with $X$-type errors, while the other does the same for excitations associated with $Z$-type errors.
Our main result combines these procedures with a random embedding (inspired by Refs.~\cite{portnoy2023local,gromov2012generalizations}) which allows us to establish (1) locality in $\mathbb{R}^3$, (2) preservation of a single logical qubit, and (3) an exponential lower bound on the memory lifetime at sufficiently low temperature.  We also present a version with an explicit embedding later in the paper. Three additional results include a more intuitive version of random embedding construction that we provide in the introduction, a construction based on a hypergraph product that we provide in Appendix~\ref{appB:product-construction}, and an analogous classical code that we construct in Sec.~\ref{sec:explicit-embedding}.

\subsection{Main results}

We now state our main results. Our construction produces a family of CSS codes from a family of cochain complexes.  We begin with a finite ``seed'' code $C_{(0)}$ encoding a single logical qubit, and define the code $C_{(k)}$ iteratively through $C_{(k-1)}$.
Let $H_{k}$ be the stabilizer Hamiltonian associated with $C_{(k)}$, acting on $n_{k}$ qubits.  For every $k$, the code $C_{(k)}$ encodes a single logical qubit; equivalently, the ground-state degeneracy of $H_{k}$ is two.

To model thermalization, we consider a quantum many-body system with local Hamiltonian $H$ coupled weakly to a bath, so that its time evolution is described by the Markovian master equation~\cite{alicki2009thermalization,alicki2010thermal,chesi2010thermodynamic}
\begin{equation}
    \dot{\rho}(t) = -i[H,\rho(t)] + \cL(\rho(t)).
\end{equation}
The stationary state is the Gibbs state
$\rho_{\beta} = e^{-\beta H}/\Tr(e^{-\beta H})$
where $\beta = 1/T$ is the inverse temperature.  We assume that the Lindbladian $\cL$ is generated by local jump operators of $O(1)$ support, although this condition can be relaxed.

We consider a family of Hamiltonians $H_{k}$ associated to codes $C_{(k)}$, and prove the following.

\begin{theorem}[Quantum memory lifetime; informal]
There exists a family of three-dimensional local Hamiltonians $\{H_k: k \in \mathbb{Z}_{\ge 0}\}$ associated with CSS stabilizer codes $\{C_{(k)}\}$ on $n_k =  \exp(\Theta(k))$ qubits
along with encoding and decoding maps $\mrm{Enc}_{k}$ and $\mrm{Dec}_{k}$
such that for every single-qubit state $\rho_L$
\begin{equation}
  \mrm{Dec}_{k} \circ \mrm{Enc}_{k}(\rho_L) = \rho_L.
\end{equation}
Moreover, there exist constants $T_c > 0$ and $\eta > 0$ independent of $k$, such that for all temperatures $T < T_c$:
\begin{enumerate}
    \item[(1)] After time $t$, the state of the system satisfies
    \begin{equation}
        \norm{\Phi_{1/T, t} \circ \mrm{Enc}_k(\rho_L) - \mrm{Enc}_k(\rho_L)}_1 \leq t  \exp(- \Theta(n_k^{\eta}))
    \end{equation}
    where $\Phi_{1/T, t}$ is the evolution generated by $\dot{\rho}(t) = -i[H, \rho] + \cL(\rho)$, and
    \item[(2)] The memory lifetime obeys
    \begin{equation}
        t_{\mathrm{mem}} \geq \exp( \Theta( n_{k}^\eta))
    \end{equation}
\end{enumerate}
\end{theorem}

\noindent In this Theorem, we used the definition of the memory lifetime
\begin{equation}
t_{\mrm{mem}} = \min\{t: \norm{\Phi_{1/T, t} \circ \mrm{Enc}_k(\rho_L) - \mrm{Enc}_k(\rho_L)}_1 \leq 1/10\}
\end{equation}
which characterizes the timescale when the encoded state is forgotten; thus, statement (2) follows trivially from statement (1).  The subscript `$L$' denotes the logical degree of freedom. The encoding that we choose is a thermal encoding defined in Refs.~\cite {alicki2010thermal,chesi2010thermodynamic}. Namely, the Hilbert space of $H_{2k}$ can be decomposed as $\cH = \cH_{L} \otimes \cH_{\mrm{else}}$ with the Hamiltonian $H_{2k} = I_L \otimes H_{2k, \mrm{else}}'$.  Under this decomposition, the thermal encoding map is defined to be $$\mrm{Enc}_k(\rho_L) = \rho_L \otimes \frac{e^{-\beta H_{2k, \mrm{else}}'}}{\Tr(e^{-\beta H_{2k, \mrm{else}}'})}.$$
We also separately prove that these Hamiltonians are local in $\mathbb{R}^3$.

\begin{theorem}[Embedding in $\mathbb{R}^3$; informal]
For every $k \geq 0$, the Hamiltonian $H_{k}$
  admits an embedding of its qubits into $\mathbb{R}^3$ such that:
\begin{enumerate}
\item (Local) The support of each Hamiltonian term is contained in a ball of radius $O(1)$.
\item (Bounded density) Every unit ball in $\mathbb{R}^3$ contains $O(1)$ qubits.
\end{enumerate}
\end{theorem}

Additionally, we believe that the spectral gap of our construction is stable under arbitrary local Hamiltonian perturbations.  This requires proving that our construction obeys TQO-1 and TQO-2, which are defined in Ref.~\cite{Bravyi_2010}.  While TQO-1 follows immediately from our code having distance scaling as a power of $n_k$, TQO-2 requires a separate proof, which will be provided in an updated version of this paper.

\subsection{Overview of the construction}

We now discuss our construction.  We begin with a construction that is not treated in the main body of the paper, but that is conceptually intuitive.  We then discuss the main construction for which we rigorously prove self-correction, which can be viewed as a simplification of the intuitive construction.   Finally, we describe a non-random variant of the construction.

A useful heuristic for self-correction is to find a construction where the size of the error syndrome grows with the size of the error.  For $X$-type errors, which correspond to 1-cochains in the associated cochain complex, this implies the following coboundary expansion property:
 \begin{equation} \label{eq:coboundary-expansion}
    \forall\, e \in C^1_{(k)}: \ \ \mrm{syn}(e) =
    |\delta e| \geq \alpha\min_{b \in B^1} |e + b|^\gamma
  \end{equation}
for constants $\alpha, \gamma > 0$.  For $Z$-type errors, which correspond to 1-chains in the dual chain complex, this implies an analogous boundary expansion property.  All of our constructions will satisfy these bounds for suitable choices of $\alpha$ and $\gamma$.

\subsubsection{Intuitive 4-step construction with random embedding}
\label{subsec:4-step}

We provide a variant of our main construction that has similar features but is more intuitive.

As a seed, we take a constant-sized surface code encoding a single logical qubit, specified by a cochain complex $C_{(0)}: C_{(0)}^0 \rightarrow C_{(0)}^1  \rightarrow C_{(0)}^2$. We then construct the code iteratively by alternating  between two procedures, denoted $\cR_X$ and $\cR_Z$:
\begin{equation}
    C_{(2k)} = (\cR_Z \circ \cR_X)^{\circ k} (C_{(0)})
\end{equation}
Operationally, $\cR_X$ and $\cR_Z$ are very similar. The main difference is that $\cR_X$ increases the syndrome cost of $X$-type errors, while $\cR_Z$ does the same for $Z$-type errors. Each procedure increases the code size by a constant factor, so that after $2k$ steps the number of physical qubits is $n_{2k} = \exp(\Theta(k))$.

There is a natural map from the Tanner graph of $C_{(i)}$ to a square complex, which we call the Tanner square complex $T_{(i)}$. Thus, at level $i$ we can specify the construction by the triple
$$(C_{(i)},T_{(i)},I_{(i)})$$
where
$I_{(i)}:T_{(i)}\to \mathbb{R}^3$
is a local bounded-density embedding. More precisely, if $t$ is the maximum number of squares intersecting any unit ball in $\mathbb{R}^3$, and $\ell$ is an upper bound on the embedded edge lengths, then we say that $I_{(i)}$ is a $(t,\ell)$-embedding.

The idea of alternating between increasing the syndrome cost of $Z$-operators and increasing the syndrome cost of $X$-operators is a common theme for all the constructions we provide in this paper. For the construction in this subsection, an iteration consists of 4 steps: perturbation, subdivision, thickening, and degree reduction.

Suppose that at level $i$ we are given a triple $(C_{(i)},T_{(i)},I_{(i)})$. One iteration produces the next triple
$\left ( C_{(i+1)},T_{(i+1)},I_{(i+1)} \right)$.
We make an induction hypothesis that the square complex $T_{(i)}$ has maximum edge and vertex degrees bounded by sufficiently large universal constants $z_e$ and $z_v$, and that $I_{(i)}$ is a $(t,1)$-embedding for some sufficiently large constant $t$ independent of $i$.  We then perform the following 4 steps:
\begin{enumerate}

\item \textbf{Perturbation:} Scale the current embedding map $I_{(i)}$ by a constant $\ell > 1$, then randomly perturb the images of vertices in $T_{(i)}$  by distance $\asymp \ell$ according to a procedure similar to that in Refs.~\cite{portnoy2023local} and \cite{gromov2012generalizations}. Such a procedure relies on the Lov\'asz local lemma to show the existence of a perturbation that reduces the density to a sufficiently small $t' < t$.
This results in a  $(t',\ell)$-embedding map without altering the structure of $T_{(i)}$, see Fig.~\ref{fig:perturbation} for an illustration.

\begin{figure}[h]
    \centering
\includegraphics[width=0.95\textwidth]{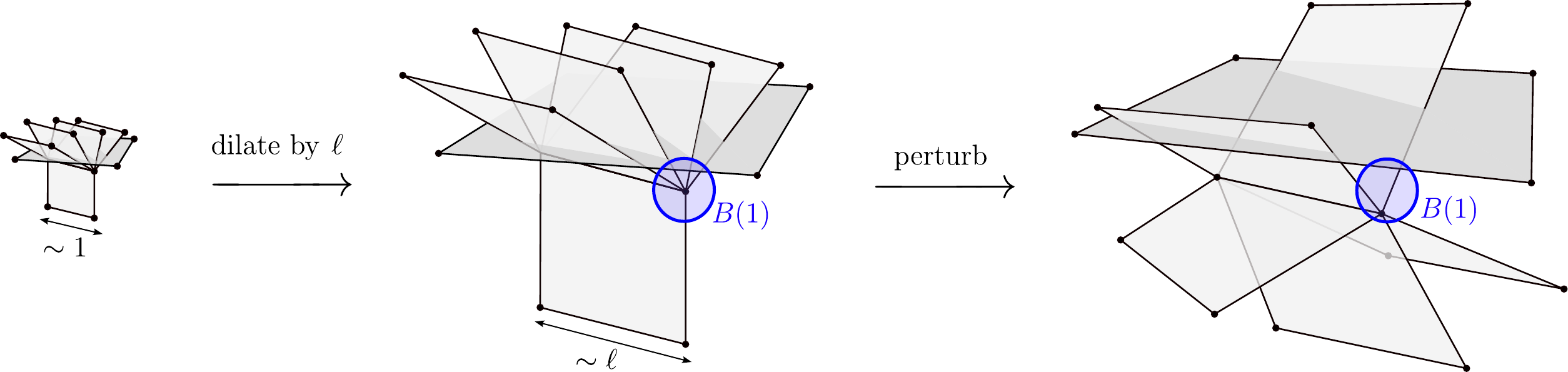}
    \caption{An illustration of the perturbation step and its effect on the embedding density.}
    \label{fig:perturbation}
\end{figure}

\item \textbf{Subdivision:} Subdivide each square of $T_{(i)}$ into $\asymp \ell \times \ell$ squares, obtaining the subdivided code $C_{(i)}'$ and its associated Tanner square complex $T_{(i)}'$.  This restores the unit-scale geometry, so that the resulting complex is $(t'', 1)$-embedded in $\mathbb{R}^3$ for some $t' <t''< t$. One can think of the code at this stage as a defect network of $\asymp \ell \times \ell$-sized surface codes.

\item \textbf{Thickening:} We now ``thicken'' the entire code; locally, this means that patches of the surface code turn into thickness-1 slabs of 3D surface code, and the newly introduced boundary surfaces are membrane-condensing. In the $\cR_X$ iteration, the membranes are $X$-type, and in the $\cR_Z$ iteration, the membranes are $Z$-type. This increases the energy cost of either $X$-type or $Z$-type excitations, depending on the iteration.

More formally, in the $\cR_X$ iteration this step corresponds to taking a tensor product $C'_{(i)} \otimes \comB$ where $\comB: \comB^0 \rightarrow \comB^1$ is the cochain complex of a 2-bit repetition code.
We then only keep the terms in the product complex of degree 2 and below.
In the $\cR_Z$ iteration this step corresponds to the tensor product $C'_{(i)} \otimes \wt \comB$ where $ \wt \comB: \wt  \comB^{-1} \rightarrow \wt  \comB^0$ is the dual cochain complex of a 2-bit repetition code.
We then only keep the terms in the product complex of degree 0 and above.

After thickening, the code is $(t, 1)$-embedded in $\mathbb{R}^3$ where $t > t''$. However, both the maximum edge and maximum vertex degrees in the Tanner graph have increased.

\item \textbf{Degree reduction}:   If we were to repeat the first three steps many times, the edge and vertex degrees would eventually become unbounded. We claim that there exists a degree reduction procedure that restores the maximum edge and vertex degrees to the universal constants $z_e$ and $z_v$ without affecting properties of the code. More precisely, for the cochain complex $C_{(i)}''$ obtained after thickening, there exists a cochain-homotopic complex $C_{(i+1)}$ whose degrees are reduced\footnote{A similar effect can be achieved by starting from the code $C_{(i)}''$ and constructing the new complex code $C_{(i+1)}$ by removing some of the redundant checks and some of the bits. One can check that there exists such a procedure that reduces the degrees while preserving the number of logical qubits. The resulting complex is not chain-homotopic to the original one, but all essential properties remain unchanged. }; an example is illustrated below:
\begin{figure}[h]
    \centering
    \includegraphics[width=1\textwidth]{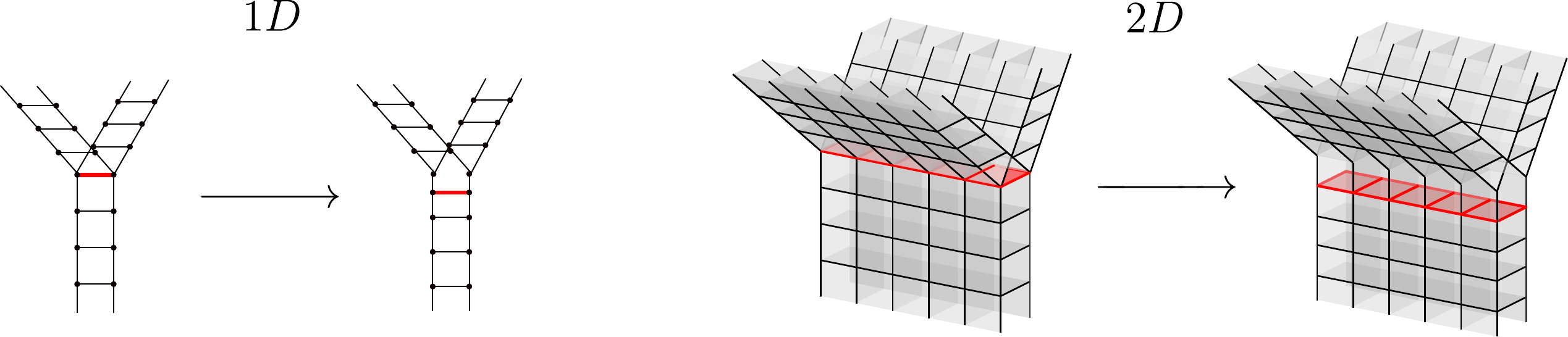}
    \caption{
    Local effect of the degree reduction step. On the left, we show a toy example of degree reduction for a local neighborhood in a $1$D complex. On the right, we show an analogous procedure for a $2$D square complex. In both cases, the thickening step increases the degrees near the junction where several subcomplexes meet. The cells responsible for this increase are shown in red.}
    \label{fig:degree-reduction}
\end{figure}
\end{enumerate}

We claim that the embedding map $I_{(i+1)}$ at the end is a $(t,1)$-embedding and that maximum edge and vertex degrees are $z_e$ and $z_v$, which completes the induction step. We also claim that each iteration preserves the total number of logical qubits. The self-correcting property arises due to the thickening step: each application of $\cR_X$ or $\cR_Z$ multiplies the syndrome size of one error type by 2.
Since the two operations are alternated, this occurs on all scales and for both Pauli error types.

\subsubsection{Main construction with random embedding} \label{sec:intro-main-construction}

The idea behind the main construction in this paper is similar to that of the 4-step one, but is structurally simpler. The iterations $\cR_X$ and $\cR_Z$ only consist of 2 steps: perturbation and replacement.
Suppose that at level $i$ we are given a triple $(C_{(i)},T_{(i)},I_{(i)})$.
We make the same induction hypothesis as before.

\begin{enumerate}
\item \textbf{Perturbation:} Performed in the same way as above. This produces a new $(t',\ell)$-embedding with sufficiently small $t' < t$.
\item \textbf{Replacement:}
Replace each square $s \in T_{(i)}(2)$ with a copy of a square complex $U_X$ in the $X$-iteration and with a copy of $U_Z$ in the $Z$-iteration, with the boundary subcomplexes of the adjacent copies identified:
$$T_{(i+1)} = \left ( \bigsqcup_{s \in T_{(i)}(2)} U_{X,s} \right )/ \sim,$$
The complexes $U_X$ ($U_Z$) combine and simplify the net effect of subdivision, thickening, and degree reduction. Both consist of two parallel  $\asymp \ell\times \ell$ layers joined by a single transverse strip of thickness $2$ parallel to the thickened direction, as shown in \eqref{fig:replacement}. This already suffices to double the syndrome of either $Z$-type or $X$-type Pauli errors, depending on the iteration.
For our choice of $U_X$ and $U_Z$, the vertex and edge degrees in $T_{(i+1)}$ remain uniformly bounded. By an argument similar to that described in the 4-step construction, one can obtain a $(t, 1)$-embedding map $I_{(i+1)}$.
\end{enumerate}

\subsubsection{Memory lifetime proof}

To prove an exponential memory lifetime at low temperatures, we apply a more intricate version of a Peierls argument. For this, we need a decoder $\cD$ that, given a Pauli error $e$, applies a correction that depends only on the syndrome $\sigma = \delta e$. We say that $\sigma$ is \emph{unstable} if there exists a single-qubit error $e_1$ such that $\sigma$ and $\sigma' = \sigma + \delta e_1$ decode to different logical outcomes under $\cD$.  Let $\Sigma(\beta)$ be the Gibbs distribution on syndrome configurations at inverse temperature $\beta$ with the energy function $E(\sigma) = |\sigma|$.  By  Theorem 1 from Ref.~\cite{chesi2010thermodynamic}, if the decoder $\cD$ has the property that $$\mathbb{P}_{\sigma \sim \Sigma(\beta)}(\sigma \text{ unstable under }\cD) \leq \exp(-n^{\eta}),$$ then the memory lifetime is $t_{\mrm{mem}} \gtrsim \exp(n^{\eta})$.

We choose $\cD$ to be a renormalization group (RG)-like decoder that acts by performing a correction at each scale, from the smallest scale to the largest, while applying a ``coarse-graining'' operation between each correction.  Because of the simple structure of the local subcomplexes $U_{X}$ ($U_{Z}$), the decoder can be defined explicitly.  Under this decoder, one can show that syndrome $\sigma$ is unstable if any syndrome remains at the largest scale.

Any unstable configuration in the $2k$-level construction must contain a connected error cluster with syndrome weight $\asymp 2^{k}$. This follows because, roughly speaking, the $X$-syndrome weight doubles during the $\cR_X$ iterations, while the $Z$-syndrome weight doubles during the $\cR_Z$ iterations.  On the other hand, the number of connected clusters with syndrome weight $m$ is at most $\asymp B^m$ for some constant $B >1$, by reduction to counting connected size-$m$ subgraphs in a bounded-degree ``decoding graph''.
Combining these two properties with a standard Peierls argument shows that, setting $k \asymp \log n$, the probability of an unstable syndrome configuration is $\leq \exp(-n^{\eta})$ when $T < T_c$.  The current bound (which we believe is tight up to constants) is $T_c \gtrsim 1/\log \ell$, where $\ell$ is the scale parameter of the perturbation step.  This results in a relatively large value for $T_c$ even with poor bounds on $\ell$.

\subsubsection{Construction with explicit embedding}

We now describe the third construction,
  which comes with an explicit embedding.
The motivation for this construction is to make the perturbation step deterministic, thus providing exact locations of all the local configurations in the code.  By choosing a deterministic perturbation, we can also optimize how the code is packed in $\mathbb{R}^3$,
  providing much smaller bounds on the scale factor $\ell$ compared to random construction.  This results in larger bounds on the temperature for thermal stability and a parametrically longer memory lifetime.

As before,
  the construction consists of two types of iterations,
  $\cR_X$ and $\cR_Z$.
Each iteration consists of two steps,
  which we call \emph{refinement} and \emph{doubling}.
These steps play roles analogous to perturbation and replacement
  in the previous construction:
  the refinement step modifies the geometry without changing the underlying
    topological structure,
  while the doubling step implements a 2-bit repetition-code structure to increase the syndrome cost of errors.

In the explicit construction, we will manipulate a geometric complex $X$,
  which plays the same role as the Tanner square complex $T$
  in the previous constructions.  We arrange the embedding so that all faces in $X$ are parallel to coordinate planes.
\begin{enumerate}
  \item \textbf{Refinement:} $X \mapsto X'$.
        We scale and subdivide the complex,
          and then apply a deterministic local perturbation.
        This step preserves the underlying homotopy:
          the refined complex is homotopy equivalent to the original one,
          and the refined chain complex is chain homotopy equivalent to the original one.
        In particular, the resulting code stores the same number of logical qubits.
  \item \textbf{Replacement:} $X' \mapsto Y = (X' \sqcup X'_\mrm{copy} \sqcup X_{\mrm{cyl}})/{\sim}$.
        We form two copies of $X'$.
        One copy remains in place,
          while the other is translated by the vector $(1,1,1)$.
        The remaining complex $X_{\mrm{cyl}}$ connects the two copies.
        Conceptually,
          the two copies duplicate the encoded information,
          while $X_{\mrm{cyl}}$ mediates between them and enforces their agreement.
\end{enumerate}

An example of this process is illustrated below.
\begin{figure}[H]
  \centering
  \vspace*{-3em}
  \includegraphics[width=0.9\textwidth]{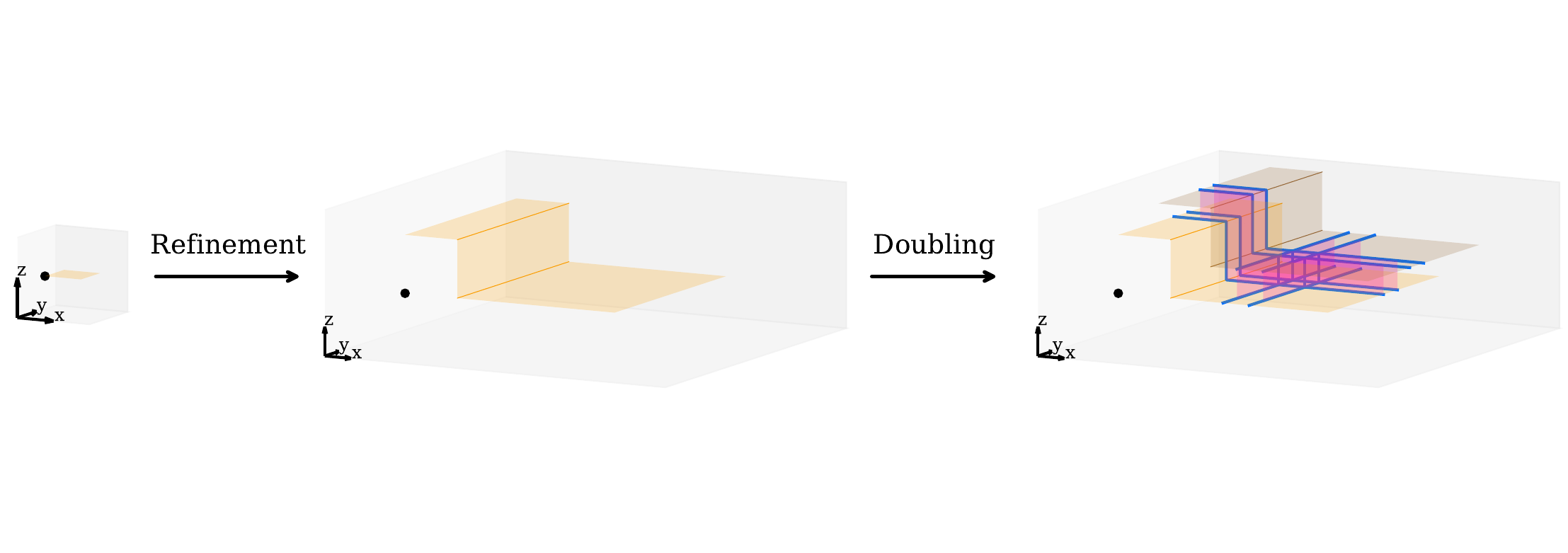}
  \vspace*{-3em}
\end{figure}

The construction is specified by listing all possible local structures
  that can appear at each stage,
  together with local rules describing how these structures transform
    under refinement and doubling.
In total,
  there are four refinement gadgets and three doubling gadgets.
The proof of locality follows by showing that these rules form a closed system:
  every allowed local structure can only be transformed into a configuration of other allowed local structures, with no new unaccounted configurations produced.  We leave the complete proof of the exponential memory lifetime for this version of the construction to future work.  However, because the explicit construction behaves similarly to the random one, we believe that the proof should not deviate much from the proof for the random construction.

\subsection{Outlook}

\begin{enumerate}[itemsep=-0.1em, topsep=0.05em]
\item \textbf{Other constructions}:  We hope that the existence of a 3D self-correcting quantum memory will motivate alternative constructions.    Since there likely exists no Pauli stabilizer code in 3D that is self-correcting and translation invariant~\cite{haah2013commuting}, one could explore translation-invariant codes that are not stabilizer codes or that violate the TQO axioms.  Other directions may involve utilizing exotic kinds of domain walls such as the magic domain walls from Ref.~\cite{li2024domain}.
\item \textbf{Initialization}: Our proof uses a thermal encoding $\mrm{Enc}(\rho_L) = \rho_L \otimes \rho_\beta$, but we do not address efficient passive initialization of this state.  A natural strategy is to start with an easily preparable state, such as  $\ketbra{0}{0}_L \otimes I_X \otimes \ketbra{\varnothing}{\varnothing}_Z $ (which can be prepared from an all-0 product state by measuring $X$ stabilizers), and rely on thermalization to mix the state within the logical sector\footnote{A simpler question is whether rapid initialization is possible from the state $\ketbra{0}{0} \otimes \ketbra{\varnothing}{\varnothing}_X \otimes \ketbra{\varnothing}{\varnothing}_Z $ as in Ref.~\cite{bergamaschi2026rapid}.  However, for their results to hold, the metacheck matrix must be LDPC, which is not the case in our code.  Thus some generalization of their results will be needed.}. The $X$-syndrome sector must heat up, while the $Z$-syndrome sector will cool down from infinite temperature.
This is believed to be efficient in the 4D toric code, although not proven rigorously. For our construction, energetic bottlenecks might be an obstruction for rapid mixing to $\ketbra{0}{0} \otimes \rho_\beta$.
\item \textbf{Thermally stable quantum computer}: A natural question one can ask is whether a fault tolerant and passive quantum computer exists in three dimensions.  The main bottleneck in this approach is an efficient implementation of a fault-tolerant non-Clifford gate via code-switching while coupled to a thermal bath.
\item \textbf{``Circuit to Gibbs state'' mapping}:
The Feynman-Kitaev method maps polynomial-depth quantum circuits to Hamiltonians whose ground state encodes the computation history.  However, states at low energies fail to encode the computation history.  Finding circuit-to-Hamiltonian mappings where states over a larger energy window can encode the computation history may provide insight into the quantum PCP conjecture (see Ref.~\cite{Anshu_2024} for recent ideas).  A simpler but still interesting question is whether nontrivial \emph{distributions} over eigenstates of the Hamiltonian can still encode the computation history.  One natural distribution would come from the Gibbs state of a local Hamiltonian at non-zero temperature.  Ideas from self-correction may help answer this question.

\item \textbf{Relation to phases of matter}:  Since our result produces a scaling family of Hamiltonians with a well-defined thermodynamic limit, one could ask how it fits within the current landscape of topological phases of matter.  In particular, the excitations in our code are deformable, and we expect that braiding and fusion rules can be appropriately defined. Can one prove that our Hamiltonian is distinct from any known translation-invariant phase of matter?
\item \textbf{Optimal parameters}:
It would be interesting to understand the optimal parameters achievable by a 3D self-correcting quantum memory.
The key exponent governing the memory lifetime is $\eta$,
  which is closely related to the exponent $\gamma$ in the coboundary-expansion bound, \cref{eq:coboundary-expansion}.
Our work establishes nontrivial lower bounds on the optimal value of $\eta$,
  but finding matching upper and lower bounds is open.
\end{enumerate}

\section{Preliminaries}

\subsection{Quantum CSS  and cochain complexes}
\label{prelim:codes}

We briefly review the connection between quantum CSS codes and cochain complexes.
A CSS code on $n$ qubits is specified by two classical codes, $C_x= \ker(H_x)$ and $C_z = \ker(H_z)$.  Here, $H_x$ and $H_z$ denote the parity-check matrices
\begin{equation}
H_x:\mathbb{F}_2^n \to \mathbb{F}_2^{m_x},
\qquad
H_z:\mathbb{F}_2^n \to \mathbb{F}_2^{m_z},
\end{equation}
which satisfy $H_z H_x^T = 0$. The $X$ and $Z$-type Pauli operators that commute with all stabilizers correspond to $C_z$ and $C_x$, respectively. The $X$-type stabilizers correspond to $C_x^\perp = \operatorname{Im}(H_x^T)$, and the $Z$-type stabilizers correspond to $C_z^\perp = \operatorname{Im}(H_z^T)$.
The number of encoded qubits is $k = \dim C_x - \dim C_z^\perp = \dim C_z - \dim C_x^\perp$.
The $X$- and $Z$-distances are
\begin{equation}
d_x = \min_{c \in C_z \setminus C_x^\perp} |c|,
\qquad
d_z = \min_{c' \in C_x \setminus C_z^\perp} |c'|,
\end{equation}
and the code distance is $d = \min(d_x,d_z)$.
We call a quantum code a low-density parity-check (LDPC) if each check acts on $O(1)$ qubits, and each qubit participates in $O(1)$ checks.

A CSS code naturally gives rise to the cochain complex $C^\bullet$:
\begin{equation}
  \begin{tikzcd}
    \mathbb{F}_2^{m_x}
      \arrow[r,"\delta^0 = H_x^T"]
    & \mathbb{F}_2^n
      \arrow[r,"\delta^1 = H_z"]
    & \mathbb{F}_2^{m_z}.
  \end{tikzcd}
\end{equation}
with $\delta^1 \delta^0 = 0$.
We will frequently write this as
\begin{equation}
    C^\bullet: C^0 \overset{\delta^0}{\longrightarrow}C^1 \overset{\delta^1}{\longrightarrow} C^2.
\end{equation}
The $X$-type operators commuting with all $Z$-checks are $1$-cocycles, $Z^1(C^\bullet) := \ker(\delta^1)$,
while the $X$-type stabilizers are $1$-coboundaries, $B^1(C^\bullet) := \operatorname{Im}(\delta^0)$. For an $X$-type error $e_X \in C^1$, its syndrome is $\delta^1 e_X \in B^2(C^\bullet) := \operatorname{Im}(\delta^1)$.
Equivalence classes of distinct $X$-type logical operators are labeled by the first cohomology group
\begin{equation}
H^1(C^\bullet)=Z^1/B^1=\ker(H_z)/\operatorname{Im}(H_x^T)=C_z/C_x^\perp.
\end{equation}
whose dimension is precisely the number of encoded qubits.

Let $C_\bullet := \operatorname{Hom}(C^\bullet, \ff_2)$ be the dual chain complex, which we can also write as
\begin{equation}
    C_\bullet: \ C_2 \overset{(\delta^1)^T}{\longrightarrow} C_1  \overset{(\delta^0)^T}{\longrightarrow} C_0
\end{equation}
with $C_i \cong C^i$. The $Z$-type operators that commute with all $X$-checks are 1-cycles, i.e. $Z_1(C_\bullet) := \ker(\delta^0)^T$, while $Z$-type stabilizers are 1-boundaries, $B_1(C_\bullet) := \operatorname{Im}(\delta^1)^T$. For a $Z$-type error $e_Z \in C_1$, its syndrome is ${\delta^0}^T e_Z \in B_0(C_\bullet) := \operatorname{Im}(\delta^0)^T$.
Equivalence classes of distinct $Z$-type logical operators are labeled by the first homology group
\begin{equation}
H_1(C_\bullet)=Z_1/B_1=\ker(H_x)/\operatorname{Im}(H_z^T)=C_x/C_z^\perp.
\end{equation}

\subsection{Tanner graphs and Tanner square complexes} \label{sec:Tanner-prelims}

A convenient way to represent classical and quantum codes is via their Tanner graphs.
The Tanner graph  $G = (V,E)$ of a classical code is bipartite, with $V = V_b \cup V_c$, where $V_b$ contains one vertex per bit and $V_c$ one vertex per check. An edge is placed between $v_b \in V_b$ and $v_c \in V_c$ if the bit corresponding to $v_b$ participates in the check corresponding to $v_c$.

The Tanner graph  $G = (V,E)$ of a quantum CSS code is tripartite, with $V = V_X \cup V_Q \cup V_Z$, where $V_Q$ contains one vertex per qubit, $V_X$ contains one vertex per $X$-check, and $V_Z$ contains one vertex per $Z$-check. A vertex in $V_Q$ is connected to a vertex in $V_X$ or $V_Z$ if the corresponding qubit participates in the corresponding $X$- or $Z$-check.

It was observed in Ref.~\cite{li2024transform} that the commutativity condition of a quantum CSS code, namely $H_z H_x^T = 0$, can be naturally captured by associating a $2$-dimensional square complex to the Tanner graph of the code, with the Tanner graph as its $1$-skeleton.
While a Tanner square complex determines a unique quantum code, the converse does not hold: a given code may correspond to multiple square complexes.  However in our construction, there is a natural choice of square complex that we will use.

We review how to construct such a square complex.
Given a Tanner graph $G = (V,E)$ of a quantum CSS code,
a square complex $T$ associated with this Tanner graph shares the same vertex set $V = V_X \cup V_Q \cup V_Z$ and edge set $E$, but comes with a set of faces $F$ that we will now describe. Fix $x \in V_X$ and $z \in V_Z$, and let $Q_{xz} \subseteq V_Q$ be the set of qubits shared by these checks. Note that $|Q_{xz}|$ must be even, since the corresponding $X$- and $Z$-checks commute. Thus, we may choose a collection of pairs $\{q_i,q_j\}$ with $q_i,q_j \in Q_{xz}$ such that every shared qubit participates in at least one pair. For each such pair $\{q_i,q_j\}$,  we assign a face $f \in F$ whose boundary is the cycle $(x,q_i,z,q_j)$. Repeating this for all $x \in V_X$ and $z \in V_Z$ produces a square complex $T$ whose $1$-skeleton is the Tanner graph of the CSS code. The parity-check matrices $H_x$ and $H_z$ are recovered from the incidence relations between $V_X$ and $V_Q$, and between $V_Z$ and $V_Q$, respectively. The square faces encode the commutativity condition $H_z H_x^T = 0$.

\subsection{Quantum code embedding framework}
\label{sec:quantum-code-embedding-prelims}

In our iterative construction, the relation between the code properties at levels $i$ and $i+1$ is most transparent in the formalism of quantum code embedding\footnote{Here ``embedding'' means that one code is realized inside the algebraic structure associated with a different code. This terminology should not be confused with geometric embedding, which is a different concept that we use to realize our construction locally in $\mathbb{R}^3$.  For certain codes like layer codes~\cite{williamson2023layer}, the code embedding formalism also provides a natural embedding in Euclidean space, but this is not the case for our construction.}, introduced in Ref.~\cite{yuan2026unified}. In Ref.~\cite{yuan2026unified}, this framework was formulated in terms of iterated mapping cones. Here, we reformulate it using spectral sequences, which more readily allows us to argue chain-homotopy equivalence between different steps of our construction.
We also make notational deviations from Ref.~\cite{yuan2026unified} to match the other conventions in this paper.

The setting of quantum code embedding is that one has a global code $C$ composed of smaller local codes $C_v$ indexed by $v \in V$, which are connected via some maps $g$ and $s$.  The local codes are typically simple and have well-understood cohomological properties, which allows one to compute the cohomology of the larger code $C$ in terms of the cohomology of the local codes and the connecting maps. The quantum code embedding framework provides a systematic way to do this analysis.

Given a cochain complex
$C\!:  \ C^0 \xrightarrow{\delta^0} C^1 \xrightarrow{\delta^1} C^2$
constructed from a direct sum of local cochain complexes
\begin{equation}
    C_x: \  C_x^0 \to C_x^1 \to C_x^2,
    \quad C_q: \  C_q^0 \to C_q^1 \to C_q^2,
    \quad C_z: \  C_z^0 \to C_z^1 \to C_z^2
\end{equation}
indexed by $x \in V_X$, $q \in V_Q$, $z \in V_Z$.
We will use uppercase letters $X, Q, Z$ to denote the direct sums of the local complexes, and lowercase letters $x, q, z$ to denote the individual local complexes and maps:
\begin{equation}
  C^\bullet_X = \bigoplus_{x \in V_X} C^\bullet_x,\quad
  C^\bullet_Q = \bigoplus_{q \in V_Q} C^\bullet_q,\quad
  C^\bullet_Z = \bigoplus_{z \in V_Z} C^\bullet_z
\end{equation}
In this notation, the complex $C$ has the form
\begin{equation} \label{eq:embedded-complex-1}
    C: \ \ C_X^0 \oplus  C_Q^0 \oplus  C_Z^0
     \overset{\delta^0}{\longrightarrow} C_X^1 \oplus  C_Q^1 \oplus  C_Z^1
      \overset{\delta^1}{\longrightarrow}  C_X^2 \oplus  C_Q^2 \oplus  C_Z^2.
\end{equation}
The coboundary maps $\delta^0$ and $\delta^1$ of the global code are constructed from coboundary maps of the local complexes as well as the connecting maps.
For $m=0,1$, the local coboundary maps are
\begin{equation} \label{eq:local-boundary-maps}
  \delta_x^m\colon C_x^m \to C_x^{m+1},\quad
  \delta_q^m\colon C_q^m \to C_q^{m+1},\quad
  \delta_z^m\colon C_z^m \to C_z^{m+1},
\end{equation}
and the connecting maps are
\begin{equation}
\label     {eq:connecting-maps}
  g_{xq}^m\colon C_x^m \to C_q^{m+1},\quad
  g_{qz}^m\colon C_q^m \to C_z^{m+1},\quad
  g_{xz}^m\colon C_x^m \to C_z^{m+1}.
\end{equation}
Equivalently, the full structure of the cochain complex $C$ can be visualized as the following diagram:
\begin{equation}
\label{eq:height-2-diagram}
\begin{tikzpicture}[baseline]
\matrix(a)[matrix of math nodes, nodes in empty cells, nodes={minimum size=25pt},
row sep=2em, column sep=2em,
text height=1.25ex, text depth=0.25ex]
{&& \blue{C_X^0}  & \blue{C_X^1} & \blue{C_X^2}\\
& {C_Q^0}  & {C_Q^1}  & {C_Q^2} &\\
\red{C_Z^0} & \red{C_Z^1} & \red{C_Z^2} &&\\};
\path[->,blue,font=\scriptsize]
(a-1-3) edge node[above]{$\delta_X^0$} (a-1-4)
(a-1-4) edge node[above]{$\delta_X^1$} (a-1-5);
\path[->,font=\scriptsize]
(a-2-2) edge node[above]{$\delta_Q^0$} (a-2-3)
(a-2-3) edge node[above]{$\delta_Q^1$} (a-2-4);
\path[->,red,font=\scriptsize]
(a-3-1) edge node[above]{$\delta_Z^0$} (a-3-2)
(a-3-2) edge node[above]{$\delta_Z^1$} (a-3-3);
\path[->,font=\scriptsize]
(a-1-3) edge node[right]{$g_{XQ}^0$} (a-2-3)
(a-1-4) edge node[right]{$g_{XQ}^1$} (a-2-4)
(a-2-2) edge node[right]{$g_{QZ}^0$} (a-3-2)
(a-2-3) edge node[right]{$g_{QZ}^1$} (a-3-3);
\path[->,font=\scriptsize]
(a-1-3) edge[bend right=70, dashed] node[left]{$g_{XZ}^0$} (a-3-2)
(a-1-4) edge[bend left=70, dashed] node[right]{$g_{XZ}^1$} (a-3-3);
\end{tikzpicture}
\end{equation}
Where $C$ is the total complex of the structure shown in the diagram. We also used the uppercase notation $X, Q, Z$ to denote the direct sums of the connecting maps, namely:
\begin{equation} \label{eq:gluing-local-global}
  g^{\bullet}_{XQ} := \bigoplus_{x \in V_X} \bigoplus_{q \in V_Q} g^{\bullet}_{xq},\quad
  g^{\bullet}_{QZ} := \bigoplus_{q \in V_Q} \bigoplus_{z \in V_Z} g^{\bullet}_{qz},\quad
  g^{\bullet}_{XZ} := \bigoplus_{x \in V_X} \bigoplus_{z \in V_Z} g^{\bullet}_{xz}.
\end{equation}
Note that we assume there are no direct maps from $C_Z$ to $C_Q$ or $C_X$, nor from $C_Q$ to $C_X$.

Because $\delta^1 \circ \delta^0 = 0$, the maps $g_{xq}, g_{qz}, g_{xz}$ must satisfy compatibility conditions, for example $\delta_q^1 \circ g_{xq}^0 = g_{xq}^1 \circ \delta_x^0$ for any $x \in V_X$, $q \in V_Q$.  In our construction, the structure is even simpler because there is no direct map from $C_x$ to $C_z$, i.e. $g^{\bullet}_{xz} = 0$.

The distinction between horizontal and vertical maps in the structure shown in Eqn.~\ref{eq:height-2-diagram} will play an important role throughout the paper. Therefore, we split the coboundary map of the full cochain complex $C$ as $\delta = \delta_{\mrm{loc}} + g$, where the local part $\delta_{\mrm{loc}}$ consists of maps $\delta_X$, $\delta_Q$, and $\delta_Z$ acting within individual local complexes, and $g$ are the connecting maps. In other words, the full coboundary map
\begin{equation}
  \delta^m =
  \begin{pmatrix}
    \delta_X^m & 0 & 0 \\
    g_{XQ}^m & \delta_Q^m & 0 \\
    g_{XZ}^m & g_{QZ}^m & \delta_Z^m
  \end{pmatrix},
  \qquad \text{where} \ m=0,1,
\end{equation}
can decomposed as $\delta^m = \delta_{\mrm{loc}}^m + g^m$, where
\begin{equation}
  \delta^m_\mrm{loc} =
  \begin{pmatrix}
    \delta_X^m & 0 & 0 \\
    0 & \delta_Q^m & 0 \\
    0 & 0 & \delta_Z^m
  \end{pmatrix},
  \   \hspace{0.5cm} \  g^m =
  \begin{pmatrix}
    0 & 0 & 0 \\
    g_{XQ}^m & 0 & 0 \\
    g_{XZ}^m & g_{QZ}^m & 0
  \end{pmatrix}.
\end{equation}
From a direct computation, it follows that $\delta_\mrm{loc}^2 = 0$, $g^2 = 0$ and $\delta_\mrm{loc} g + g \delta_\mrm{loc} = 0$.

\subsubsection{Spectral sequence and homological perturbation lemma}
\label{subsec:spectral-sequence}

Now, we depart from the discussion in \cite{yuan2026unified} and introduce the spectral sequence perspective.  To our knowledge, this perspective is new and it naturally captures the code embedding framework. We first describe the general setup before relating it to the cochain complex $C$ in Eqn.~\ref{eq:height-2-diagram}.
We will use the notation $(C, \delta)$ to specify the cochain complex, where $C$ denotes the vector spaces and $\delta$ denotes the differential operator.

Let us review some of the facts we will need. Let $E_r^{p,q}$ denote the vector space on the $r$-th page of the spectral sequence ($r \ge 0$), where $p$ is the horizontal grading and $q$ is the vertical grading. Define the total complex
\begin{equation}
  \Tot(E_r)^n := \bigoplus_{p+q=n} E_r^{p,q}
\end{equation}
which is the complex obtained by summing along the diagonals of the $r$-th page.
On the $r$-th page, we also define the differential operator
\begin{equation}
  D_r: \Tot(E_r)^n \to \Tot(E_r)^{n+1},
\end{equation}
which satisfies $D_r \circ D_r = 0$. Furthermore, $D_r$ decomposes into
\begin{equation}
  D_r = d_{r,0} + d_{r,1} +...,
\end{equation}
where $d_{r,i}: E_r^{p,q} \to E_r^{p-r-i+1, q+r+i}$ ($i \ge 0$).
Here is an example showing the $0$-th page.
\begin{equation}
\label{eq:E0-page-with-differentials}
\begin{tikzpicture}[baseline]
\matrix(a)[matrix of math nodes, nodes in empty cells,
  nodes={minimum size=22pt},
  row sep=2.2em, column sep=2.8em,
  text height=1.5ex, text depth=0.25ex]
{
E_0^{p-1,q-1} & E_0^{p,q-1} & E_0^{p+1,q-1} & E_0^{p+2,q-1} \\
E_0^{p-1,q}   & E_0^{p,q}   & E_0^{p+1,q}   & E_0^{p+2,q}   \\
E_0^{p-1,q+1} & E_0^{p,q+1} & E_0^{p+1,q+1} & E_0^{p+2,q+1} \\
E_0^{p-1,q+2} & E_0^{p,q+2} & E_0^{p+1,q+2} & E_0^{p+2,q+2} \\
};

\path[->,font=\scriptsize]
(a-2-2) edge node[above]{$d_{0,0}$} (a-2-3)
(a-2-2) edge node[right]{$d_{0,1}$} (a-3-2)
(a-2-2) edge node[below right]{$d_{0,2}$} (a-4-1);

\draw[->]
([xshift=-1.8em,yshift=1.6em]a-1-1.north west) --
node[above] {$p$}
([xshift=1.8em,yshift=1.6em]a-1-4.north east);

\draw[->]
([xshift=-1.8em,yshift=1.6em]a-1-1.north west) --
node[left] {$q$}
([xshift=-1.8em,yshift=-1.6em]a-4-1.south west);
\end{tikzpicture}
\end{equation}
We now describe how to define the vector spaces on the next page of the spectral sequence. Because $D_r \circ D_r = 0$, we have $\sum_{i+j=k} d_{r,i} \circ d_{r,j} = 0$ for any $k \ge 0$. In particular, when $k=0$, we have
\begin{equation}
  d_{r,0} \circ d_{r,0} = 0.
\end{equation}
The differential $d_{r,0}$ induces the next page of the spectral sequence, namely
\begin{equation}
    E_{r+1}^{p,q} = H(E_r^{p,q}, d_{r,0}).
\end{equation}
We now describe how the total differential $D_{r+1}$ on the next page is obtained. This relies on the homological perturbation lemma \cite{nlab:homological_perturbation_theory}, which we will briefly review here.
More discussions can be found in \Cref{app:homological-pert-lemma}.
Suppose that we have:
\begin{itemize}
  \item[(1)] A cochain complex $(C, d)$,
  \item[(2)] A smaller cochain complex $(H, d_H)$,
  \item[(3)] A map $h: C \to C[-1]$ and chain maps $i: H \to C$, $p: C \to H$ satisfying
        \begin{equation}
        p i = \id_H
        \end{equation}
        and
        \begin{equation}
        i p = \id_C + h d + d h.
        \end{equation}
\end{itemize}
Then $H$ is a deformation retract of $C$.
We note an abuse of notation in that $p$ is simultaneously used both for the grading index, e.g., in $E^{p,q}_r$, as well as the map $p$.  Now we perturb the differential on $C$ by adding $d_\mrm{pert}$, so the new differential is
\begin{equation}
  d' = d+d_\mrm{pert},
\end{equation}
with $(d')^2=0$. Under a suitable smallness condition, we can transfer this perturbation to a new $d_H'$ on $H$.
We then define $(C,d')$ and $(H,d_H')$ and we can compute the new maps $  d_H',\ i',\ p',\ h'$,
such that $(H,d_H')$ is again a deformation retract of $(C,d')$.

Under the smallness condition, if there exists a finite $N$ for which $(hd_\mrm{pert})^N = 0$, we can compute the transferred differential $d_H'$ as
\begin{equation}
  d_H'
  = d_H + p\,d_\mrm{pert}\,(1-hd_\mrm{pert})^{-1} i
  = d_H + p\,d_\mrm{pert}\,(1 + hd_\mrm{pert} + (hd_\mrm{pert})^2 +...) i,
\end{equation}
where the geometric series stops after finitely many terms.

Returning to the spectral sequence, we will apply the homological perturbation lemma with:
\begin{equation}
  \begin{gathered}
    (C, d) = (\Tot(E_r), d_{r,0}), \\
    (H, d_H) = (\Tot(E_{r+1}),0).
  \end{gathered}
\end{equation}
Notice that because we work over a field (not just a ring), the complex $(E_r^{p,q}, d_{r,0})$ contracts onto its cohomology, thus giving $(E_{r+1}^{p,q}, 0) = (H(E_r^{p,q}, d_{r,0}),0)$. Thus,  $(\Tot(E_{r+1}),0)$ is a deformation retract of $(\Tot(E_r), d_{r,0})$. We then set
\begin{equation}
     d_{\mathrm{pert}} = D_r - d_{r,0}
\end{equation}
which produces $(C, d + d_\mrm{pert}) = (\Tot(E_r), D_r)$, and define $(\Tot(E_{r+1}),D_{r+1}) = (H,d_H')$, which allows us to find $D_{r+1}$ using the homological perturbation lemma, assuming that the perturbation satisfies the smallness condition.
This gives us a $D_{r+1}$ such that $(\Tot(E_{r+1}), D_{r+1})$ is a deformation retract of $(\Tot(E_r), D_r)$:
\begin{equation}
  (\Tot(E_r), D_r) \simeq (\Tot(E_{r+1}), D_{r+1}).
\end{equation}
The reason the smallness condition is satisfied in our case is that our complexes are bounded, so the geometric series terminates after finitely many terms.

\subsubsection{The spectral sequence description for code embedding framework}

We now use spectral sequences to reformulate the code embedding framework.
We focus primarily on the ``height-$2$ cone'' described in \cite{yuan2026unified},
  although the same idea applies to the ``height-$n$ cone'' as well.

Let $C$ be the cochain complex in \eqref{eq:height-2-diagram},
  and let $C^g$ be the embedded complex defined in \cite{yuan2026unified},
  whose definition we recall below.
The original framework relates $C^g$ and $C$ at the level of cohomology groups.
Here we give a chain-level refinement of this relation.
In particular,
  we obtain a natural cochain map from $C^g$ to $C$.

To do so, we first relate $C$ to the spectral sequence from the previous subsubsection.
We have
\begin{equation}
  (\Tot(E_0), D_0) = (C,\delta).
\end{equation}
Here the nonzero entries on the $0$-th page are
\begin{equation}
    E^{p,0}_0 = C^p_X, \quad E^{p,1}_0 = C^{p+1}_Q, \quad  E^{p,2}_0 = C^{p+2}_Z,
\end{equation}
and all other vector spaces $E^{p,q}_0$ are set to zero.
The horizontal map $d_{0,0}$ is identified with $\delta_\mrm{loc}$.
The vertical map $d_{0,1}$ is identified with the maps
$g_{XQ}$ and $g_{QZ}$,
while $d_{0,2}$ is identified with $g_{XZ}$.
All other maps are zero.

Thus, the complex $C$ fits into the $0$-th page of the spectral sequence
as follows:
\begin{equation}
\label{eq:page-0-ours}
\begin{tikzpicture}[baseline]
\matrix(a)[matrix of math nodes, nodes in empty cells,
  nodes={minimum size=22pt},
  row sep=2.2em, column sep=2.8em,
  text height=1.5ex, text depth=0.25ex]
{
\color{gray}E_0^{-2,-1} & \color{gray} E_0^{-1,-1} & \color{gray} E_0^{0,-1} & \color{gray} E_0^{1,-1} & \color{gray} E_0^{2,-1} \\
\color{gray} E_0^{-2,0}   & \color{gray} E_0^{-1,0}   & \color{blue} E_0^{0,0} = C^0_X  & \color{blue} E_0^{1,0}  = C^1_X  & \color{blue} E_0^{2,0}  = C^2_X \\
\color{gray} E_0^{-2,1} & E_0^{-1,1}  = C^0_Q & E_0^{0,1}  = C^1_Q & E_0^{1,1}  = C^2_Q & \color{gray} E_0^{2,q-1} \\
\color{red} E_0^{-2,2}  = C^0_Z & \color{red} E_0^{-1,2}  = C^1_Z & \color{red} E_0^{0,2}  = C^2_Z & \color{gray} E_0^{1,2} & \color{gray} E_0^{2,2} \\
};

\path[->,blue,font=\scriptsize]
(a-2-3) edge node[above]{$d_{0,0} = \delta_\mrm{loc}$} (a-2-4);

\path[->,font=\scriptsize]
(a-2-3) edge node[right]{$d_{0,1}=g_{XQ}$} (a-3-3);

\path[->,font=\scriptsize]
(a-2-3) edge[dashed] node[left,pos=0.2]{$d_{0,2}=g_{XZ}$} (a-4-2);

\draw[->]
([xshift=-2.6em,yshift=1.6em]a-1-1.north west) --
node[above] {$p$}
([xshift=1.8em,yshift=1.6em]a-1-5.north east);

\draw[->]
([xshift=-2.6em,yshift=1.6em]a-1-1.north west) --
node[left] {$q$}
([xshift=-1.5em,yshift=-1.6em]a-4-1.south west);
\end{tikzpicture}
\end{equation}
By the homological perturbation lemma,
  the cochain complex is chain-homotopy equivalent to the complex on the 1-st page.
\begin{equation} \label{eq:chain-homotopy-eqv}
  (C, \delta) \simeq  (H(C, \delta_{\mrm{loc}}), H(g)).
\end{equation}
Here $(H(C,\delta_{\mrm{loc}}),H(g))$
  denotes the cochain complex on the 1-st page,
  namely $(\Tot(E_1), D_1)$.
We later refer to this as the \emph{induced complex}.
When there is no ambiguity,
  we abbreviate this equivalence as $C \simeq H(C, \delta_\mrm{loc})$.

We now recall the definition of $C^g$.
It is the subcomplex of $\Tot(E_1) = H(C, \delta_\mrm{loc})$ given by the column
\begin{equation*}
  C^g: E_1^{0,0} \to E_1^{0,1} \to E_1^{0,2}.
\end{equation*}

From this perspective,
  we can upgrade the original code embedding formalism
  from a statement about cohomology groups to a chain-level statement.
Indeed,
  the embedded code $C^g$ appears naturally as a subcomplex of $H(C,\delta_{\mrm{loc}})$,
  and the homological perturbation lemma provides a chain-homotopy equivalence
  from $H(C,\delta_{\mrm{loc}})$ to $C$:
\begin{equation*}
  C^g \hookrightarrow H(C, \delta_\mrm{loc}) \xrightarrow{\simeq} C.
\end{equation*}

As an application of this perspective,
  we recover the isomorphism from \cite[Theorem I.1]{yuan2026unified}.
\begin{lemma}
  Assume that $H^1(C^X) = H^1(C^Z) = 0$,
  then
  \begin{equation*}
    H^1(C^g) \cong H^1(C).
  \end{equation*}
\end{lemma}
\begin{proof}
  First,
  \begin{equation*}
    H^1(H(C, \delta_\mrm{loc})) \cong H^1(C),
  \end{equation*}
  since $H(C,\delta_{\mathrm{loc}})$ is chain-homotopy equivalent to $C$.

  It remains to show $H^1(C^g) \cong H^1(H(C, \delta_\mrm{loc}))$.
  The assumptions $H^1(C^X)=H^1(C^Z)=0$ imply that
    $E_1^{-1,2} = 0$ and $E_1^{1,0} = 0$.
  Consequently, the induced complex reduces to
  \begin{equation*}
    H(C, \delta_\mrm{loc}):
    E_1^{0,0} \oplus (E_1^{-1,1} \oplus E_1^{-2,2})
    \to
    E_1^{0,1}
    \to
    E_1^{0,2} \oplus (E_1^{2,0} \oplus E_1^{1,1}).
  \end{equation*}
  Since the differential $D_1 = H(g)$ only has components of the form
    $d_{1,i}: E_1^{p,q} \to E_1^{p-i,q+1+i}$ for $i \ge 0$,
    the source and target degrees
      force all components in $H(g)$ to vanish except for
    $E_1^{0,0} \to E_1^{0,1} \to E_1^{0,2}$.
  This is precisely the subcomplex $C^g$.
  Therefore, the inclusion $C^g \hookrightarrow H(C,\delta_{\mrm{loc}})$
    induces an isomorphism
  \begin{equation*}
    H^1(C^g)
    \cong
    H^1(H(C, \delta_\mrm{loc})).
  \end{equation*}
\end{proof}

\subsection{Topological defect networks and sheaf codes}\label{prelims:sheaf}

Typically, a code is specified combinatorially,
  by giving a set of qubits, a set of checks, and the incidence relations between them.
There is, however, a more geometric way to encode the same information,
  through the language of topological defect networks~\cite{Aasen_2020,williamson2023layer}
  or sheaf codes~\cite{first2024good2querylocallytestable,panteleev2024maximallyextendablesheafcodes,lin2024transversalnoncliffordgatesquantum,li2025poincar},
  which are closely related frameworks for quantum CSS codes.
In what follows,
  we use the sheaf-code formalism as developed in
  \cite{lin2024transversalnoncliffordgatesquantum,li2025poincar}.

We will use this formalism in \Cref{sec:explicit-embedding}
  to construct codes with an explicit embedding in $\mathbb{R}^3$.
The advantage is that the operations involved in these explicit constructions have natural geometric descriptions.

Readers who prefer to avoid the formalism may skip ahead to
  \Cref{sec:explicit-embedding}.
Our use of sheaf codes is fairly minimal there:
  the relevant data are summarized explicitly
  and rephrased in terms of assigning colors to elements of a cell complex
  (see the discussion at the beginning of each subsection in \Cref{sec:explicit-embedding}).

Let $X$ be a $t$-dimensional cell complex,
  and let $X(i)$ denote the set of $i$-cells in $X$, for $i=0,1,...,t$.
For cells $\sigma,\tau \in X$,
  we write $\sigma \le \tau$
  to mean that $\sigma$ is a face of $\tau$,
  and we write $\sigma \lessdot \tau$ if, in addition,
  their dimensions differ by $1$.
Let
\begin{equation*}
  X_{\ge \sigma}(i) = \{\tau \in X(i): \tau \ge \sigma\}
\end{equation*}
  denote the set of $i$-cells that contain $\sigma$ as a face.

A sheaf $\cF$ assigns a finite-dimensional vector space $\cF_\sigma$ over $\ff_2$
  to each cell $\sigma\in X$,
  together with restriction maps
\begin{equation*}
  \res_{\sigma,\tau}: \cF_\sigma\to \cF_\tau \quad \text{ for each } \sigma\le\tau,
\end{equation*}
satisfying the condition
\begin{equation*}
  \res_{\pi,\tau} \circ \res_{\sigma,\pi} = \res_{\sigma,\tau} \quad \text{ for all } \sigma \le \pi \le \tau.
\end{equation*}

The cell complex $X$ together with its sheaf data $\cF$ defines a sequence of maps
\begin{equation} \label{eq:sheaf-cochain-complex}
  C^\bullet(X, \cF): \quad
  C^0(X, \cF)
  \overset{\delta^0}{\longrightarrow}
  C^1(X, \cF)
  \overset{\delta^1}{\longrightarrow}
  \cdots
  \overset{\delta^{t-1}}{\longrightarrow}
  C^t(X, \cF),
\end{equation}
where
\begin{equation*}
  C^i(X,\cF)=\bigoplus_{\sigma\in X(i)}\cF_\sigma,
\end{equation*}
and
\begin{equation*}
  (\delta \alpha)(\tau)
  =
  \sum_{\sigma \in X(i): \ \sigma\lessdot \tau} \res_{\sigma,\tau}(\alpha(\sigma)),
  \quad \text{ for every }
  \tau\in X(i+1).
\end{equation*}

We want these maps to form a cochain complex,
  i.e., to satisfy $\delta \circ \delta = 0$.
This is guaranteed whenever the cell complex
  with the constant sheaf forms a cochain complex:
  \begin{equation}
    C^\bullet(X,\ff_2): \quad
    C^0(X,\ff_2)
    \longrightarrow
    C^1(X,\ff_2)
    \longrightarrow
    \cdots
    \longrightarrow
    C^t(X,\ff_2).
  \end{equation}
Here, the constant sheaf assigns $\ff_2$ to every cell,
  with all restriction maps equal to the identity.

We specialize to the case where $\cF_\tau = \ff_2$ for every $t$-cell $\tau$.
Thus, for any cell $\sigma$,
  the restrictions of an element of $\cF_\sigma$
  to the top-dimensional cells containing $\sigma$ define a map
\begin{equation*}
  \iota_\sigma:
  \cF_\sigma
  \to
  \prod_{\tau \in X_{\ge \sigma}(t)} \cF_\tau
  \cong \ff_2^{X_{\ge \sigma}(t)}
\end{equation*}
given by
\begin{equation*}
  \iota_\sigma(\alpha)
  =
  \prod_{\tau \in X_{\ge \sigma}(t)} \res_{\sigma,\tau}(\alpha).
\end{equation*}

The sheaves appearing in our construction satisfy the strong sheaf axiom.
Under this condition,
  the map $\iota_\sigma$ is injective,
  so we may regard $\cF_\sigma$ as a subspace of
  $\ff_2^{X_{\ge \sigma}(t)}$.
Moreover, the sheaf is determined by its values on the $(t-1)$-cells \cite[Lemma 3.5]{li2025poincar}.
More explicitly, for any cell $\sigma$,
  one has the characterization \cite[Theorem 3.3]{li2025poincar}
\begin{equation}\label{eq:strong-sheaf-characterization}
  \iota_\sigma \cF_{\sigma}
  =
  \{\alpha \in \prod_{\tau \in X_{\ge \sigma}(t)} \cF_{\tau}: \forall \tau' \in X_{\ge \sigma}(t-1), \alpha|_{X_{\ge \tau'}(t)} \in \iota_{\tau'} \cF_{\tau'}\}.
\end{equation}
Consequently, in \Cref{sec:explicit-embedding},
  to specify the sheaf, it suffices to assign
  a subspace $\cF_{\sigma} \subseteq \ff_2^{X_{\ge \sigma}(t)}$
  for each $(t-1)$-cell $\sigma$.
In the following, we drop $\iota$ from the notation
  and simply regard $\cF_{\sigma}$ as a subspace of $\ff_2^{X_{\ge \sigma}(t)}$.

The sheaves we consider also satisfy local acyclicity.
We say that the cell complex $X$ is locally acyclic if,
for every $i$-cell $\sigma \in X$,
\begin{equation*}
  H_j(X_{\le \sigma},\mathbb{F}_2)=0
  \quad
  \text{ for } 0<j\le i.
\end{equation*}
We say that the sheaf $\cF$ is locally acyclic if,
for every $i$-cell $\sigma \in X$,
\begin{equation*}
  H^j(X_{\ge \sigma},\cF)=0
  \quad
  \text{ for } i\le j<t.
\end{equation*}

These local acyclicity conditions imply a Poincar\'e duality statement for sheaf codes.
This duality lets us study the dual chain complex,
  obtained by transposing the coboundary maps,
  in terms of the cochain complex associated with a dual sheaf.
To define this dual sheaf,
  for each $(t-1)$-cell $\sigma$
  let $\cF^{\perp}_{\sigma} \subseteq \ff_2^{X_{\ge \sigma}(t)}$
  be the orthogonal subspace of $\cF_{\sigma}$.
The sheaf $\cF^{\perp}$ is then generated from these values on the $(t-1)$-cells
  using the characterization in \Cref{eq:strong-sheaf-characterization}.

We are ready to state the Poincar\'e duality.
\begin{theorem}[{\cite[Theorem 3.17]{li2025poincar}}]
  \label{thm:poincare-duality}
  Let $X$ be a $t$-dimensional sparse cell complex,
    and let $\cF$ be a locally acyclic sheaf on $X$.
  Then there is a duality between $\cF$ and $\cF^{\perp}$
    in terms of logical qubits, code distances,
    (co)boundary expansion, and decoders.
  In particular, for every $0 \le i \le t$, there is an isomorphism:
  \begin{equation}
    H^i(X,\cF^\perp) \cong H_{t-i}(X,\cF).
  \end{equation}
\end{theorem}
We expect this duality to extend to memory lifetime.\footnote{
A possible unifying proof strategy is to apply the homological perturbation lemma
to the double complex arising from \v{C}ech cohomology.
Because the cell complex is sparse and finite-dimensional,
the resulting perturbation expansion terminates
after finitely many terms.}

We summarize the data that must be specified:
\begin{itemize}
  \item the cell complex $X$;
  \item the sheaf data on $(t-1)$-cells, namely $\cF_{\sigma} \subseteq \ff_2^{X_{\ge \sigma}(t)}$ for $\sigma \in X(t-1)$.
\end{itemize}
The data on the $(t-1)$-cells determine the sheaf values on all other cells
  through the characterization in \Cref{eq:strong-sheaf-characterization}.

We must then check the following conditions:
\begin{itemize}
  \item the constant sheaf $C^\bullet(X, \ff_2)$ satisfies $\delta \circ \delta = 0$,
  \item the cell complex $X$ is locally acyclic,
  \item the sheaf $\cF$ is locally acyclic.
\end{itemize}
Once these conditions hold,
  we can apply Poincar\'e duality to study the dual code
  by analyzing the cochain complex associated with the dual sheaf,
  which is often easier to analyze.

We also briefly comment on the relation to topological defect networks.
One way to interpret the quantum CSS code associated with
\begin{equation*}
  C^{i-1}(X, \cF) \overset{\delta^{i-1}}{\longrightarrow} C^i(X, \cF) \overset{\delta^i}{\longrightarrow} C^{i+1}(X, \cF)
\end{equation*}
is to regard each $t$-cell of $X$ as a copy of the $(i,t-i)$-toric code in $t$ dimensions.
The sheaf data on the $(i-1)$-cells then specifies the corresponding
  gapped boundary conditions between adjacent copies of these toric codes.  In the topological defect network language, boundary conditions can be described in terms of condensation data
  for the $e$- and $m$-type excitations,
  which in the sheaf language is encoded by the vector spaces $\cF_\sigma$.

For example, when $i=1$ and $t=2$,
the choice
\begin{equation*}
  \cF_\sigma = \{000,111\},
  \qquad
  \cF_\sigma^\perp = \{000,011,101,110\},
\end{equation*}
describes a gapped boundary between three copies of the two-dimensional surface code.
In this case, the condensed excitations are generated by
\begin{equation*}
  \langle e_1 e_2 e_3,\; m_1 m_2,\; m_2 m_3 \rangle.
\end{equation*}
See \cite{lin2024transversalnoncliffordgatesquantum} for further discussion.

\subsection{Memory lifetime}\label{sec:review-memorylifetime}

We now discuss how to prove that the memory lifetime of our code is exponentially long in a power of the system size.  Normally, the model studied in this setting is a quantum system coupled to a bath in a weak-coupling Markovian limit.  This means that the density matrix of a quantum system  $\rho$ evolves according to the Master equation
\begin{equation}
\frac{d}{dt} \rho = -i [H, \rho] + \sum_{i} (S_i \rho S_i^\dagger - \frac{1}{2}\{\rho, S_i^\dagger S_i\})
\end{equation}
with $H$ the system Hamiltonian and the last term is often called the Lindbladian and written as $\cL(\rho)$.  The operators $\{S_i\}$ are referred to as jump operators.  The solution to this equation of motion is written as $\rho(t) = \exp(-it[H,\cdot] + t\cL) [\rho(0)]$. We will be interested in a choice of jump operators where $\cL(\rho_\beta) = 0$ with $\rho_{\beta} = e^{-\beta H}/\Tr(e^{-\beta H})$ an equilibrium Gibbs state.  We may additionally like to impose a detailed balance condition on the jump operators, but the main theorem we use does not require detailed balance.  More details can be found in Refs.~\cite{alicki2009thermalization, alicki2010thermal, chesi2010thermodynamic}.

Since our construction has a single encoded qubit, we restrict the analysis that follows to this case.  Call $X_L$ and $Z_L$ the logical operators associated to the encoded qubit.

We now construct so-called dressed operators $\widetilde{X}_L$ and $\widetilde{Z}_L$,
which will be useful for proving self-correction.
We assume that our code is a CSS quantum code, where the stabilizers are only $X$- or $Z$-type, and the two Pauli sectors effectively decouple.
In our setting, $\widetilde{X}_L$ and $\widetilde{Z}_L$ take the form
\begin{align} \label{eq:dressed-logical}
\widetilde{X}_L &= \sum_{ \sigma_Z}\Pi_{\sigma_Z} \cO_Z(\sigma_Z) X_L \cO_Z(\sigma_Z) \Pi_{\sigma_Z}  \nonumber \\
\widetilde{Z}_L &= \sum_{ \sigma_X}\Pi_{\sigma_X} \cO_X(\sigma_X) Z_L \cO_X(\sigma_X) \Pi_{\sigma_X},
\end{align}
where the syndrome $\sigma = (\sigma_X, \sigma_Z)$, where $\sigma_X \in \mathbb{F}_2^{r_Z}$ and $\sigma_Z \in \mathbb{F}_2^{r_X}$, is split into the violations of $Z$-type and $X$-type stabilizers, respectively.
Here $\Pi_{\sigma_X}$ and $\Pi_{\sigma_Z}$ project onto the subspace of states with syndrome $\sigma_X$ and $\sigma_Z$, respectively, and $\cO_X(\sigma_X)$ and $\cO_Z(\sigma_Z)$ are Pauli correction operators chosen to remove the syndrome of a particular kind.  Since the choice of Pauli correction is not unique, the definition above is not yet well defined.  We will require a decoder to determine $\cO_X(\sigma_X)$ and $\cO_Z(\sigma_Z)$. It will be a composite map relying on maps $\dec_X(\sigma_X)$ and $\dec_Z(\sigma_Z)$, which determine the supports of $Z$-type correction and $X$-type correction to be applied based on the given syndrome. Then:

\begin{definition} [Dressed observable]
Suppose we are given logical operators $(X_L, Z_L)$ and Pauli corrections output by the decoder $\cO_Z(\sigma_Z) = Z^{\dec_Z(\sigma_Z)} $ and $\cO_X(\sigma_X) = X^{\dec_X(\sigma_X)}$.  The dressed observables with respect to this decoder are given by Eq.~\eqref{eq:dressed-logical}.
\end{definition}
We then define the notion of an observable being protected from a set of errors:
\begin{definition}
    Define $P_{\mrm{stable}}$ to be a projector onto the subspace satisfying $P_{\mrm{stable}} \widetilde{X}_L = \widetilde{X}_L P_{\mrm{stable}}$ for dressed observable $\widetilde{X}_L$. We say that $\widetilde{X}_L$ is protected from a set of operators $\cE$ on the subspace $P_{\mrm{stable}}$ if and only if
    \begin{equation}
    [E, \widetilde{X}_L] P_{\mrm{stable}} =0 \,\,\, \forall E \in \cE.
    \end{equation}
\end{definition}

To illustrate this definition, suppose $\cE$ is a set of Pauli operators (the main theorem will hold for non-Pauli jump operators as well).  Define $\ket{\sigma, \pm_L} = \cO(\sigma)   \ket{\pm_L}$ as a state with syndrome $\sigma = (\sigma_X, \sigma_Z)$, where $\cO(\sigma) = \cO_X(\sigma_X) \cO_Z(\sigma_Z)$ is the Pauli correction operator and $\ket{\pm_L}$ is a codestate that is an eigenstate of $X_L$ with eigenvalue $\pm 1$.  It follows from a simple computation that $\widetilde{X}_L \ket{\sigma, \pm_L} = \pm \ket{\sigma, \pm_L}$.  Define $P_{\mrm{stable}}$ as a subspace of states protected from an error $E$.  This means that $[E, \widetilde{X}_L] \ket{\sigma, \pm_L} = 0$ if $\ket{\sigma, \pm_L} \in P_{\mrm{stable}}$.  By another simple computation, the only states for which this is not true is those where $\cO(\sigma) \ket{\sigma, \pm_L} \neq \cO(\sigma + \sigma_E) E\ket{\sigma, \pm_L}$, where $\sigma_E$ is the syndrome associated to $E$.  One may say that these states are those which decode differently upon applying a small error.    A similar result holds for $Z_L$ and $\widetilde{Z}_L$.
\begin{claim}\label{claim:stableconfig}
Assuming $\cE$ is a set of Pauli operators, if $\exists E \in \cE$ and $\exists \ket{\sigma, \cdot} \in \mrm{supp}(P_{\mrm{stable}})$  such that $\cO(\sigma)\cO(\sigma + \sigma_E) E$ is a logical operator, then $\widetilde{X}_L$ ($\widetilde{Z}_L$) is not protected from $\cE$ on subspace $P_{\mrm{stable}}$.
\end{claim}
The main theorem we will need is (from \cite{chesi2010thermodynamic}, but implicitly appeared in \cite{alicki2010thermal}):
\begin{theorem}[Theorem 1 in \cite{chesi2010thermodynamic}]\label{thm:memorylifetime}
Let $\cL$ be a Lindbladian with jump operators in the set $\cE$, whose stationary state is $\rho_{\beta} = e^{-\beta H}/\Tr(e^{-\beta H})$ for stabilizer Hamiltonian $H$.  Define $e^{\cL t}$ to be the generator of time evolution by time $t \geq 0$.  Construct dressed observables $\widetilde{X}_L$ and $\widetilde{Z}_L$ as well as a subspace $P_{\mrm{stable}}$ in which $\widetilde{X}_L$ and $\widetilde{Z}_L$ are protected from the set of errors $\cE$.  Assume that $\widetilde{X}_L$, $\widetilde{Z}_L$, and $P_{\mrm{stable}}$ commute with $H$: then there exists a CPTP logical encoding map $\mathrm{Enc}: L(\mathbb{C}^2) \to L(\cH_S)$ and a logical decoding map $\mathrm{Dec}: L(\cH_S) \to L(\mathbb{C}^2)$ such that
\begin{equation}
\norm{e^{-it[H,\cdot] + \cL t} \circ \mathrm{Enc}(\tau) - \mathrm{Enc}(\tau)}_1 \leq 8 \norm{\cL}_1 t \Tr(1-P_{\mrm{stable}}) \rho_{\beta}
\end{equation}
for any single-qubit density matrix $\tau$, and $\mathrm{Dec} \circ \mathrm{Enc}(\tau) = \tau$.
\end{theorem}
The proof of this theorem provides explicit logical encoding and logical decoding maps. Note that the total Hilbert space for the system can be factorized as $\cH_S = \cH_{\mrm{log}} \otimes \cH_{\mrm{syn}}$ where $\cH_{\mrm{log}}$ is the logical subspace.  In this basis, the Hamiltonian is $H = I \otimes H'$ where $H'$ is diagonal in the syndrome basis with matrix elements equal to the size of the syndrome.  We choose a thermal encoding map to encode the logical state $\tau$, which acts as
\begin{equation}
\mrm{Enc}(\tau) = \frac{\tau \otimes e^{-\beta H'}}{\Tr(\tau \otimes e^{-\beta H'})}
\end{equation}

The goal is to retrieve $\tau$ after evolving the system for a long time $t$.  In the infinite time limit, the system thermalizes to the state $I \otimes \frac{e^{-\beta H'}}{\Tr(e^{-\beta H'})}$, and thus while we cannot recover $\tau$ at infinite time, we want to be able to do so after time $\sim \exp(n_k^{\eta})$ where $\eta$ is some positive constant and $n_k$ is the number of qubits in the construction at level $2k$.  The logical decoding map $\mrm{Dec}(\rho) = \Tr_{\cH_{\mrm{syn}}} (\rho)$ extracts the state of the encoded qubit after this time.  It is immediate that $\mathrm{Dec} \circ \mathrm{Enc}(\tau) = \tau$.

In summary, our goal will be to construct decoding maps $\dec_X(\sigma_X)$ and $\dec_Z(\sigma_Z)$ and a subspace of states $P_{\mrm{stable}}$ so that $\widetilde{X}_L$ and $\widetilde{Z}_L$ are protected from the set of single-qubit Pauli operators on $P_{\mrm{stable}}$\footnote{We restrict our analysis to single-site Pauli jump operators, but we emphasize that Theorem~\ref{thm:memorylifetime} applies to non-Pauli jump operators, as well as jump operators supported over regions moderately smaller than the code distance. }.  Then, if we can prove that $\Tr(1-P_{\mrm{stable}}) \rho_{\beta} \leq \exp(-n_k^{\eta})$, the above theorem proves that the memory lifetime is exponentially long.

\section{Construction with random embedding}
\label{sec:construction}

Our code is a CSS code associated with a three-term cochain complex that we construct iteratively.
We denote the code in the $i$-th iteration by $C_{(i)}^\bullet$ (we may also refer to this as the code at ``level $i$''),
  which is associated with a Tanner square complex $T_{(i)}$ (see Subsec.~\ref{sec:Tanner-prelims}).

Each iteration maps $T_{(i)} \mapsto T_{(i+1)}$, and increases the size of the code by a constant factor while effectively doubling the size of the syndrome associated with one type of errors, Pauli $X$ or $Z$.  There are two different iterations: the $\cR_X$ iteration, which acts as $\cR_X: T_{(2i)} \mapsto T_{(2i+1)}$, and the $\cR_Z$ iteration, which acts as $\cR_Z: T_{(2i+1)} \mapsto T_{(2i+2)}$. The $\cR_X$ iteration effectively doubles the size of the syndrome for $X$-type errors, while the $\cR_Z$ iteration does the same for $Z$-type errors. The final code is constructed by alternating between these iterations.
Each iteration consists of 2 steps, which we call perturbation and replacement.

We use $I_{(i)}: T_{(i)} \rightarrow \mathbb{R}^3$ to denote the embedding map, which maps each vertex of $T_{(i)}$ to a point in $\mathbb{R}^3$.  The image of edges and higher-dimensional cells are then obtained via linear interpolation. At the $i$-th iteration, we assume an induction hypothesis that $I_{(i)}$ is bounded-density and local. Then we show that after one iteration, the resulting embedding $I_{(i+1)}$ remains bounded-density and local.

We briefly collect some important notation for defining the construction.

\noindent General notation:
\begin{itemize}[itemsep = -1pt]
    \item $C^\bullet_{(i)}$: the cochain complex of the code at iteration $i$.
    \item $T_{(i)}$: Tanner square complex of the code at iteration $i$.
    \item $\cR_X(T_{(i)}) = T_{(i+1)}$, for $i$ even: $X$-iteration.
    \item $\cR_Z(T_{(i)}) = T_{(i+1)}$, for $i$ odd: $Z$-iteration.
    \item $U_X$ and $U_Z$: the new square complex each square is replaced with during the $\cR_X$ and $\cR_Z$ iterations, respectively.
\end{itemize}

\noindent Geometry-related notation (all parameters will be universal constants independent of the iteration index $i$):
\begin{itemize}[itemsep = 0pt]
\item $I_{(i)}: T_{(i)} \rightarrow \mathbb{R}^3 $: embedding map at iteration $i$.
    \item $z_v, z_e$: maximum vertex degree (edges incident to a vertex) and maximum edge degree (faces incident to an edge) in $T_{(i)}$ at any step $M$.
    \item  $t, t'$: ``density'' of embedding at different steps, namely the maximum number of 2-cells of the complex that intersect any unit ball in $\mathbb{R}^3$. We will have $t' < t$.
    \item $\lambda$: scale factor; during the perturbation step, we scale the square complex by a constant $\lambda$, and then perturb the vertices.
    \item $\ell$: a geometric parameter related to $\lambda$, which sets the subdivision scale, such that the edges in the square complex after the replacement step become unit length again.
    \item $\alpha, \alpha_1, \alpha_2, \gamma$: geometric $O(1)$ parameters controlling the perturbation appearing in embedding.  Roughly, $\alpha \lambda$ is the distance we perturb by, $\alpha_2 \lambda$ is the distance between the perturbation vectors of neighboring 0-cells, and $\alpha_1 \lambda$ captures the range of possible perturbation vectors for a 0-cell; our procedure randomly selects a vector in this range.
\end{itemize}

\subsection{Step 1: perturbation} \label{subsec:perturbation}

An iteration $\cR_X$ ($\cR_Z$) consists of two steps: perturbation and replacement. The two steps transform the Tanner square complex and embedding as $(T_{(i)}, I_{(i)}) \xmapsto{\text{perturb}} (T_{(i)}, I'_{(i)}) \xmapsto{\text{replace}} (T_{(i+1)}, I_{(i+1)})$. In particular, the perturbation step does not change the square complex, only the embedding map. The replacement step updates both the square complex and the embedding map.

We now state the induction hypothesis. In particular, we assume that the embedding $I_{(i)}$ is local and has bounded density (defined below).  We also assume that the maximum vertex degree (edges incident to a vertex) and maximum edge degree (faces incident to an edge) in $T_{(i)}$ are $z_v = 7$ and $z_e = 4$.  For the induction step, we show that $I_{(i+1)}$ is local and bounded-density characterized by the same universal constants, and that $T_{(i+1)}$ has the same values for $z_v$ and $z_e$.

The definition of a bounded-density local embedding of a square complex is provided below:
\begin{definition}[$(t,\ell)$-embedding in $\mathbb{R}^3$]
A map $I:X \to \mathbb{R}^3$ is called a $(t,\ell)$-embedding of a cell complex $X$ in $\mathbb{R}^3$ if it satisfies the following conditions:
\begin{enumerate}
\item (Bounded density) the preimage of a unit ball in $\mathbb{R}^3$ intersects with $\leq t$ 2-cells in $X$.
\item (Geometric locality) Neighboring 0-cells in $X$ map to points separated a distance $\leq \ell$ in $\mathbb{R}^3$.
\end{enumerate}
\end{definition}

\begin{remark}
  In \Cref{sec:explicit-embedding},
    we use a variant in which the bounded-density condition is imposed on balls of diameter $1$ rather than unit balls.
  That is, the preimage of any ball of diameter $1$ in $\mathbb{R}^3$ intersects $\le t$ 2-cells of $X$.
\end{remark}

If the square complex is $(t,\ell)$-embedded in $\mathbb{R}^3$ for constants $t$ and $\ell$, then the associated code is local in $\mathbb{R}^3$.  The main statement regarding the perturbation step is:
\begin{lemma} \label{lemma:embedding-main}
    If $I_{(i)}$ is a $( t, 1)$-embedding of $T_{(i)}$ in $\mathbb{R}^3$, there exists another embedding map $I'_{(i)}$ of $T_{(i)}$ in $\mathbb{R}^3$ that is a $(t', \ell)$-embedding, where $t$, $t' < t$ and $\ell > 1$ are universal $O(1)$ constants independent of $i$.
\end{lemma}

Before proving this lemma, we provide a high-level description of what the embedding achieves.  We assume that $I_{(i)}$ is a $( t, 1)$-embedding, and in the perturbation step, we construct a new map $I'_{(i)}$ that is a $( t', \ell)$-embedding. The constant $t'$ can be made small enough, so that we can then show that the embedding map $I_{(i+1)}$ for the complex after replacement will be a $(t,1)$-embedding, for $t$ large enough. In this case, $\ell$ will depend on $t$ and $t'$ (this dependence is derived in the proof of Lemma~\ref{lemma:embedding-main}).  We will avoid explicitly specifying the relationship between these $O(1)$ constants in statements of lemmas and claims, as they are cumbersome and the reader can deduce them from the proofs.  Instead, we provide explicit numerical values for them in Subsec.~\ref{sec:explicit-choice-constants}.

The perturbation procedure is similar to those used in the embedding theorems of Refs.~\cite{portnoy2023local, gromov2012generalizations}. The full procedure is defined as follows:
\begin{enumerate}
    \item  We formally add diagonal lines to the square 2-cells so they are standard triangular simplices; these can be removed later on. Therefore, we will refer to the embedded complex as a simplicial complex for the rest of the discussion. This simplicial complex is $(2 t, 1)$-embedded in $\mathbb{R}^3.$
    \item We scale the image of the simplicial complex in $\mathbb{R}^3$ by a sufficiently large constant $\lambda$ (which is $O(1)$); this results in a $(2 t, \lambda)$-embedding.
    \item Finally, we perturb the vertices of the embedded simplicial complex. For this, select $A>1$ distinct labels, and assign labels to vertices so that no simplex contains two vertices of the same label.   Consider a ball of radius $\alpha \lambda$ for $\alpha$ a large integer. On this ball, select $A$ caps of size (arc length) $\alpha_1 \lambda$, each separated by arc-distance $\geq 2\alpha_2 \lambda$ for positive integers $1 < 2\alpha_2 < \alpha$ and $\alpha_1$.  Assign a distinct label to each cap.  For a vertex with label $i$ at position $v$, choose a uniformly distributed vector $c(v)$ in the cap labeled $i$ (perturbations of different vertices with the same label are independent random variables); perturb the vertex to position $v + c(v)$.  As we prove below, after perturbation, the new map will be a $(t',\ell)$-embedding for some $t' < t$ that we can choose to be small enough.
\end{enumerate}

To prove Lemma~\ref{lemma:embedding-main}, we first need to establish that simplices do not become too distorted after perturbation.
\begin{claim}\label{claim:dist}
After the perturbation step, the image of any 1-simplex in $\mathbb{R}^3$ has length $\in [(\alpha_2-1) \lambda, ( 2\alpha +1)\lambda]$.
\end{claim}
\begin{proof}
Call $r_1$ and $r_2$ the locations of the endpoints of the 1-simplex before perturbation.  Call $p_1$ and $p_2$ the vectors they were perturbed by.  Then,
\begin{equation}
\norm{(r_1 + p_1) - (r_2 + p_2)}_2 \leq \norm{r_1 - r_2}_2 + \norm{p_1}_2 + \norm{p_2}_2 \leq  \lambda + \alpha \lambda + \alpha \lambda.
\end{equation}
using the fact that $p_1$ and $p_2$ lie on a radius $\alpha \lambda$ sphere.  To achieve a lower bound, we use
\begin{equation}
\norm{(r_1 + p_1) - (r_2 + p_2)}_2 \geq \norm{p_1 - p_2}_2 - \norm{r_2 - r_1}_2 \geq \alpha_2 \lambda -  \lambda
\end{equation}
which uses the fact that each endpoint is assigned a different label, and the distance between caps is $\geq 4 \alpha_2 \lambda/\pi \geq \alpha_2 \lambda $.
\end{proof}

\begin{claim} \label{claim:simplex-thickness}
Define the width of a 2-simplex as the shortest distance from a vertex to its opposite side, minimized over all vertices.
The width of any 2-simplex is $\geq \gamma \lambda$ after perturbation, where $\gamma = \alpha(1-\cos\frac{\alpha_2}{\alpha}) - 1$.
\end{claim}
\begin{proof}
We provide the crudest possible argument showing $\gamma > 0$; significant improvements can be made to our lower bound on $\gamma$ for certain placements of the caps.  Using notation from the proof of the previous claim, the distance between a vertex and the opposite side is
\begin{align}
d_1 &= \min_{\epsilon \in [0,1]} \norm{r_1 + p_1 - \epsilon(r_2 + p_2) - (1-\epsilon)(r_3+p_3)}_2 \nonumber \\ &\geq \min_{\epsilon \in [0,1]} \norm{p_1 - \epsilon p_2 - (1-\epsilon)p_3}_2 - \norm{r_1 - \epsilon r_2 - (1-\epsilon)r_3}_2 \nonumber \\
&\geq \min_{\epsilon \in [0,1]} \norm{p_1 - \epsilon p_2 - (1-\epsilon)p_3}_2 - \lambda
\end{align}
where in the third line we use the fact that the width of the original simplex is less than the maximum side length.

Call $q_\epsilon = \epsilon p_2 + (1-\epsilon)p_3$. Then we can write
\begin{align} \label{eq:d1}
d_1 \geq \min_{\epsilon \in [0,1]} \norm{p_1 - q_\epsilon}_2 - \lambda \geq \frac{1}{\alpha \lambda} \min_{\epsilon \in [0,1]} |p_1 \cdot (p_1 - q_\epsilon) | - \lambda,
\end{align}
by Cauchy-Schwartz. Because the arc distance between the caps is at least $\alpha_2 \lambda$, the angle between any two points in different caps is at least $\frac{\alpha_2}{\alpha}$. Therefore, $p_1 \cdot p_2 \leq \alpha\lambda^2 \cos \frac{\alpha_2}{\alpha}$, and $$|p_1 \cdot (p_1 - q_\epsilon)| = \alpha \lambda - \frac{p_1\cdot q_\epsilon}{\alpha \lambda} \geq \alpha \lambda \left ( 1 - \cos \frac{\alpha_2}{\alpha} \right).$$
Combining this with \Cref{eq:d1}, we obtain $\gamma \geq \alpha \left ( 1 - \cos \frac{\alpha_2}{\alpha} \right) - 1$.
\end{proof}

Now, we will argue that the density of simplices reduces after perturbation.  We will rely on some of the arguments in Theorem 7 of Ref.~\cite{portnoy2023local}, which is a modification of a proof first presented by Gromov and Guth~\cite{gromov2012generalizations}.  However, we will depart from their arguments in a few places to provide a sharper analysis with explicit constants.

\begin{proof}[ Proof of Lemma~\ref{lemma:embedding-main}.]

We will now show that after scaling by $\lambda$ and perturbation by scale $\alpha \lambda$, the simplicial complex is $(t',( 2 \alpha + 1) \lambda)$-embedded for $\lambda$ large enough, giving the desired embedding map $I'_{(i)}$ with $\ell =  (2 \alpha +1)\lambda$.

The proof involves showing two results: first that the density reduces like $t \to t'$ and second that distances increase like $1 \to \ell$.  The second result follows from Claim~\ref{claim:dist}, so we only need to prove the first result.  The intuition behind the proof will be to identify ``bad events'' where the density is too large and show that the random perturbation decreases the probability of bad events.  Invoking the Lov\'asz local lemma shows the existence of a perturbation where no bad event occurs.

The scaling step caused the complex to be $(2 t,\lambda)$-embedded.  Mark each simplex with one of $A' = O(1)$ labels so that simplices of the same label do not share any vertices.  Define $\Delta$ to be some set of 2-simplices with the same label.  Call $P_{k, \Delta}$ the probability that there exists a unit ball such that all 2-simplices in the set $\Delta$ intersect it after perturbation, and $|\Delta| = k$. Call $(k, \Delta)$ the corresponding event that this occurs.  For the event $(k, \Delta)$ to occur, all simplices in $\Delta$ must be contained in a ball of radius $\alpha \lambda$ (denoted $B(\alpha \lambda)$).  We now state the following claim; the proof is deferred to Appendix~\ref{app:extraemb}, where we provide bounds on the constants:
\begin{claim}
Given a fixed unit ball $B'(1) \subset B(\alpha \lambda)$, the probability that a 2-simplex in $B(\alpha \lambda)$ intersects $B'(1)$ after perturbation is $\leq \Gamma/\lambda$ for some constant $\Gamma > 0$.
\end{claim}

Since each 2-simplex in $\Delta$ is perturbed independently of the others, the probability that all simplices in $\Delta$ intersect $B'(1)$ is $P_{B'(1), \Delta} \leq \Gamma^k \lambda^{-k}$ by independence. Define  $\mathfrak{C}$  to be a  cover of $B(\alpha \lambda)$ with unit balls such that $|\mathfrak{C}| = \alpha^3 \lambda^3$. By a union bound
$$P_{k, \Delta} \leq \max_{\mathfrak{C} } \sum_{B'(1) \subset \mathfrak{C}} P_{B'(1), \Delta} \leq \alpha^3 \Gamma^k \lambda^{3-k}. $$

  Given that the preimage of any radius-$\lambda$ ball overlaps with $\leq 2 t$ simplices, the event $(k, \Delta)$ depends on $\leq (2 \alpha^3  t)^{k}$ events, and if
\begin{equation}
(2 \alpha^3  t)^{k} P_{k, \Delta} \leq \alpha^3 (2 \alpha^3 \Gamma t)^{k} \lambda^{3-k} \leq 1/4
\end{equation}
then there is a non-zero probability that $(k,\Delta)$ does not occur for any $\Delta$, by the Lov\'asz local lemma.  Thus by setting $k A' = t'$, we ensure that each unit ball overlaps with $\leq t'$ simplices;  this requires choosing $\lambda \geq (4\alpha^3)^{A'/(t'-3A')}( 2 \alpha^3 \Gamma t)^{t'/(t' - 3A')}$.  A numerical bound is provided in Sec.~\ref{sec:explicit-choice-constants}.  This implies that there exists an event where the complex is $(t',\ell)$-embedded in $\mathbb{R}^3$ where  $\ell =(2 \alpha +1)\lambda$.
\end{proof}
Even though the perturbation step only proves the existence of a good perturbation, there exists an algorithm running in expected polynomial time that finds a good perturbation.  This algorithm is simply the generalized Moser-Tardos algorithm (see Ref.~\cite{moser2010constructive,moser2009constructive}), which is generally applicable as long as some version of the Lovasz local lemma is satisfied.

\subsection{Step 2: replacement}

We present two equivalent formulations for studying replacement.
The first replaces each square in the Tanner square complex with a square complex,
  while the second replaces each vertex with a corresponding cochain complex.

The square-based formulation offers clearer intuition from the perspective of embeddings.
In contrast, the vertex-based formulation aligns with the code embedding framework in \Cref{sec:quantum-code-embedding-prelims}.
It is better suited for analyzing code properties,
  such as the number of encoded qubits,
  and for constructing chain maps between codes,
  which will play a central role in the construction of the decoder.

\subsubsection{Replacement in the Tanner square complex framework}

From here onward, we assume that $\ell$ is odd and will choose it to be the smallest odd integer larger than $(2\alpha + 1) \lambda$.

In the replacement step of iteration $\cR_X$ each square  $s \in T_{(i)}(2)$ is replaced by the square complex $U_X$ that is shown in Fig.~\ref{fig:replacement} (in the illustration, we set $\ell = 5$). Similarly, in the dual iteration $\cR_Z$, each square is replaced by the square complex $U_Z$ shown in Fig.~\ref{fig:replacement}.  Vertices of both $U_X$ and $U_Z$ are 3-colored as indicated in the caption.

\begin{figure} \label{fig:replacement}
    \includegraphics[width = 1\textwidth]{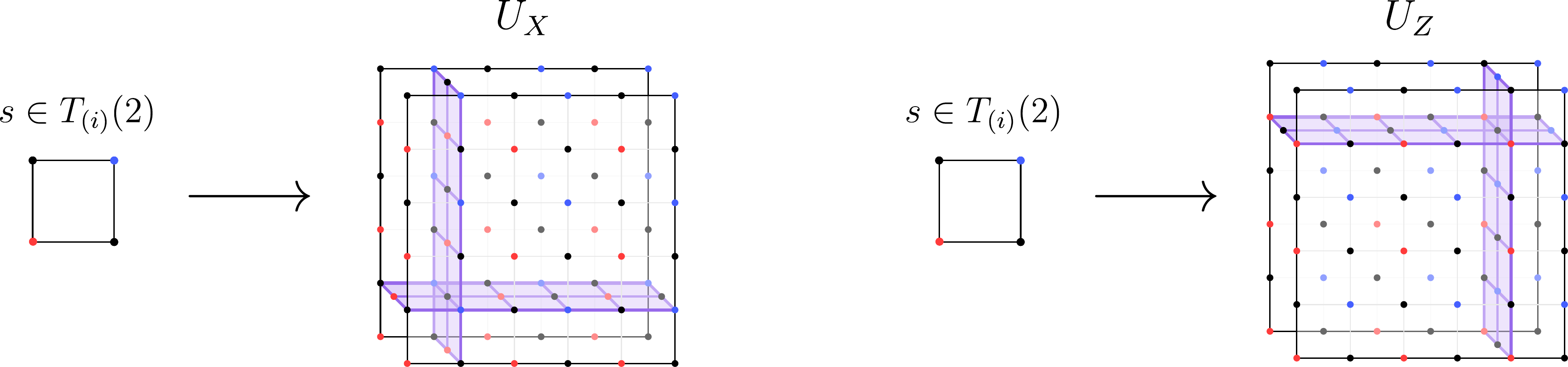}
    \caption{An illustration of the replacement process.  A square $s$ in the Tanner square complex $T_{(i)}$ is replaced with a local structure $U_X$ during the $\cR_X$ iteration, and with $U_Z$ during the $\cR_Z$ iteration.  The coloring convention is as follows: red dots denote $z \in V_{Z,(i)}$, black dots denote $q \in V_{Q,(i)}$, and blue dots denote $x \in V_{X,(i)}$.}
\end{figure}

Both complexes consist of two $\ell \times \ell$ layers along with a perpendicular square subcomplex of thickness 2 (shown in purple).  Furthermore, both complexes come with four 1-dimensional boundary subcomplexes, associated with four edges of the square being replaced. Adjacent boundary complexes meet at a pair of vertices, which form a corner subcomplex associated with corners of the square being replaced.

More formally, for each square $s \in T_{(i)}(2)$, let $U_{X,s}$ be the copy of $U_X$ that will replace $s$ in the $\cR_X$ iteration. The new Tanner square complex $T_{(i+1)}$ is obtained by gluing these copies according to incidence relations in $T_{(i)}$. In the $\cR_X$ iteration,
$$T_{(i+1)} = \left ( \bigsqcup_{s \in T_{(i)}(2)} U_{X,s} \right )/ \sim$$
where $\sim$ identifies the edges and vertices at common boundaries of neighboring structures $U_X$.  Namely, if two squares $s,s'\in T_{(i)}(2)$ are adjacent along an edge $e$, then we identify the boundary subcomplexes of $U_{X,s}$ and $U_{X,s'}$ labeled by $e$. Similarly, we identify all corner subcomplexes labeled by the same vertex $v \in T_{(i)}(0)$. The replacement step for $\cR_Z$ is defined identically, with $U_X$ instead replaced by $U_Z$.
By construction, $T_{(i+1)}$ is again a square complex. Each square of $T_{(i+1)}$ has two diagonally opposing black vertices in $V_{Q,(i+1)}$, one red vertex in $V_{Z,(i+1)}$, and one blue vertex in $V_{X,(i+1)}$.

\begin{remark}
    Note that the subcomplex $U_X$ resembles a piece of a thickness-1 3D surface code slab with Pauli-$X$ membrane operators and membrane-condensing boundary conditions in the front and back of the slab. Similarly, the subcomplex $U_Z$ resembles a piece of a thickness-1 3D surface code slab with Pauli-$Z$ membrane operators and membrane-condensing boundary conditions in the front and back of the slab. This illustrates the ``doubling'' effect of the $\cR_X$ and $\cR_Z$ iterations on the $X$-type and $Z$-type errors, respectively.
\end{remark}

For the new Tanner square complex $T_{(i+1)}$ after replacement, the embedding map $I_{(i+1)}:T_{(i+1)}\to\mathbb{R}^3$ can be obtained from the perturbed embedding $I'_{(i)}:T_{(i)}\to\mathbb{R}^3$. In the $\cR_X$ iteration, define a  ``flat'' realization of $U_X$, which we denote $ U_X'$. $U_X'$ is geometrically a projection of $U_X$ onto $\mathbb{R}^2$ along the doubling direction, and it has the same four boundary subcomplexes and the same four corner subcomplexes as $U_X$. As a result of this flattening, several cells can collapse onto one location.  This is illustrated below:

\begin{equation} \label{fig:replacement-flatten}
    \includegraphics[width = 1\textwidth]{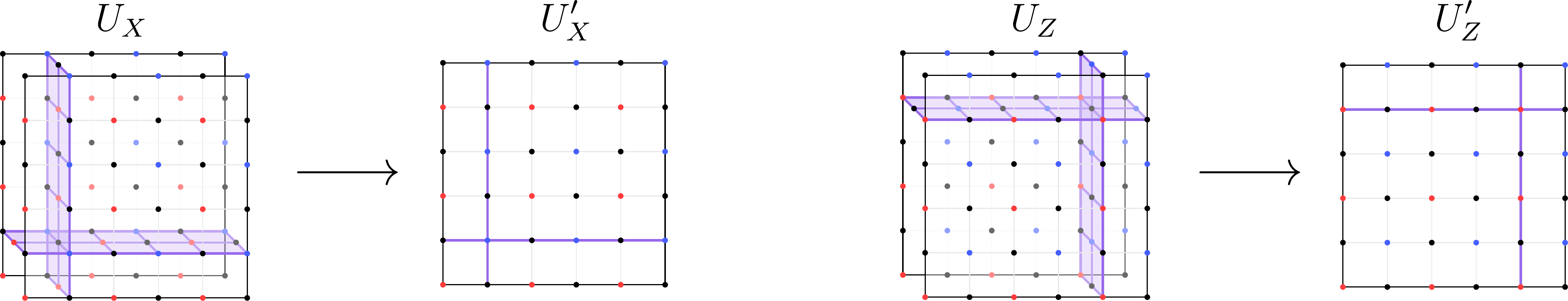}
\end{equation}
For example, the corner subcomplex of $U_{X}'$ is now entirely placed at the same point.

Now, we will describe how to obtain the map $I_{(i+1)}$. Recall that in the perturbation step, we formally added a diagonal to each square $s \in T_{(i)}(2)$, turning it into a pair of adjacent triangular simplices.
Thus, if an embedded square $I'_{(i)}(s)$ is not planar, it nevertheless consists of two embedded triangular simplices, where each simplex is trivially planar.
For each $s \in T_{(i)}(2)$, we add an associated diagonal to $U_{X,s}'$ that turns it into a pair of subdivided triangles. We can now define a piecewise-linear embedding map $$I'_{X,s}: U'_{X,s} \rightarrow \mathbb{R}^3$$ by sending the four corners of $U'_{X,s}$ to the locations of the the four vertices of $s$ under $I'_{(i)}(s)$. The images of the rest of the vertices, edges, and 2-cells in $U'_{X,s}$ are obtained by linear interpolation.  An example of the action of this local auxiliary map is shown below:

\begin{equation} \label{fig:replacement-fold}
    \includegraphics[width = 1\textwidth]{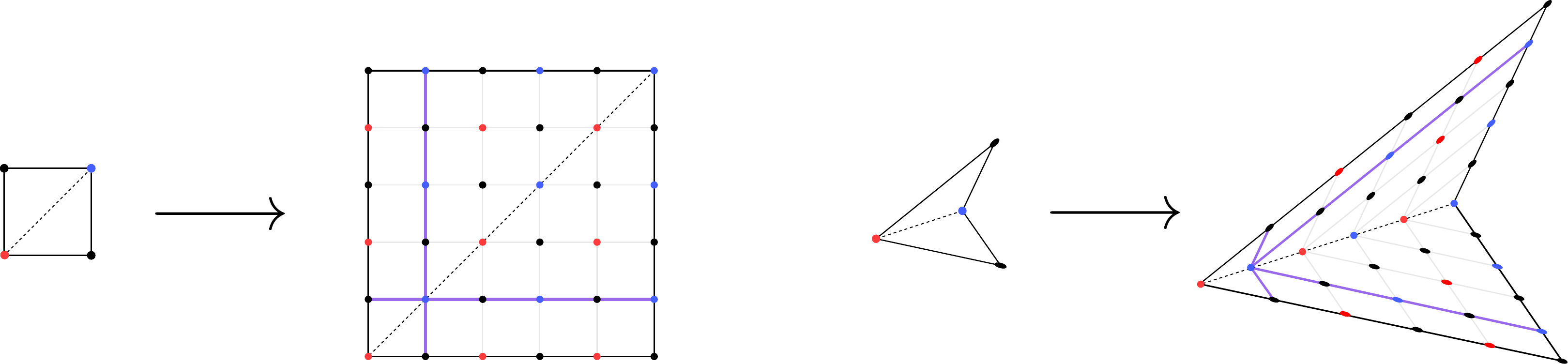}
\end{equation}
The maps $I'_{(i)}(s)$ are compatible on the boundary subcomplexes, since on every common edge and corner they are determined by the same embedded edge or vertex of $T_{(i)}$ under $I'_{(i)}$. Therefore, this defines a unique piecewise-linear map
$$I_{(i+1)}:T_{(i+1)}\to\mathbb{R}^3.$$
Finally, we remove the auxiliary diagonals. In subsec.~\ref{subsec:embedding-replacement}, we show that $I_{(i+1)}$ introduced here is a $(t,1)$-embedding.

\subsubsection{$\cR_X$ in the code embedding framework} \label{sec:embedding-tanner-graph}

We now give an alternative description of the replacement step via the code embedding formalism.
The main implication from this perspective is twofold:
  the number of logical qubits is preserved under the replacement,
\begin{equation}
  H^1(C_{(i)}) \cong H^1(C_{(i+1)})
\end{equation}
and there exists a chain map
\begin{equation}
  \cF_{\cR_X}: C_{(i)} \to C_{(i+1)}
\end{equation}
which will be proved in \Cref{sec:properties-Rx}.
The analogous statement for $\cR_Z$
  will be proved in \Cref{sec:properties-Rz}.

Let us first discuss replacement during the $\cR_X$ iteration. Recall that the vertices of the Tanner graph of the code $C_{(i)}: \ff_2^{V_X} \rightarrow \ff_2^{V_Q} \rightarrow \ff_2^{V_Z}$ can be decomposed as $T_{(i)}(0) = (V_X,V_Q,V_Z)$. We will associate a local cochain complex to each $v \in T_{(i)}(0)$. In the $\cR_X$ step, we define the local complexes $C_x: C_x^0 \to C_x^1 \to C_x^2$,
$C_q: 0 \to C_q^1 \to C_q^2$,
$C_z: 0 \to 0 \to C_z^2$,
which are indexed by $x \in V_X$, $q \in V_Q$, $z \in V_Z$. The code at the next level $C_{(i+1)} = \cR_X(C_{(i)})$ is an embedding complex as defined in Subsec.~\ref{sec:quantum-code-embedding-prelims} with the following structure:
\begin{equation} \label{eq:graded-complex-RX}
\begin{tikzpicture}[baseline]
\matrix(a)[matrix of math nodes, nodes in empty cells, nodes={minimum size=25pt},
row sep=2em, column sep=2em,
text height=1.25ex, text depth=0.25ex]
{&& \blue{C_X^0}  & \blue{C_X^1} & \blue{C_X^2}\\
& {0}  & {C_Q^1}  & {C_Q^2} &\\
\red{0} & \red{0} & \red{C_Z^2} &&\\};
\path[->,blue,font=\scriptsize]
(a-1-3) edge node[above]{$\delta_X^0$} (a-1-4)
(a-1-4) edge node[above]{$\delta_X^1$} (a-1-5);
\path[->,font=\scriptsize]
(a-2-2) edge node[above]{$\delta_Q^0$} (a-2-3)
(a-2-3) edge node[above]{$\delta_Q^1$} (a-2-4);
\path[->,red,font=\scriptsize]
(a-3-1) edge node[above]{$\delta_Z^0$} (a-3-2)
(a-3-2) edge node[above]{$\delta_Z^1$} (a-3-3);
\path[->,font=\scriptsize]
(a-1-3) edge node[right]{$g_{XQ}^0$} (a-2-3)
(a-1-4) edge node[right]{$g_{XQ}^1$} (a-2-4)
(a-2-2) edge node[right]{$g_{QZ}^0$} (a-3-2)
(a-2-3) edge node[right]{$g_{QZ}^1$} (a-3-3);
\end{tikzpicture}
\end{equation}
which differs from the discussion in Subsec.~\ref{sec:quantum-code-embedding-prelims} in that some vector spaces are zero and $g^{\bullet}_{XZ} = 0$. The full code $C_{(i+1)}: C^0_{(i+1)} \xrightarrow{\delta^0} C^1_{(i+1)} \xrightarrow{\delta^1} C^2_{(i+1)}$ is the total complex of the diagram above, with
\begin{equation}
    C^\bullet_{(i+1)} = C_X^\bullet \oplus C_Q^\bullet \oplus C_Z^\bullet
\end{equation}
Its coboundary maps are direct sums of local boundary maps and connecting maps, which we will  write as $\delta = \delta_{\mrm{loc}} + g$. The local part $\delta_{\mrm{loc}}$ consists of maps $\delta_x$, $\delta_q$, and $\delta_z$ acting within individual local complexes as in Eqn.~\eqref{eq:connecting-maps}. The connecting maps $g$ are defined between local complexes as in Eqn.~\eqref{eq:connecting-maps}.

The rest of this discussion formally defines the local complexes $C_x$, $C_q$, and $C_z$ as well as gluing maps in detail. On a first pass, the reader can skip these formal discussions and refer to Eqn.~\eqref{fig:replacement-full}.

For $x \in V_X$, the local complex $C_x: C_x^0 \rightarrow C_x^1 \rightarrow C_x^2$  is determined by the local 2D structure of the square complex $T_{(i)}$ in the vicinity of $x$.
Let $X_x$ be the 2-cell subcomplex of $T_{(i)}$ formed by 2-cells that contain $x$:
\begin{equation}
  X_x = \overline{\{f \in T_{(i)}(2): \  x \in \overline{f} \}}
\end{equation}
and let  $\partial X_x \subset X_x$ be the 1-cell complex that does not contain $x$:
\begin{equation}
  \partial X_x = \overline{\{ e \in X_x(1): \ x \not\in \overline{e}\}}
\end{equation}
where the $\overline {\{\ldots\}}$ notation stands for closure. In particular, if $X_x$ contains a square $f$, it also contains its boundary edges and vertices, and if $\partial X_x$ contains an edge $e$, it also contains its boundary vertices. Both $X_x$ and $\partial X_x$ are illustrated in \eqref{fig:replacement-Cx}.

We now construct a new local complex $Y_x$ from $X_x$ in two steps. First, subdivide each square in $X_x$ into $(\ell-1)/2 \times (\ell-1)/2$ squares (recall that $\ell$ is odd) to obtain square complex $X_x'$. Similarly, we subdivide each edge in $ \partial X_x$ into $(\ell-1)/2$ edges to obtain $\partial X'_x$.
Second, define $Y_x$ by gluing two copies of $X_x'$ along the cylinder on their common boundary:
$$Y_x = \left ( X'_x \times \{0\} \right )\sqcup \left ( X'_x \times \{ 1 \} \right) \sqcup \left ( \partial X'_x \times [0,1] \right )/ \sim,$$
where $\sim$ identifies the two copies of $\partial X'_x$ in $ \left ( X'_x \times \{0\} \right )\sqcup \left ( X'_x \times \{ 1 \} \right) $ with the two copies of $\partial X'_x$ in $\partial X'_x \times [0,1]$.
We now define $C_x$ to be the cellular cochain complex of $Y_x$, i.e. $C_x^\bullet := C^\bullet(Y_x)$. We illustrate an example below in \eqref{fig:replacement-Cx}, where we chose $\ell = 5$.

\begin{equation} \label{fig:replacement-Cx}
    \includegraphics[width = 0.9\textwidth]{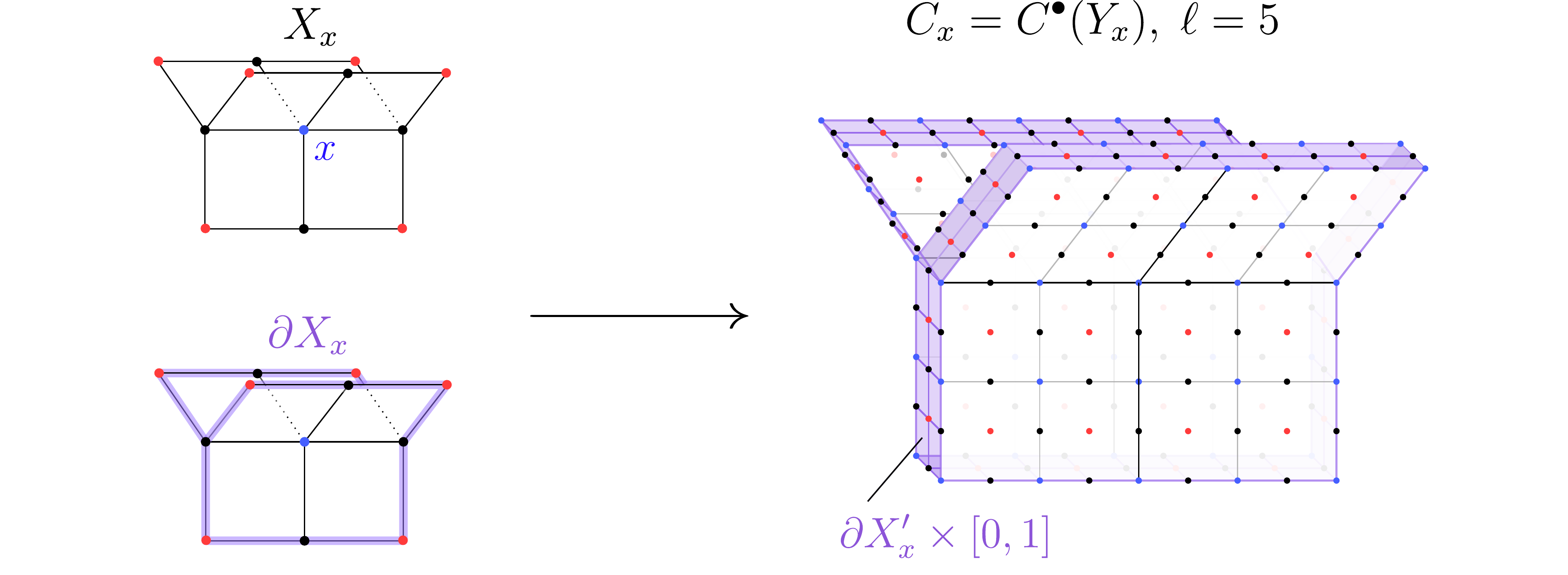}
\end{equation}

Similarly, $C_q: 0 \rightarrow C_q^1 \rightarrow C_q^2$ is induced by the local 1D structure of $T_{(i)}$ in the vicinity of $q$.  Let $X_q$ be the 1-cell subcomplex of $T_{(i)}$ formed by the 1-cells that connect the vertex $q \in V_Q$ with vertices in $V_Z$:

\begin{equation}
  X_q = \overline{\{e \in T_{(i)}(1): q \in \overline{e}, \ \exists z \in V_Z \text{ s.t. } z \in \overline{e} \}}.
\end{equation}
Let $\partial X_q$ be the 0-cell subcomplex of $X_q$ that does not contain $q$:
\begin{equation}
  \partial X_q = \{v \in X_q(0): q \neq v\}.
\end{equation}
Both $X_q$ and $\partial X_q$ are illustrated in \eqref{fig:replacement-Cq}.  We then similarly obtain $Y_q$ in two steps: first, we subdivide each edge of $X_q$ into $(\ell-1)/2$ edges, obtaining $X'_q$. Second, we set
$$Y_q = \left ( X'_q \times \{ 0 \} \right) \sqcup  \left (X'_q  \times \{1\} \right )\sqcup \left ( \partial X'_q \times [0,1] \right )/ \sim,$$
where $\sim$ identifies the two copies of $\partial X'_q$ in $\left ( X'_q \times \{ 0 \} \right) \sqcup  \left (X'_q  \times \{1\} \right )$ with the two copies of $\partial X'_q$ in $\partial X'_q \times [0,1]$.
Finally, we define $C_q$ to be the shifted cellular cochain complex of $Y_q$, i.e. $C_q^\bullet := C^\bullet(Y_q)[-1]$,
so that $C_q^m = C^{m-1}(Y_q)$.

\begin{equation} \label{fig:replacement-Cq}
    \includegraphics[width = 0.55\textwidth]{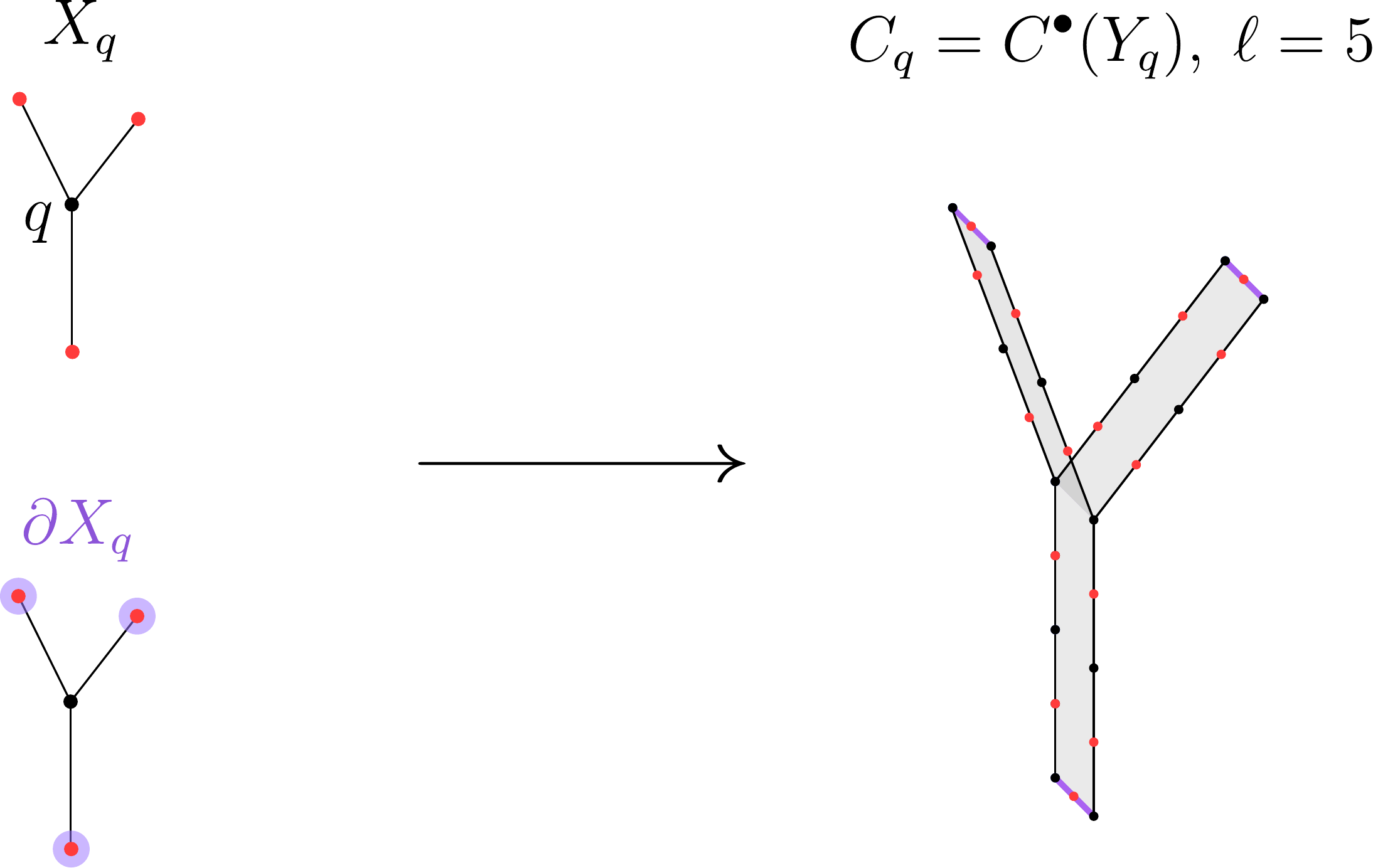}
\end{equation}

Finally, $C_z: 0 \rightarrow 0 \rightarrow C_z^2$ is defined to be two copies of the point $z$, which form a basis of the vector space $C_z^2$. More formally, we define a 0-dimensional square complex $Y_z$ consisting of two points and set $C_z^\bullet = C^\bullet(Y_z)[-2]$.

To fully define $C_{(i+1)}$, we also need to define the connecting maps $g_{XQ}$ and $g_{QZ}$.  Below, we define $g_{xq}$ and $g_{qz}$; these can then be used to define $g_{XQ}$ and $g_{QZ}$ via Eqn.~\eqref{eq:gluing-local-global}.

The map $g_{xq}^m: C^m_x \rightarrow C^{m+1}_q$ for $m = 0,1$ is only nonzero when $q \in V_Q$ and $x \in V_X$ are adjacent in $T_{(i)}$.  Consider an adjacent pair $(x,q)$, and let $X'_x$ and $X'_q$ be the associated cell complexes as defined above.  Note that $Y_x$ contains a one-dimensional subcomplex that is isomorphic to $Y_q$.
We therefore have the map
\begin{equation}
  Y_q \hookrightarrow Y_x,
\end{equation}
which induces a cochain map corresponding to restriction:
\begin{equation}
  C(Y_x) \rightarrow C(Y_q).
\end{equation}
Combined with the definitions $C_x \cong C(Y_x)$ and $C_q \cong C(Y_q)[-1]$, we obtain
\begin{equation}
  g_{xq}: C_x \cong C(Y_x) \rightarrow C(Y_q) \cong C_q[1].
\end{equation}

Next, we define $g_{qz}^m: C^m_q \rightarrow C^{m+1}_z$. We will only need this map for $m = 1$. It only acts nontrivially for adjacent $q \in V_Q$ and $z \in V_Z$ in $T_{(i)}$. Consider such an adjacent pair $(q,z)$. Note that $Y_q$ includes a 0-dimensional subcomplex that is isomorphic to $Y_z$.
The construction therefore gives a map
\begin{equation}
  Y_z \hookrightarrow Y_q,
\end{equation}
which induces a cochain map corresponding to restriction:
\begin{equation}
  C(Y_q) \rightarrow C(Y_z).
\end{equation}
Combined with the definitions $C_q \cong C(Y_q)[-1]$ and $C_z \cong C(Y_z)[-2]$, we obtain
\begin{equation}
  g_{qz}: C_q[1] \cong C(Y_q) \rightarrow C(Y_z) \cong C_z[2],
\end{equation}

The structure of the full complex $C_{(i+1)} = \cR_X(C_{(i)})$ obtained this way is shown below for one example of a local neighborhood. The specification of all local complexes and gluing maps is identical to the replacement procedure defined on the Tanner square complex in the previous subsubsection.

\begin{equation} \label{fig:replacement-full}
    \includegraphics[width = 0.6\textwidth]{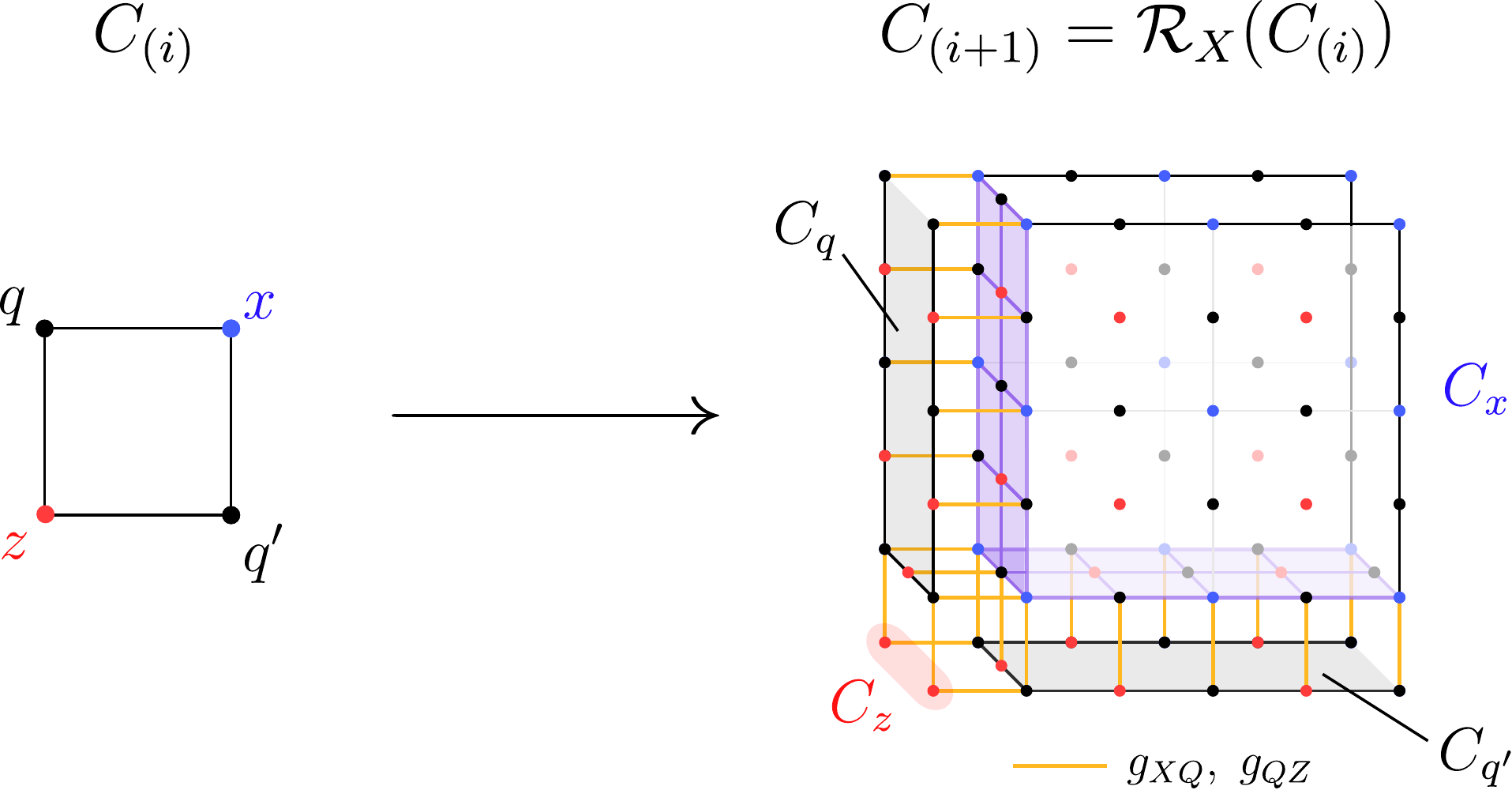}
\end{equation}

\subsubsection{The induced code $(H(C_{(i+1)},\delta_\mrm{loc}), H(g))$}

Given the construction of $C_{(i+1)}$ from local complexes,
  we apply the code embedding formalism from \Cref{subsec:spectral-sequence} to analyze its structure.
We begin by describing the induced complex $H(C_{(i+1)},\delta_{\mathrm{loc}})$
  and its induced coboundary maps $H(g)$,
  which will serve as the main tool for studying the properties of $C_{(i+1)}$
  in the next subsubsection.  The main lemma we will prove in this subsubsection is the following:

\begin{lemma} \label{lemma:HC-structure}
  The cochain complex $(H(C_{(i+1)}, \delta_\mrm{loc}), H(g))$ has the following structure:
  \begin{equation}\label{eq:RXcochainstructure}
    \begin{tikzcd}[column sep=3.5em, row sep=1em]
      {}
      &
      H^0(C_X,\delta_X)
        \arrow[r, "{H(g^0_{XQ})}"]
      &
      H^1(C_Q,\delta_Q)
        \arrow[r, "{H(g^1_{QZ})}"]
      &
      H^2(C_Z,\delta_Z) \oplus H^2(C_X,\delta_X) \oplus H^2(C_Q,\delta_Q)
      \\
      \hspace{-2em}\cong \hspace{-6em}
      &
      \ff_2^{V_{X,i}}
        \arrow[r, "{\delta^0_{(i)}}"]
      &
      \ff_2^{V_{Q,i}}
        \arrow[r, "{\Delta \circ \delta^1_{(i)} \;\oplus\; 0 \;\oplus\; 0}"]
      &
      (\ff_2^2)^{V_{Z,i}} \oplus H^2(C_X,\delta_X) \oplus H^2(C_Q,\delta_Q)
    \end{tikzcd}
  \end{equation}
  Here $\delta^0_{(i)}$ and $\delta^1_{(i)}$
    are the cochain maps in $C_{(i)}: \ff_2^{V_{X,i}} \xrightarrow{\delta^0_{(i)}} \ff_2^{V_{Q,i}} \xrightarrow{\delta^1_{(i)}} \ff_2^{V_{Z,i}}$
    and $\Delta : \ff_2^{V_{Z,i}} \to (\ff_2^2)^{V_{Z,i}} \cong
    \ff_2^{V_{Z,i}} \oplus \ff_2^{V_{Z,i}}$ is the diagonal map $b \mapsto (b,b)$.

  In the map $\Delta \circ \delta^1_{(i)} \;\oplus\; 0 \;\oplus\; 0$,
    the only nonzero component is $\Delta \circ \delta^1_{(i)}$ to $(\ff_2^2)^{V_{Z,i}}$.
  The components to $H^2(C_X,\delta_X)$ and $H^2(C_Q,\delta_Q)$ are zero maps.
\end{lemma}

The local complexes in $C_{(i+1)}$ (see Eqn.~\ref{eq:graded-complex-RX}) are labeled by the basis elements of vector spaces in $C_{(i)}:   \ff_2^{V_X} \rightarrow \ff_2^{V_Q}  \rightarrow \ff_2^{V_Z}$. Readers familiar with the quantum code embedding formalism might already see that the code ``embedded'' in $C_{(i+1)} = \cR_X(C_{(i)})$  must be related to the code in the previous iteration $C_{(i)}$. The ``embedded code'' is $(H(C_{(i+1)},\delta_\mrm{loc}), H(g))$; the lemma above implies that it contains a subcomplex of the form
\begin{equation} \label{eq:subcomplex-Ci}
  \ff_2^{V_X} \rightarrow
 \ff_2^{V_Q}  \rightarrow
  \left (\ff_2^2 \right)^{V_Z}
\end{equation}
This subcomplex almost looks like $C_{(i)}$ except that the boundary of $X$-type operators (1-cochains) is doubled, which is reflected in the last vector space $\left(\ff_2^2 \right)^{V_Z}$.  In the next subsubsection, we will argue that there is a chain map from $C_{(i)}$ to $H(C_{i+1}, \delta_{\mrm{loc}})$, and this will help us prove an isomorphism between the first cohomology groups of $C_{(i)}$ and $C_{(i+1)}$


\begin{proof}[Proof of Lemma~\ref{lemma:HC-structure}]
We start by determining the vector spaces and the maps in complex $H(C_{(i+1)},\delta_\mrm{loc})$. For this, we will first need the cohomology groups of $C_x$, $C_q$, and $C_z$ with respect to their local coboundary maps.
First, consider $C_x$. Note that it is the cellular cochain complex of $Y_x$. Further, observe that
$X'_x$ is topologically the cone of $\partial X'_x$ and $Y_x$ is topologically the suspension of $\partial X'_x$.
For reduced cohomology, suspension shifts the degree by one:
\begin{equation}
  \wt H^m(Y_x) \cong \wt H^{m-1}(\partial X'_x).
\end{equation}
In addition, we trivially have $\wt H^{m-1}(\partial X'_x) \cong \wt H^{m-1}(\partial X_x)$. For $m > 0$, the reduced cohomology and the usual cohomology agree,
and $H^0(Y_x) \cong \ff_2$ because $Y_x$ is connected. Therefore,
\begin{equation}
 H^0(Y_x) \cong \ff_2, \qquad
  H^1(Y_x) \cong \wt H^0(\partial X_x) = 0,
\end{equation}
where the second identity holds because $\partial X_x$ is connected.
This translates into the cohomological properties of $C_x$:
\begin{equation}
  H^0(C_x, \delta_x) \cong \ff_2, \qquad H^1(C_x, \delta_x) \cong 0.
\end{equation}
In addition, $H^2(C_x)$ is nontrivial, but we will not need to evaluate it.

Next, we consider $C_q$. Notice that
$X'_q$ is topologically the cone of $\partial X'_q$,
and $Y_q$ is topologically the suspension of $\partial X'_q$.
Thus, we have the same relation between the reduced cohomology of $Y_q$ and $\partial X'_q$ as we had before, i.e. $\wt H^m(Y_q) \cong \wt H^{m-1}(\partial X'_q)$.
$Y_q$ is still connected,
so $H^0(Y_q) \cong \ff_2$, but $\partial X_q'$ is not connected so $ \wt H^1(Y_q) \cong \wt H^0(\partial X'_q)$ will now be nontrivial.  However, $H^2(C_q) \cong H^1(Y_q)$, and we will not need to evaluate this for our purposes. Including the shift by one in the definition of $C_q$, we have:
\begin{equation}
  H^0(C_q, \delta_q) \cong 0, \qquad H^1(C_q, \delta_q) \cong \ff_2.
\end{equation}

Finally, the homological properties of $C_z$ are
\begin{equation}
  H^0(C_z, \delta_z) \cong 0, \qquad H^1(C_z, \delta_z) \cong 0, \qquad H^2(C_z, \delta_z) \cong \ff_2^2,
\end{equation}

and the complex $H(C_{(i+1)}, \delta_\mrm{loc})$ has the following structure:
\begin{equation}
\begin{tikzpicture}[baseline]

\def\longnodefont{\fontsize{10}{9.5}\selectfont} \def\arrowfont{\fontsize{7}{9.5}\selectfont}

\matrix (a) [
  matrix of math nodes,
  nodes in empty cells,
  row sep=2em,
  column sep=3.2em,
  nodes={minimum size=25pt, align=center, text width=2cm, font=\longnodefont}
] {
&& |[text width=2.6cm]| \blue{H^0(C_X,\delta_X) \cong \ff_2^{V_X}}
   & \blue{0}
   & \blue{H^2(C_X,\delta_X)} \\
& { 0}
  & |[text width=2.6cm]| {H^1(C_Q,\delta_Q) \cong \ff_2^{V_Q}}
  & {H^2(C_Q,\delta_Q)} & \\
\red{0}
  & \red{0}
  & |[text width=2.8cm]| \red{H^2(C_Z,\delta_Z) \cong (\ff_2^2)^{V_Z}}
  && \\
};

\path[->,font=\arrowfont]
(a-1-3.south) edge node[right]{$H(g_{XQ}^0)$} (a-2-3.north)
(a-1-4.south) edge node[right]{$H(g_{XQ}^1)$} (a-2-4.north)
(a-2-2.south) edge node[left]{$H(g_{QZ}^0)$} (a-3-2.north)
(a-2-3.south) edge node[left]{$H(g_{QZ}^1)$} (a-3-3.north);

\path[->,font=\scriptsize]
([xshift=-4pt]a-1-3.south) edge[dashed] ([xshift=4pt]a-3-2.north)
([xshift=-4pt]a-1-4.south) edge[dashed] ([xshift=4pt]a-3-3.north);
\end{tikzpicture}
\end{equation}

A direct verification shows that
$$H(g_{xq}^0): H^0(C_x, \delta_x) \to H^1(C_q,\delta_q)$$  maps $1 \mapsto 1$
for adjacent $x$ and $q$. To see this, notice that $H^0(C_x, \delta_x) \cong \ff_2$ is generated by the all-1 0-cochain on $C_x^0$, and $H^1(C_q, \delta_q) \cong \ff_2$ is generated by the all-1 1-cochain on $C_q^1$. For adjacent $x $ and $q$, $g_{xq}^0$ restricts $C_x$ to a subcomplex isomorphic to $C_q$, therefore mapping one generator to the other. Thus, the map $H(g_{xq}^0)$ is equivalent to the adjacency map between $V_X$ and $V_Q$ in $T_{(i)}$, i.e. the coboundary map $\delta^0_{(i)}$ of $C_{(i)}$.

Similarly, a direct verification shows that
$$H(g_{qz}^1): H^1(C_q, \delta_q) \to H^2(C_z, \delta_z)$$ maps $1 \mapsto (1,1)$
for adjacent $q$ and $z$. Consider again the all-1 1-cochain in $C_q$ that is a generator of $H^1(C_q, \delta_q) \cong \ff_2$. Restricting it to the subcomplex isomorphic to $C_z$ (again, with a degree shift by one), gives an all-1 2-cochain on $C_z^2$ when $q$ and $z$ are adjacent. This 2-cochain corresponds to the diagonal class $(1,1) \in H^2(C_z, \delta_z) \cong \ff_2^2$. Thus, the map  $H(g_{qz}^1)$ is the composition of the adjacency map between $V_Q$ and $V_Z$ in $T_{(i)}$ (i.e. the coboundary map $\delta^1_{(i)}$ of $C_{(i)}$) with the diagonal map $1 \mapsto (1,1)$.

Thus, the cochain complex $H(C_{(i+1)}, \delta_\mrm{loc})$ can be summarized by Eqn.~\ref{eq:RXcochainstructure}.

\end{proof}

\begin{remark}
  We see that the subcomplex in eq.~\eqref{eq:subcomplex-Ci} and the full complex $H(C_{(i+1)},\delta_\mrm{loc})$ agree up to the extra vector spaces $H^2(C_X, \delta_X)$ and $H^2(C_Q, \delta_Q)$. These capture relations between the extra   $Z$-type stabilizers added in
  $C_{(i+1)} = \cR_X(C_{(i)})$ in comparison to $C_{(i)}$.

  Additionally, the fact that we end up with two copies of $\ff_2^{V_Z}$ means that applying  $\cR_X$ results in doubling of the syndrome of $X$-type errors.
\end{remark}

\subsubsection{Properties of the code $C_{(i+1)} = \cR_X (C_{(i)})$}
\label{sec:properties-Rx}

We now analyze the properties of the code $C_{(i+1)}$ obtained from the replacement process.
\begin{lemma} \label{lemma:chain-map}
  For the replacement $C_{(i+1)} = \cR_X(C_{(i)})$, there exists a chain map
  \begin{equation}
    \cF_{\cR_X,i}: C_{(i)} \to C_{(i+1)}.
  \end{equation}

  Furthermore,
    $\cF_{\cR_X,i}$ induces an isomorphism on the first cohomology groups
  \begin{equation}
    H^1(C_{(i)}) \cong H^1(C_{(i+1)}).
  \end{equation}
  In particular, the replacement preserves the number of logical qubits.
\end{lemma}
\begin{proof}
  We prove the lemma by factoring the desired map through the induced complex
    $H(C_{(i+1)},\delta_{\mathrm{loc}})$.
  Specifically, we construct $\cF_{\cR_X,i}$ as a composition
  \begin{equation}
    C_{(i)}
    \hookrightarrow
    H(C_{(i+1)},\delta_{\mathrm{loc}})
    \xrightarrow{\simeq}
    C_{(i+1)}.
  \end{equation}
  We will show that both maps induce isomorphisms on first cohomology,
    which implies that $\cF_{\cR_X,i}$ also induces an isomorphism on first cohomology.

  We first define the chain map $C_{(i)} \hookrightarrow H(C_{(i+1)},\delta_\mrm{loc})$.
    The induced complex $H(C_{(i+1)},\delta_{\mrm{loc}})$ has the structure shown in Lemma~\ref{lemma:HC-structure}, and this leads to the following maps:
  \begin{equation}
    \begin{tikzcd}[column sep=5em]
      C_{(i)}:\hspace{-4em}
      &
      \ff_2^{V_{X,i}}
        \arrow[r, "{\delta^0_{(i)}}"]
        \arrow[d, hook]
      &
      \ff_2^{V_{Q,i}}
        \arrow[r, "{\delta^1_{(i)}}"]
        \arrow[d, hook]
      &
      \ff_2^{V_{Z,i}}
        \arrow[d, hook, "{\Delta \;\oplus\; 0 \;\oplus\; 0}"]
      \\
      H(C_{(i+1)}, \delta_\mrm{loc}):\hspace{-4em}
      &
      \ff_2^{V_{X,i}}
        \arrow[r, "{\delta^0_{(i)}}"]
      &
      \ff_2^{V_{Q,i}}
        \arrow[r, "{\Delta \circ \delta^1_{(i)} \;\oplus\; 0 \;\oplus\; 0}"]
      &
      (\ff_2^2)^{V_{Z,i}} \oplus H^2(C_X,\delta_X) \oplus H^2(C_Q,\delta_Q)
    \end{tikzcd}
  \end{equation}
  Recall that $\Delta$ is the diagonal map $\ff_2^{V_Z} \to (\ff_2^2)^{V_Z} \cong \ff_2^{V_Z} \oplus \ff_2^{V_Z}$, $b \mapsto (b,b)$,
    which is injective,
    so the last vertical map is also injective.
  It is clear that the above maps commute with the coboundary maps in both complexes,
    so we have defined a chain map $C_{(i)} \hookrightarrow H(C_{(i+1)},\delta_\mrm{loc})$.

  We now argue that this chain map induces an isomorphism on first cohomology.
  Since the degree-$0$ to degree-$1$ maps coincide in both complexes, we have
  \begin{equation}
    B^1(C_{(i)})
    = \operatorname{im} \delta^0_{(i)}
    = B^1(H(C_{(i+1)}, \delta_\mrm{loc})).
  \end{equation}
  Moreover, since $\Delta$ is injective,
  \begin{equation}
    Z^1(H(C_{(i+1)}, \delta_\mrm{loc}))
    = \ker (\Delta \circ \delta^1_{(i)})
    = \ker \delta^1_{(i)}
    = Z^1(C_{(i)}).
  \end{equation}
  It follows that
  \begin{equation}
    H^1(C_{(i)}, \delta_{(i)}) \cong H^1(H(C_{(i+1)}, \delta_\mrm{loc}), H(g)).
  \end{equation}

  We now define the chain map $H(C_{(i+1)}, \delta_\mrm{loc}) \xrightarrow{\simeq} C_{(i+1)}$.
  By the code embedding formalism and the spectral sequence analysis in \Cref{subsec:spectral-sequence},
    the complex $C_{(i+1)}$ can be realized as the total complex of the $E_0$ page of a spectral sequence.
  In particular, we identify
  \begin{equation}
    (\mrm{Tot}(E_0),D_0)
    =
    (C_{(i+1)},\delta_{(i+1)}),
  \end{equation}
  where $E^{p,0}_0 \cong C^p_X$, $E^{p,1}_0 \cong C^{p+1}_Q$, $E^{p,2}_0 \cong C^{p+2}_Z$, and $E^{p,q}_0 = 0$ otherwise.
  The horizontal maps are the local maps $\delta_X$, $\delta_Q$, and $\delta_Z$,
    while the vertical maps are the connecting maps $g_{XQ}$ and $g_{QZ}$.

  Similarly, the induced complex $H(C_{(i+1)}, \delta_\mrm{loc})$
    is the total complex of the $E_1$ page:
  \begin{equation}
    (\mrm{Tot}(E_1), D_1) = (H(C_{(i+1)}, \delta_\mrm{loc}), H(g)).
  \end{equation}
  By the general theory of spectral sequences, as discussed in \Cref{subsec:spectral-sequence},
    we have
  \begin{equation}
    (\mrm{Tot}(E_0),D_0) \simeq (\mrm{Tot}(E_1), D_1).
  \end{equation}
  Hence, this yields a chain homotopy equivalence
  \begin{equation}
    H(C_{(i+1)}, \delta_\mrm{loc}) \xrightarrow{\simeq} C_{(i+1)},
  \end{equation}
  which induces an isomorphism on first cohomology:
  \begin{equation}
    H^1(H(C_{(i+1)}, \delta_\mrm{loc}), H(g)) \cong  H^1(C_{(i+1)}, \delta_{(i+1)}).
  \end{equation}
  This completes the proof.
\end{proof}

\subsubsection{$\cR_Z$ in the code embedding formalism} \label{sec:RZ-embedding-formalism}

We now define replacement during the $Z$-iteration.  We will keep the discussion brief since the definition of $\cR_Z$ is, in a sense, dual to that of $\cR_X$.  In particular, the roles of $x$ and $z$ vertices will be effectively exchanged, and we will take the dual of certain intermediate cell complexes.

To each $v \in T_{(i)}(0)$,  associate a local cochain complex, which now has a different structure. The code at the next level $C_{(i+1)} = \cR_Z(C_{(i)})$ can be visualized as:
\begin{equation}
\begin{tikzpicture}[baseline]
\matrix(a)[matrix of math nodes, nodes in empty cells, nodes={minimum size=25pt},
row sep=2em, column sep=2em,
text height=1.25ex, text depth=0.25ex]
{&& \blue{C_X^0}  & \blue{0} & \blue{0}\\
& {C_Q^0}  & {C_Q^1}  & {0} &\\
\red{C_Z^0} & \red{C_Z^1} & \red{C_Z^2} &&\\};
\path[->,blue,font=\scriptsize]
(a-1-3) edge node[above]{$\delta_X^0$} (a-1-4)
(a-1-4) edge node[above]{$\delta_X^1$} (a-1-5);
\path[->,font=\scriptsize]
(a-2-2) edge node[above]{$\delta_Q^0$} (a-2-3)
(a-2-3) edge node[above]{$\delta_Q^1$} (a-2-4);
\path[->,red,font=\scriptsize]
(a-3-1) edge node[above]{$\delta_Z^0$} (a-3-2)
(a-3-2) edge node[above]{$\delta_Z^1$} (a-3-3);
\path[->,font=\scriptsize]
(a-1-3) edge node[right]{$g_{XQ}^0$} (a-2-3)
(a-1-4) edge node[right]{$g_{XQ}^1$} (a-2-4)
(a-2-2) edge node[right]{$g_{QZ}^0$} (a-3-2)
(a-2-3) edge node[right]{$g_{QZ}^1$} (a-3-3);
\end{tikzpicture}
\end{equation}
The local complex $C_z: C_z^0 \rightarrow C_z^1 \rightarrow C_z^2$ plays a similar role to that of $C_x$ in the $X$-iteration. It is determined by the local 2D structure of the square complex $T_{(i)}$ in the vicinity of $z \in V_Z$. We define auxiliary complexes $X_z$, $\partial X_z$, as well as subdivided $X'_z$, $\partial X_z'$ by analogy with the definition for $C_x$ during the $X$-iteration.  We then define
$$Y_z = \left (X'_z \times \{ 0 \} \right ) \sqcup \left (X'_z \times \{ 1 \} \right )  \sqcup \left ( \partial X'_z \times [0,1] \right )/ \sim.$$  We further define $C_z$ to be the dual cochain complex of $Y_z$, i.e. $C_z^\bullet = (C^\bullet(Y_z))^T$.  We illustrate an example below.

\begin{equation} \label{fig:replacement-Cz-RZ}
    \includegraphics[width = 0.75\textwidth]{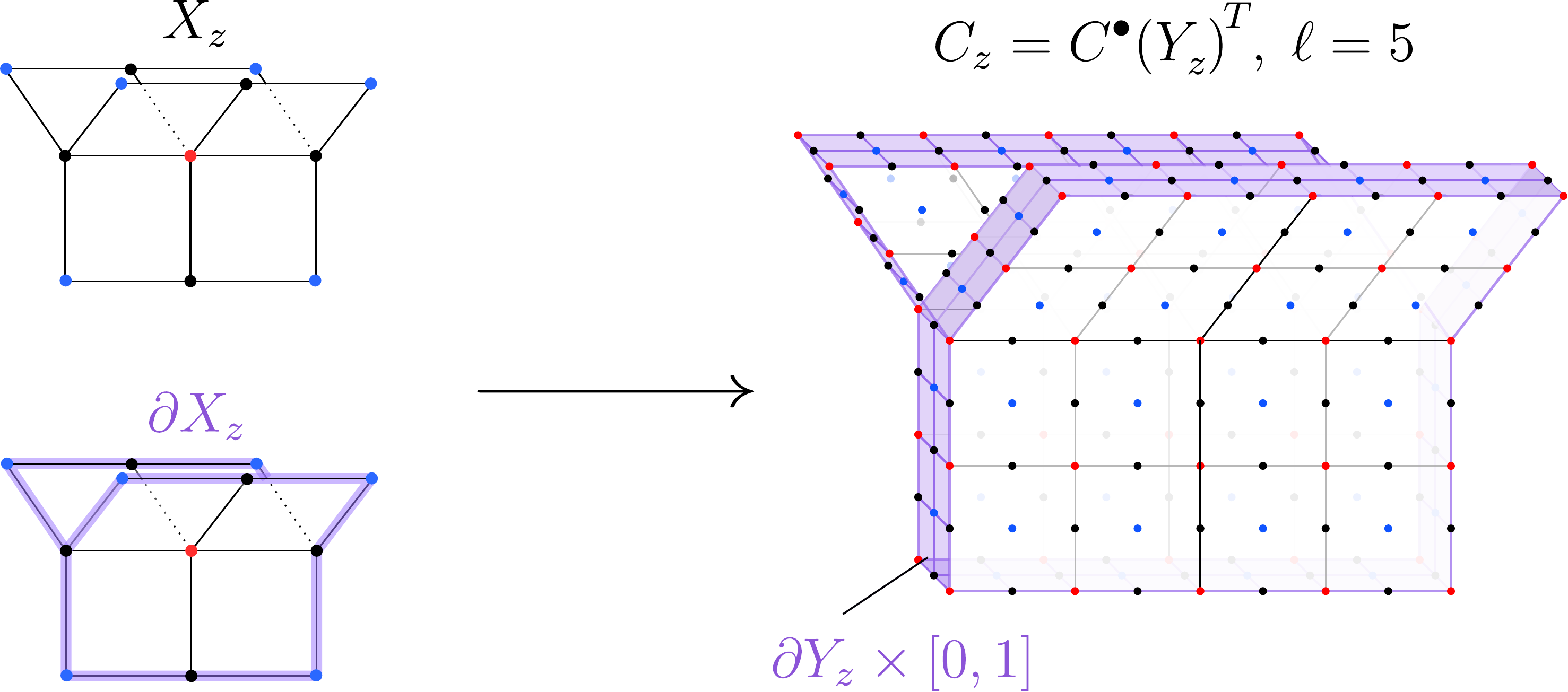}
\end{equation}

Similarly, the local complex $C_q: C_q^0  \rightarrow C_q^1 \rightarrow 0$ is induced by the local 1D structure of $T_{(i)}$ in the vicinity of $q$, where now we keep the edges connecting the vertex $q \in V_Q$ to blue vertices $x \in V_X$. We define $Y_q$ by analogy as before and set $C_q^\bullet = (C^\bullet (Y_q))^T$. Note that, unlike for $X$-iteration, there is no degree shift in the definition, as it is taken care of due to the transposition.
Finally, $C_x: C_x^0 \rightarrow 0 \rightarrow 0$ is defined to be two copies of the point $x$, and $C_x^\bullet = (C^\bullet (Y_x))^T$.

The gluing maps $g_{XQ}$ and $g_{QZ}$ are defined similarly to those in $X$-iteration. The structure of the full complex $C_{(i+1)} = \cR_Z(C_{(i)})$ obtained this way is shown below for one example of a local neighborhood.

\begin{equation} \label{fig:replacement-full-RZ}
    \includegraphics[width = 0.6\textwidth]{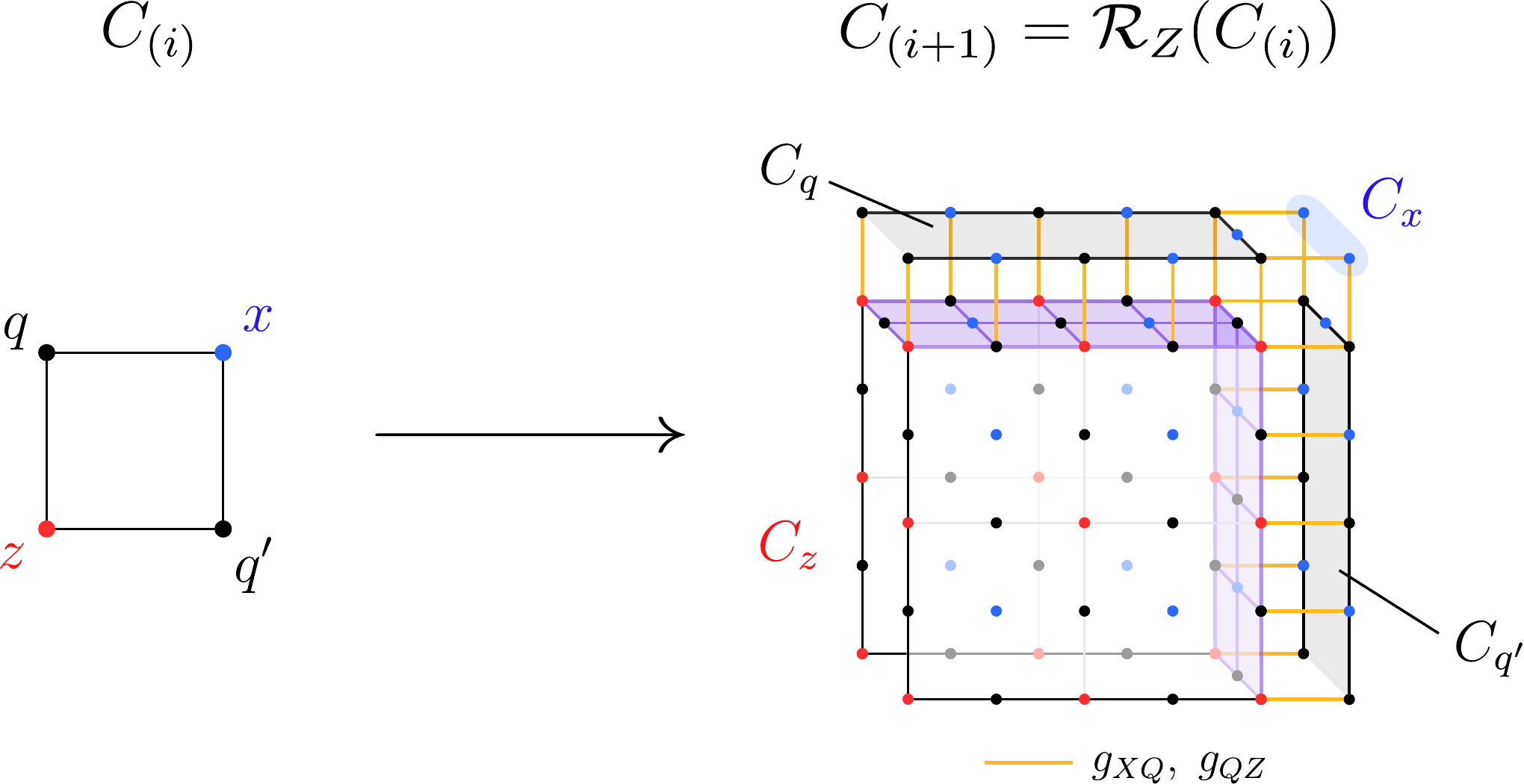}
\end{equation}

\subsubsection{Properties of the code $C_{(i+1)}= \cR_Z (C_{(i)})$}
\label{sec:properties-Rz}

As before, we first study the structural property of the cochain complex $\big(H(C_{(i+1)},\delta_\mrm{loc}), H(g)\big)$.

\begin{lemma} \label{lemma:HC-structure2}
  The cochain complex $(H(C_{(i+1)}, \delta_\mrm{loc}), H(g))$ has the following structure:
  \begin{equation}\label{eq:RZcochainstructure}
    \begin{tikzcd}[column sep=3.7em, row sep=1em]
    {}
    &
      H^0(C_Z,\delta_Z) \oplus H^0(C_X,\delta_X) \oplus H^0(C_Q,\delta_Q)
        \arrow[r, "{H(g^0_{XQ})}"]
      &
      H^1(C_Q,\delta_Q)
        \arrow[r, "{H(g^1_{QZ})}"]
      &
      H^2(C_Z,\delta_Z)
      \\
      \hspace{-2em}\cong \hspace{-6em}
      &
      (\ff_2^2)^{V_{X,i}} \oplus H^0(C_X,\delta_X) \oplus H^0(C_Q,\delta_Q)
        \arrow[r, "{\delta^0_{(i)} \circ \Delta'} \oplus \; 0 \; \oplus \; 0"]
      &
      \ff_2^{V_{Q,i}}
        \arrow[r, "{\delta^1_{(i)}}"]
      &
      \ff_2^{V_{Z,i}}
    \end{tikzcd}
  \end{equation}
  Here $\delta^0_{(i)}$ and $\delta^1_{(i)}$
    are the cochain maps in $C_{(i)}: \ff_2^{V_{X,i}} \xrightarrow{\delta^0_{(i)}} \ff_2^{V_{Q,i}} \xrightarrow{\delta^1_{(i)}} \ff_2^{V_{Z,i}}$
    and $\Delta' : (\ff_2^2)^{V_{X,i}} \cong
    \ff_2^{V_{X,i}} \oplus \ff_2^{V_{X,i}} \to \ff_2^{V_{X,i}}$ is the codiagonal map $(b,0) \mapsto b$, $(0,b) \mapsto b$.
\end{lemma}
Thus, the cochain complex $H(C_{(i+1)}, \delta_\mrm{loc})$ contains a subcomplex of the following form:
\begin{equation} \label{eq:subcomplex-Ci-RZ}
  \left (\ff_2^2 \right)^{V_X} \rightarrow
 \ff_2^{V_Q}  \rightarrow
 \ff_2^{V_Z}
\end{equation}
This subcomplex almost looks like $C_{(i)}$, and the part of this structure relevant to $X$-operators is identical to that in $C_{(i)}$. The doubling in the first vector space $\left (\ff_2^2 \right)^{V_X}$ reflects the fact that the boundary of $Z$-type errors, which are 1-chains in the dual chain complex, is doubled.

\begin{proof}[Proof of Lemma~\ref{lemma:HC-structure2}]
After computing all the local cohomology groups in a similar manner to the proof of Lemma~\ref{lemma:HC-structure}, we find that the complex $H(C_{(i+1)}, \delta_\mrm{loc})$ has the following structure:
\begin{equation}
\begin{tikzpicture}[baseline]

\def\longnodefont{\fontsize{10}{9.5}\selectfont} \def\arrowfont{\fontsize{7}{9.5}\selectfont}

\matrix (a) [
  matrix of math nodes,
  nodes in empty cells,
  row sep=2em,
  column sep=3.2em,
  nodes={minimum size=25pt, align=center, text width=2cm, font=\longnodefont}
] {
&& |[text width=2.6cm]| \blue{H^0(C_X,\delta_X) \cong (\ff_2^2)^{V_X}}
   & \blue{0}
   & \blue{0} \\
& { H^0(C_Q,\delta_Q)}
  & |[text width=2.6cm]| {H^1(C_Q,\delta_Q) \cong \ff_2^{V_Q}}
  & {0} & \\
\red{H^0(C_Z,\delta_Z)}
  & \red{0}
  & |[text width=2.8cm]| \red{H^2(C_Z,\delta_Z) \cong \ff_2^{V_Z}}
  && \\
};

\path[->,font=\arrowfont]
(a-1-3.south) edge node[right]{$H(g_{XQ}^0)$} (a-2-3.north)
(a-1-4.south) edge node[right]{$H(g_{XQ}^1)$} (a-2-4.north)
(a-2-2.south) edge node[left]{$H(g_{QZ}^0)$} (a-3-2.north)
(a-2-3.south) edge node[left]{$H(g_{QZ}^1)$} (a-3-3.north);

\path[->,font=\scriptsize]
([xshift=-4pt]a-1-3.south) edge[dashed] ([xshift=4pt]a-3-2.north)
([xshift=-4pt]a-1-4.south) edge[dashed] ([xshift=4pt]a-3-3.north);
\end{tikzpicture}
\end{equation}

A direct verification shows that
$$H(g_{xq}^0): H^0(C_x, \delta_x) \to H^1(C_q,\delta_q)$$  maps $(1,0) \mapsto 1$ and $(0,1) \mapsto 1$
for adjacent $x$ and $q$. To see this, notice that $H^1(C_q, \delta_q) \cong \ff_2$ is generated by the 1-cochain that is supported on a single (arbitrarily chosen) basis element in $C_q^1$. $H^0(C_x, \delta_x) \cong \ff_2^2$ is generated by $(1,0)$ and $(0,1)$, and each of these maps into a nontrivial representative of $H^1(C_q, \delta_q) \cong \ff_2$ under $H(g_{xq}^0)$ if $x$ and $q$ are adjacent. Thus, the map $H(g_{xq}^0)$ is the direct sum of two copies of the adjacency map between $V_X$ and $V_Q$ in $T_{(i)}$ (i.e. a direct sum of two copies of the coboundary map $\delta^0_{(i)}$ of $C_{(i)}$).

We also directly check that
$$H(g_{qz}^1): H^1(C_q, \delta_q) \to H^2(C_z, \delta_z)$$ maps $1 \mapsto 1$
for adjacent $q$ and $z$. Note that $H^2(C_z, \delta_z) \cong \ff_2$ is generated by the 2-cochain that is supported on a single (arbitrarily chosen) basis element in $C_z^2$.  For adjacent $q$ and $z$, the nontrivial generator of $H^1(C_q, \delta_q) \cong \ff_2$ in $C_q$ maps onto the nontrivial generator of $H^2(C_z, \delta_z) \cong \ff_2$ under $H(g_{qz}^1)$. Thus, $H(g_{qz}^1)$ is the adjacency map between $V_Q$ and $V_Z$ in $T_{(i)}$ (i.e. the coboundary map $\delta^1_{(i)}$ of $C_{(i)}$).

Thus, the cochain complex $H(C_{(i+1)}, \delta_\mrm{loc})$ can be summarized by Eqn.~\ref{eq:RZcochainstructure}.

\end{proof}

\begin{remark}
  We see that the subcomplex in eq.~\eqref{eq:subcomplex-Ci-RZ} and the full complex $H(C_{(i+1)},\delta_\mrm{loc})$ agree up to the extra vector spaces $H^0(C_Q, \delta_Q)$ and $H^0(C_Z, \delta_Z)$. These capture relations between the extra  $X$-type stabilizers added in
  $C_{(i+1)} = \cR_Z(C_{(i)})$ in comparison to $C_{(i)}$.

  Additionally, the fact that we end up with two copies of $\ff_2^{V_X}$ means that applying  $\cR_Z$ results in doubling of the syndrome of $Z$-type errors.
\end{remark}

\begin{lemma} \label{lemma:chain-map2}
  For the replacement $C_{(i+1)} = \cR_Z(C_{(i)})$, there exists a chain map
  \begin{equation}
    \cF_{\cR_Z,i}: C_{(i)} \to C_{(i+1)}.
  \end{equation}

  Furthermore,
    $\cF_{\cR_Z,i}$ induces an isomorphism on the first cohomology groups
  \begin{equation}
    H^1(C_{(i)}) \cong H^1(C_{(i+1)}).
  \end{equation}
  In particular, the replacement preserves the number of logical qubits.
\end{lemma}
\begin{proof}
Analogous to that of Lemma~\ref{lemma:chain-map}.
\end{proof}

\subsubsection{Embedding in $\mathbb{R}^3$ after replacement} \label{subsec:embedding-replacement}

Here, we will show that the embedding map $I_{(i+1)}: T_{(i+1)} \rightarrow \mathbb{R}^3$ defined earlier in subsection~\ref{sec:embedding-tanner-graph} is a $(t,1)$-embedding.
First, we establish a useful geometric property of our construction, namely, that the edge and vertex degrees of  $T_{(i+1)}$ remain bounded by a universal constant:

\begin{lemma}
    If $T_{(i)}$ has maximum vertex and edge degrees $z_{v}^{(i)} \leq 7$ and $z_e^{(i)} \leq 4$, then $T_{(i+1)}$ similarly has maximum vertex and edge degrees $z_v^{(i+1)} \leq 7$ and $z_e^{(i+1)} \leq 4$.
\end{lemma}
\begin{proof}
For the purpose of this proof, we will call $U_{X,s} \subseteq T_{(i+1)}$ minus all its boundary and corner subcomplexes the \emph{bulk subcomplex} of $U_{X,s}$.

Then, for any vertex in $T_{(i+1)}(0)$, the following holds:
\begin{enumerate}
    \item[(1)] $v \in T_{(i+1)}(0)$ is inside the bulk subcomplex associated with some $U_{X,s}$. Then, the maximum degree of this vertex is 6, which occurs at the crossing point of two subcomplexes labeled in purple in the figure below.
    \item[(2)] $v' \in T_{(i+1)}(0)$  is inside the boundary subcomplex labeled by some $e \in T_{(i)}(1)$. Any vertex in the boundary subcomplex neighbors a single vertex for each adjacent bulk subcomplexes along with at most 3 other vertices that belong to the boundary. This is illustrated in the figure below, where an adjacent vertex in the bulk is connected to vertex $v$ via an edge highlighted in red. The number of adjacent bulk subcomplexes equals $z_e^{(i)} \leq 4$. Thus, the degree of $v'$ is at most 7.
    \item[(3)] $v'' \in T_{(i+1)}(0)$ is inside the corner subcomplex labeled by some $w \in T_{(i)}(0)$. The connectivity of this vertex is the same as for $w$, and thus, has degree at most $z_v^{(i)} \leq 7$.
\end{enumerate}
Thus, we have directly verified that $z_v^{(i+1)} \leq 7$.

\begin{equation} \label{fig:degree-vertex}
    \includegraphics[width = 1\textwidth]{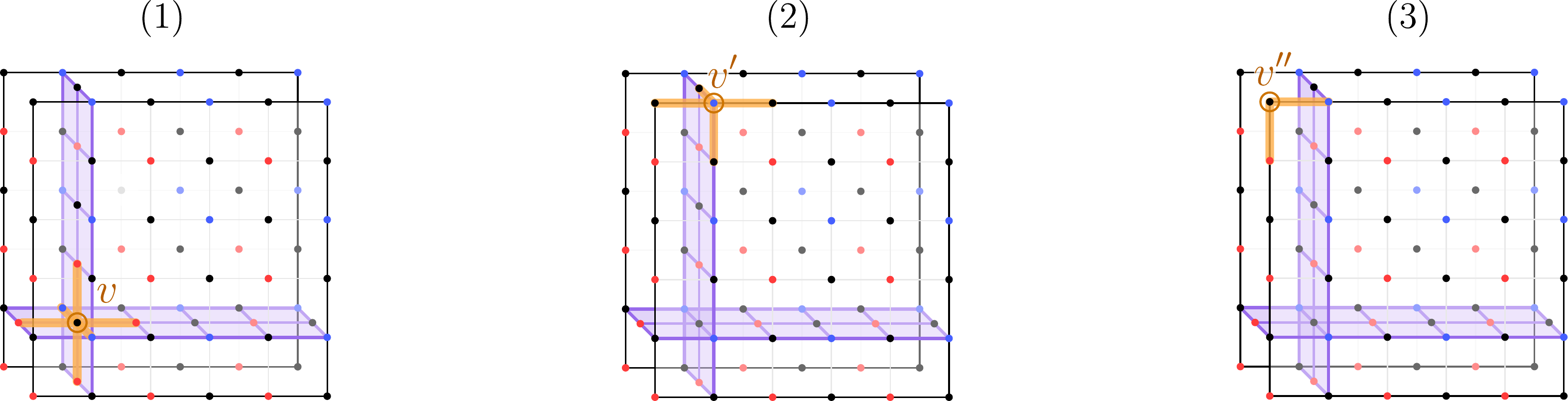}
\end{equation}

For any edge in $T_{(i+1)}(1)$, the following holds:
\begin{enumerate}
    \item[(1)]  $e \in T_{(i+1)}(1)$ is inside the bulk subcomplex associated with some $U_{X,s}$. The highest edge degree is 4, which is shown in the figure below.
    \item[(2)] $e' \in T_{(i+1)}(1)$  is inside the boundary subcomplex labeled by some $e \in T_{(i)}(1)$. Any edge in the boundary subcomplex has one incident square per adjacent  bulk subcomplex; the number of adjacent bulk subcomplexes is $z_e^{(i)} \leq 4$.
\end{enumerate}
Thus, we have directly verified that $z_e^{(i+1)} \leq 4$.

\begin{equation} \label{fig:degree-edge}
    \includegraphics[width = 0.7\textwidth]{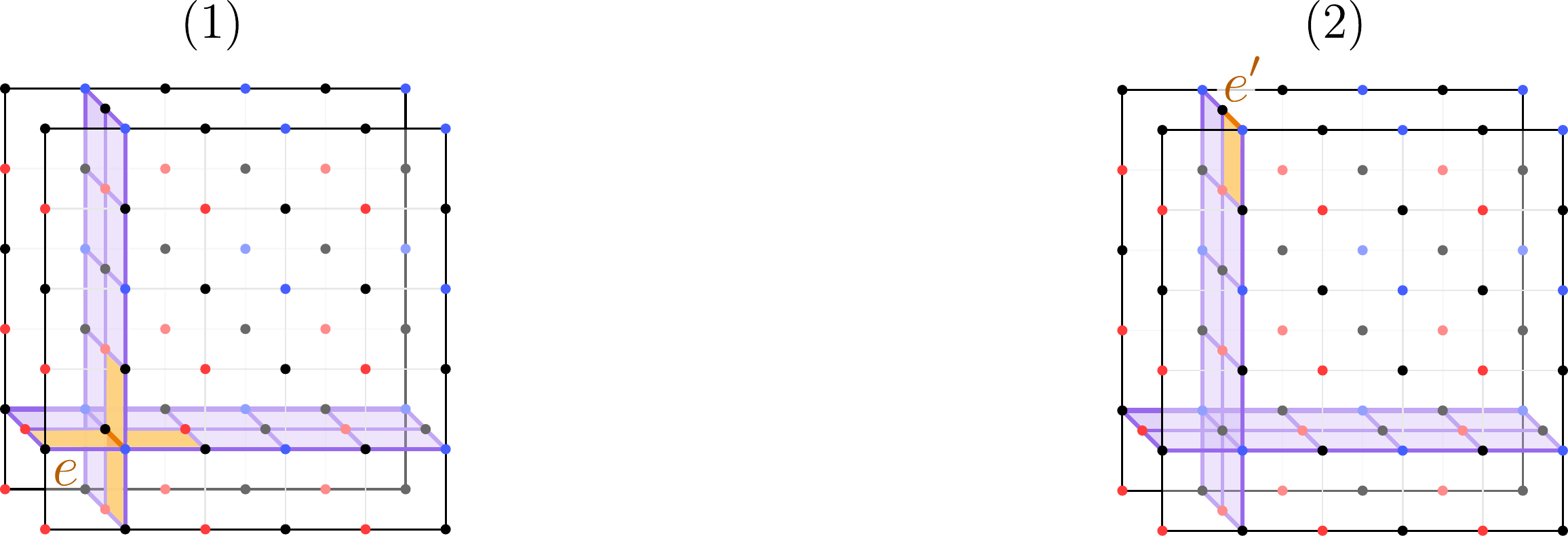}
\end{equation}

\end{proof}

\begin{lemma} \label{lemma:step2-embedding}
    Given $I'_{(i)}: T_{(i)} \rightarrow \mathbb{R}^3$ as defined in Subsec.~\ref{subsec:perturbation} and $I_{(i+1)}: T_{(i+1)} \rightarrow \mathbb{R}^3$ as defined in Subsec~\ref{sec:embedding-tanner-graph}, if
    $I'_{(i)}$ is a $(t', \ell)$-embedding, then $I_{(i+1)}$ is a $(t, 1)$-embedding where $t = 64 \pi t' (2 \alpha + 1)^2 / (\gamma (\alpha_2 -1))$.
\end{lemma}
\begin{proof}
Let us first recall several facts about the embedding map $I'_{(i)}: T_{(i)} \rightarrow \mathbb{R}^3$.
In lemma~\ref{lemma:embedding-main}, we showed that $I'_{(i)}$ is a $(t', \ell)$-embedding with $\ell = \ell(t,t')$ satisfying $$\ell (t,t') \geq (2 \alpha +1)\lambda \geq (4\alpha^3)^{A'/(t'-3A')}( 2 \alpha^3 \Gamma t)^{t'/(t' - 3A')}.$$ The constants $A'$, $\alpha$ and $\Gamma$ are $O(1)$ universal constants.

To obtain $(T_{(i+1)}, I_{(i+1)})$, we perform replacement where each square in $T_{(i)}$ is subdivided into a pair of (triangular) simplices and is replaced with a flat square complex $U'_X$ (or $U'_Z$ for the $\cR_Z$ iteration). We additionally subdivide some of the squares in $U'_X$ ($U'_Z$) into a pair of simplices as shown in \eqref{fig:replacement-fold}. Upon replacement, there will be two types of 2-cells, ones with nonzero area (these are the squares and triangles shown in \eqref{fig:replacement-fold} that are in-plane with the original 2-cells to be substituted) and the ones with zero area (these are the squares shaded purple in \eqref{fig:replacement-flatten}).  Let us first consider only 2-cells with nonzero area.

Since the sidelengths of the original triangular simplices are $\in [(\alpha_2-1) \lambda, (2\alpha+1)\lambda]$, the sidelengths of 2-cells with nonzero area upon replacement are $\in [\frac{\alpha_2-1 }{2 \alpha + 1}, 1]$. Since each of the original triangular simplices had width $\geq \gamma \lambda$, any diagonal added to the complex $U'_X$ ($U'_Z$) has length $\in [\frac{2\gamma}{2 \alpha+1}$,1]. Therefore, the 2-cells in $U'_X$ ($U'_Z$) have area  $\geq \frac{\gamma (\alpha_2 - 1)}{ (2 \alpha +1)^2}$. Since number of 2-cells that intersect in a unit ball is $\leq t'$ before replacement, the number of new 2-cells with non-zero area that intersect with a unit ball after replacement is at most
\begin{equation}
8\pi t'\frac{(2 \alpha+1)^2}{\gamma (\alpha_2 - 1)}.
\end{equation}

Now let us add the zero-area squares back into consideration. All sidelengths of these squares are $\leq 1$. The density of 2-cells intersecting a unit ball now increases by at most a factor of 8, and a unit ball intersects
\begin{equation}
t \leq 64\pi t' \frac{(2 \alpha+1)^2}{\gamma (\alpha_2 - 1)}
\end{equation}
simplices in total.
\end{proof}

\subsection{Main statements}

We will start with the base case. Consider a 2-dimensional cochain complex
\begin{equation}
C_{(0)}: C^0_{(0)} \overset{\delta_{(0)}^0 }{\longrightarrow} C^1_{(0)} \overset{\delta_{(0)}^1 }{\longrightarrow}  C^2_{(0)}
\end{equation}
corresponding to the $[\![ 5,1,2 ]\!] $ surface code that contains two $X$-type and two $Z$-type stabilizers, where:
\begin{equation}
    C^0_{(0)}  \cong \ff_2^2, \ \ C^1_{(0)}  \cong \ff_2^5, \ \  C^2_{(0)}  \cong \ff_2^2.
\end{equation}
Its associated Tanner square complex is shown below:
\begin{equation}
    \includegraphics[width = 0.15\textwidth]{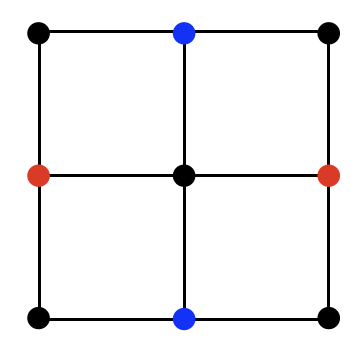}
\end{equation}

Starting from this base complex, we define the level-$2k$ code by alternating $X$- and $Z$-iterations:
\begin{equation}
C_{(2k)} := (\cR_Z \circ \cR_X)^{\circ k}(C_{(0)}).
\end{equation}
Let $T_{(2k)}$ be the associated Tanner square complex, and let $H_{(2k)}$ be the corresponding CSS stabilizer Hamiltonian after embedding $T_{(2k)}$ in $\mathbb{R}^3$. The size of the code at level $2k$ is $n_{2k} = \exp(\Theta(k))$. Because the base case had a single logical qubit, our code will always have one logical qubit by Lemmas~\ref{lemma:chain-map} and \ref{lemma:chain-map2}.

\begin{theorem}[3D self-correcting quantum memory] \label{thm:main}
    For the code $C_{(2k)}$ at any level $k \geq 0$, there exist constants $t$, $\beta_c < \infty$ and $\eta >0$ independent of $k$ such that:
    \begin{enumerate}
        \item [(1)] There exists a $(t,1)$-embedding map
            \begin{equation}
            I_{(2k)}: T_{(2k)} \rightarrow \mathbb{R}^3
            \end{equation}
        which can be constructed by a randomized algorithm in expected time $\poly(n_{2k})$.
        \item [(2)] The code $C_{(2k)}$ encodes a single logical qubit, equivalently $H^1(C_{(2k)}) \cong \ff_2$.
        \item [(3)] The pair $(T_{(2k)}, I_{(2k)})$ defines a local Hamiltonian $H_{(2k)}$ in $\mathbb{R}^3$.
        At $\beta > \beta_c$, the memory lifetime defined in Subsec.~\ref{sec:review-memorylifetime} satisfies
        \begin{equation}
            t_{\mathrm{mem}} \geq \exp( \Theta( n_{2k}^\eta))
        \end{equation}
        for some $\eta > 0$.
    \end{enumerate}
\end{theorem}

\begin{proof}

For our base case $C_{(0)}$, all three items are immediately satisfied.

We now prove statement (1) by induction. Assume $T_{(i)}$ admits a $(t,1)$-embedding $I_{(i)}$ for some $i$-independent constant $t$. The perturbation step produces a $(t',\ell)$-embedding by Lemma~\ref{lemma:embedding-main}. The replacement step then turns this into a $(t,1)$-embedding $I_{(i+1)}$ of $T_{(i+1)}$ by Lemma~\ref{lemma:step2-embedding}, after choosing the geometric constants as in Subsec.~\ref{sec:explicit-choice-constants}. Thus, every $T_{(2k)}$ admits a $(t,1)$-embedding with $t$ independent of $k$.

Statement (2) follows from the base case combined with Lemmas~\ref{lemma:chain-map} and \ref{lemma:chain-map2}.  Proving statement (3) will be the main goal of the next two sections.
\end{proof}

\subsection{Explicit choice of constants} \label{sec:explicit-choice-constants}

We now argue that there is a choice of all the constants that fit the constraints.  Our constants are very far from optimal as the goal is merely to show that the construction is well-defined, and we leave the task of optimizing these constants to the future.

First we note that $A' \leq 4z_v + 1$ and $A \leq 2 z_v + 1$.  By the induction hypothesis, $z_e = 4$ and $z_v = 7$ so $A' \leq 29$ and $A \leq 15$.  Select $\alpha = 1000$ as well as $\alpha_1 = 2\alpha_2 = 139$, which allows for an arrangement of $A$ caps on $B(\alpha \lambda)$ with enough separation.
We can compute
\begin{equation}
\gamma = \alpha\left(1-\cos\frac{\alpha_2}{\alpha}\right) - 1 \geq 1
\end{equation}
We can now compute the constant $\Gamma(\alpha, \alpha_1, \alpha_2, \gamma)$ determined in Appendix~\ref{app:extraemb}
\begin{equation}
\Gamma \leq \frac{18 (2\alpha+1)}{\alpha \gamma \left(1-\cos\frac{\alpha_1}{2\alpha}\right)}  \leq 14920
\end{equation}
Following Lemma~\ref{lemma:step2-embedding}, we find
\begin{equation}
\frac{t}{t'} \leq 64 \pi \frac{(2\alpha+1)^2}{\gamma  (\alpha_2 - 1)} \leq 1.18 \times 10^7
\end{equation}
Finally, in order to apply the Lovász local lemma, we need
\begin{equation}
\alpha^3 (2\alpha^3 \Gamma  t)^k \lambda^{3-k} \leq \alpha^3 ((6.84 \times 10^8) \alpha^3 \Gamma  k)^k \lambda^{3-k}  \leq 1/4
\end{equation}
where we use the upper bound for $t/t'$ as well as $t' \leq k A' \leq 29 k$.  The choice of $k$ which minimizes $\lambda$ gives $t' \leq 29k \leq  5595$ and $\lambda \approx 5.4 \times 10^{24}$.    Thus, with these parameters, we can guarantee that after each iteration, the construction is $(6.6 \times 10^{10}, 1)$-embedded in $\mathbb{R}^3$.  It is likely that one additional round of subdivision and scaling one can improve this to a $(t,1)$-embedding for $t$ close to 1; proving this will require some additional geometric properties of the random perturbation, which we leave to the future.  However, in this case the scaling factor $\ell$ will still need to be quite large.

\section{Decoder}
\label{sec:decoder}

In this section, we describe the decoder $\dec_X$ for the Pauli-$X$ errors and the decoder $\dec_Z$ for the Pauli-$Z$ errors for the code family $\{ C_{(2k)} \}$ constructed in the previous section.
These decoders will be used in the next section to analyze the memory lifetime of the code family.

The goal of the decoder $\dec_X$ for the code $C_{(2k)}$ is the following: given a syndrome $\sigma_X = \delta^1_{(2k)} e_X \in B^2_{(2k)}$ induced by an error $e_X \in C^1_{(2k)}$, find a correction $f_X \in C^1_{(2k)}$ such that
\begin{equation}
  \delta^1_{2k} f_X = \sigma_X.
\end{equation}
The correction need not coincide with the error. However, decoding is successful if the difference between the correction and the error is an $X$-stabilizer, i.e.
\begin{equation} \label{eq:decoder-correct}
  f_X - e_X \in B^1_{(2k)}.
\end{equation}
In this case, we say that the decoder successfully corrects the error.

The decoder $\dec_Z$ for Pauli-$Z$ errors is defined analogously using
the dual complex. Since the $X$ and $Z$ sectors can be analyzed
independently and behave equivalently in our construction, we focus
on just the $X$-decoder $\dec_X$ and suppress the superscript $X$.

As we will discuss in the next subsection,
  our decoder follows a renormalization group (RG)-like approach and satisfies two key properties: (1) syndrome reduction and (2) local computability.
Both conditions are essential for establishing the memory lifetime in the next section.

\subsection{Decoder via coarse-graining maps}
\label{sec:decoder-CG-maps}

We construct the decoder through a sequence of RG-like coarse-graining steps.   Given an initial syndrome
$\sigma_{2k}\triangleq \sigma\in B^2_{(2k)}$, the coarse-graining step at
level $i>0$ takes a syndrome $\sigma_i\in B^2_{(i)}$ as input and outputs
a level-$i$ correction $f_i\in C^1_{(i)}$ together with a coarser syndrome
$\sigma_{i-1}\in B^2_{(i-1)}$, satisfying
\begin{equation}\label{eq:coarse-graining-relation}
  \sigma_i
  =
  \delta^1_i f_i
  +
  \cF_{i-1}(\sigma_{i-1}).
\end{equation}
Here $\cF_{i-1}:C_{(i-1)}\to C_{(i)}$ is either
  the cochain map $\cF_{\cR_X,i-1}$ defined in \Cref{sec:properties-Rx}
  or the cochain map $\cF_{\cR_Z,i-1}$ defined in \Cref{sec:properties-Rz},
  depending on the replacement step at level $i-1$.
We suppress the replacement label for conciseness.

Thus, each coarse-graining step decomposes the syndrome $\sigma_i$ into a locally correctable part $\delta_i^1 f_i$ and a residual part $\cF_{i-1}(\sigma_{i-1})$ represented by the coarse-grained syndrome $\sigma_{i-1}$. Applying this recursively produces the pairs
\begin{equation}
  (\sigma_{2k-1},f_{2k}),
  (\sigma_{2k-2},f_{2k-1}),
...,
  (\sigma_0,f_1).
\end{equation}
At the base level $i=0$, the code has finite size, so the decoder directly outputs a correction $f_0 \in C^1_{(0)}$ such that
\begin{equation}
  \sigma_0 = \delta^1_0 f_0.
\end{equation}
Using these corrections $f_0, f_1,..., f_{2k}$, we construct a global correction $f$ for the original syndrome $\sigma$ by lifting each lower-level correction to level $2k$ and summing the contributions:
\begin{equation}
  f = f_{2k}
  + \cF_{2k-1}(f_{2k-1})
  + \cF_{2k-1} \circ \cF_{2k-2}(f_{2k-2})
  + \cdots
  + \cF_{2k-1} \circ \cdots \circ \cF_{0}(f_0).
\end{equation}
By direct calculation, one can verify that $\delta^1_{2k} f=\sigma$.

The full level-$2k$ decoder is obtained by composing the coarse-graining steps
\begin{equation}
  \sigma_i \mapsto (\sigma_{i-1},f_i),
\end{equation}
together with the base-level decoder for $\sigma_0$. The remainder of this section constructs these coarse-graining steps explicitly.

Before describing the coarse-graining steps, we discuss two additional properties that the coarse-graining map must satisfy: (1) syndrome reduction and (2) local computability. These two properties will play a central role in the proof of memory lifetime.

\subsubsection{Syndrome reduction}
\label{sec:syndrome-reduction}

In the proof of the memory lifetime bound, we compare two competing effects: the Boltzmann suppression coming from the energy of a syndrome, and the entropic/combinatorial contribution coming from the number of syndrome configurations of a certain energy. Syndrome reduction controls the Boltzmann factor, while local computability controls the entropic factor.

In particular, we need the syndrome weight to decrease sufficiently under coarse-graining. Equivalently, if a syndrome persists through many coarse-graining steps, then the original syndrome at level $2k$ must have had a large weight.

Specifically, we will prove the following two bounds:
\begin{itemize}
\item for $\cR_Z$ iterations, the coarse-graining step satisfies $|\sigma_{i-1}| \le |\sigma_i|$;
\item for $\cR_X$ iterations, the coarse-graining step satisfies $|\sigma_{i-1}| \le \frac{1}{2}|\sigma_i|$.
\end{itemize}

The implication of these results is that the syndrome weight decreases throughout the coarse-graining process: $\cR_Z$ steps are weight non-increasing, while $\cR_X$ steps are weight-contracting.

\subsubsection{Local computability and the decoding graph}

To control the entropic factor, we want both the coarse-grained syndrome $\sigma_{i-1}$ and the correction $f_i$ at a given location to depend only on the values of $\sigma_i$ in a constant-size neighborhood of this location.
Because of \Cref{eq:coarse-graining-relation}, this controls the number of syndrome configurations $\sigma_i$ that can be directly responsible for a given syndrome configuration $\sigma_{i-1}$ after coarse-graining.

The local computability definition will rely on the notion of a ``decoding graph'', which we use to capture the history of the coarse-graining steps applied by the decoder. The nodes of the decoding graph index the local complexes where either the syndrome or the corrections at each level can be supported. The simplest definition of local computability uses a ``directed'' graph distance on the decoding graph as a measure of locality. While we could have used spatial locality instead, graph locality will be more convenient for the proofs in Sec.~\ref{sec:memlifetime}.

The decoding graph $G$ is defined to be a union of subgraphs $G_i$, one for each level $i$ of the construction. Because the local complexes $C_{(i), x}$, $C_{(i), q}$, and $C_{(i), z}$ at level $i$  are labeled by the vertices of the Tanner square complex at level $i-1$, i.e. $x \in V_{X,(i-1)}$, $q \in V_{Q,(i-1)}$ and $z \in V_{Z,(i-1)}$, $G_i$ is isomorphic to the Tanner square complex at level $i-1$.
Nodes $v$ in $G_i$ and $v'$ in $G_{i-1}$ are connected if the $v$ is in the local complex $Y_{(i-1),v'}$.
A schematic example of a decoding graph is shown in Fig.~\ref{fig:subgraph}.  A more formal definition is given below.
\begin{definition}[Decoding graph]
\label{def:decoding-graph}
For $1\leq i\leq 2k$, define $G_i = (V_i, E_i)$ to be a graph isomorphic to the 1-skeleton of $T_{(i-1)}$. Label vertices in $G_i$ by vertices in $T_{(i-1)}(0)$ and edges in $G_i$ by edges in $T_{(i-1)}(1)$.
Define $G_0$ to be a single vertex, labeled $v^{(0)}$.

The decoding graph $G = (V,E)$ contains a disjoint union $\bigcup_{i=0}^{2k} G_i$ as well as additional edges connecting consecutive levels $G_{i-1}$ and $G_i$.
A node $v \in G_{i-1}$ is connected to a node $v' \in G_i$ if the latter satisfies $v' \in Y_{(i-1),v}(0)$.
Finally, the root node $v^{(0)}$ is connected to all nodes in $G_1$.
\end{definition}
We will now define local computability with respect to a ``directed'' graph distance on $G$.
The reason for using a directed distance, rather than the usual graph distance, is that the coarse-graining map has a natural directed dependence: the correction or the coarse-grained syndrome in one local region may depend only on the input syndrome in regions that can reach it along directed paths in $G$.

\begin{definition}[Directed distance]
    Define a directed structure on $G$ at each level $G_i$ in the following way. The edge $(u,v) \in E_i$ is directed from $u$ to $v$ if the corresponding vertices in the Tanner square complex $T_{(i-1)}(0)$ satisfy $u \in V_{X, (i-1)}$ and $v \in V_{Q, (i-1)}$ or $u \in V_{Q, (i-1)}$ and $v \in V_{Z, (i-1)}$.  Define $\mrm{dist}_i(u,v)$ to be the smallest length of the directed path from $u$ to $v$ in $G_i$. We will call $\mrm{dist}_i(u,v)$ the directed distance between $u$ and $v$.  If no directed path exists, $\mrm{dist}_i(u,v) = \infty$.
\end{definition}

\begin{definition}[Locally computable] \label{def:locality}
For node $v \in T_{(i-1)}(0)$, define
\begin{equation}
\xi_{i, v}(m) = \sum_{\substack{v' \in T_{(i-1)}(0): \\ \mrm{dist}_i(v',v) \leq m }} \sigma_i \big|_{Y_v}.
\end{equation}
Given a coarse-graining step at level $i$ where $\sigma_i \mapsto (\sigma_{i-1}, f_i)$, we say that the correction $f_i$ is locally computable if it decomposes into a sum of local components, i.e.
    \begin{equation}
        f_i = \sum_{v \in T_{(i-1)}(0)} f_{i,v},
    \end{equation}
such that  $f_i\big|_{Y_v} = f_{i,v}$, and each local component $f_{i,v}$ depends only on the the values of the syndrome $\sigma_i$ in the local complexes $C_{v'}$ for which $\mrm{dist}_i(v',v) \leq 1$. More formally,
\begin{equation}
    f_{i,v}(\sigma_i) = f_{i,v}(\xi_{i, v}(1)).
\end{equation}
We say that the coarse-grained syndrome $\sigma_{i-1}$ is locally computable if its value $\sigma_{i-1,z} = \sigma_{i-1}\big|_{z}$ at location $z \in V_{Z,(i-1)}$ only depends on the values of syndrome $\sigma_i$ in the local complexes $C_{v'}$ for which $\mrm{dist}_i(v',v) \leq 2$. More formally,
\begin{equation}
    \sigma_{i-1,z}(\sigma_i) = \sigma_{i-1,z}(\xi_{i, v}(2)).
\end{equation}
\end{definition}

\begin{figure}
\centering
\includegraphics[width=0.75\textwidth]{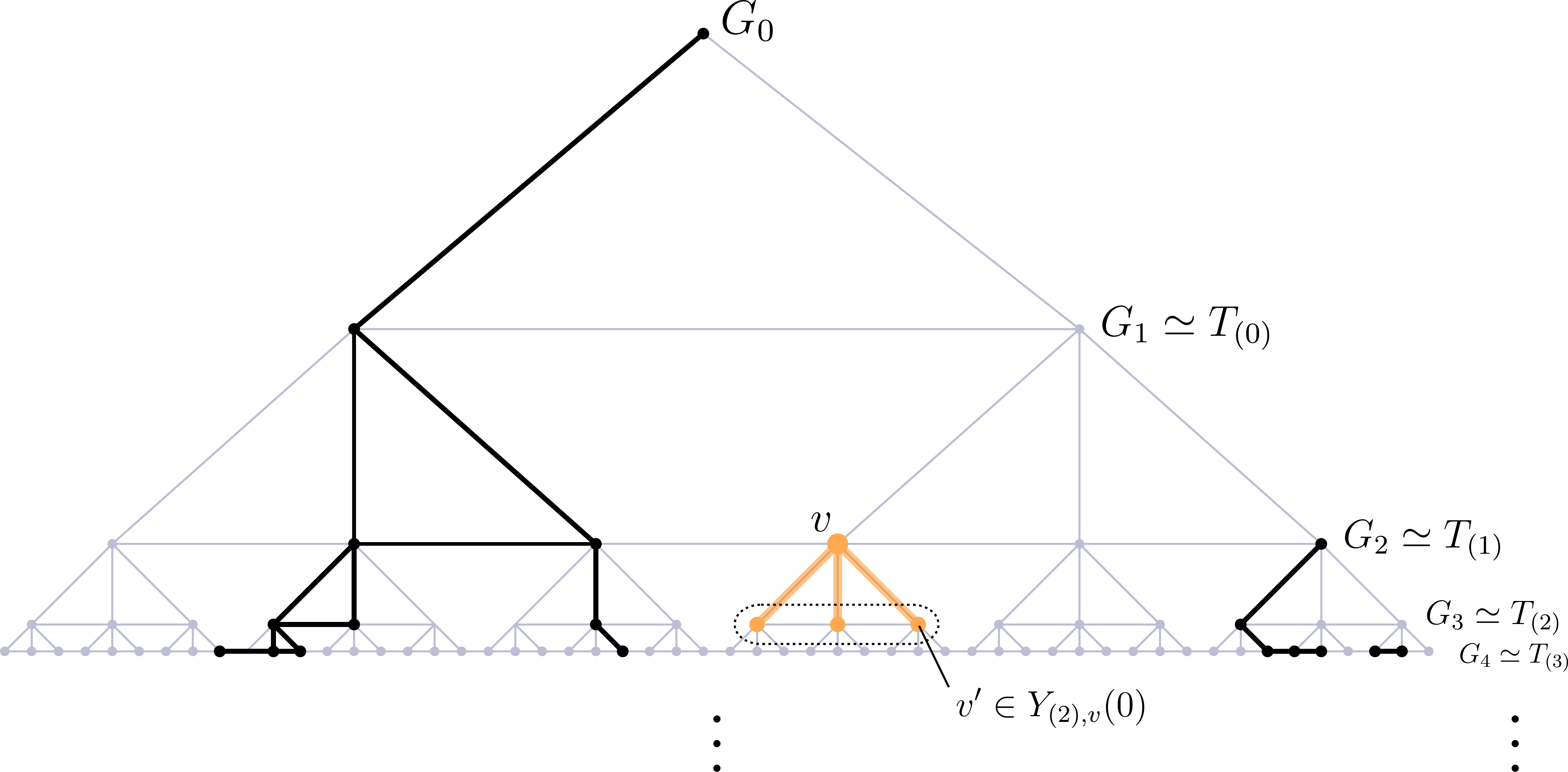}
\caption{A schematic illustration of the nodes and edges in the decoding graph. It contains subgraphs at each level $G_i \cong T_{(i-1)}$. The subgraph at level 0 $G_0$ consists of a single node. The nodes at different levels are connected; the node labeled by $v$ at level $i$ is connected to the node $v'$ at level $i+1$ if $v' \in Y_{(i-1),v}(0)$. There is a single node at level 0, which is connected to all nodes at level 1.
The marking of nodes and edges in black will be discussed in more detail in Sec.~\ref{sec:memlifetime}.  At a high level, the marked nodes and edges illustrate the history of coarse-graining steps starting from some syndrome $\sigma_{2k}$. A node $v \in G_i \simeq T_{(i-1)}$ is marked black if $C_{(i),v}$ contains the support of either syndrome $\sigma_i$, or the coboundary of any local correction $f_{i,s}$ for some $w \in T_{(i-1)}(0)$. An edge between two marked nodes is also marked. }
\label{fig:subgraph}
\end{figure}

\begin{remark}
    Because of the condition $\sigma_i = \delta^1_i f_i + \cF_{i-1}(\sigma_{i-1})$, if $f_i$ is locally computable, $\sigma_{i-1}$ is also locally computable. This is because the syndrome points in $\cF_{i-1}(\sigma_{i-1})$ must be either in $\sigma_i$ (which is graph-distance 0 away) or in the coboundary of some local correction $f_{i,v}$. In the latter case, the correction $f_{i,v}$ depends on the syndrome within graph-distance 1 away, and taking the coboundary of the correction increases graph-distance by at most 1 because $\delta^1 = \delta^1_\loc + g^1$.
\end{remark}
\begin{remark}
The use of the directed distance implies the following dependencies.  We start with subgraphs associated with the $\cR_X$ iteration.  Any correction $f_{i,x}$ will only depend on $\sigma_i\big|_{Y_x}$ and $f_{i,q}$ will only depend on
$$\sigma_i\big|_{Y_q} + \sum_{x \sim q} \sigma_i\big|_{Y_x}.$$
Any syndrome $\sigma_{i-1,z}$ will only depend on
$$\sigma_i\big|_{Y_z} + \sum_{q \sim z} \sigma_i\big|_{Y_q} + \sum_{\substack{ x \sim q \\ q \sim z }} \sigma_i\big|_{Y_x}.$$
The syndrome and correction will not have support in any of the other regions.  For subgraphs associated with the $\cR_Z$ iteration, $f_{i,z}$ will only depend on $\sigma_i\big|_{Y_z}$ and $\sigma_{i-1}\big|_{Y_z}$ will only depend on $\sigma_i\big|_{Y_z}$.  The syndrome and correction will not have support in any of the other regions.
\end{remark}

\subsection{Coarse-graining for $\cR_Z$}
\label{sec:coarse-graining-RZ}

We first describe the coarse-graining step at levels corresponding to the $\cR_Z$ iteration. This is simpler than the coarse-graining step for the $\cR_X$ iteration.
\begin{theorem} \label{thm:coarse-graining-RZ}
Given a level-$i$ syndrome $\sigma \in B^2_{(i)}$, there exist a locally computable coarse-graining map $\sigma \mapsto (\sigma', f)$ where the coarse-grained syndrome $\sigma' \in B^2_{(i-1)}$ and correction $f \in C^1_{(i)}$ satisfy
  \begin{equation} \label{eq:thm-RZ-CG}
    \sigma = \delta^1_{i} f + \cF_{\cR_Z, i-1}(\sigma').
  \end{equation}
Moreover,
  \begin{equation}
    |\sigma'| \le |\sigma|.
  \end{equation}
\end{theorem}
Moving forward, we will drop the subscript $i$ in the notation for the local complexes; for example, we will write $C_z$ instead of $C_{(i),z}$.

The key idea is that the syndrome can be coarse-grained by considering configurations locally.
Each local part of the syndrome is contained in a local subcomplex $C_z = (C^\bullet(Y_z))^T$ for some $z \in V_{Z,(i-1)}$, and can be corrected independently of the others.
In each $C_z$, we choose a correction supported only on $C_z$ that brings the local syndrome to the image of $\cF_{\cR_Z,i-1}$, which consists of a single point inside $C_z$.
Geometrically, this correction is a sum of open strings on $Y_z$ that bring all syndromes inside $Y_z$ to the coarse-grained point.
This correction is local, and the resulting syndrome weight cannot increase: each local cluster contributes at most one syndrome point after coarse-graining.

We now formalize this argument.
In the remainder of this subsection, we suppress the subscript $\cR_Z$ in $\cF_{\cR_Z, i-1}$ for notational simplicity.

\begin{proof}[Proof of Theorem~\ref{thm:coarse-graining-RZ}]
We split the proof into four parts:

\vspace{5 pt}
\noindent 1. \textbf{Defining the coarse-grained syndrome $\sigma'$:}
Recall that $C^2_{(i)} \cong C_Z^2 = \bigoplus_z C_z^2$, and $C_z \triangleq (C^\bullet(Y_z))^T$. Thus, the syndrome $\sigma \in C^2_{(i)}$ decomposes as
\begin{equation} \label{eq:sigma-sum-RZ}
    \sigma = \sum_{z \in V_{Z,(i-1)}} \sigma |_{Y_z}
\end{equation}
where $\sigma|_{Y_z}$ denotes the restriction of $\sigma$ to the local complex $Y_z$, viewed as a vector in $(C_z^2)^T$. In this expression and similar expressions below, it is implied that the local cochains can be included in the full complex via extension by zero.

We define the coarse-grained syndrome $\sigma' \in C_{(i-1)}^2 \cong \ff_2^{V_{Z,(i-1)}}$ via
  \begin{equation}
    \sigma'|_z \triangleq [\sigma|_{Y_z}]_{\loc} \in H^2(C_z, \delta_{z}) \cong \ff_2,
  \end{equation}
for each $z \in V_{Z,(i-1)}$. Here, $[\alpha]_{\loc}$ denotes the local cohomology class $\alpha + B^2_z$ in $H^2(C_z, \delta_z)$.
Notice that by $C^3_z = 0$, we have $C^2_z = Z^2_{z}$, and $\sigma|_{Y_z}$ is automatically a local cocycle; thus, the class $[\sigma|_{Y_z}]_{\loc}$ is well-defined.

\begin{figure}[t]
  \centering
  \includegraphics[width=0.5\textwidth]{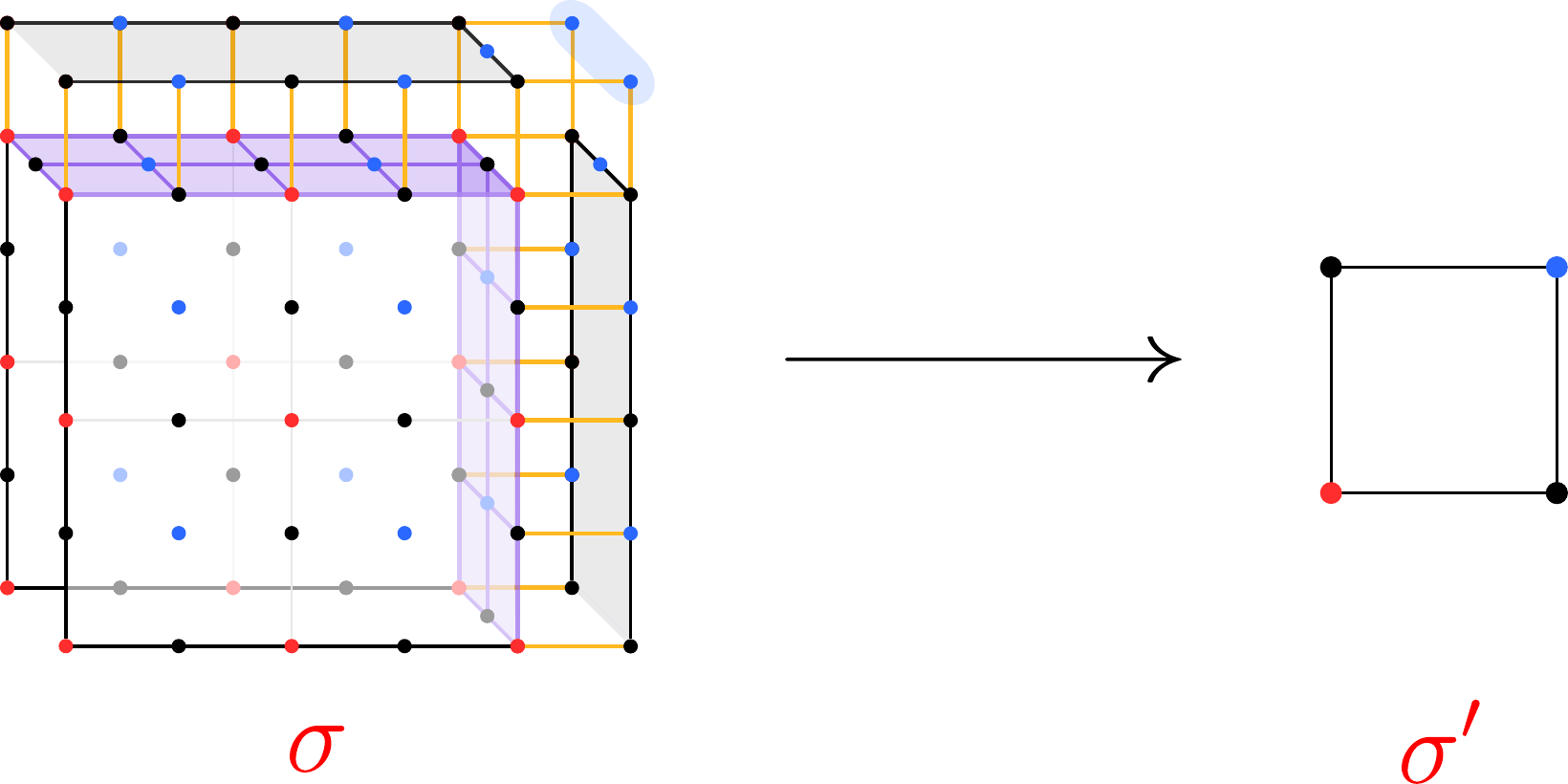}
  \caption{The syndrome $\sigma$ is supported on the level-$i$ complex
    $C_{(i)}$, while the coarse-grained syndrome $\sigma'$ is supported on the level-$(i-1)$ complex $C_{(i-1)}$.  }
\end{figure}

Equivalently,
  \begin{equation}
    \sigma'|_z = \sum_{p  \in Y_z(2)} \sigma|_{p}.
  \end{equation}
namely, $\sigma'|_z$ is the parity of the entries of the vector $\sigma|_{Y_z}$. The local computability of the map $\sigma \mapsto \sigma'$ follows directly from the definition.

\vspace{5 pt}
\noindent 2. \textbf{Verifying that $\sigma' \in B^2_{(i-1)}$:}
We will show that there exists $e' \in C^1_{(i-1)}$ such that $\sigma' = \delta^1_{i-1} e'$.
Since $\sigma \in B^2_{(i)}$, there exists some $e \in C^1_{(i)}$ such that $\sigma = \delta_i^1 e$. We define $e' \in C^1_{(i-1)} \cong \ff_2^{V_{Q,(i-1)}}$ via
  \begin{equation}
    e'|_q \triangleq [e|_{Y_q}]_{\loc} \in H^1(C_q, \delta_q) \cong \ff_2,
  \end{equation}
for each vertex $q \in V_{Q,(i-1)}$.
Since $C^2_q = 0$, we have $C^1_q = Z^1_q$, so $e|_{Y_q}$ is automatically a local cocycle, and the class $[e|_{Y_q}]_{\loc}$ is well-defined.

We now show that $\sigma' = \delta_{i-1}^1 e'$.
Recall from \Cref{sec:properties-Rz} that the induced map $H(g^1_{QZ})$ is precisely the coboundary map $\delta_{i-1}^1$ in the coarse-grained complex:
  \begin{align}
    H(g^1_{QZ}) \cong \delta_{i-1}^1: \  \ff_2^{V_{Q,(i-1)}} \to \ff_2^{V_{Z,(i-1)}}
  \end{align}
Therefore,
  \begin{align}
    (\delta_{i-1}^1 e')|_z
    &= \sum_{q \in V_{Q,(i-1)}} H(g_{qz}^1)\!\left([e|_{Y_{q}}]_{\loc}\right) \\
    &= \left[\sum_{q \in V_{Q,(i-1)}} g_{qz}^1\!\left(e|_{Y_{q}}\right)\right]_{\loc} \\
    &= \left[(\delta_i^1 e)|_{Y_{z}}\right]_{\loc} \\
    &= [\sigma|_{Y_{z}}]_{\loc} \\
    &= \sigma'|_z.
  \end{align}
Thus, $\sigma'  = \delta^1_{i-1} e' \in B^2_{(i-1)}$.

The identity above can also be checked more directly.
From the definition of $e'$, we have
  \begin{equation}
    e'|_q = \sum_{a \in Y_q{(1)}} e_{a},
  \end{equation}
so $e'|_q$ is simply the parity of the entries of vector $e$. Since each map $g_{qz}^1$ sends a basis element of $Y_q{(1)}$ to a basis element of $Y_z{(2)}$, the relation $\sigma' = \delta_{i-1}^1 e'$ follows by counting in $\ff_2$.

\vspace{5 pt}
\noindent  3. \textbf{Constructing the correction $f$:}
A natural candidate for the correction $f\in C^1_{(i)}$ would be $f = e - \cF_{i-1}(e')$, which indeed satisfies \Cref{eq:thm-RZ-CG}. However, this expression depends on the unknown error $e$, whereas the decoder has access only to $\sigma$.
We therefore construct $f$ directly from $\sigma$ via a local map.

By definition of $\sigma'$, for each $z \in V_{Z,(i-1)}$ we have
  \begin{equation}
    \sigma|_{Y_{z}} - \cF_{i-1}(\sigma')|_{Y_{z}}
    \in B^2(C_{z}, \delta_z).
  \end{equation}
Hence there exists a local correction $f_z \in C^1_{z}$ such that
  \begin{equation}
    \delta_{z}^1 f_z = \sigma|_{Y_{z}} - \cF_{i-1}(\sigma')|_{Y_{z}},
  \end{equation}
that is locally computable by definition. We then define
  \begin{equation}
    f \triangleq \sum_{z \in V_{Z,(i-1)}} f_z.
  \end{equation}
We now verify that $f$ satisfies $\sigma = \delta_i^1 f + \cF_{i-1}(\sigma')$. Recall that $\delta_i = \delta_{\loc} + g_{XQ} + g_{QZ}$.
Since $f$ is supported on the local complexes $C_z$ only, the terms $g_{XQ}$ and $g_{QZ}$ vanish on $f$. Therefore,
  \begin{equation}
    \delta_i^1 f =  \delta_\loc^1 f
    = \sum_{z \in V_{Z,(i-1)}} \delta_{z}^1 f_z
    = \sum_{z \in V_{Z,(i-1)}}
    \left(\sigma|_{Y_z} - \mathcal{F}_{i-1}(\sigma')|_{Y_z}\right)
    = \sigma - \cF_{i-1}(\sigma').
  \end{equation}

\vspace{5 pt}
\noindent  4. \textbf{Verifying that $|\sigma'| \le |\sigma|$:}
Since $\sigma'|_z \in \ff_2$ is a single bit, its weight is either $0$ or $1$. Therefore,
  \begin{equation}
    \Big| \sigma'|_{z} \Big| \le \Big| \sigma|_{Y_{z}} \Big|.
  \end{equation}
Using \Cref{eq:sigma-sum-RZ} and that $C_z$ are disjoint for different $z \in V_{Z,(i-1)}$, it follows from the above equation that  $|\sigma'| \le |\sigma|$.

\end{proof}

\subsection{Coarse-graining for $\cR_X$}
\label{sec:coarse-graining-RX}

We now state the main result for the coarse-graining steps associated with $\cR_X$.
\begin{theorem} \label{thm:coarse-graining-RX}
Given a level-$i$ syndrome $\sigma \in B^2_{(i)}$, there exist a locally computable coarse-graining map $\sigma \mapsto (\sigma', f)$ where the coarse-grained syndrome $\sigma' \in B^2_{(i-1)}$ and a correction $f \in C^1_{(i)}$ satisfy
  \begin{equation} \label{eq:thm-RX-CG}
    \sigma = \delta^1_{i} f + \cF_{\cR_X, i-1}(\sigma').
  \end{equation}
Furthermore,
  \begin{equation} \label{eq:thm:syn-red}
    |\sigma'| \le  \frac{1}{2} |\sigma|.
  \end{equation}
\end{theorem}
\noindent In the remainder of this subsection, we suppress the superscript $\cR_X$ in $\cF_{\cR_X, i-1}$. We will also suppress the subscript $i$ in many of the variables.  For example, we will write $C_X$ instead of $C_{(i),X}$, $C_x$ instead of $C_{(i),x}$, and $f_x$ instead of $f_{i,x}$.

In contrast to the case of $\cR_Z$, the coarse-graining step for $\cR_X$ is more complicated.
For $\cR_Z$, the coarse-grained syndrome $\sigma'$ can be defined directly from the local cohomology classes of the original syndrome $\sigma'|_z = [\sigma|_{Y_z}]$.
For $\cR_X$, however, such a direct construction is not possible.
Instead, the map is constructed as a sequence of steps:
\begin{enumerate}
  \item \textbf{Cleaning the syndrome from $C_X$:}
           For each $x\in V_{X,(i-1)}$, we find a local correction $f_x\in C_x^1$, so that the updated syndrome and error are
          \begin{equation}
              \sigma_1 \triangleq \sigma + \sum_{x \in V_{X,(i-1)}} \delta^1 f_x, \qquad   e_1 \triangleq e + \sum_{x \in V_{X,(i-1)}} f_x.
          \end{equation}
        where $\sigma = \delta^1_i e$, and $\sigma_1$ is supported on $Y_Q \cup Y_Z$.
  \item  \textbf{Cleaning the error from $C_X$:}
            This step is not performed by the decoder; instead, this is an auxiliary step used to prove that the final coarse-grained syndrome is a coboundary. For each local complex $C_x$, we show that there exists a stabilizer $s_x \in C^0_x$ that removes the error within $C_x$:
          \begin{equation}
              e_1' \triangleq e_1 + \sum_{x \in V_{X,(i-1)}} \delta^0_i s_x
          \end{equation}
          so that $e'_1 \in C_Q^1$.

  \item \textbf{Cleaning the syndrome from $C_Q$:}
          For each $q \in V_{Q,(i-1)}$, we find a local correction $f_q \in C^1_q$ so that the final cleaned syndrome and the new error are:
          \begin{equation}
              \sigma_2 \triangleq \sigma_1 + \sum_{q \in V_{Q,(i-1)}} \delta^1 f_q, \qquad e_2 \triangleq e_1' + \sum_{q \in V_{Q,(i-1)}} f_q,
          \end{equation}
        which satisfy $\sigma_2 \in C_Z^2$ and $e_2 \in C_Q^1$.
  \item \textbf{Defining the coarse-grained syndrome $\sigma'$ and correction $f$:}
            Finally, the new syndrome $\sigma'$ can be defined directly from the cleaned syndrome $\sigma_2 \in B^2_{(i)}$. We will show that the following is well-defined:
            \begin{equation}
            \sigma' = \cF_{i-1}^{-1}(\sigma_2),
            \qquad e' = \cF_{i-1}^{-1}(e_2).
            \end{equation}
            The total correction is given by
            \begin{equation}
                f = \sum_{x \in V_{X,(i-1)}} f_x + \sum_{q \in V_{Q,(i-1)}} f_q.
            \end{equation}
\end{enumerate}
Throughout the section, we also track the syndrome weight. In particular, we show that
\begin{equation}
    |\sigma_1| \leq |\sigma|, \quad
    |\sigma_2| \leq |\sigma_1|, \quad
    |\sigma'| = \frac{1}{2} |\sigma_2|.
\end{equation}
Together, these inequalities imply the desired syndrome-reducing property from \Cref{eq:thm:syn-red},
\begin{equation}
  |\sigma'| \leq \frac{1}{2} |\sigma|.
\end{equation}

The first three steps of the coarse-graining procedure are established in separate lemmas, while the final step is carried out in the proof of Theorem~\ref{thm:coarse-graining-RX}.

\subsubsection{Cleaning the syndrome from $C_X$}

We begin with the step that cleans the syndrome from the local complexes
$C_x$. This is the most involved part of the decoding step.

First observe that if $\sigma\in B^2(C_{(i)})$, then its restriction to
$Y_x$ for $x \in V_{X,(i-1)}$ is a local coboundary. Indeed, if $\sigma=\delta_i^1 e$, then
\begin{equation}
  \sigma_x:=\sigma|_{Y_x} = (\delta_i^1 e)|_{Y_x} = \delta_x^1(e|_{Y_x}),
\end{equation}
so $\sigma_x\in B_x^2$.
Thus, the syndrome can be cleaned from each local complex $C_x$ independently.

\begin{lemma}[Cleaning the syndrome from $C_x$]
\label{lemma:cleaning-syndrome-Yx}

Let $\sigma_x = \sigma|_{Y_x} \in B^2_x$ be a local syndrome.
Then there exists a local correction $f_x \in C^1_x$ such that
\begin{enumerate}
    \item[(1)] $f_x$ removes the syndrome on $Y_x$, i.e. $\delta^1_{x} f_x = \sigma_x$;
    \item[(2)] $\sum_{q \in V_{Q,(i-1)}} |g^1_{xq}(f_x)|
    \le |\sigma_x|$.
\end{enumerate}
\end{lemma}
\noindent Before proving the lemma, we explain why the second condition guarantees that cleaning the syndrome from $C_x$ does not increase the total weight of the updated syndrome.
\begin{corollary} \label{cor:CX}
Suppose that for each $x \in V_{X,(i-1)}$, we found a local correction $f_x \in C^1_x$ satisfying Lemma~\ref{lemma:cleaning-syndrome-Yx}. Then
\begin{equation} \label{eq:RX-total-counting-Cx-1}
    |\sigma_1| = \big|\sigma + \delta^1_i \sum_{x \in V_{X,(i-1)}} f_x\big| \le |\sigma|.
\end{equation}
\end{corollary}
\begin{proof}
Recall that $\delta^1_i = \delta^1_\loc + g^1_{XQ} + g^1_{QZ}$, where $g^1_{XQ}: C^1_X \rightarrow C^2_Q$ and $g^1_{QZ}: C^1_Q \rightarrow C^2_Z$. Let $f_X:=\sum_x f_x$. Since $f_X$ is supported on $C_X^1$, we have
\begin{equation}
  \delta_i^1 f_X = \sum_{x \in V_{X,(i-1)}} \left ( \delta_x^1 f_x + \sum_{q \in V_{Q,(i-1)}} g^1_{xq}(f_x) \right).
\end{equation}
The first term cleans the syndrome out of $C_X = \bigoplus_x C_x$ because
$\delta_x^1f_x=\sigma_x$.  Therefore, the net change in the syndrome weight is
  \begin{align}
    |\sigma_1|  &= \big|\sigma + \sum_x \sigma_x + \sum_{x,q}  g^1_{xq}(f_x) \big| \\
    &\leq \big|\sigma + \sum_x \sigma_x\big| + \sum_{x,q}  |g^1_{xq}(f_x)| \\
    &\leq |\sigma| - \sum_x |\sigma_x| + \sum_x |\sigma_x| = |\sigma|
  \end{align}
where the last line follows from the definition $\sigma_x :=\sigma|_{Y_x}$ and statement (2) of Lemma~\ref{lemma:cleaning-syndrome-Yx}.
\end{proof}
\begin{proof}[Proof of Lemma~\ref{lemma:cleaning-syndrome-Yx}]
We fix $x$ and drop the subscript $x$ throughout the proof; for example, $\sigma_x \in C_x^2$ and $f \in C_x^1$ will be denoted as $\sigma$ and $f$.

The complex $Y_x$ decomposes into pieces, each corresponding to a face $j \in X_x(2)$; see \Cref{fig:region} for one such piece. On each piece $Y_{x,j}$, if it contains an even syndrome weight, we can choose a correction $f_j$ consisting of open strings that pair the syndrome points within that piece. It is straightforward to verify that the choice can be made so that statement~(2) of the Lemma holds locally for each such piece.

The case when the number of syndrome points in $Y_{x,j}$ is odd is more subtle and is addressed separately. We construct a correction $f_0$, such that $\sigma + \delta^1 f_0$ has even syndrome parity on every piece $Y_{x,j}$. We then choose a correction $f_j$ on each piece. Finally, we show that
\begin{equation}
  f = f_0 + \sum_{j \in X_x(2)} f_j
\end{equation}
satisfies the desired properties.

\begin{figure}[t]
  \centering
  \includegraphics[width=0.37\textwidth]{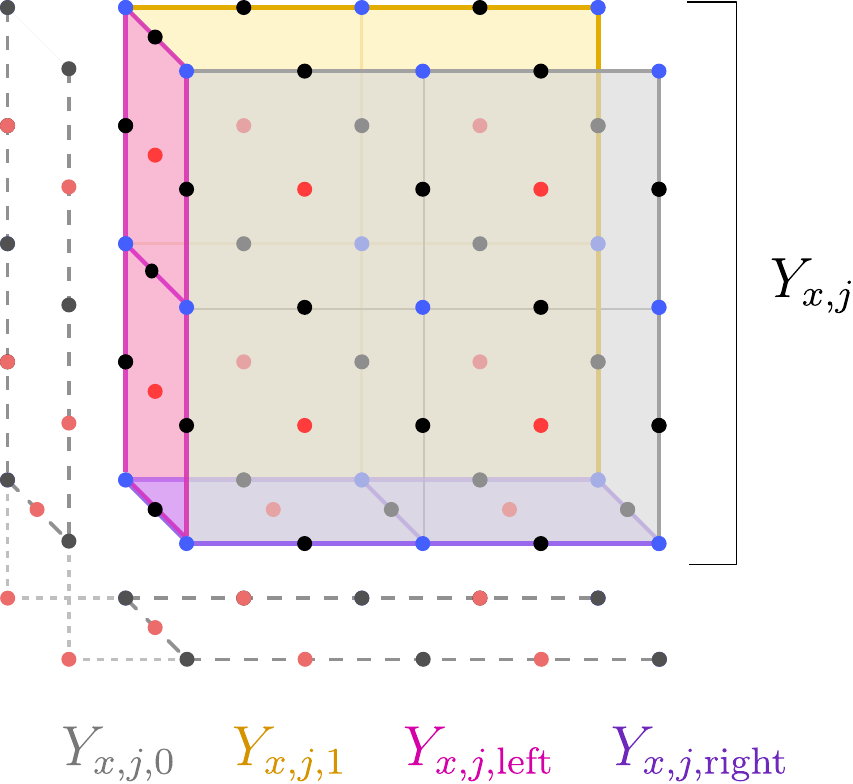}
  \caption{Decomposition of a piece $Y_{x,j}$ associated with a face $j \in X_x(2)$. Each piece consists of $Y_{x,j,0}$ (gray) and $Y_{x,j,1}$ (orange), corresponding to layers $0$ and $1$, together with $Y_{x,j,\mathrm{left}}$ (pink) and $Y_{x,j,\mathrm{right}}$ (purple) which connect the two layers.}
  \label{fig:region}
\end{figure}

\vspace{5 pt}
\noindent 1. \textbf{Defining the regions in $Y_x$.}
Recall that
\begin{equation}
  Y_x = \left(X'_x \times \{0\}\right) \sqcup \left(X'_x \times \{1\}\right) \sqcup \left(\partial X'_x \times [0,1]\right) / \sim,
\end{equation}
where $X'_x$ is a subdivision of $X_x$.
Each face $j \in X_x(2)$
  corresponds to a \emph{piece} $Y_{x,j}$ of $Y_x$, so that
\begin{equation}
  Y_x
  = \bigcup_{j \in X_x(2)} Y_{x,j}.
\end{equation}
We further decompose each piece into four regions:
\begin{equation}
  Y_{x,j}
  = \left(Y_{x,j,0} \cup Y_{x,j,1} \cup Y_{x,j,\mathrm{left}} \cup Y_{x,j,\mathrm{right}}\right),
\end{equation}
which are illustrated in \Cref{fig:region}. We call the one-dimensional subcomplex separating two adjacent two-dimensional regions an interface. We use both the local interfaces inside a fixed $Y_{x,j}$ and the interfaces between adjacent pieces $Y_{x,j}$ and $Y_{x,j'}$.

Accordingly, we decompose the syndrome as
\begin{align}
  \sigma = \sum_{j \in X_x(2)} \sigma_j
  = \sum_{j \in X_x(2)} \left(\sigma_{j,0} + \sigma_{j,1}
  + \sigma_{j,\mathrm{left}} + \sigma_{j,\mathrm{right}}\right),
\end{align}
where each term in the sum is supported in the region denoted by the summation index.

We will use the notation $\wt{\alpha_R} \in \ff_2$ to denote the parity of the syndrome $\alpha$ in a given region $R$; for example,
\begin{equation}
  \wt{\sigma_j} \coloneqq \sum_{s \in Y_{x,j}(2)} \sigma_j \big|_s,
\end{equation}
where the sum runs over all coordinates in the corresponding region.

\vspace{5 pt}
\noindent 2. \textbf{Constructing $f_0$.}
Let $G = (V_G, E_G)$ be the link graph for the vertex $x$ in the complex $X_x$. Thus, the edge set $E_G \cong X_x(2)$ corresponds to the faces of $X_x$ incident to $x$, while the vertex set ${V_G} \cong X_x(1) \setminus\partial X_x$ corresponds to the edges of $X_x$ incident to $x$. Note that the elements of ${V_G}$ index the interface regions between adjacent pieces $Y_{x,j}$.

Let $\wt \sigma \in \ff_2^{E_G}$ be the vector recording the parity of $\sigma$ on each piece $Y_{x,j}$:
\begin{equation}
  \wt{\sigma} = (\wt{\sigma_j})_{j \in {E_G}}.
\end{equation}
We first show that $\wt \sigma \in \mathrm{Im} \, \delta_G$. Since $\sigma \in B^2_x$, there exists some $e \in C^1_x$ such that $\delta^1_x e = \sigma$. We decompose $e$ according to the regions defined above:
\begin{equation}
  e = \sum_{j \in {E_G}} e_j + \sum_{v \in {V_G}} e_v,
\end{equation}
where $e_v$ is supported on the interface region shared by $\{Y_{x,j}: j \ni v\}$, while $e_j$ is supported on the remaining part of the piece corresponding to $j \in {E_G}$, excluding interface regions between the pieces.

Let $\wt e_\int \in \ff_2^{V_G}$ be the vector recording the parity of $e$ on each interface region:
\begin{equation}
  \wt{e}_\int = (\wt{e_v})_{v \in {V_G}}.
\end{equation}
Let $\delta_G: \ff_2^{V_G} \to \ff_2^{E_G}$ be the coboundary map of $G$. We claim that
\begin{equation}
  \wt \sigma = \delta_G \wt e_\int.
\end{equation}
This holds because a basis element of $C_x^1$ supported in an interface region indexed by $v$ changes the syndrome parity on the pieces indexed by $j \ni v$, whereas a basis element supported inside $Y_{x,j}$ leaves the syndrome parity on $Y_{x,j}$ and all other pieces unchanged. Applying this to the definition of $\wt e_\int$ yields the claim.

Thus, $\wt{\sigma}$ lies in $\operatorname{Im}\delta_G$. Although the decoder does not know $e$, because $\wt{\sigma}  \in \operatorname{Im}\delta_G$, we choose any $\alpha \in \ff_2^{V_G}$ satisfying
\begin{equation}
  \delta_G \alpha = \wt{\sigma}.
\end{equation}
We now choose $f_0$ so that its parity on the interface indexed by $v\in V_G$ is $\alpha_v$, i.e. it satisfies
\begin{equation}
  \wt f_0 = \bigl(\wt{(f_0)_k}\bigr)_{k \in {V_G}} = \alpha.
\end{equation}

Our choice of $f_0$ will be such that its support entirely lies in only one of the two layers, 0 or 1, and choose the correction that gives the smaller contribution to $g^1_{xq}(f)$.
By symmetry between the two layers, we may assume
\begin{equation} \label{eq:proof-assumption-1}
  \sum_{j \in X_x(2)} |\wt{\sigma}_{j,1}|
  \leq
  \sum_{j \in X_x(2)} |\wt{\sigma}_{j,0}|.
\end{equation}
Under this assumption, we take $f_0$ to be supported near the center of the copy $X'_x \times \{0\}$, on the qubits adjacent to the central $X$-check, which is shared by all faces. These qubits are in one-to-one correspondence with vertices in ${V_G}$, so we choose the value of $f_0$ on them to be $\alpha$.

Let the updated syndrome after applying $f_0$ be
\begin{equation}
  \tau = \sigma + \delta^1_x f_0.
\end{equation}
By construction,
\begin{equation}
  \wt{\tau_j} = 0
\end{equation}
for every $j \in X_x(2)$. Thus, the updated syndrome $\tau$ has even syndrome parity on each piece $Y_{x,j}$.

\vspace{5 pt}\noindent 3. \textbf{Constructing $f_j$ for each piece.}
For each $j \in X_x(2)$, we choose a correction $f_j$ consisting of open ``strings'' that pair the syndrome points inside each $Y_{x,j}$. See \Cref{fig:pairing} for an example of such a pairing.

\begin{figure}[t]
  \centering
  \includegraphics[width=0.32\textwidth]{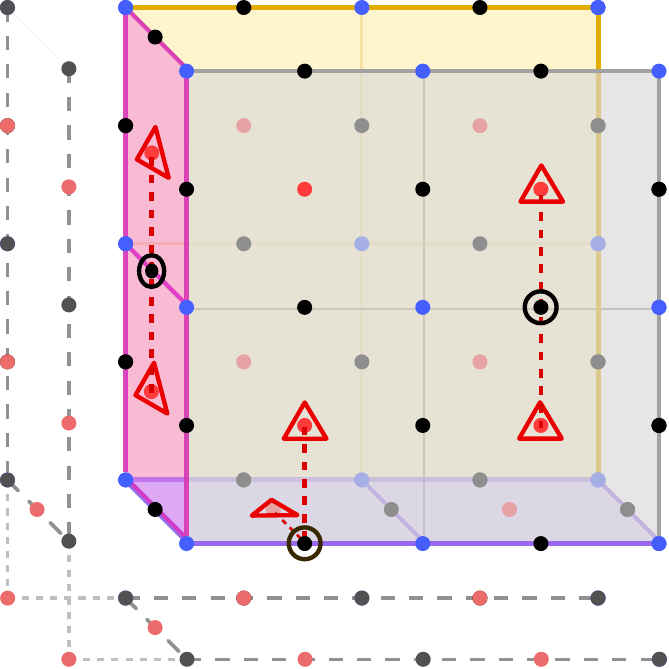}
  \caption{An example of pairing syndromes on a piece $Y_{x,j}$.
    The red triangles indicate the syndromes to be paired.
    The circles indicate the correction 1-cochains.
    The pairing is chosen to minimize the number of crossings between the regions
    $Y_{x,j,0}$, $Y_{x,j,1}$, $Y_{x,j,\mathrm{left}}$, and $Y_{x,j,\mathrm{right}}$.}
  \label{fig:pairing}
\end{figure}

We choose the strings so that they cross the interfaces between the regions $Y_{x,j,0}$, $Y_{x,j,1}$,  $Y_{x,j,\mathrm{left}}$, and $Y_{x,j,\mathrm{right}}$ as few times as possible because each crossing contributes to $g^1_{xq}(f)$, which we want to minimize.
Our choice will be determined by the syndrome parities within each region, which we denote
\begin{equation}
  \wt{\tau}_{j,0}, \
  \wt{\tau}_{j,1}, \
  \wt{\tau}_{j,\mathrm{left}}, \
  \wt{\tau}_{j,\mathrm{right}} \in \ff_2.
\end{equation}
Since $\wt{\tau_j}= 0$, we have
\begin{equation}
  b = \wt{\tau}_{j,0}+\wt{\tau}_{j,1}
  =
\wt{\tau}_{j,\mathrm{left}}+\wt{\tau}_{j,\mathrm{right}} \in \ff_2.
\end{equation}
First pair syndrome points within each region, using strings that do not cross any interface. This leaves exactly one unpaired syndrome in each region whose
parity is $1$, and no unpaired syndrome in regions whose parity is $0$. The remaining correction strings are chosen according to the value of $b$.
\begin{itemize}
  \item If $b=0$, the possible unpaired syndromes in $Y_{x,j,0}$ and $Y_{x,j,1}$ can be paired by a string crossing the interfaces between $0$ and $\mathrm{left}$, and between $\mathrm{left}$ and $1$. Similarly, the possible unpaired syndromes in $Y_{x,j,\mathrm{left}}$ and $Y_{x,j,\mathrm{right}}$ can be paired by a string crossing the interface between $\mathrm{left}$ and $\mathrm{right}$.
  \item If $b=1$, then there is one unpaired syndrome in one of $Y_{x,j,0}$ and $Y_{x,j,1}$, and one unpaired  syndrome in one of $Y_{x,j,\mathrm{left}}$ and $Y_{x,j,\mathrm{right}}$. These two syndromes can be paired by a string crossing a single interface.
\end{itemize}

\begin{figure}[t]
  \centering
  \includegraphics[width=0.32\textwidth]{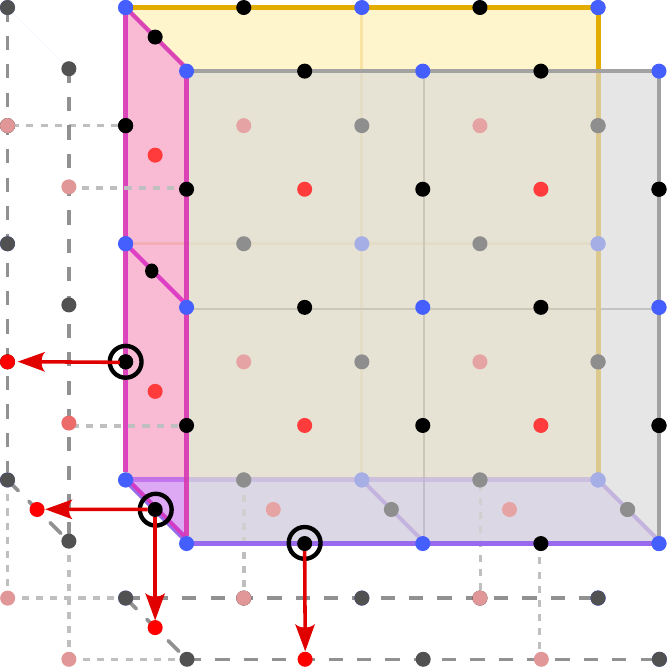}
  \caption{Illustration of the contribution to $g^1_{xq}(f)$ from crossings of the interfaces.
            Each crossing involves qubits supported on the interface,
              which induce syndrome outside $Y_x$ through $g^1_{xq}$.
            Crossings along the $\mathrm{left}$--$\mathrm{right}$ interface contribute weight $2$,
              while crossings along other interfaces contribute weight $1$.}
  \label{fig:g_xq}
\end{figure}

A crossing between $Y_{x,j,\mathrm{left}}$ and $Y_{x,j,\mathrm{right}}$ corresponds to a qubit at the corner which contributes $2$ to $|g^1_{xq}(f_x)|$, whereas crossing any of the other interfaces contributes $1$; see \Cref{fig:g_xq} for an example.
Therefore, we obtain
\begin{equation} \label{eq:total-parities}
  \sum_{q \in V_{Q,(i-1)}} |g^1_{xq}(f_j)|
  =
  |\wt{\tau}_{j,0}|
  + |\wt{\tau}_{j,1}|
  + |\wt{\tau}_{j,\mathrm{left}}|
  + |\wt{\tau}_{j,\mathrm{right}}|
  - |\wt{\tau}_{j,0} + \wt{\tau}_{j,1}|.
\end{equation}
To verify this expression, it suffices to check all possible distinct parity configurations up to the symmetries exchanging $Y_{x,j,0}$ with $Y_{x,j,1}$ and $Y_{x,j,\mathrm{left}}$ with $Y_{x,j,\mathrm{right}}$:
\begin{itemize}
  \item $(\wt{\tau}_{j,0}, \wt{\tau}_{j,1}, \wt{\tau}_{j,\mathrm{left}}, \wt{\tau}_{j,\mathrm{right}}) = (0,0,0,0)$: no crossing is needed, and both sides of \Cref{eq:total-parities} equal $0$.
  \item $(\wt{\tau}_{j,0}, \wt{\tau}_{j,1}, \wt{\tau}_{j,\mathrm{left}}, \wt{\tau}_{j,\mathrm{right}}) = (1,1,0,0)$: two crossings are needed, and both sides of \Cref{eq:total-parities} equal $2$.
  \item $(\wt{\tau}_{j,0}, \wt{\tau}_{j,1}, \wt{\tau}_{j,\mathrm{left}}, \wt{\tau}_{j,\mathrm{right}}) = (0,0,1,1)$: one crossing between $\mathrm{left}$ and $\mathrm{right}$ is needed, and both sides of \Cref{eq:total-parities} equal $2$.
  \item $(\wt{\tau}_{j,0}, \wt{\tau}_{j,1}, \wt{\tau}_{j,\mathrm{left}}, \wt{\tau}_{j,\mathrm{right}}) = (1,0,1,0)$: one crossing between $0$ and $\mathrm{left}$ is needed, and both sides of \Cref{eq:total-parities} equal $1$.
\end{itemize}
Finally, we set
\begin{equation}
    f_1 = \sum_{j \in X_x(2)} f_j,
\end{equation}
and the full correction is $f = f_0 + f_1$.

\vspace{5 pt}\noindent 4. \textbf{Verifying statement (2) of the Lemma.}
Statement (1) of the Lemma follows from the construction of $f_0$ and $f_j$.
It remains to show that
\begin{equation}
  \sum_{q \in V_{Q,(i-1)}} |g^1_{xq}(f)| \le |\sigma|.
\end{equation}
First, by the decomposition of $\sigma$, we have
\begin{equation}
  |\sigma|
  =
  \sum_{j \in X_x(2)}
  \left(
    |\sigma_{j,0}|
    + |\sigma_{j,1}|
    + |\sigma_{j,\mathrm{left}}|
    + |\sigma_{j,\mathrm{right}}|
  \right).
\end{equation}
In particular,
\begin{equation} \label{eq:sigma-lower-bound}
  |\sigma|
  \ge
  \sum_{j \in X_x(2)}
  \left(
    |\wt{\sigma}_{j,0}|
    + |\wt{\sigma}_{j,1}|
    + |\wt{\sigma}_{j,\mathrm{left}}|
    + |\wt{\sigma}_{j,\mathrm{right}}|
  \right),
\end{equation}
since the weight of $\sigma$ on each region is at least the parity of the syndrome in that region.

By step 3 of the proof and $g_{xq}^2(f_0) = 0$, we have
\begin{equation} \label{eq:gxq-fj-count}
\begin{split}
     \sum_{q \in V_{Q,(i-1)}} |g^1_{xq}(f)| &=   \sum_{q \in V_{Q,(i-1)}} \sum_{j \in X_x(2)} |g^1_{xq}(f_j)|
     \\
  &=
  \sum_{j \in X_x(2)}
  \left(
    |\wt{\tau}_{j,0}|
    + |\wt{\tau}_{j,1}|
    + |\wt{\tau}_{j,\mathrm{left}}|
    + |\wt{\tau}_{j,\mathrm{right}}|
    - |\wt{\tau}_{j,0}+\wt{\tau}_{j,1}|
  \right).
\end{split}
\end{equation}
By our choice of correction $f_0$, it only changes the parities in the regions $Y_{x,j,0}$. Therefore
\begin{equation} \label{eq:f0-property}
    \wt{\tau}_{j,1} = \wt{\sigma}_{j,1},
  \qquad
  \wt{\tau}_{j,\mathrm{left}} = \wt{\sigma}_{j,\mathrm{left}},
  \qquad
  \wt{\tau}_{j,\mathrm{right}} = \wt{\sigma}_{j,\mathrm{right}}.
\end{equation}
Comparing \eqref{eq:sigma-lower-bound} and \eqref{eq:gxq-fj-count}, we obtain
\begin{equation}
   \sum_{q \in V_{Q,(i-1)}} |g^1_{xq}(f)| + \sum_{j \in X_x(2)}
  \left(
     |\wt{\sigma}_{j,0}| - |\wt{\tau}_{j,0}|
    + |\wt{\tau}_{j,0}+\wt{\tau}_{j,1}|
  \right)
  \leq |\sigma|.
\end{equation}
Thus, to prove statement~(2), it is enough to show that
\begin{equation} \label{eq:remaining-layer0-bound}
  \sum_{j \in X_x(2)}
  \left(
    |\wt{\tau}_{j,0}| - |\wt{\tau}_{j,0}+\wt{\tau}_{j,1}|
  \right)
  \le
  \sum_{j \in X_x(2)} |\wt{\sigma}_{j,0}|.
\end{equation}
This follows from
\begin{equation}
  \sum_{j \in X_x(2)}
  \left(
    |\wt{\tau}_{j,0}| - |\wt{\tau}_{j,0}+\wt{\tau}_{j,1}|
  \right)
  \le
  \sum_{j \in X_x(2)} |\wt{\tau}_{j,1}|
  =
  \sum_{j \in X_x(2)} |\wt{\sigma}_{j,1}|
  \le
  \sum_{j \in X_x(2)} |\wt{\sigma}_{j,0}|.
\end{equation}
The first inequality follows from the triangle inequality; the equality uses \Cref{eq:f0-property}, and the last inequality follows from the assumption in \Cref{eq:proof-assumption-1}. This proves statement~(2).
\end{proof}

\subsubsection{Cleaning the error from $C_X$}
After cleaning the syndrome from $C_X$, we have
\begin{equation}
    e_1 \triangleq e + \sum_{x \in V_{X,(i-1)}} f_x,
    \qquad
    \delta_i^1 e_1 \in C_Q^2 \oplus C_Z^2.
\end{equation}
For $x \in V_{X,(i-1)}$, let $e_{1,x}\triangleq e_1|_{Y_x}$. Since
$\delta_i^1e_1$ has no component in $C_X^2$, we have
$\delta_x^1 e_{1,x}=0$, so $e_{1,x}\in Z^1(C_x,\delta_x)$.

\begin{lemma} \label{lemma:cleaning-error-CX}
Given an error with no syndrome in the local complex, i.e. $e_{1,x} \in Z^1 (C_x, \delta_x)$, there exists a stabilizer $s_x \in C_x^0$ such that $e_{1,x} + \delta^0_x s_x = 0$.
\end{lemma}
\begin{proof}
  Because $H^1(C_x, \delta_x) = 0$, every local 1-cocycle is a local coboundary, i.e. $Z^1 (C_x, \delta_x) = B^1 (C_x, \delta_x)$, and the statement immediately follows.
\end{proof}

\begin{corollary} [Cleaning the error from $C_X$]
    Given an error $e_1$ such that $\delta^1_i e_1  \in C_Q^2 \oplus C_Z^2$, there exists a stabilizer
    \begin{equation}
        s_X = \sum_{x \in V_{X,(i-1)}} s_x
    \end{equation}
    that cleans the error from $C_X$, i.e.
    \begin{equation}
        e_1' \triangleq e_1 + \delta^0_i s_X \ \in C_Q^1.
    \end{equation}
\end{corollary}
\begin{proof}
    For each $x \in V_{X,(i-1)}$, choose $s_x$ by Lemma~\ref{lemma:cleaning-error-CX}.
    Recall that $\delta^0_i = \delta^0_\loc + g^0_{XQ} + g^0_{QZ}$, where $g^0_{XQ}: C^0_X \rightarrow C^1_Q$ and $g^0_{QZ}: C^0_Q \rightarrow C^1_Z$. Denoting $Y_X = \bigsqcup_x Y_x$, we have:
    \begin{equation}
        e_1'|_{Y_X} = e_1|_{Y_X} + \delta_\loc^0 \big( \sum_x s_x \big) =   \sum_x ( e_{1,x} + \delta^0_x s_x) = 0.
    \end{equation}
    Since $e_1'$ is not supported in $C_X^1$ and $C_Z^1 \cong 0$, this implies $e_1' \in C_Q^1$.
\end{proof}

\subsubsection{Cleaning the syndrome from $C_Q$}

After cleaning the syndrome from $C_X$, we  obtained $$\sigma_1 = \sigma + \delta^1_i \sum_x f_x \in C_Q^2\oplus C_Z^2.$$
For each $q\in V_{Q,(i-1)}$, let $\sigma_q:=\sigma_1|_{Y_q}$. After cleaning the error from $C_X$, we may assume
$\sigma_1=\delta_i^1 e_1'$ with $e_1'\in C_Q^1$. Therefore,
$\sigma_q\in B^2 (C_q, \delta_q)$.

\begin{lemma}[Cleaning the syndrome from $C_Q$]
\label{lemma:cleaning-syndrome-CQ}
Given a syndrome $\sigma_q\in B^2 (C_q, \delta_q)$, there exists a local correction
$f_q\in C_q^1$ that:
\begin{enumerate}
    \item[(1)] removes the syndrome from $C_q$, i.e. $\delta_q^1 f_q=\sigma_q$;
    \item[(2)] satisfies $\sum_{z\in V_{Z,(i-1)}} |g^1_{qz}(f_q)| \le |\sigma_q|$.
\end{enumerate}
\end{lemma}

\begin{proof}
We will drop the subscript $q$ in many of the variables in this proof for the sake of succinctness. For example, we will write $\sigma$ for $\sigma_q$ and $f$ for $f_q$.

Since $\sigma \in B^2(C_q, \delta_q)$, there exists $h \in C^1_q$ such that  $\delta^1_q h = \sigma$. We choose any such $h$ and decompose it into a part supported on the boundary and a part supported in the interior of $C_q$: $h = h_{\partial}+h_{\int}$. See \Cref{fig:decoderCq} for an illustration.
Only the $h_{\partial}$ part can create syndrome in the neighboring $C_z$ complexes, so
\begin{equation}
  \sum_{z\in V_{Z,(i-1)}} |g^1_{qz}(h)|=|h_{\partial}|.
\end{equation}
Thus, we would like to show that there exists a choice of $h$ such that $|h_\partial| \leq |\sigma|$.

Let $z \in \partial X_q \subseteq V_{Z,(i-1)}$ index the boundaries of the complex $C_q$. We denote the values of the 1-cochain $h$ on the two qubits of boundary $z$ as $h_{\partial,z,0}$ and $h_{\partial,z,1}$.  Call $\sigma_{\mid,z}$ the syndrome on the $z$-th boundary. Then, we have
$$\sigma_{\mid, z} = h_{\partial,z,0} + h_{\partial,z,1}.$$
There are two cases that we would like to consider:
\begin{itemize}
    \item $\sigma_{\mid,z} = 1$: In this case, $|h_{\partial,z,0}|$ or $|h_{\partial,z,1}|$ is $1$, and we automatically have $|h_{\partial,z}| = |h_{\partial,z,0}| + |h_{\partial,z,1}| \leq |\sigma_{\mid,z}|$.
    \item $\sigma_{\mid,z} = 0$: the pair $(h_{\partial,z,0},h_{\partial,z,1})$ is either $(0,0)$ or $(1,1)$. Recall that $C_q$ has an all-1 1-cocycle. Adding this all-1 vector to $h$ does not change its coboundary inside $C_q$, but exchanges the $(0,0)$ and $(1,1)$ pairs. Thus, we may choose $h$ such that, among the pairs with $\sigma_{\mid,z}=0$, the number of $(1,1)$ pairs is less or equal to the number of $(0,0)$ pairs. Pair each $(1,1)$ pair with a distinct $(0,0)$ pair. Consider one s uch pair, indexed by $z$ and $z'$. Since $h_{\partial,z,0} \neq h_{\partial,z',0}$, there must be at least one syndrome between these two points; similarly for $h_{\partial,z,1} \neq h_{\partial,z',1}$. Hence, for each $(h_{\partial,z,0},h_{\partial,z,1}) = (1,1)$ there exists two distinct syndrome points in $\sigma$.
\end{itemize}
Collecting everything together, we obtain
\begin{equation}
    |h_\partial| \leq \sum_{z \in \partial X_q} |h_{\partial,z}| \leq |\sigma|.
\end{equation}
Taking $f_q=h$ proves the lemma.
 \begin{figure}[t]
    \centering
    \includegraphics[width=0.45\textwidth]{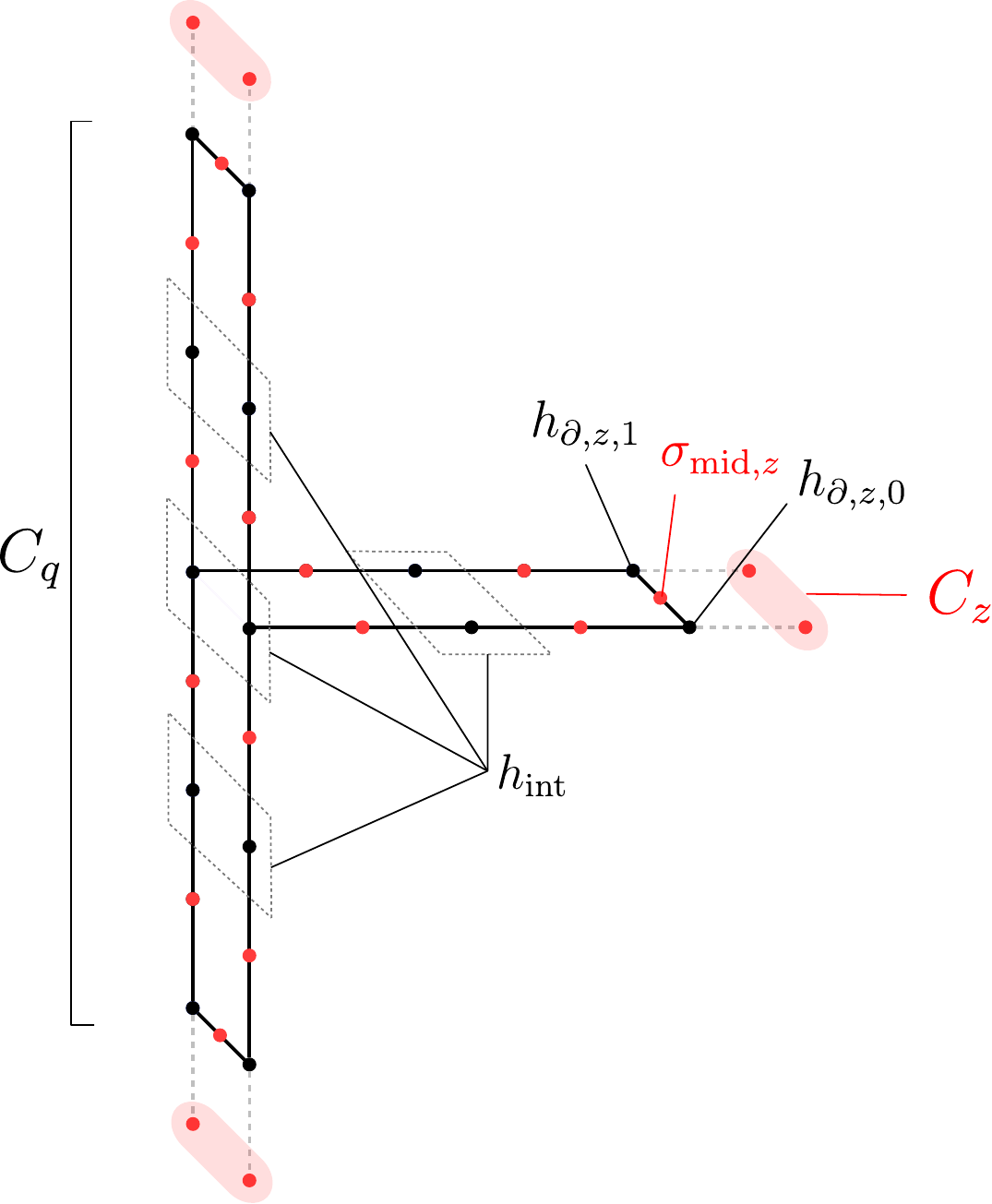}
    \caption{An example of complex $C_q$ illustrating the regions $\mid$, $\partial$, and $\int$. }
    \label{fig:decoderCq}
  \end{figure}
\end{proof}
\begin{corollary} \label{cor:CQ}
    Suppose that for each $q \in V_{Q,(i-1)}$, we found a local correction $f_q \in C^1_q$ satisfying Lemma~\ref{lemma:cleaning-syndrome-CQ}. Define $$\sigma_2 = \sigma_1 + \delta^1_i\sum_{q\in V_{Q,(i-1)}} f_q.$$
    Then
    \begin{equation} \label{eq:RX-total-counting-Cx}
        |\sigma_2|  \leq |\sigma_1|.
    \end{equation}
\end{corollary}
\begin{proof}
Recall that $\delta^1_i = \delta^1_\loc + g^1_{XQ} + g^1_{QZ}$, where $g^1_{XQ}: C^1_X \rightarrow C^2_Q$ and $g^1_{QZ}: C^1_Q \rightarrow C^2_Z$. Then, we write
\begin{align}
     |\sigma_2| &= \big| \sigma_1 + \sum_{q} \delta^1_q (f_q)   + \sum_{q,z} g^1_{qz} f_q \big|
     \\
     &= \big| \sigma_1 + \sum_{q} \sigma_q  + \sum_{q,z} g^1_{qz} (f_q) \big|
     \\
     &\leq \big| \sigma_1 + \sum_{q} \sigma_q \big| + \sum_{q,z} |g^1_{qz} (f_q)|
     \\
     &\leq |\sigma_1| - \sum_{q} |\sigma_q| + \sum_{q} |\sigma_q| = |\sigma_1|,
\end{align}
where the last lines holds from the definition $\sigma_q := \sigma_1|_{Y_q}$ and statement (2) of Lemma~\ref{lemma:cleaning-syndrome-CQ}.
\end{proof}

\subsubsection{Coarse-grained syndrome $\sigma'$ and correction $f$}

After the first three steps of coarse-graining, we obtained the cleaned syndrome $\sigma \in B^2(C_{(i)})$
\begin{equation}
\sigma_2 = \sigma + \delta^1_i \sum_{x \in V_{X,(i-1)}} f_x +  \delta^1_i \sum_{q \in V_{Q,(i-1)}} f_q
\end{equation}
that satisfies $\sigma_2 \in C_Z^2$, and updated error
\begin{equation}
e_2 = e +  \sum_{x \in V_{X,(i-1)}} f_x + \delta^0_i \sum_{x \in V_{X,(i-1)}} s_x  + \sum_{q \in V_{Q,(i-1)}} f_q
\end{equation}
that satisfies $e_2 \in C_Q^1$ and $\delta^1_i e_2 = \sigma_2$.

\begin{lemma}[RG step]
\label{lemma:RX-coarse-grained-syndrome}
Given an error $e_2\in C_Q^1$ and $\delta_i^1e_2=\sigma_2\in C_Z^2$, then:
\begin{enumerate}
    \item[(1)] $\sigma_2$ is in the image of $\cF_{i-1}^2$, and thus, we can define $\sigma' = \cF^{-1}_{i-1}(\sigma_2)$.
    Moreover, it satisfies $\sigma' \in B^2(C_{(i-1)})$.
    \item[(2)] Futhermore, we have $  |\sigma'| \leq \frac{1}{2}|\sigma|$.
\end{enumerate}
\end{lemma}

\begin{proof}
Recall that $H^1(C_q, \delta_q) \cong \ff_2$. For each  $q \in V_{Q,(i-1)}$, the component $e_q = e_2 |_{Y_q}$ is a cocycle in the local complex $C_q$, i.e. $e_q \in Z^1(C_q,\delta_q)$. Thus, we can consider
\begin{equation}
    [e_2]_\loc \in \bigoplus_{q \in V_{Q,(i-1)}} H^1(C_q, \delta_q) \cong \ff_2^{V_{Q,(i-1)}} \cong C_{(i-1)}^1
\end{equation}
Therefore, we can define $e' = \cF^{-1}_{i-1}(e_2)$.
We will define the coarse-grained syndrome as $\sigma' = \delta^1_{i-1} e'$
which is automatically in $B_{(i-1)}^2$.  Consider now
\begin{equation}
\begin{split}
     \cF_{i-1}(\sigma') = \cF_{i-1}(\delta^1_{i-1} e') &= \delta^1_{i} \cF_{i-1}( e') \\
     &= \delta^1_{i} \cF_{i-1}( \cF^{-1}_{i-1}(e_2)) = \delta^1_{i} e_2 \equiv \sigma_2
\end{split}
\end{equation}
which establishes statement (1).

Statement (2) follows from combining the action of $\cF^{-1}$ with the results of Corollaries~\ref{cor:CX} and~\ref{cor:CQ}.
\end{proof}
Finally, we prove the main theorem of this subsection.
\begin{proof}[Proof of Theorem~\ref{thm:coarse-graining-RX}]
Using Lemmas~\ref{lemma:cleaning-syndrome-Yx} and ~\ref{lemma:cleaning-syndrome-CQ}, we set
\begin{equation}
  f
  =
  \sum_{x\in V_{X,(i-1)}} f_x
  +
  \sum_{q\in V_{Q,(i-1)}} f_q.
\end{equation}
and $\sigma_2=\sigma+\delta_i^1 f$. $f_x$ and $f_q$ are locally computable by construction.

Using Lemma~\ref{lemma:RX-coarse-grained-syndrome}
$\sigma_2=\cF_{i-1}(\sigma')$, we obtain $\sigma' = \cF_{i-1}^{-1}(\sigma_2)$ and obtain
\begin{equation}
  \sigma
  =
  \delta_i^1 f+\cF_{i-1}(\sigma'),
\end{equation}
and
\begin{equation}
    |\sigma'| \leq \frac{1}{2}|\sigma|
\end{equation}
as required.
\end{proof}

\section{Memory lifetime}\label{sec:memlifetime}

We now use the decoder of the previous section to prove the memory lifetime bound stated in \Cref{thm:main}. In Sec.~\ref{sec:review-memorylifetime}, we reviewed some of the known memory lifetime bounds in the operator notation.  Given a CSS code $C_{(2k)}$, we consider the stabilizer Hamiltonian $H_{2k} = -\sum_i s_{(2k),i}$ where $s_{(2k),i}$ are the Pauli stabilizer generators of the code. The system evolves under the master equation shown in Sec.~\ref{sec:review-memorylifetime}, with Hamiltonian $H_{2k}$ and the set of jump operators $\cE$ for the Lindbladian. We take $\cE$ to be the set of all single-qubit Pauli operators, although the proof idea can be readily generalized to larger support and non-Pauli jump operators.

We use the following convention. The support of an Pauli-$X$ error is given by a 1-cochain $e_X \in C^1_{(2k)}$ and its syndrome is $\sigma_X = \delta^1_{2k} e_X \in B^2(C_{(2k)})$. The support of a Pauli-$Z$ error is represented by a 1-chain in the dual chain complex, and its syndrome is its boundary. Given a state with syndrome $\sigma = (\sigma_X,\sigma_Z)$, the decoder constructed in Sec.~\ref{sec:decoder} finds a correction 1-cochain $f_X$ for Pauli-$X$ errors and a correction 1-chain $f_Z$ for Pauli-$Z$ errors. This determines the Pauli correction operator $X^{\dec_X(\sigma_X)}Z^{\dec_Z(\sigma_Z)}$.

The main idea of the memory lifetime proof is to construct the projector $P_\mrm{stable}$ onto a subspace where the dressed logical operators are protected from all jump operators in $\cE$. By Claim~\ref{claim:stableconfig}, it can be defined as a subspace of states with syndrome $\sigma$ satisfying
$$X^{\dec_X(\sigma_Z)}Z^{\dec_Z(\sigma_X)}  \ket{\sigma, \pm_L} = X^{\dec_X(\sigma_X+\sigma_{E,X})}Z^{\dec_Z(\sigma_Z+\sigma_{E,Z})} E \ket{\sigma, \pm_L}$$
for any $E \in \cE$, where $\sigma_E = (\sigma_{E,X},\sigma_{E,Z})$ is the syndrome of the error $E$. This definition explicitly relies on specifying a decoder.
We will then show that the probability of the complement of this subspace under the Gibbs measure on syndrome configurations is exponentially small:
$$\Tr(1-P_\mrm{stable}) \rho_\beta \leq \exp(-\Theta(n_k^\eta)).$$
Combining this with \Cref{thm:memorylifetime} yields the desired lower bound on the memory lifetime.

\subsection{Unstable and stable syndrome configurations}

Because $X$ and $Z$ sectors can be treated separately, we first analyze Pauli-$X$ errors and their syndromes. We therefore suppress the $X$ subscript throughout this discussion; for example, we will use $\dec(\sigma)$ to denote $\dec_X(\sigma_X)$.

Recall from \Cref{sec:decoder-CG-maps} that the decoder operates via a sequence of coarse-graining maps. Starting from the level-$2k$ syndrome $\sigma_{2k} := \sigma \in B^2(C_{(2k)})$, the decoder maps it to the syndrome at each level to the coarse-grained syndrome at the next level $\sigma_{2k-1} \in B^2(C_{(2k-1)})$ and the correction $f_{2k} \in C^1_{(2k)}$. Doing this successively produces the coarse-graining map $\sigma_i \mapsto (\sigma_{i-1},f_i)$, which satisfies
\begin{equation}
    \sigma_i = \delta^1_i f_i + \cF_{i-1}(\sigma_{i-1}),
\end{equation}
where $\cF_{i-1}: C_{(i-1)} \rightarrow C_{(i)}$ is the cochain map $\cF_{\cR_X, i-1}$ or $\cF_{\cR_Z, i-1}$ depending on parity of the iteration index.

We now give a sufficient condition for a state with level-$2k$ syndrome $\sigma_{2k}$ to lie in the stable subspace $P_{\mrm{stable}}$. This is simply a (co)homological formulation of Claim~\ref{claim:stableconfig}.
\begin{definition} \label{def:alternative-stable}
A syndrome $\sigma_{2k} \in C^2_{(2k)}$ is unstable with respect to the decoder
if there exists a 1-cochain $e^1 \in C^1_{(2k)}$ with $|e^1| = 1$ (i.e., a single X flip) such that $$f(\sigma_{2k}) +  f(\sigma_{2k} + \delta^1_{2k} e^1)  + e^1$$
belongs to a nontrivial cohomology class $H^1(C_{(2k)})$.  Otherwise, it is called stable.
\end{definition}
This definition depends on the choice of the decoder. Informally, if a single $X$-flip acting on some syndrome changes the decoding outcome, then that syndrome is said to be unstable with respect to that decoder.

In this subsection, we will formulate a convenient necessary condition for a syndrome configuration to be unstable with respect to a decoder, which we will use in our memory lifetime bound. Before that, we will slightly modify our decoder.

\subsubsection{Modification of the decoder}
\label{sec:decoder-modification}

One can show that the syndrome reduction property of the decoder discussed in Sec.~\ref{sec:syndrome-reduction} guarantees the error-correction properties of the decoder, namely that the decoder will correct the errors up to size $\mrm{poly}(n_{2k})$, where $n_{2k}$ is the size of the code. However, for simplicity of some of the proofs, we would like to equip the decoder with a more explicit error-correction guarantee. For this, we will simply change the definition of the coarse-graining operations at level 2 and above (i.e. at levels 2,1, and 0).

Namely, at level $2$, we define the coarse-graining step to be the map $\sigma_2 \mapsto f_2 $ where the correction $f_2\in C^1_{(2)}$ satisfies
\begin{equation}
    \delta_2^1 f_2 = \sigma_2.
\end{equation}
We note that the code $C_{(2)} = \cR_Z \circ \cR_X (C_{(0)})$ has constant size (namely, $O(\ell^4)$), and the distance is at least $O(\ell)$. We choose any decoder (which can be done in many ways, including using a lookup table) for this finite-size code that fulfills the condition above and corrects all sufficiently small errors. For the proofs below, it is sufficient that the decoder corrects all errors whose syndrome fits inside a radius-2 ball in graph distance.

With this modification, we automatically have $\sigma_1 = 0$, $\sigma_0 = 0$, and thus, we do not need to define coarse-graining at these levels.

\subsubsection{Necessary condition for unstable configurations}

We now prove the following condition for a syndrome to be stable with respect to the decoder we constructed.

\begin{lemma}
Syndrome $\sigma_{2k}$ is stable if, upon applying the coarse-graining steps of the decoder iteratively, $\sigma_{2} = 0$.
\end{lemma}
\begin{proof}
We first prove some structural properties of the coarse-graining procedures.  These properties will be used afterwards to prove the lemma.

\vspace{0.15cm}
\noindent
\textbf{Level $i$ is even:  $\cR_Z$ coarse-graining.}   Consider two syndrome configurations $\sigma_i$ and $\wt \sigma_{i}$  that differ inside a subcomplex $X_{(i),x} \subseteq T_{(i)}$ for some $x \in V_{X,(i)}$.  Then, there exists an $x_{i-1} \in V_{X,(i-1)}$ such that $\sigma_i$ and $\wt \sigma_{i}$ only differ in the set $\{ Y_{(i),z}:  z \in \partial X_{(i-1),x_{i-1}}\}$. This set contains all complexes $Y_{(i),x'}$ such that $\mrm{dist}_i(x',x_{i-1})\leq 2$ (recall that $\mrm{dist}_i(\cdot,\cdot)$ is the directed distance on the decoding graph defined in Def.~\ref{def:locality}). Therefore, by local computability of the decoder, the corrections for $\sigma_i$ and $\wt \sigma_{i}$, which we call $f_i$ and $\wt f_i$, will only differ in the same set. Thus, after applying the correction and $\cF^{-1}_{\cR_Z,i-1}$, the coarse-grained syndrome $\sigma_{i-1}$ and $\wt \sigma_{i-1}$ will differ only in $\partial X_{(i-1),x_{i-1}} \subseteq X_{(i-1), x_{i-1}}$.

\vspace{0.15cm}
\noindent
\textbf{Level $i-1$ is odd:  $\cR_X$ coarse-graining.}
Consider now level $i-1$. There exists $x_{i-2} \in V_{X,(i-2)}$ such that the subcomplex $X_{(i-1), x_{i-1}}$ is contained in the set
$$ \{ Y_{(i-1),x_{i-2}} \} \cup \{Y_{(i-1),q}: q \in V_{Q,(i-2)}, \, q\sim x_{i-2}\} \cup\{Y_{(i-1),z}: z \in V_{Z,(i-2)}, \exists q \in V_{Q,(i-2)}, \, z \sim q\sim x_{i-2} \}.$$ This union already contains all complexes $Y_{(i),w}$ such that $\mrm{dist}_i(w,x_{i-2})\leq 2$. Therefore, by the local computability of the decoder, the associated corrections $f_{i-1}$ and $\wt f_{i-1}$ only differ within the union of these sets. Thus, after applying the correction and $\cF^{-1}_{\cR_X,i-2}$, the coarse-grained syndrome $\sigma_{i-2}$ and $\wt \sigma_{i-2}$ will differ only in $X_{(i-2),x_{i-2}}$.

\vspace{0.15cm}

\noindent \textbf{Proving the Lemma.} Now, we use the intermediate results above to prove the Lemma. Consider $\sigma_{2k} \in B^2(C_{(2k)})$; call $e_{2k} \in C^1_{(2k)}$ such that $\sigma_{2k} = \delta^1_{2k} e_{2k}$. Let $e^1 \in C^1_{(2k)}$ be an arbitrary weight-1 error $|e^1| = 1$, and call $\wt \sigma_{2k} = \delta^1_{2k} \wt e_{2k} = \delta^1_{2k} (e_{2k} + e^1)$. The syndrome difference $\delta^1_{2k} e^1$ is supported in $\partial X_{(2k),q_{2k}} \subseteq X_{(2k),x_{2k}}$ for some $x_{2k} \in V_{X, (2k)}$ and $q_{2k} \in V_{Q, (2k)}$.

Thus, we can repeatedly apply the conclusions from the two intermediate results above.  In particular, we can conclude that at level 2, we arrive at syndrome configurations $\sigma_2$ and $\wt \sigma_2$ which differ in $X_{(2),x_2}$ for some $x_2 \in V_{X,(2)}$.  By assumption, $\sigma_2 = 0$, so $\wt \sigma_2$ must be entirely supported in $X_{(2),x_2}$.

Next, we claim that there exists some $h_2 \in C^1_{(2)}$ that is supported entirely in $X_{(2),x}$ such that $\delta_2^1 h_2 = \wt \sigma_2$.  To show this, note that coarse-grained errors obey
\begin{align}
e_{j} &= \cF^{-1}_{j}(e_{j+1} + f_{j+1} + \delta^0 s_{j+1}) \\
\wt e_{j} &= \cF^{-1}_{j}(\wt e_{j+1} + \wt f_{j+1} + \delta^0 \wt s_{j+1})
\end{align}
From the intermediate results and by local computability, at each level, $f_{j+1} + \wt f_{j+1}$, and $\delta^0 s_{j+1} + \delta^0 \wt s_{j+1}$ are supported in $X_{(i+1), x_{i+1}}$ for some $x_{i+1} \in V_{X, (i+1)}$. Using this, we find that then $e_{j} + \wt e_{m}$ is also supported in $X_{(i), x_i}$ for some $x_i \in V_{X, (i)}$.
Repeatedly applying this argument starting from level $2k$ until level 2, we obtain $e_2 + \wt e_2$, which must be supported in $X_{(2),x_2}$. We define $h = e_2 + \wt e_2$ and obtain $$\delta_2^1 h = \delta_2^1 e_2 + \delta_2^1 \wt e_2 = \delta_2^1 \wt e_2 = \wt \sigma_2$$ as desired.

Since both $e_2 + \wt e_2$ and $\wt \sigma_2 = \delta_2^1 (e_2 + \wt e_2)$ are supported in $X_{(2),x_2}$, there exists a correction $\wt f_2$ applied by the level-2 decoder that is supported in $X_{(2),x_2}$ so that $e_2 + \wt e_2 + \wt f_2$ has no syndrome and is supported in $X_{(2),x_2}$.  For large enough $\ell$, $X_{(2),x}$ cannot contain a logical operator at level 2 and thus $[e_2 + \wt e_2 + \wt f_2]_2$ is a trivial element of $H^1(C_{(2)})$, where $[\alpha]_k$ denotes the homology class of $\alpha$ in  $H^1(C_{(k)})$.  No corrections will be applied at level 1 and level 0.  Repeatedly applying the map $\cF_i$ the way back to level $2k$, we get
\begin{equation}
[e_2 + \wt e_2 + \wt f_2]_2 = [e_{2k} + \wt e_{2k} + f(\sigma_{2k}) +  f(\sigma_{2k} + \delta^1_{2k} e^1)]_{2k} = [e^1 + f(\sigma_{2k}) +  f(\sigma_{2k} + \delta^1_{2k} e^1)]_{2k}
\end{equation}
and $e^1 + f(\sigma_{2k}) +  f(\sigma_{2k} + \delta^1_{2k} e^1)$ must in the trivial homology class of $H^1(C_{(2k)})$, implying that $\sigma_2$ is stable by Def.~\ref{def:alternative-stable}.
\end{proof}
Define $\Sigma_{2k}(\beta)$  to be the classical probability distribution over syndrome configurations such that $$\mathbb{P}(\sigma_{2k}) = \frac{e^{-\beta |\sigma_{2k}|}}{Z}, \quad \text{where} \ \ Z = \sum_{\sigma_{2k} \in B^2(C_{(2k)})} e^{- \beta|\sigma_{2k}|}.$$
\begin{corollary}\label{cor:projbound}
Define $\widetilde{P} = \sum_{\sigma_{2k}: \sigma_2 = 0} \ketbra{\sigma_{2k}}{\sigma_{2k}}$.  Then,
$$\mathbb{P}_{\sigma_{2k}\sim \Sigma_{2k}(\beta)}(\sigma_2 \neq 0) = \Tr (1 - \widetilde{P}) \rho_{\beta} \geq \Tr (1 - P) \rho_{\beta}.$$
\end{corollary}
\begin{proof}
Follows from the previous Lemma.
\end{proof}
Thus, we simply need to find an upper bound on $\mathbb{P}_{\sigma_{2k}\sim \Sigma_{2k}(\beta)}(\sigma_2 \neq 0)$.  To do so, we will show that the event $\sigma_2 \neq 0$ requires $\sigma_{2k}$ to be of a very large weight.  By bounding the number of distinct large-weight configurations, we will show that for $\beta$ large enough, $\mathbb{P}_{\sigma_{2k}\sim \Sigma_{2k}(\beta)}(\sigma_2 \neq 0)$ is doubly exponentially small in $k$.  In physics, this proof method often goes under the name of a Peierls argument.

\subsection{Recording coarse-graining history in the decoding graph}

In Sec.~\ref{def:locality}, we defined the decoding graph $G$. We will use it to track the history of the coarse-graining steps in decoding, i.e.
$$\sigma_{2k}, \ (\sigma_{2k-1},f_{2k}),
  (\sigma_{2k-2},f_{2k-1}),
...,
  (\sigma_0,f_1), \ f_0.$$
More specifically, we will store the indices of local complexes that contain the support of the syndrome and the coboundaries of local corrections. In particular, the decoding graph at level $i$, denoted $G_i$, has the same structure as the 1-skeleton of the Tanner graph $T_{(i-1)}$.
Note that $G$ is bounded-degree, and we will call the maximum degree $\Delta = O(\ell^2)$.  We will also define $\Delta_{\mrm{level}}$ to be the maximum degree of any subgraph $G_i$, which is independent of $\ell$.

For a particular choice of decoder, the highest-level syndrome $\sigma_{2k}$ uniquely determines the history of coarse-graining steps under decoding. At each level $i$, we track the supports of the syndrome as well as where the syndrome intermittently moved when we add local corrections $\delta f_{i,v}$ for $v \in T_{(i-1)}(0)$; we then mark the nodes in $G$ consistent with these supports. These induce a subgraph of $G$ that we call $G_{\sigma_{2k}}$.  A formal definition is given below:
\begin{definition}[Subgraph induced by syndrome $\sigma_{2k}$]
\label{def:subgraph}
Given $\sigma_{2k}$ and the associated coarse-graining data, define an induced subgraph $G_{\sigma_{2k}} \subseteq G$.  The edges and vertices of $G_{\sigma_{2k}}$ are determined by the following procedure:
\begin{enumerate}
    \item For each $2 \leq i \leq 2k$:
    \begin{enumerate}
        \item[(1.1)] Include each $v \in G_i$ if $\sigma_i \big|_{Y_v} \neq 0$.
        \item[(1.2)]Include each $v \in G_i$ if $f_{i,v} \neq 0$  or $\delta f_{i,v'} \big|_{Y_v} \neq 0$ for some $v'$ in $T_{(i-1)}(0)$.
        \item[(1.3)] If two nodes $v$ and $v'$ in $G_{\sigma_{2k}}$  are connected by an edge $(v,v')$ in $G_i$, include $(v,v')$ in $G_{\sigma_{2k}}$.
    \end{enumerate}
    \item If any node $v \in G_2$ is included in $G_{\sigma_{2k}}$, include all nodes $v' \in G_1$ where $v \in Y_{(1),v'}$.
    \item If any node $v \in G_1$ is included in $G_{\sigma_{2k}}$, include the node $v^{(0)} \in G_0$.
    \item For each $0 \leq i \leq 2k$: if $v \in G_{i-1}$ and $v' \in G_i$ are included in $G_{\sigma_{2k}}$ and are connected by an edge $(v,v')$ in $G$, include $(v,v')$ in $G_{\sigma_{2k}}$.
\end{enumerate}
\end{definition}
A schematic example of decoding subgraph is marked black in \Cref{fig:subgraph}. An illustration of how to determine a decoding subgraph is shown in Fig.~\ref{fig:level} for the $\cR_X$ iteration. For $\cR_Z$ iterations, the structure of the subgraph $G_{\sigma_{2k}}$ is simpler since corrections are only ever applied in $C_z$ subcomplexes; thus, only $z$-type nodes can be added to $G_{\sigma_{2k}}$ at such iterations.

\begin{figure}
    \centering
\includegraphics[width=0.8\textwidth]{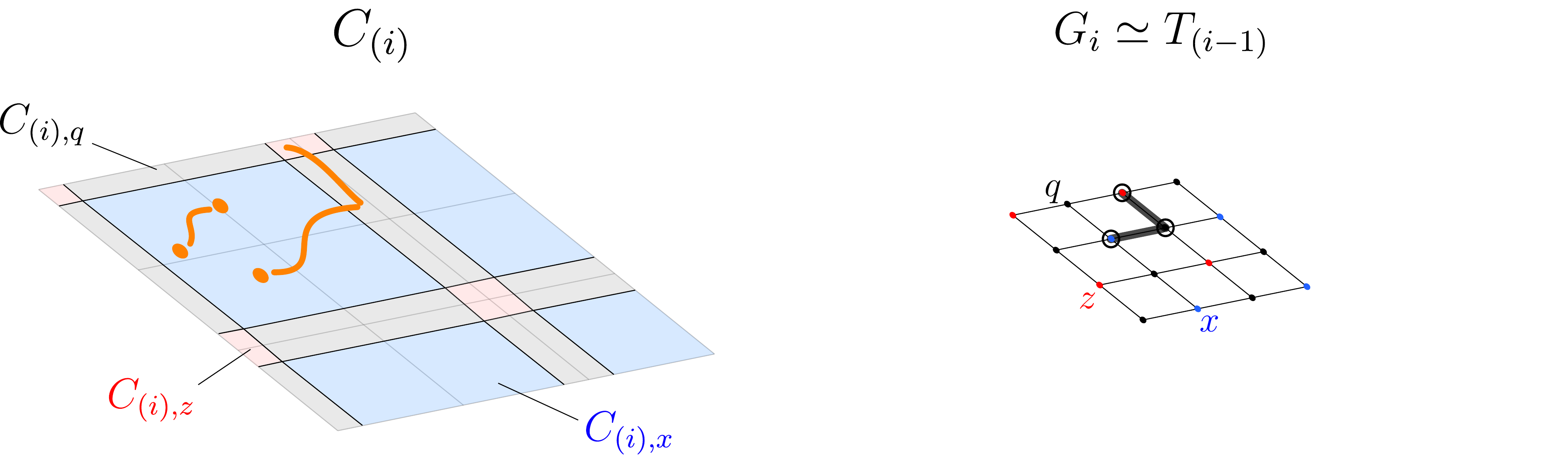}
\caption{Left: the orange points indicate the syndrome at level $i$, while the orange lines indicate the local corrections at the associated coarse-graining step.  Right: the circled vertices and bold links on the right indicate the nodes and edges to be added to the induced subgraph $G_{\sigma_{2k}}$ at level $i$.}
\label{fig:level}
\end{figure}

The induced subgraph $G_{\sigma_{2k}}$ consists of a set of connected components. We will now show that because of the local computability property, the weight of the syndrome in each connected component of $G_{\sigma_{2k}}$ is reduced independently.  Call $K$ a connected component of $G_{\sigma_{2k}}$, and decompose it into $K = \bigcup_i K_i$ where $K_i =  K \cap G_i$.  Define $\sigma_{K_i}$ to be the syndrome supported in $\bigcup_{v \in K_i} Y_{(i), v}$.
\begin{claim} \label{claim:disjoint-clusters}
    Consider two different connected components of $G_{\sigma_{2k}}$, which we denote $K$ and $K'$. The syndrome $\sigma_{K_{i-1}}$ only depends on $\sigma_{K_i}$ and not on $\sigma_{K'_{i-1}}$.
\end{claim}
\begin{proof}
    Consider any $v \in K_{i }\subseteq K$ labeling local complex $Y_{(i), v}$. Suppose (for the sake of contradiction) that there exists $u \in K'_{i-1} \subseteq K'$ corresponding to $Y_{(i-1), u}$ such that the syndrome restricted to $Y_{(i-1), u}$, denoted $\sigma_{i-1, u} = \sigma_{i-1}\big|_{Y_{(i-1), u}}$, depends on the syndrome restricted to $Y_{(i), v}$, denoted $\sigma_{i, v}$.  We can compute the functional dependence of $\sigma_{i-1, u}$ on the syndrome in $\sigma_i$ via
    \begin{align}
    \cF^{-1}(\sigma_{i-1}\big|_{Y_{(i-1),u}}) = \cF^{-1}(\sigma_{i-1}) \big|_{U_u} = \big( \sigma_i + \sum_{v' \in T_{(i-1)}(0)} \delta^1_i f_{i,v'}\big) \big|_{U_u}
    \end{align}
    where $U_u = \{Y_{(i),w}:  w \in Y_{(i-1),u}\}$.

    We now have two cases to consider.
    First, if $\sigma_i \big|_{U_u}$ depends on $\sigma_{i, v}$, then $Y_{(i), v}$ must be contained in $U_u$. Note that both $v$ and $u$ must be in $G_{\sigma_{2k}}$. By Def.~\ref{def:subgraph}, $v$ must be connected to $u$, contradicting the assumption that $K$ and $K'$ are different connected components.

    Second, suppose that for some $v' \in T_{(i-1)}(0)$, $(\delta^1_i f_{i,v'} )\big|_{U_u}$ depends on $\sigma_{i, v}$.
    By local computability, $v$ must be within directed distance $1$ away from $v'$. Because $f_{i,v'}$ is nonzero, $v'$ must be in $G_{\sigma_{2k}}$. By Def.~\ref{def:subgraph}, the neighboring $v'$ and $v$ must be in the same connected component.

    For $(\delta^1_i f_{i,v'}) \big|_{U_u}$ to be nonzero,  $\delta^1_i f_{i,v'}$ must have support in $Y_{(i), w} \in U_u$ for some $w$. Thus, $v'$ and $w$ are both in $G_{\sigma_{2k}}$ and are in the same connected component; however, by definition, $w$ is in the same connected component as $u$.
    From Def.~\ref{def:subgraph}, this implies that $u$ and $v$ must be in the same connected component,  contradicting the assumption.

\end{proof}

The main result of this subsection is below, which shows that we can analyze the action of the decoder independently on each connected component.

\begin{lemma} \label{lemma:locality-decoder}
 Consider a connected component $K$ of $G_{\sigma_{2k}}$. Denote $\sigma'_i$ the part of syndrome  $\sigma_i$ contained in local regions associated with $K$, and denote $\sigma'  = \cup_i \sigma'_i$.  Define $f_{i}'$, to be the sum of local corrections in $f_i$ that depend on $\sigma'_i$.
    \begin{enumerate}
        \item[(1)] $\sigma'_i + \delta_i^1 f'_i$ is in the image of $\cF_{i-1}$, and $\cF^{-1}_{i-1}(\sigma'_{i-1}) = \sigma'_i + \delta_i^1 f'_i.$
        \item[(2)] For the $\cR_Z$ iteration,
        \begin{equation}
            |\sigma'_{i-1}| \leq |\sigma'_i|.
        \end{equation}
    and for the $\cR_X$ iteration,
        \begin{equation}
            |\sigma'_{i-1}| \leq \frac{1}{2} |\sigma_{i}'|.
        \end{equation}
    \end{enumerate}
\end{lemma}
\begin{proof}
Follows from local computability of the coarse-graining map, Claim~\ref{claim:disjoint-clusters}, as well as \Cref{thm:coarse-graining-RZ} and \Cref{thm:coarse-graining-RX}.
\end{proof}

\subsection{Witness}

The subgraph $G_{\sigma_{2k}}$ induced by syndrome $\sigma_{2k}$  can be split into disjoint connected components. If it contains a component that contains the node $v^{(0)}$, we call such a component a \emph{witness subgraph}. Because of how the subgraph $G_{\sigma_{2k}}$ is defined, if a witness subgraph is present, the syndrome configuration $\sigma_{2k}$ must be unstable. We formalize this below and use this definition to upper bound the number of such subgraphs and lower bound their minimum size.

\begin{definition}[Witness subgraph and witness syndrome]
Consider a syndrome configuration $\sigma_{2k}$ and construct the associated subgraph $G_{\sigma_{2k}}$ according to Def.~\ref{def:subgraph}.  We call the connected component of the subgraph containing $v^{(0)}$ the witness subgraph of $\sigma_{2k}$ (which could be empty), which we denote $G_\tau$.

Define $\tau_i$ to be the syndrome in $\sigma_i$  restricted to the regions used to define the witness subgraph at level $i$.
Define the witness syndrome $\tau$ of $\sigma_{2k}$ to be the union $\tau = \bigcup_i \tau_i$. We call $\tau_i$ the witness syndrome at level $i$.
\end{definition}
We note that different witness syndrome configurations (we will call these ``witnesses'' moving forward) can correspond to the same witness subgraph.  We will account for this later on in the proof.
\begin{lemma}\label{lem:bndfoot}
The level-$2k$ witness syndrome of $\sigma_{2k}$, denoted $\tau_{2k}$, is a coboundary of some 1-cochain $h^1 \in C^1_{(2k)}$.
\end{lemma}
\begin{proof}
We provide a construction of $h^1$.  Recall that the vertices in the witness subgraph of $\sigma_{2k}$, denoted $G_{\tau}$, form a connected component.
By Lemma~\ref{lemma:locality-decoder},
$\tau_i + \delta^1_i\dec_{i}(\tau_{i}) = \cF(\tau_{i-1})$.
Define
\begin{equation}
 p = \sum_{i \geq 2} \cF_{2k - 1} \circ \cF_{2k - 2}  \circ \dots \circ \cF_{i}  \left( f_{i}(\tau_{i})\right) \in C^1_{(2k)}
\end{equation}
Applying Lemma~\ref{lemma:locality-decoder} iteratively, we obtain that $\cF_{2k-1} \circ \dots \circ \cF_{2} (\tau_2) = \tau_{2k} + \delta^1_{2k} p$.  By definition of the witness subgraph and witness syndrome, $\tau_2 = \sigma_2 \in B^2(C_{(2)})$. Using that $\cF_i$ is a cochain map, it follows $\cF_{2k-1} \circ \dots \circ \cF_{2} (\tau_2)$ is also a coboundary of some 1-cochain, which we call $e \in C^1_{(2k)}$.

Finally, we have $\tau_{2k} = \delta^1_{2k} e + \delta^1_{2k}  p$. Identifying $h = e + p$ shows the statement of the Lemma.
\end{proof}
We will now bound the number of all possible subgraphs corresponding to valid witnesses.  We then argue that the size of the witness subgraph can be related to the weight of the level-$2k$ witness $|\tau_{2k}|$.
\begin{lemma}\label{lem:memlifetimelem3}
The number of witness subgraphs with $m$ nodes is $\leq |V_2| (\Delta+1)^{2m}$.
\end{lemma}
\begin{proof}
For a connected component of $G_{\sigma_{2k}}$ to be a witness subgraph, it is sufficient for it to contain any $v^{(2)} \in G_2$. Thus, we will need to bound the number of connected subgraphs of size $m$ containing $v^{(2)}$ and multiply by $|V_2|$ for an upper bound.

To prove this, perform a depth-first traversal of the graph $G_{\tau}$ starting at a given node $v^{(2)}$.  More specifically, for a given node with $\Delta$ incident edges, we assign distinct integers $1,2,\cdots, \Delta$ to each edge.  A traversal of $G_{\tau}$ can then be labeled by a list of integers $\{n_1, n_2, n_3, \cdots\}$, with $n_i \in [0,\Delta]$.  This list can be interpreted the following way: if one is at node $u$ and the element at the front of the list is $n_i$, we move along the $n_i$-th edge if $n_i \in [1,\Delta]$, and backtrack along the edge that brought us to $u$ if $n_i = 0$.  We then pop $n_i$ from the list and repeat until the list is empty.  It follows that for a depth-first traversal of a graph with $m$ nodes, the number of elements in the list is at most $2m$.  The number of possible lists is then at most $\leq (\Delta+ 1)^{2m}$.

Finally, since $G_{\tau}$ is an induced subgraph, different witness subgraphs will have different traversal lists, implying the thesis.
\end{proof}
Now we provide a lower bound on the number of nodes in the witness subgraph corresponding to an unstable configuration.  This is the key lemma in which the results of the previous section are used.
\begin{lemma}\label{lem:memlifetimelem4}
An unstable configuration of the level-$2k$ construction has a witness subgraph with $\geq 2^{k-1}$ nodes.
\end{lemma}
\begin{proof}
Call $V_{\tau}$ its vertex set of the witness subgraph $G_\tau$.  We can decompose $V_{\tau} = \bigcup_i V_{\tau, i}$ where $V_{\tau, i} = V_{\tau} \cap V_i$ denotes the vertices of the witness subgraph at level $i$.  Consider the witness $\tau_{i}$ at level $i$.
Each syndrome point in $\tau_{i}$ accounts for at least one node in the witness subgraph in $V_{\tau,i+1}$, namely the one labeled by $v \in V_{Z,(i)}$. This is because, if there is a syndrome point in $\tau_i$, some subcomplex at a lower level $Y_{(i+1),v}$ must contain either $\tau_{i+1}$, or the correction at level $i+1$, or the coboundary of the correction at level $i+1$.

Thus, the total number of nodes in the witness subgraph obeys
\begin{equation}
|V_{\tau}| = \sum_{i=2}^{2k} |V_{\tau, i}| \geq \sum_{i=3}^{2k} |\tau_{i-1}|.
\end{equation}
By the cluster syndrome reduction property in Lemma~\ref{lemma:locality-decoder}, we have $|\tau_{i-1}| \leq |\tau_i|/r$ where $r = 1$ for the $\cR_Z$ iteration and $r = 2$ for the $\cR_X$ iteration.  Since an unstable configuration must have at least one syndrome at level 2, summing the geometric series gives $|V_{\tau}| \geq 2^{k-1}$.
\end{proof}

\begin{lemma}\label{lem:memlifetimelem6}
The witness $\tau$ of syndrome $\sigma_{2k}$ satisfies $\sum_{i=0}^{2k} |\tau_i| \leq 4 |\tau_{2k}|$.
\end{lemma}
\begin{proof}
We first note that $\tau_{i}$ is a union of connected components; we label each $\tau_i^c$. From Lemma~\ref{lemma:locality-decoder}, after applying a coarse-graining step, the syndrome supported in the same units that $\tau_{i}$ is in, denoted $\tau_i'$, satisfies
\begin{equation}
|\tau_i'| \leq \sum_c |\tau_i'^c| \leq \sum_c |\tau_i^c| \leq |\tau_i|
\end{equation}
By Lemma \ref{lemma:locality-decoder} $|\tau_{i-1}| \leq  |\tau_i|/r$ where $r$ equals $1$ or $2$ depending on the iteration.
Therefore,
\begin{equation}
\sum_{i=0}^{2k} |\tau_i| \leq 2 |\tau_{2k}| \sum_{i=0}^{\infty} \left(\frac{1}{2}\right)^i \leq 4 |\tau_{2k}|.
\end{equation}
\end{proof}

\begin{lemma}~\label{lem:memlifetimelem5}
Suppose the witness subgraph of $G_{\sigma_{2k}}$ has $m$ nodes.  Then, the total number of nodes in the witness subgraph whose associated local subcomplex contains at least one syndrome point is $\geq m/(1+\Delta_{\mrm{level}}^4)$.
\end{lemma}
\begin{proof}
Call $G_{\sigma_{2k},i} = G_{\sigma_{2k}} \cap G_i$ and $G_{\tau,i} = G_{\tau} \cap G_i$.  Recall that a node is added to $G_{\sigma_{2k}}$ if either the subcomplex labeled by this node contains syndrome at that level or if it contains the correction 1-cochain (or the coboundary of the correction) at this level. By construction of $G_{\sigma_{2k}}$, if $v \in G_{\sigma_{2k}, i}$ but the associated subcomplex $Y_{(i),v
}$ does not have syndrome points, then the correction $f_{i,v}$ must be nontrivial.  By local computability of the correction, there must be a node $u$ in the same connected component a graph distance $1$ from $v$ such that $Y_{(i),u
}$ has at least one syndrome point associated with $\sigma_i$. The correction $f_{i,v}$ can deposit syndrome in $Y_{(i),w
}$ where $w$ is distance 1 from $v$ and in the same connected component; thus, $u$ and $w$ are at most graph distance 2 from each other.  Thus, any node $v \in G_{\sigma_{2k}, i}$ must be within graph-distance 2 away from some node $v' \in G_{\sigma_{2k}, i}$ in the same connected component whose associated subcomplex $Y_{(i),v
}$ contains at least one syndrome point in $\sigma_i$ (if the distance is 0, then $v' = v$).
Because $G_\tau \subseteq G_{\sigma_{2k}}$ and forms a single connected component, this statement is also true for $G_\tau$.

Suppose there are $m_i$ nodes in $G_{\tau,i}$, and thus, $\sum_i m_i = m$.  The number of nodes responsible for at least one syndrome point is lower-bounded by the distance-$4$ domination number of $G_{\tau,i}$, denoted $\gamma_4(G_{\tau,i})$.  Then, the total number of such nodes is
\begin{equation}
\geq \sum_i \gamma_4(G_{\tau,i}) \geq \sum_i \frac{m_i}{\Delta_{\mrm{level}}^4 + 1} = \frac{m}{\Delta_{\mrm{level}}^4 + 1}.
\end{equation}
\end{proof}

\subsection{Main result}

Finally, we use the above results to argue that the probability of an unstable configuration is very small, by combining the bound on the number of possible distinct witnesses with the minimum energy of such a witness.  This uses a standard Peierls argument.
\begin{theorem}\label{thm:mainresult}
The probability $\mathbb{P}_{\sigma_{2k}\sim \Sigma_{2k}(\beta)}(\sigma_2 \neq 0)$ is $\leq 2\cdot2^{-2^{k-1}}$ for $\beta > \beta_c$ with $\beta_c$ finite.
\end{theorem}
\begin{proof}
Suppose $\sigma_2$ has at least one syndrome.  Then, $\sigma_{2k}$ must contain a level-$2k$ witness, i.e. $\tau_{2k} \subseteq \sigma_{2k}$. Assuming $\sigma_{2k}\sim \Sigma_{2k}(\beta)$, the probability of $\sigma_{2k}$ containing a fixed $\tau_{2k}$ is
\begin{align}
\mathbb{P}_{\sigma_{2k}\sim \Sigma_{2k}(\beta)}(\tau_{2k} \subseteq \sigma_{2k} ) &= \frac{1}{Z} \sum_{\sigma_{2k}: \tau_{2k} \subseteq \sigma_{2k}} e^{-\beta |\sigma_{2k}|} \\
&\leq e^{-\beta |\tau_{2k}|} \frac{\sum_{\sigma_{2k}: \tau_{2k} \subseteq \sigma_{2k}} e^{-\beta |\tau_{2k} + \sigma_{2k}|}}{Z} \\
&\leq e^{-\beta |\tau_{2k}|}
\end{align}
 In the third line, we use the fact that $\tau_{2k} + \sigma_{2k}$ is a 2-coboundary (which follows from Lemma~\ref{lem:bndfoot} and $\sigma_{2k}$ being a 2-coboundary by definition); thus, all terms in the numerator are also present in the denominator.

Next, we use the Lemmas from the previous subsection.  From Lemma~\ref{lem:memlifetimelem4} the number of nodes in the witness subgraph is $\geq 2^{k-1}$.  Then, we may write
\begin{align}
\mathbb{P}_{\sigma_{2k}\sim \Sigma_{2k}(\beta)}(\sigma_2 \neq 0) &\leq \sum_{\text{valid } \tau_{2k}} \mathbb{P}_{\sigma_{2k}\sim \Sigma_{2k}(\beta)}(\tau_{2k} \subseteq \sigma_{2k}) \\
&\leq \sum_{\substack{\text{subgraphs} \\ K}} \ \sum_{ \substack{\tau_{2k} \,  \text{consistent} \\ \text{ with } K}} \mathbb{P}_{\sigma_{2k}\sim \Sigma_{2k}(\beta)}(\tau_{2k} \subseteq \sigma_{2k}) \\
&\leq \sum_{\substack{\text{subgraphs} \\ K}} \ \sum_{ \substack{\tau_{2k} \,  \text{consistent} \\ \text{ with } K}} e^{-\beta |\tau_{2k}|} \\
&\leq \sum_{m \geq 2^{k-1}} \, \sum_{K: |V_K| = m} \ \sum_{ \substack{\tau_{2k} \,  \text{consistent} \\ \text{ with } K}} e^{-\beta |\tau_{2k}|}
\end{align}
where: $K$ denotes a witness subgraph, the statement ``$\tau_{2k}$ consistent with $K$'' means that $K$ is a validly constructed subgraph given $\tau_{2k}$, and the Lemma is used in the last line.  Now, we focus on the inner sum.
\begin{equation}
\sum_{ \substack{\tau_{2k} \,  \text{consistent} \\ \text{ with } K}} e^{-\beta |\tau_{2k}|} \leq \sum_{ \substack{\tau_{2k} \,  \text{consistent} \\ \text{ with } K}} e^{-\frac{\beta}{4} \sum_{i=1}^{2k}|\tau_{i}|}
\end{equation}
where we used Lemma~\ref{lem:memlifetimelem6}.  Suppose for the fixed choice of $K$ with vertex set $V_K$ satisfying $|V_K| = m$, the set of nodes whose associated subcomplexes contain at least one syndrome point is $\widetilde{V}_{K}$.  Then, we may write
\begin{align}
\sum_{ \substack{\tau_{2k} \,  \text{consistent} \\ \text{ with } K}} e^{-\frac{\beta}{4} \sum_{i=1}^{2k}|\tau_{i}|} &\leq \sum_{\widetilde{V}_K} \sum_{ \substack{\tau_{2k} \,  \text{consistent} \\ \text{ with } K, \wt V_K}}  e^{-\frac{\beta}{4} \sum_{i=1}^{2k}|\tau_{i}|}  \\
&\leq \sum_{\widetilde{V}_K} \prod_{v \in \widetilde{V}_K} W_v
\end{align}
where $W_v$ is the weighted sum over all valid syndrome configurations $\sigma_v$ with $|\sigma_v| \geq 1$ in the subcomplex corresponding to $v$:
\begin{equation}
W_v \leq \sum_{j = 1}^{M_\ell} \binom{M_\ell}{j} e^{-\beta j/4} = (e^{-\beta/4} + 1)^{M_\ell} - 1.
\end{equation}
where $M_\ell = O(\ell^2)$ is the maximum size of the local complex.
This implies
\begin{align}
\sum_{ \substack{\tau_{2k} \,  \text{consistent} \\ \text{ with } K}} e^{-\beta |\tau_{2k}|} &\leq \sum_{\widetilde{V}_K} ((e^{-\beta/4} + 1)^{M_\ell} - 1)^{|\widetilde{V}_K|} \\
&\leq \sum_{\widetilde{V}_K} ((e^{-\beta/4} + 1)^{M_\ell} - 1)^{m/(1+\Delta^4_{\mrm{level}})} \\
&\leq 2^m  ((e^{-\beta/4} + 1)^{M_\ell} - 1)^{m/(1+\Delta^4_{\mrm{level}})}
\end{align}
where in the second line we assume $(e^{-\beta/4} + 1)^{M_\ell} - 1 < 1$ and invoked Lemma~\ref{lem:memlifetimelem5}.  Putting this together,
\begin{align}
\mathbb{P}_{\sigma_{2k}\sim \Sigma_{2k}(\beta)}(\sigma_2 \neq 0) &\leq
\sum_{m \geq 2^{k-1}} \sum_{K: |V_K| = m} 2^m  ((e^{-\beta/4} + 1)^{M_\ell} - 1)^{m/(1+\Delta^4_{\mrm{level}})} \\
&\leq \sum_{m \geq 2^{k-1}} (\Delta + 1)^{2m} 2^m  ((e^{-\beta/4} + 1)^{M_\ell} - 1)^{m/(1+\Delta^4_{\mrm{level}})}.
\end{align}
where in the second line we use Lemma~\ref{lem:memlifetimelem3}.  Finally, if
\begin{align}
2 \big((e^{-\beta/4} + 1)^{M_\ell} - 1 \big)^{1/(1+\Delta^4_{\mrm{level}})}(\Delta + 1)^2 \leq \frac{1}{2}
\end{align}
which upon using $\log(1 + x) \leq x \log 2$ for $x \in [0,1]$ implies
\begin{equation}
e^{-\beta} \leq \frac{1}{\Delta_{\mrm{level}}^2 M_\ell} \frac{\log\left(1 + \left(\frac{1}{4(\Delta + 1)^2}\right)^{\Delta_{\mrm{level}}^4+1}\right)}{ \log 2}
\end{equation}
or $\beta \geq \beta_c \asymp  \Delta_{\mrm{level}}^4 \log( \ell)$, then
\begin{align}
\mathbb{P}_{\sigma_{2k}\sim \Sigma_{2k}(\beta)}(\sigma_2 \neq 0) &\leq
 \sum_{m \geq 2^{k-1}} \left(\frac{1}{2}\right)^m \leq 2\cdot 2^{-2^{k-1}}.
\end{align}
\end{proof}
\begin{corollary}
The memory lifetime of the 3D passive quantum memory is $\gtrsim 2^{2^{k-1}}$ at inverse temperatures $\beta \geq \beta_c$, with $\beta_c$ determined from Thm.~\ref{thm:mainresult}.
\end{corollary}
\begin{proof}
From Cor.~\ref{cor:projbound} and Thm.~\ref{thm:mainresult}, we have $$2\cdot 2^{-2^{k-1}} \geq \Tr(1-P_X)\rho_{\beta},$$ where the subscript indicates we are restricting to $X$-type errors.  Next, we must treat $Z$-type errors.  The analysis is identical apart from the $\cR_X$ and $\cR_Z$ iterations being exchanged, and we get a similar bound $2\cdot 2^{-2^{k-1}} \geq \Tr(1-P_Z)\rho_{\beta}$.  The right hand side of Thm.~\ref{thm:memorylifetime} can be bounded by $$\Tr(1-P_\mrm{stable})\rho_{\beta} \leq 4\cdot 2^{-2^{k-1}}.$$ Thus, from Thm.~\ref{thm:memorylifetime}, the norm $\norm{e^{-i[H,\cdot] t + \cL t} \circ \mathrm{Enc}(\rho) - \mathrm{Enc}(\rho)}_1 \leq \epsilon$ for all $t \lesssim \epsilon \cdot 2^{2^{k-1}}$.
\end{proof}

\section{Construction with explicit embedding}
\label{sec:explicit-embedding}

Up to this point, we have established the existence of a 3D self-correcting quantum memory,
  where the embedding in $\RR^3$ is achieved via a random perturbation method.
This construction can be made deterministic
  using a deterministic algorithm for the Lov\'asz local lemma \cite{chandrasekaran2013deterministic}.
However, it still incurs a large scale factor $\ell \lesssim 10^{24}$.  As a result, the exponent $\eta$ defined via $t_{\mrm{mem}}\sim \exp(n^\eta)$ is only $\sim 10^{-2}$.
While this estimate can be improved with a more careful analysis of the random perturbation,
  it remains desirable to obtain a more explicit 3D construction that might achieve an even larger value for $\eta$.

In this section,
  we present an explicit construction with scale factor $\ell = 8$.
We achieve this through an inductive argument.
Rather than perturbing the code randomly,
  we instead identify a finite number of ``local structures''
  and specify an explicit replacement rule for each structure.  These rules allow us to generate the code at the next step while ensuring locality in $\mathbb{R}^3$.

As a warm-up, we first illustrate the underlying idea in a simpler classical setting
  before turning to the quantum code.
While the construction in Sec.~\ref{sec:construction} is based on the Tanner square complex,
  here we instead use the sheaf codes formalism,
  discussed in Sec.~\ref{prelims:sheaf}.  An equivalent framework is that of topological defect networks~\cite{williamson2023layer,Aasen_2020}.

\subsection{Warm-up: classical code}\label{sec:classicalcode}

The authors of Ref.~\cite{lin2024proposals}
posed the following question: does there exist a 2D classical code with parity-check matrix $H$ such that both the code with parity-check matrix $H$ and the code with parity-check matrix $H^T$ are self-correcting classical memories? We answer this question in the affirmative.

\begin{theorem}\label{thm:classical-memory}
There exists a scaling family of two-dimensional classical LDPC codes with parity-check matrices $H_k$ such that both the code associated with $H_k$ and the code associated with $H_k^T$ are classical memories. Moreover, their memory lifetime satisfies
$t_{\mathrm{mem}} \ge \exp \left ( \Theta(n_k^\eta) \right ) $,
where $n_k \to \infty$ as $k \to \infty$ is the number of physical bits and $\eta > 0$ is a constant.
\end{theorem}
The goal of this section is to construct such a code family.
As in the quantum construction described in Sec.~\ref{sec:construction},
  the construction again consists of two types of iterations,
  denoted by $\cR_{\mrm{Rep}}$ and $\cR_{\mrm{Par}}$.

\subsubsection{Notation and framework}

In this section, we specialize to the case of 1-dimensional sheaf codes,
  which correspond to a classical geometrically local code
  $\cC(X, \cF, I)$ defined by the following data:

\begin{enumerate} [itemsep=-0.05em,topsep=4pt]
  \item A 1-dimensional cell complex $X$,
        with vertex set $X(0)$ and edge set $X(1)$.
  \item Sheaf data $\cF$,
        which assigns to each vertex $v$ of degree $\ne 2$ a color,
          either black or white.
  \item An embedding $I: X \to \R^2$.
\end{enumerate}

We will assume that each edge of $X$ is incident to two distinct vertices. For a vertex $v$, denote
$X_{\ge v}(1)$
to be the set of edges incident to $v$.
If $v$ is black, then the associated local code
  $\cF_v \subseteq \ff_2^{X_{\ge v}(1)}$
  is taken to be the repetition code;
  for example, when $|X_{\ge v}(1)|=3$, we set $\cF_v = \{000,111\}$.
If $v$ is white, then $\cF_v$ is taken to be the parity-check code;
  for example, when $|X_{\ge v}(1)|=3$, we set $\cF_v = \{000,011,101,110\}$.
For vertices of degree $2$ we always set
  $\cF_v=\{00,11\},$
  which is both a repetition code and a parity-check code.
For this reason, we do not explicitly assign a color to degree-$2$ vertices.

As we will see, the embedding $I$ is local and has bounded density. In our construction, each vertex is mapped to a lattice point in $\zz^2$, and each edge is mapped to a unit-length segment joining two adjacent lattice points. Moreover, no two edges of $X$ are mapped to the same edge of the integer lattice.

We again use $\ell$ to denote the scaling factor at each iteration.
While the construction is valid for all $\ell \ge 3$,
  we focus on the case $\ell = 4$.

\subsubsection{Construction}

We begin with a finite-size code $\cC_{(0)}$ and define two procedures, $\cR_{\mrm{Rep}}$ and $\cR_{\mrm{Par}}$. This produces a family of codes
\begin{equation}
  \cC_{(0)},\,
  \cC_{(1)}=\cR_{\mrm{Rep}}(\cC_{(0)}),\,
  \cC_{(2)}=\cR_{\mrm{Par}}(\cC_{(1)}),\,
  \cdots,\,
  \cC_{(2k-1)}=\cR_{\mrm{Rep}}(\cC_{(2k-2)}),\,
  \cC_{(2k)}=\cR_{\mrm{Par}}(\cC_{(2k-1)}),\,
  \cdots
\end{equation}

\paragraph{Base case.}

The initial code $\cC_{(0)}$ is specified by the following data:
\begin{enumerate}[itemsep=-0.05em,topsep=4pt]
  \item The $1$-cell complex $X_{(0)}$ consists of four $0$-cells and four $1$-cells arranged in a loop.
  \item Each vertex $v$ has degree $2$, and hence the sheaf structure at every vertex is fixed to be $\cF_{(0)}(v):=\{00,11\}$.
  \item The four edges are embedded in $\RR^2$ as the boundary of a unit square.
\end{enumerate}

It is clear that the classical code $\cC_{(0)}$
  encodes a single bit.

\paragraph{Iterations $\cR_{\mrm{Rep}}$ and $\cR_{\mrm{Par}}$.}

Each iteration consists of two steps:
\begin{enumerate}[itemsep=-0.05em,topsep=4pt]
  \item Refinement: $X \mapsto X'$.
  \item Doubling: $X' \mapsto Y=(X' \sqcup X'_\mrm{copy} \sqcup X_{\mrm{cyl}})/{\sim}$.
\end{enumerate}

During refinement, the geometry is modified from $X$ to $X'$, while the underlying topological structure remains unchanged.
During the doubling step, the complex $X'$ is replaced by a new complex $Y$ consisting of two copies of $X'$ together with a cylinder $X_{\mrm{cyl}}$ joining them. The equivalence relation $\sim$ identifies the two boundaries of $X_{\mrm{cyl}}$ with the corresponding attachment regions in the two copies of $X'$. These structures will be specified explicitly below.

The two types of iteration, $\cR_{\mrm{Rep}}$ and $\cR_{\mrm{Par}}$, differ only in the sheaf data assigned during the doubling step. We therefore begin by describing the common geometric part of the construction, namely the complex $X$ and its embedding $I$, and then specify the corresponding sheaf data for the $\cR_{\mrm{Rep}}$ and $\cR_{\mrm{Par}}$ iterations.

\paragraph{Refinement.}

The refinement step is straightforward,
  consisting of a scaling and subdivision procedure.
We scale the entire complex by a factor of $\ell$ and subdivide each edge into $\ell$ segments. The resulting complex $X'$ inherits a natural embedding $I'$, obtained by scaling the original embedding $I$.

The sheaf data $\cF'$ is induced directly from $\cF$: each vertex of $X'$ of degree different from $2$ inherits the color of the corresponding vertex of $X$.

\paragraph{Doubling.}

Starting from the subdivided complex $X'$, we form a second copy, denoted $X'_{\mrm{copy}}$, and translate it by the vector $(1,1)$. We then connect $X'$ and $X'_{\mrm{copy}}$ by an auxiliary cylinder-like complex, denoted $X_{\mrm{cyl}}$.

To define $X_{\mrm{cyl}}$, we first specify the attaching regions $X'_0 \subset X'$ and $X'_{\mrm{copy},0} \subset X'_{\mrm{copy}}$.
These are $0$-dimensional subcomplexes. The vertices in $X'_0$ are paired with corresponding vertices in $X'_{\mrm{copy},0}$, and $X_{\mrm{cyl}}$ is obtained by a collection of unit-length edges between each pair.

More concretely, suppose that $(a,b)-(a+1,b)$ is an edge of $X$.
In the subdivided complex $X'$,
this edge is replaced by a path of $\ell$ unit segments joining
$(\ell a,\ell b)$ to $(\ell(a+1),\ell b)$.
In the translated copy $X'_{\mathrm{copy}}$,
the corresponding path joins
$(\ell a+1,\ell b+1)$ to $(\ell(a+1)+1,\ell b+1)$.
We define the corresponding attaching regions to consist of the interior
vertices on these two paths that can be joined vertically:
\begin{align}
  X'_0
  &\supseteq
  \{(\ell a + 2,\ell b),(\ell a + 3,\ell b),\ldots,
  (\ell a + (\ell-1),\ell b)\}, \\
  X'_{\mathrm{copy},0}
  &\supseteq
  \{(\ell a + 2,\ell b + 1),(\ell a + 3,\ell b + 1),\ldots,
  (\ell a + (\ell-1),\ell b + 1)\}.
\end{align}
These vertices pair naturally to form the unit-length segments in
$X_{\mathrm{cyl}}$:
\begin{equation}
  X_{\mathrm{cyl}}
  \supseteq
  \{
    (\ell a + m,\ell b) - (\ell a + m,\ell b + 1)
    : m \in \{2,3,\ldots,\ell-1\}
  \}.
\end{equation}
Thus, for each edge of $X$, the corresponding paths in $X'$ and $X'_{\mrm{copy}}$, together with the added edges in $X_{\mrm{cyl}}$, form a subcomplex shaped like a ladder.

\paragraph{Sheaf data.}

It remains to assign a color to each vertex of $Y$ whose degree is not $2$. There are two types of such vertices.
The first consists of vertices in $X'$ and $X'_{\mrm{copy}}$ that do not lie in the attaching regions.
The second consists of vertices lying in the attaching regions.

Vertices of the first type inherit the color of the corresponding vertex in $X$. Vertices of the second type are colored black in the iteration $\cR_{\mrm{Rep}}$ and white in the iteration $\cR_{\mrm{Par}}$.

This completes the construction of the code $\cC_{(i+1)}$ from $\cC_{(i)}$. In particular, the size of $\cC_{(k)}$ grows exponentially in $k$.

\begin{figure}[ht]
  \centering
  \includegraphics[width=0.65\textwidth]{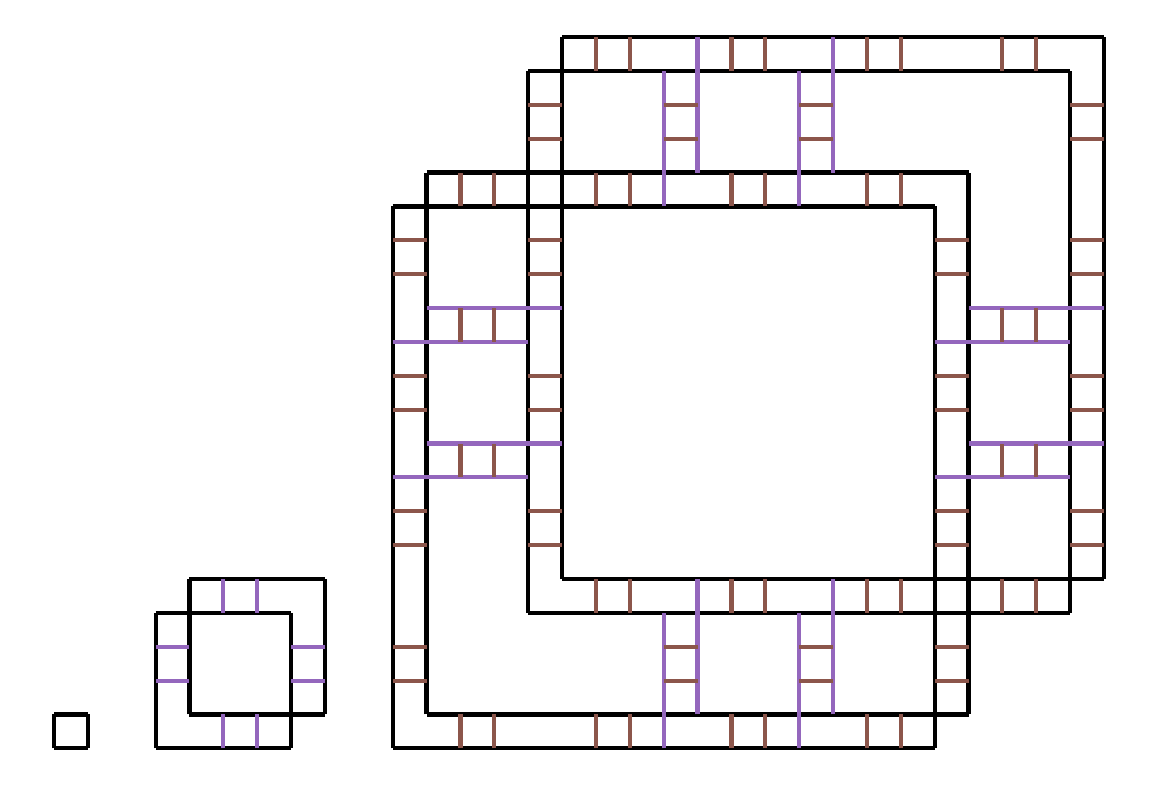}
  \caption{The first two iterations for $\ell=4$, showing $X_{(0)}$, $X_{(1)}$, and $X_{(2)}$.
    Black edges are obtained from scaling and copying $X_{(0)}$.
    Purple edges are the new cylinder edges added in $X_{(1)}$, and brown edges are the new cylinder edges added in $X_{(2)}$. \\
    \hspace*{1em}
    Every lattice point that appears as a degree-$2$ or degree-$3$ vertex in the figure indeed represents a vertex of the complex incident to all adjacent edges shown.
    By contrast, a lattice point that appears to be a degree-$4$ vertex actually represents two distinct degree-$2$ vertices that overlap under the embedding.
    Equivalently, the two lines crossing at such a point are not connected in the complex.
          }
\end{figure}

\begin{figure}[ht]
  \centering
  \includegraphics[width=0.65\textwidth]{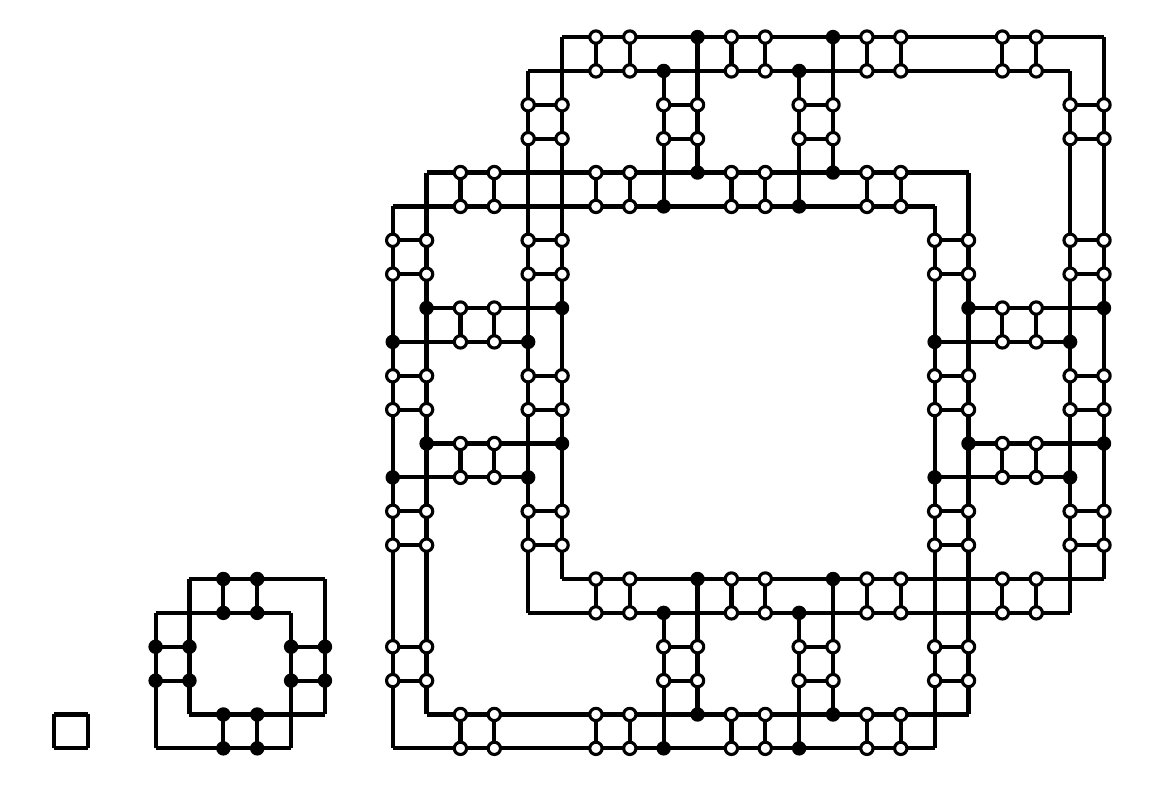}
  \caption{The first two iterations for $\ell = 4$,
            showing $\cC_{(0)}, \cC_{(1)} = \cR_{\mrm{Rep}}(\cC_{(0)}), \cC_{(2)} = \cR_{\mrm{Par}}(\cC_{(1)})$.}
\end{figure}

\clearpage

\subsubsection{Proof of embedding}

By construction, the embedding $I_{(i)}$ at each level $i$ satisfies the following basic property.
\begin{claim}
  Each vertex is mapped to a point of $\zz^2$, and each edge is mapped to a unit-length segment joining two adjacent lattice points.
\end{claim}

We use the following lemma to control possible overlaps in the embedding.
It will be used later to prove locality and bounded density.
\begin{lemma}
  No two distinct edges of $X_{(i)}$ are mapped by $I_{(i)}$ to the same edge of the integer lattice.
\end{lemma}
\begin{proof}
  We say that a lattice segment has \emph{coordinate} $a$
    if it is contained in a lattice line of the form $x=a$ or $y=a$.
  For example, the segment $(a,b)-(a,b+1)$ has coordinate $a$.

  We prove the statement by induction on $i$.
  The base case $i=0$ holds by construction.
  Assume that the statement holds for $X_{(i)}$.
  We show it also holds for
  \begin{equation}
    X_{(i+1)}
    =
    \bigl(
      X'_{(i)}
      \sqcup X'_{(i),\mathrm{copy}}
      \sqcup X_{(i),\mathrm{cyl}}
    \bigr)/{\sim}.
  \end{equation}

  By construction,
    the three pieces occupy disjoint coordinate classes:
  \begin{itemize}
    \item edges of $X'_{(i)}$ are mapped to lattice segments
          with coordinate congruent to $0 \pmod{\ell}$,
    \item edges of $X'_{(i),\mrm{copy}}$ are mapped to lattice segments
          with coordinate congruent to $1 \pmod{\ell}$,
    \item edges of $X_{(i),\mrm{cyl}}$ are mapped to lattice segments
          with coordinate congruent to one of $2,3,\ldots,\ell-1 \pmod{\ell}$.
  \end{itemize}
  Thus, edges belonging to different pieces cannot be mapped to the same lattice segment.

  It remains to rule out overlaps within each of the three pieces.
  For $X'_{(i)}$,
    this follows from the induction hypothesis:
    scaling and subdivision replace each edge of $X_{(i)}$
    by a chain of distinct lattice segments.
  The same argument applies to $X'_{(i),\mrm{copy}}$,
    since translation preserves this property.

  It remains to consider $X_{(i),\mrm{cyl}}$.
  Each edge of $X_{(i),\mathrm{cyl}}$ is associated with an edge of $X_{(i)}$.
  For example,
    if an edge of $X_{(i)}$ is mapped to the lattice segment
    $(a,b)-(a+1,b)$,
    then the corresponding edges of $X_{(i),\mathrm{cyl}}$ are
    \[
      \{
        (\ell a + m,\ell b) - (\ell a + m,\ell b + 1)
        : m \in \{2,3,\ldots,\ell-1\}
      \}.
    \]
  These lattice segments are disjoint.

  Moreover,
    edges of $X_{(i),\mathrm{cyl}}$ associated with different edges of $X_{(i)}$
    are supported in disjoint rectangular regions.
  Indeed,
    the edges associated with a horizontal segment $(a,b)-(a+1,b)$
    are contained in
    \[
      (1,\ell) \times (0,1) + (\ell a,\ell b),
    \]
    and the edges associated with a vertical segment $(a,b)-(a,b+1)$
    are contained in
    \[
      (0,1) \times (1,\ell) + (\ell a,\ell b).
    \]
  These regions are disjoint over $(a,b) \in \zz^2$.
  Thus, no two distinct edges of $X_{(i),\mathrm{cyl}}$
    map to the same lattice segment.
  This completes the induction.
\end{proof}

\begin{corollary}
  $I_{(i)}$ is a $(4,1)$-embedding of $X_{(i)}$ in $\RR^2$.
  In particular,
    every open disk of diameter $1$ in $\RR^2$ intersects at most $4$ edges of the image of $X_{(i)}$
    and adjacent $0$-cells are mapped to points with distance at most $1$.
\end{corollary}

\begin{proof}
  Locality follows since each edge of $X_{(i)}$ is mapped to a unit-length segment.  Since each lattice edge is the image of at most one edge of $X_{(i)}$,
    it follows that every open disk of diameter $1$ intersects at most $4$ edges.
\end{proof}

\subsubsection{Properties of the code}

In the previous subsubsection,
  we showed that $\cC(X_{(i)},\cF_{(i)},I_{(i)})$ is a geometrically local classical code in $\RR^2$.
In this subsubsection,
  we discuss its code properties
  including the codewords and the code associated with $H^T_i$.

One property of $\cC(X,\cF,I)$ is that it has codewords that appear at different scales.  In particular, iterations $\cC_{(i)} = \cR_{\mrm{Rep}}(\cC_{(i)})$ do not change the number of codewords.  However, iterations $\cC_{(i)} = \cR_{\mrm{Par}}(\cC_{(i)})$ add new local codewords.  Thus, after many iterations, $\cC_{(2k)}$ will have codewords at many scales that were introduced by parity iterations.  None of these smaller codewords should have a long memory lifetime, but larger codewords should.  However, to define the encoded bit for the larger codewords, one must treat the smaller codewords as ``gauge codewords'' of a subsystem code, which is discussed further in App.~\ref{appB:product-construction}.

Now, we determine the properties of the transpose code.  Let $H$ be the parity-check matrix of the code $\cC(X,\cF,I)$.
As discussed in \Cref{prelims:sheaf}
  the code associated with the parity-check matrix $H^T$
  shares all of its properties with $\cC(X,\cF^\perp,I)$
  up to constant factors,
  where $\cF^\perp$ is the sheaf obtained from $\cF$
  by setting $\cF^\perp_v = (\cF_v)^\perp$ for each vertex $v$.
In our construction,
  this operation corresponds to exchanging the black and white coloring,
  while leaving all other data unchanged.

\subsubsection{Discussion of self-correction}

We briefly sketch the argument for the exponential scaling of the memory lifetime.  We do not provide a full proof; however, it is very similar to the analysis for the quantum code with random embedding.  As in \Cref{sec:decoder}, we construct two decoders,
  one for $\cR_{\mrm{Rep}}$ and one for $\cR_{\mrm{Par}}$.
Each decoder is built from coarse-graining maps
  that take a syndrome $\sigma_i$ at level $i$
  to a syndrome $\sigma_{i-1}$ at level $i-1$, together with a correction $f_i$.
These coarse-graining maps again satisfy
  syndrome reduction and local computability.

For the code $\cC(X, \cF, I)$, we have:
\begin{itemize}
  \item for $\cR_{\mrm{Par}}$, $|\sigma_{i-1}| \le |\sigma_{i}|$;
  \item for $\cR_{\mrm{Rep}}$, $|\sigma_{i-1}| \le \frac{1}{2} |\sigma_{i}|$.
\end{itemize}
For the code $\cC(X, \cF^\perp, I)$, we have:
\begin{itemize}
  \item for $\cR_{\mrm{Par}}$, $|\sigma_{i-1}| \le \frac{1}{2} |\sigma_{i}|$;
  \item for $\cR_{\mrm{Rep}}$, $|\sigma_{i-1}| \le |\sigma_{i}|$.
\end{itemize}
This guarantees that after an $\cR_{\mrm{Rep}}$ iteration followed by an $\cR_{\mrm{Par}}$ iteration,
  the syndrome weight contracts by a constant factor:
\begin{equation}
  |\sigma_{i-2}| \le \frac{1}{2} |\sigma_{i}|.
\end{equation}
These coarse-graining maps are then incorporated into
  the decoding graph to define witnesses,
  and the remainder of the analysis (i.e. constructing the decoding graph, witness subgraph, and the Peierls argument) proceeds analogously to that of \Cref{sec:memlifetime}.

There are two main differences. One of them lies in the explicit description of the coarse-graining maps.
Since the present construction is formulated in the sheaf-theoretic setting,
  a fully rigorous treatment would require translating the earlier framework into this language.
This translation is largely cosmetic and notational, and we do not pursue it here.
Starting from $\sigma_i$,
  in contrast to the construction in \Cref{sec:coarse-graining-RZ} and \Cref{sec:coarse-graining-RX},
  we clean the syndrome within a ladder-like region.
This is the reason for choosing $\ell = 4$, which corresponds to a ladder with two rungs.  If there are two (or more) rungs, one can clean the syndrome from the interior ladder region without increasing the syndrome weight.
After this cleaning, the syndrome is supported on the boundaries of the ladder,
  i.e. precisely on the vertices corresponding to those in $X_{i-1}$,
  which allows us to apply the inverse map $\mathcal{F}^{-1}$
  to obtain $\sigma_{i-1}$.

The other difference is that, while the quantum code encodes a single logical qubit, the classical code has a hierarchy of codewords at different scales. The $\cR_{\mrm{Rep}}$ iterations are responsible for increasing the memory lifetime of large-scale codewords, whereas the $\cR_{\mrm{Par}}$ iterations introduce new local codewords at the smallest scale and do not increase the memory lifetime.  To define an encoded bit we must treat the small codewords as ``gauge codewords'', borrowing terminology from subsystem codes (see also App.~\ref{appB:product-construction}).

\subsection{Quantum code}

We now describe the quantum code construction with an explicit embedding.
Intuitively, it can be viewed as a one-dimension-higher generalization of the classical construction, but with a more involved refinement step.
As before, the code is obtained by alternating between $\cR_X$ and $\cR_Z$ iterations.

\subsubsection{Notation and framework}

In this section, we specialize to the case of 2-dimensional sheaf codes, where a geometrically local code
  $\cQ(X, \cF, I)$ is defined by the following data:

\begin{enumerate}[itemsep=-0.05em,topsep=4pt]
  \item A 2-dimensional cell complex $X$,
        with vertex set $X(0)$, edge set $X(1)$, and face set $X(2)$.
  \item Sheaf data $\cF$,
        which assigns to each edge $e$ a color, either blue or red,
          whenever the face degree of $e$,
          i.e. the number of faces incident to $e$, is $\neq 2$
  \item An embedding $I: X \to \RR^3$.
\end{enumerate}

For an edge $e$, denote
  $X_{\ge e}(2)$
  to be the set of faces that contain $e$.
If $e$ is blue, then the associated local code
  $\cF_e \subseteq \ff_2^{X_{\ge e}(2)}$
  is taken to be a repetition code;
  for example, when $|X_{\ge e}(2)|=3$, we set $\cF_e = \{000,111\}$.
If $e$ is red, then $\cF_e$ is taken to be a parity-check code;
  for example, when $|X_{\ge e}(2)|=3$, we set $\cF_e = \{000,011,101,110\}$.
Edges of face degree $2$ are treated separately:
  in that case, we set $\cF_e=\{00,11\}$,
  which is simultaneously a repetition code and a parity-check code.
For this reason, we do not explicitly assign a color to edges with face degree $2$.

As we will see, the embedding $I$ is local and has bounded density.
In our construction, each vertex is mapped to a point of $\zz^3$.
Most faces are squares and are mapped to lattice unit squares,
  with occasional exceptions given by triangular or pentagonal faces.
The triangular faces are mapped to degenerate regions of zero area,
  while the pentagonal faces are mapped to lattice unit squares.
Moreover, at most two faces of $X$ can occupy the same lattice square upon embedding;
most lattice unit squares are occupied by exactly one face of $X$.

We will use $\ell$ to denote the scaling factor at each iteration. While the construction is valid for all $\ell \ge 7$, we focus on the case $\ell=8$.

We use orthogonal projection for all 3D drawings in this section; interactive 3D renderings will be provided in a future version of the paper.
The triad in the lower-left corner of each figure indicates the orientation of the axes and the unit length in each direction.

The visual conventions are as follows. Translucent planes represent faces of the embedded square complex. Edges are drawn with three line thicknesses: thick edges have face degree not equal to $2$ and therefore carry nontrivial sheaf data (with the exception of the boundary edges, which can have a smaller face incidence and still carry sheaf data); medium edges indicate potential attaching regions (used for doubling) present in the subdivided complex $X'$ but not in $X$; thin lines are only visual guides marking creases between planes and do not correspond to edges of the complex.

A black dot marks the origin of the local coordinate system in each figure. Unless explicitly stated in the caption, colors are used only to distinguish different planes and structures, not to indicate sheaf data. However, edge colors are chosen consistently with the local code, so that edges of the same color carry the same local code. In particular, after the steps within one iteration, edges inherit the local code of the corresponding same-colored edges from the previous stage. See \Cref{fig:visual_conventions_example} for an illustration.

\begin{figure}[H]
  \centering
  \includegraphics[width=0.4\textwidth]{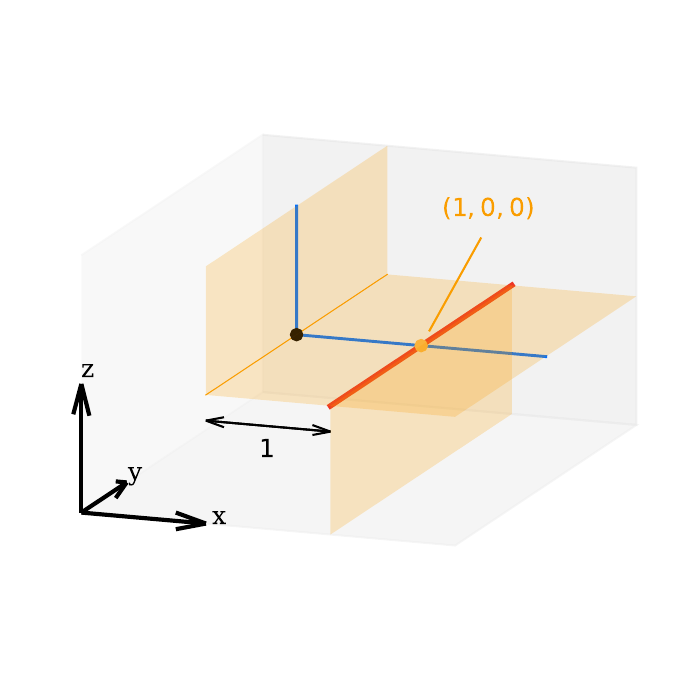}
    \caption{Illustration of the pictorial conventions used in the figures, in particular the three different edge thicknesses.}
  \label{fig:visual_conventions_example}
\end{figure}

\subsubsection{Construction overview}

We begin with a finite-size code $\cQ_{(0)}$
  and define two iterative procedures, $\cR_X$ and $\cR_Z$.
Analogously to the classical case, this produces a family of codes
\begin{equation}
  \cQ_{(0)}, \,
  \cQ_{(1)} = \cR_X(\cQ_{(0)}),\,
  \cQ_{(2)} = \cR_Z(\cQ_{(1)}),\, \cdots, \,
  \cQ_{(2k-1)} = \cR_X(\cQ_{(2k-2)}), \,
  \cQ_{(2k)} = \cR_Z(\cQ_{(2k-1)}), \, \cdots
\end{equation}

\paragraph{Base case:}
The initial code $\cQ_{(0)}$ is specified by the following data:
\begin{enumerate}[itemsep=-0.05em,topsep=4pt]
  \item The 2-cell complex $X_{(0)}$ is a single square.
  \item The two horizontal edges are colored blue,
        and the two vertical edges are colored red.
  \item The square is embedded in $\RR^3$ as a unit square.
\end{enumerate}

It is clear that the corresponding code is nontrivial and encodes a single logical qubit.

\paragraph{Iterations $\cR_X$ and $\cR_Z$:}
Each iteration consists of two steps:
\begin{itemize}[itemsep=-0.05em,topsep=4pt]
  \item Refinement: $X_{(i)} \mapsto X'_{(i)}$
  \item Doubling: $X'_{(i)} \mapsto X_{(i+1)} = (X'_{(i)} \sqcup X'_{(i), \mrm{copy}} \sqcup X_{(i), \mrm{cyl}}) / \sim$
\end{itemize}

During the refinement step, the geometry is modified from $X_{(i)}$ to $X'_{(i)}$, while the underlying topological structure remains unchanged.
This step consists of scaling, subdivision,
  and an additional (deterministic) perturbation when necessary.

During the doubling step, the complex $X'_{(i)}$ is replaced by a new complex $X_{(i+1)}$ consisting of two copies of $X'_{(i)}$, together with a cylinder-like structure $X_{(i), \mrm{cyl}}$ joining them.
One copy remains in place
  while the other, denoted $X'_{(i), \mrm{copy}}$,
  is obtained by translating $X'_{(i)}$ by the vector $(1,1,1)$.
The cylinder $X_{(i), \mrm{cyl}}$ is formed by
  connecting $X'_{(i)}$ and $X'_{(i), \mrm{copy}}$
  with a collection of unit-width bands.
The two boundaries of $X_{(i), \mrm{cyl}}$ are then identified
  with the corresponding attaching regions in the two copies of $X'_{(i)}$.

The two types of iteration, $\cR_X$ and $\cR_Z$, differ only in the sheaf data assigned during the doubling step. We therefore begin by describing the common geometric part of the construction, namely the geometry $X_{(i)}$ and the embedding $I_{(i)}$, and then specify the corresponding sheaf data assigned during the $\cR_X$ and $\cR_Z$ iterations.

Before going into further technical detail,
  we begin with three examples.
The purpose of these examples is to illustrate the basic mechanism
  before discussing the general construction.

There are several possible ways in which a unit-square face of $X_{(i)}$
  can change during one iteration.
We begin with the simplest case where no perturbation occurs in the refinement step.
We focus primarily on the interior structure of the construction.
The boundary regions will have a more complicated geometry due to their interaction with neighboring regions,
  as we will discuss later.

\begin{example}
  We direct reader to
    \Cref{fig:quantum_doubling_example} for this example.
  The panels are ordered as follows:
    (1) top left, (2) top center, (3) top right,
    (4) bottom center, and (5) bottom right.

  In panel (1), we start with a unit square in $X_{(i)}$.
  In panel (2),
    after scaling and subdivision in the refinement step,
    this square becomes an $8 \times 8$ grid of unit squares in $X'_{(i)}$.
  The candidate attaching regions, later called \emph{candidate ribs},
    are the three vertical and three horizontal segments
    running through the middle of the grid.
  In the doubling step,
    \emph{some} of these candidate ribs are promoted to actual attaching regions.
  They lie in the planes $x=3,4,5$ and $y=3,4,5$.

  In panel (3), we enter the doubling step.
  We form a copy of the subdivided complex and translate it by $(1,1,1)$.
  The figure now shows a portion of $X'_{(i)}$ in orange
    and a portion of $X'_{(i), \mrm{copy}}$ in brown.
  The candidate ribs of $X'_{(i), \mrm{copy}}$ are translated accordingly,
    so they lie in the planes $x=4,5,6$ and $y=4,5,6$.
  Thus, the candidate ribs of $X'_{(i)}$ and $X'_{(i),\mathrm{copy}}$ lying on $x=4,5$ and $y=4,5$
    are separated by distance $1$
    and are promoted to actual attaching regions.

  In panel (4),
    we connect the attaching regions at coordinate $4$ and $5$
    by the complex $X_{(i), \mrm{cyl}}$, shown in pink.
  By construction,
    this complex is contained entirely in
    planes $x=4,5$ and $y=4,5$.
  Geometrically, $X_{(i),\mrm{cyl}}$ is a hashtag-shaped subcomplex
    extruded by a unit segment in the $z$-direction.
  Note that the horizontal and vertical bands are identified at their intersection,
    rather than passing through one another.
  In particular, the complex $X_{(i),\mrm{cyl}}$ is connected.

  Finally, in panel (5), we assign colors to the newly formed edges with face degree $\ne 2$.
  These consist of the attaching regions,
    i.e. the two hashtag-shaped subcomplexes,
    together with the four segments contained in $X_{(i), \mrm{cyl}}$.
  The color is determined by the type of iteration:
    blue for $\cR_X$ and red for $\cR_Z$.

  \begin{figure}[ht]
    \centering
    \vspace*{-3em}
    \includegraphics[width=0.9\textwidth]{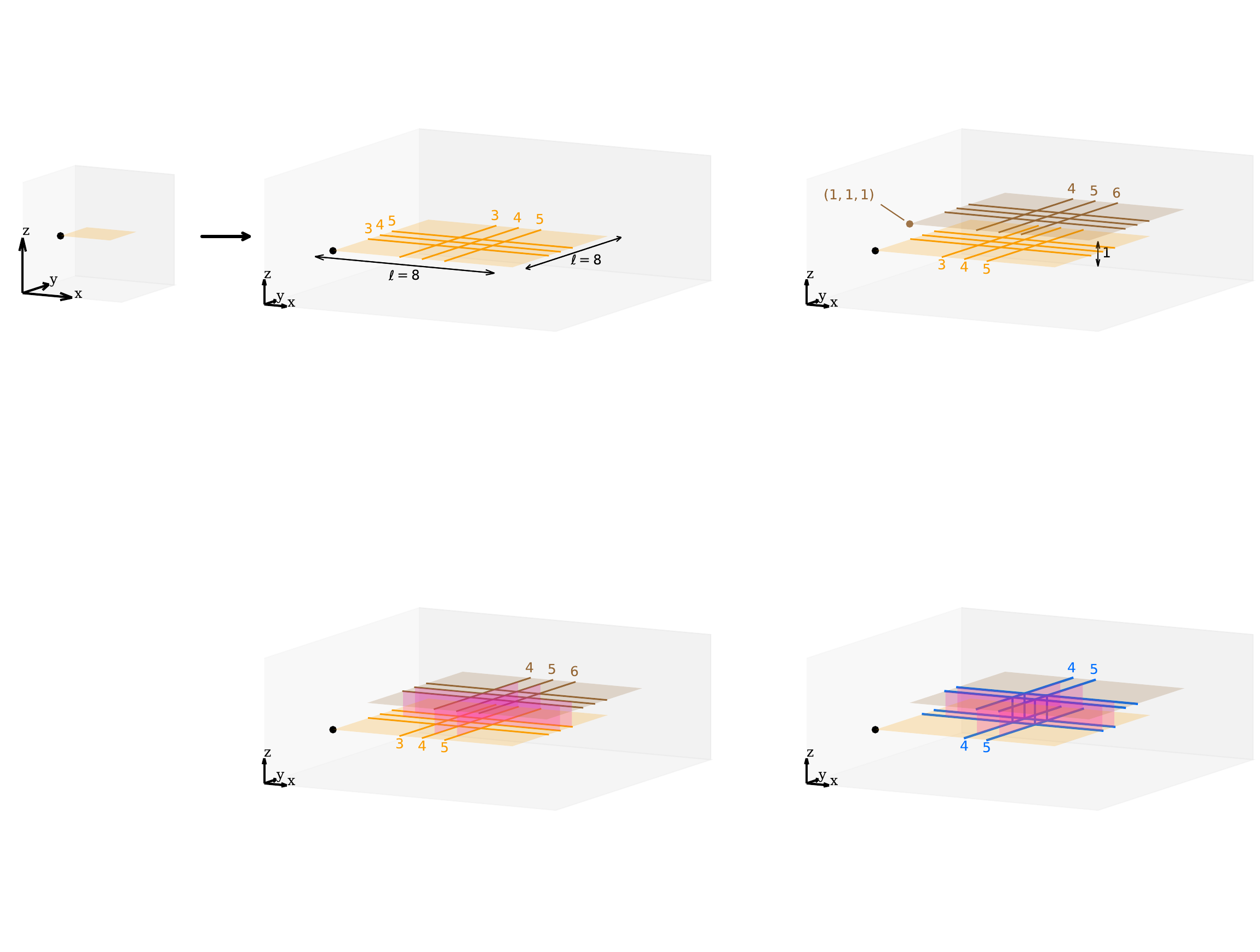}
    \vspace*{-3em}
    \caption{The doubling step for a unit square in $X_{(i)}$ in the simple case.}
    \label{fig:quantum_doubling_example}
  \end{figure}
\end{example}

In the example above,
  the portion after refinement shown in panel (2)
  is obtained solely by scaling and subdivision,
  and is therefore perfectly flat.
In general, however,
  the refinement step may include a perturbation,
  so this portion need not remain flat.
Nevertheless, these perturbations occur only near the boundary of the scaled square,
  so the candidate ribs can still be defined in the middle.
We illustrate this point with another example.

\begin{example}
  We direct the reader to  \Cref{fig:quantum_doubling_example-2} for this example.
  The panels are ordered as follows:
  (1) top left, (2) top right,
  (3) bottom center, and (4) bottom right.

  In panel (1), we begin with a perturbed portion of the subdivided complex $X'_{(i)}$,
    occurring near the boundary of the scaled square.
  The candidate ribs still lie in planes $x=3,4,5$ and $y=3,4,5$.

  In panel (2), we form a copy of the subdivided complex and translate it by $(1,1,1)$.
  The figure now shows a portion of $X'_{(i)}$ in orange
    and a portion of $X'_{(i), \mrm{copy}}$ in brown.
  The candidate ribs of $X'_{(i), \mrm{copy}}$ are translated accordingly,
    so they lie in the planes $x=4,5,6$ and $y=4,5,6$.
  Thus, the candidate ribs of $X'_{(i)}$ and $X'_{(i),\mathrm{copy}}$ lying on $x=4,5$ and $y=4,5$
    are separated by distance $1$
    and are promoted to actual attaching regions.

  In panel (3), we connect the attaching regions by the complex $X_{(i), \mrm{cyl}}$ in pink.
  By construction,
    this portion of the complex is contained entirely in planes $x=4,5$ and $y=4,5$.
  Near the center, $X_{(i), \mrm{cyl}}$ retains the form of
    a hashtag-shaped subcomplex extruded by a unit segment in the $z$-direction,
    although there is now an additional perturbation near the boundary.

  Finally, in panel (4), we assign colors, as before,
    to the newly formed edges with face degree $\ne 2$.
  The color is determined by the type of iteration:
    blue for $\cR_X$ and red for $\cR_Z$.

  \begin{figure}[ht]
    \centering
    \vspace*{-3em}
    \includegraphics[width=0.8\textwidth]{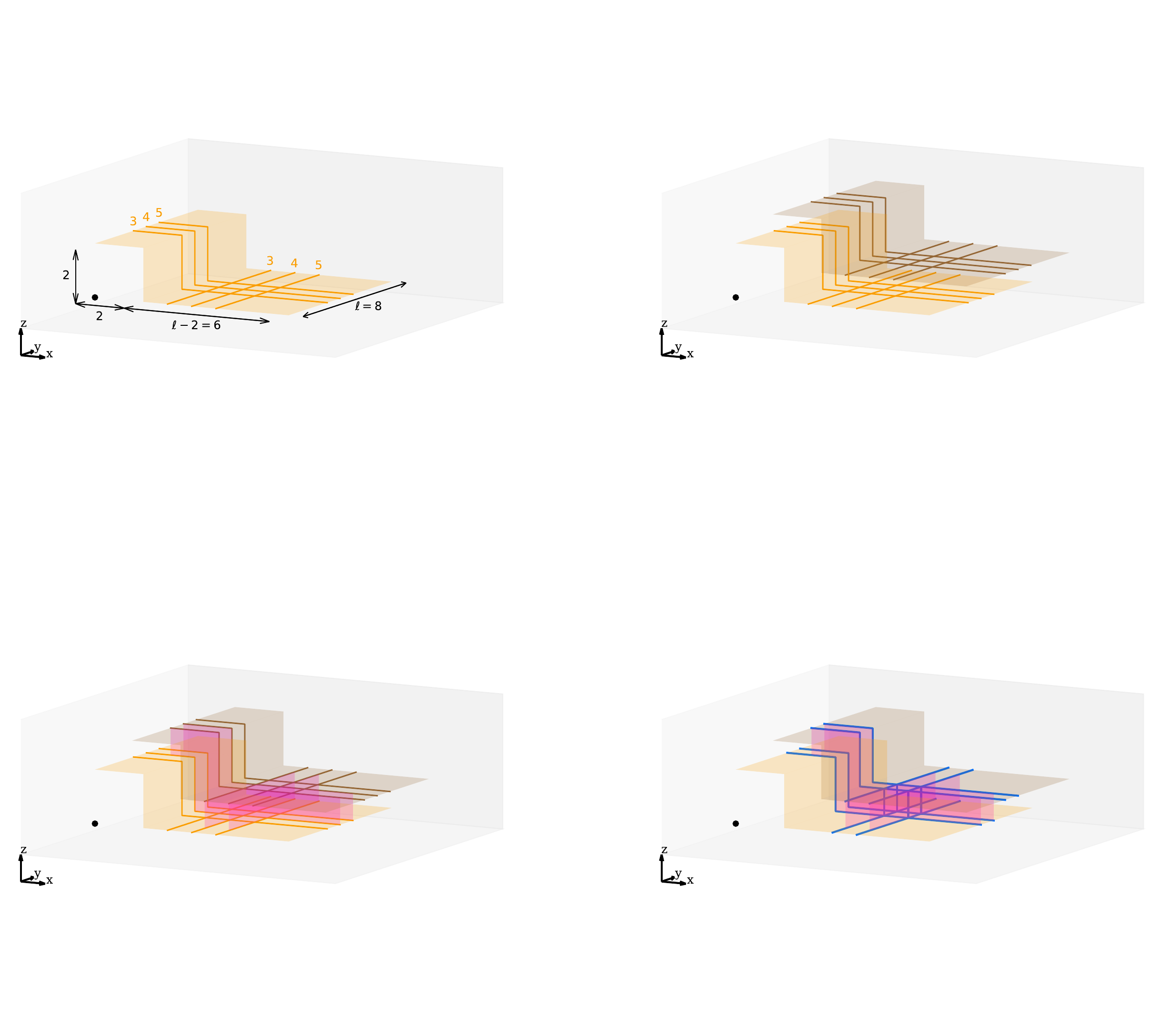}
    \vspace*{-3em}
    \caption{The doubling step with a perturbed scaled square.}
    \label{fig:quantum_doubling_example-2}
  \end{figure}
\end{example}

In the example above,
  we see that because we choose the candidate ribs to lie in the middle of the grid,
  perturbation near the boundary of the scaled square
  does not affect the general behavior of the doubling step.
As we will see later,
  the perturbation applied to such a square in $X'_{(i)}$
    occupies planes with $x\equiv -2,0,2 \pmod{8}$ and $y\equiv -2,0,2 \pmod{8}$.
  This is why we choose the candidate ribs to lie in the middle planes $x=3,4,5$ and $y=3,4,5$.

We now illustrate another type of perturbation,
  which explains the choice of the distance-$2$ shifts.
When a unit square is occupied by two faces of $X_{(i)}$,
  scaling and subdivision alone would cause
  the candidate ribs of the two faces to coincide.
During doubling,
  the corresponding portions of the complex $X_{(i), \mrm{cyl}}$
  may create additional overlapping faces\footnote{
    Technically, the additional overlap arises near the boundary
      and is not visible in the local region illustrated here.
    In particular,
      this issue appears in the doubled T-junction configuration,
      as we will see later.}.
This is undesirable:
  as the iteration proceeds,
  the number of overlapping faces would grow without bound,
  so the resulting embedding would fail to have bounded density.

The solution is to perform a perturbation:
  we shift one of the scaled squares
  by distance $2$ in the $z$-direction.
This shift is large enough to ensure that the
  corresponding portions of the complex $X_{(i), \mrm{cyl}}$
  are completely disjoint,
  as illustrated below.
\begin{example}
  We direct the reader to \Cref{fig:quantum_doubling_example-3} for this example.
  The panels are ordered as follows:
  (1) top left, (2) top center, (3) top right,
  (4) bottom center, and (5) bottom right.

  In panel (1), we begin with two copies of a unit square in $X_{(i)}$,
    stacked on top of one another.
  In panel (2), during the perturbation step,
    we shift one of the scaled squares by distance $2$
    in the $z$-direction so that the two scaled squares no longer overlap.
  In particular,
    the two families of candidate ribs become completely disjoint.

  In panel (3), we form a copy of the subdivided complex
    and translate it by $(1,1,1)$.
  Note that the four scale squares in $X'_{(i)}$ and $X'_{(i),\mrm{copy}}$
    are completely disjoint from one another.

  In panel (4), we connect the attaching regions
    by the complex $X_{(i), \mrm{cyl}}$, shown in pink.
  Since the shift between $X'_{(i)}$ and $X'_{(i),\mrm{copy}}$ is $1$,
    whereas the two scaled squares in $X'_{(i)}$ are separated by distance $2$,
    the two corresponding portions of $X_{(i), \mrm{cyl}}$
    do not intersect one another.
  In particular,
    in this portion of $X_{(i+1)}$, each unit square is occupied only once.

  Finally, in panel (5), we assign colors to the newly formed edges with face degree $\ne 2$, as in the previous examples.

  \begin{figure}[ht]
    \centering
    \vspace*{-3em}
    \includegraphics[width=0.9\textwidth]{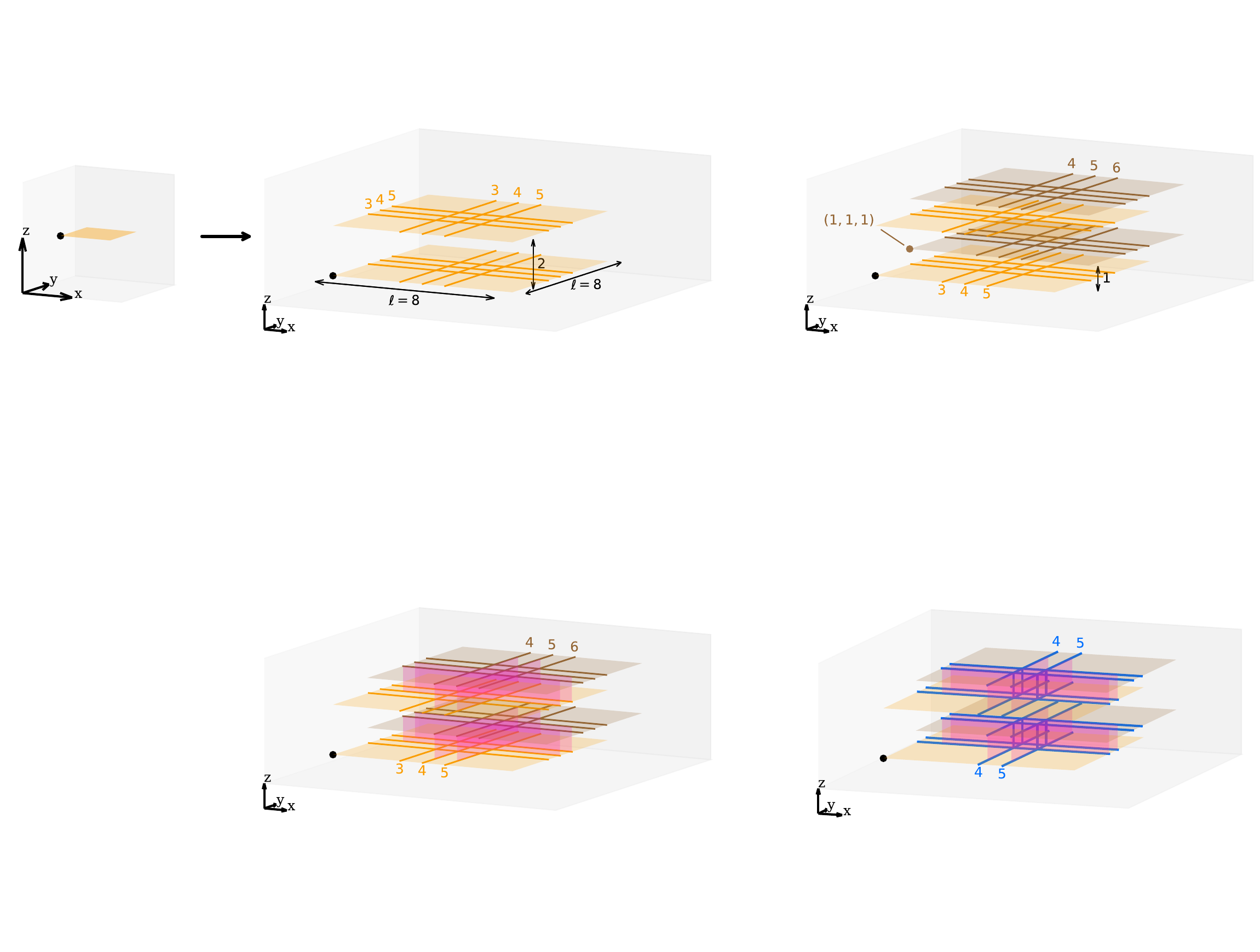}
    \vspace*{-3em}
    \caption{The doubling step for two overlapping squares.}
    \label{fig:quantum_doubling_example-3}
  \end{figure}
\end{example}

We summarize the key lessons from the examples above.
We occasionally use the expression \emph{planes with coordinates}
  to mean planes whose relevant coordinate takes the specified values.
For example, planes with coordinates $3,4,5$
  refer to planes of the form $x=3,4,5$, $y=3,4,5$, or $z=3,4,5$.
\begin{itemize}
  \item The doubling step is performed by first forming a copy of $X'_{(i)}$,
        denoted $X'_{(i),\mrm{copy}}$,
        and shifting it by $(1,1,1)$.
  \item The candidate ribs of the two copies in places with coordinates $4,5$
        are separated by distance $1$,
        which allows us to connect them by unit-width bands
        forming $X_{(i), \mrm{cyl}}$.
  \item The choice of planes with coordinates $3,4,5$ for the candidate ribs in $X'_{(i)}$
        keeps them away from perturbations near the boundary
        of the scaled square.
  \item The perturbation shift the scaled structure by distance $2$,
        so that no unit-squares in $X'_{(i)}$ and $X'_{(i),\mrm{copy}}$ are stacked on top of one another after the shift by $(1,1,1)$.
\end{itemize}

We have so far taken a bottom-up view by studying simple examples.
We now take a top-down view by listing the relevant local structures.

The following local structure in $X_{(i)}$ have the property that,
  after scaling,
  they are already in a form suitable for the doubling step,
  with no perturbation required.
We refer to them as \emph{simple local structures} in $X_{(i)}$,
  as illustrated in \Cref{fig:simple-local-structures}:
\begin{figure}[ht]
  \centering
  \includegraphics[width=0.9\textwidth]{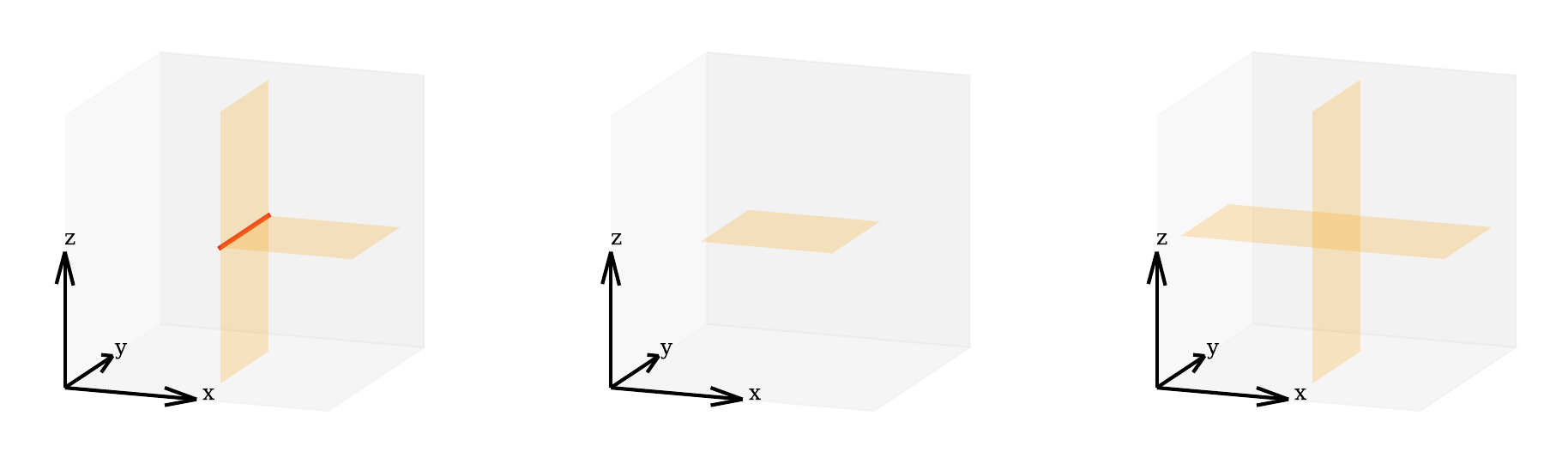}
  \caption{The three types of simple local structures in $X_{(i)}$.}
  \label{fig:simple-local-structures}
\end{figure}
\begin{itemize}
  \item a lattice unit-length segment that is the image of a single edge of $X_{(i)}$
          and is incident to at most three faces of $X_{(i)}$;
  \item a lattice unit square that is the image of a single face of $X_{(i)}$;
  \item a lattice unit-length segment that is the common image of two edges of $X_{(i)}$,
          each incident to at most two faces of $X_{(i)}$,
          such that the corresponding faces form locally flat sheets
          meeting along the segment.
\end{itemize}

After scaling,
  these simple local structures respectively produce the following local configurations
  near the attaching regions in $X'_{(i)}$,
  as illustrated in \Cref{fig:local-configurations}:
\begin{figure}[ht]
  \centering
  \includegraphics[width=0.9\textwidth]{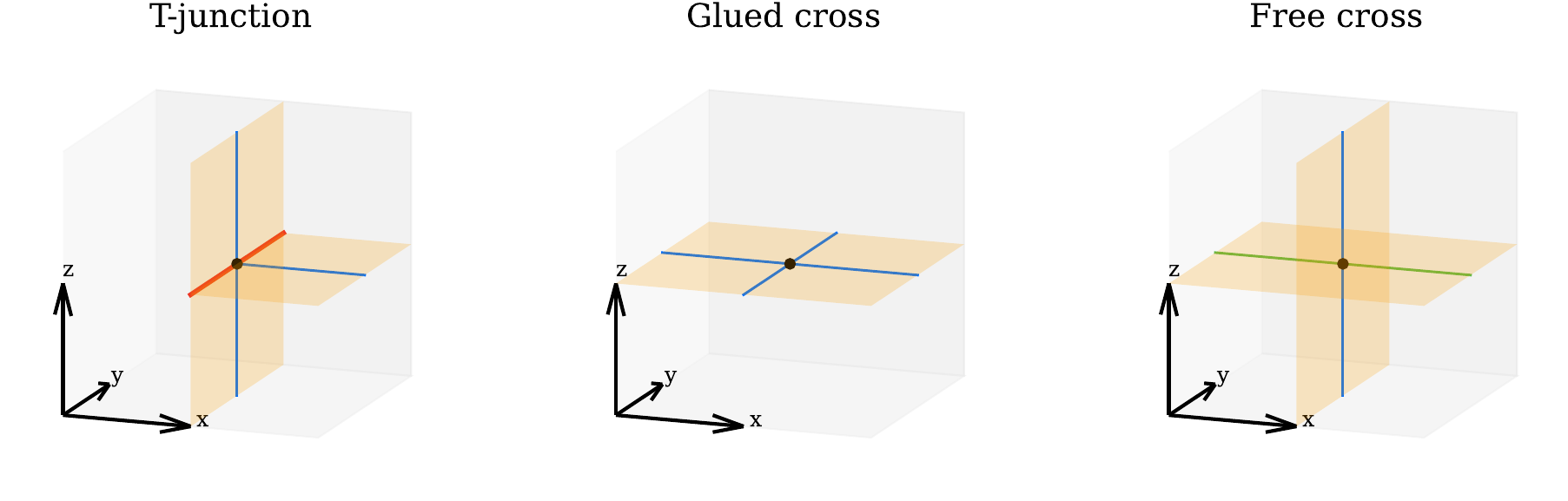}
  \caption{The three types of local configurations near the attaching regions: T-junctions, glued crosses, and free crosses.}
  \label{fig:local-configurations}
\end{figure}
\begin{enumerate}
  \item [(a)] T-junctions, where an edge is incident to three faces (shown in red), with ribs running perpendicular to the edge (shown in blue);
  \item [(b)] glued crosses, where two ribs belonging to the same face intersect at a vertex;
  \item [(c)] free crosses, where two ribs intersect at a vertex in the embedding map, but belong to two distinct faces of $X'$.
\end{enumerate}

There are also more complicated local structures in $X_{(i)}$
  beyond the simple types listed above.
These require perturbations during the refinement step.
We mark such structures by enclosing them in boxed regions;
  as described in the next subsubsection,
  there are four types of boxed local structures that require perturbation.

Finally, \Cref{fig:refinement-gadget-schematic}
  provides a schematic diagram showing how the local structures transform
  under refinement and doubling.
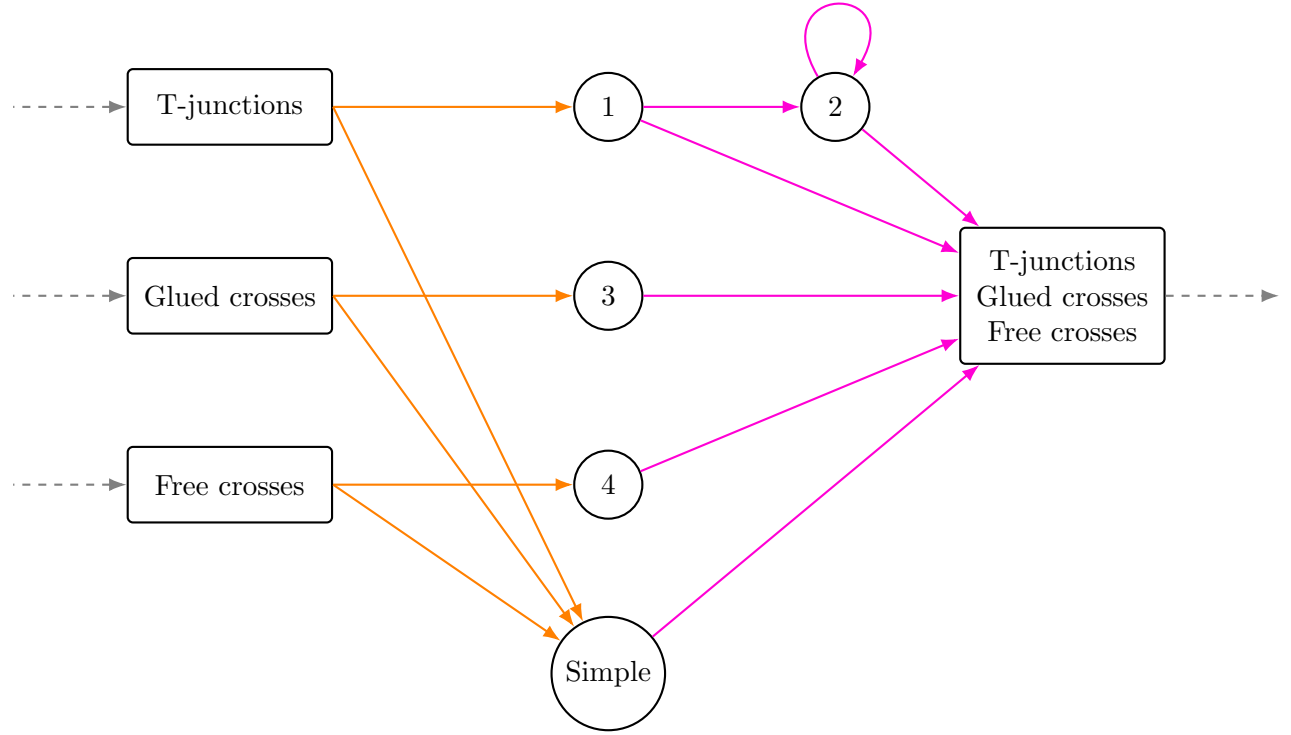
\begin{figure}[ht]
  \centering
  \begin{tikzpicture}[
    >=Latex,
    thick,
    type/.style={draw, circle,
                minimum size=9mm,
                inner sep=0pt},
    type2/.style={draw, rectangle,
                rounded corners=2pt,
                minimum width=27mm,
                minimum height=10mm,
                inner sep=2pt,
                align=center}
  ]

  \definecolor{mypink}{HTML}{FF00D4}

  \node[type2] (T) at (-5,5) {T-junctions};
  \node[type2] (GC) at (-5,2.5) {Glued crosses};
  \node[type2] (FC) at (-5,0) {Free crosses};

  \node[type2, minimum height=18mm] (all) at (6,2.5) {T-junctions\\Glued crosses\\Free crosses};

  \node[type] (one) at (0,5) {1};
  \node[type] (two) at (3,5) {2};
  \node[type] (three) at (0,2.5) {3};
  \node[type] (four) at (0,0) {4};
  \node[type, minimum size=15mm] (simple) at (0,-2.5) {Simple};


  \draw[<-,gray,dashed] (T.west) -- ++(-1.5,0);
  \draw[<-,gray,dashed] (GC.west) -- ++(-1.5,0);
  \draw[<-,gray,dashed] (FC.west) -- ++(-1.5,0);
  \draw[->,gray,dashed] (all.east) -- ++(1.5,0);

  \draw[->,orange] (T) -- (one);
  \draw[->,orange] (GC) -- (three);
  \draw[->,orange] (FC) -- (four);

  \draw[->,orange] (T.east) -- (simple);
  \draw[->,orange] (GC.east) -- (simple);
  \draw[->,orange] (FC.east) -- (simple);

  \draw[->,mypink] (one) -- (two);
  \draw[->,mypink] (two) to[out=120,in=60,looseness=8] (two);

  \draw[->,mypink] (one) -- (all);
  \draw[->,mypink] (two) -- (all);
  \draw[->,mypink] (three) -- (all);
  \draw[->,mypink] (four) -- (all);
  \draw[->,mypink] (simple) -- (all);


  \end{tikzpicture}
  \caption{A schematic diagram showing how
            local structures transforms under refinement (in pink) and doubling (in orange).
            The arrows indicate the transitions between different types of local structures under the corresponding gadgets.
            The circles labeled $1, 2, 3, 4$ represent the four types of local structures that require special treatment during refinement.
            The circle Simple represents the local structure that is ready for doubling without perturbation.
            The rectangles represent the local structures near the candidate ribs.}
  \label{fig:refinement-gadget-schematic}
\end{figure}

\subsubsection{Refinement step: $X_{(i)} \mapsto X'_{(i)}$}\label{sec:labelinggadgets}

Motivated by the examples in the previous subsubsection, we provide a formal definition of the refinement step.  Recall that in the classical construction,
  the refinement step consists simply of scaling followed by subdivision.
In the quantum construction, however, additional perturbations are required.
These perturbations are performed inside the boxes
  prescribed by the doubling step in the previous iteration.

There are four types of local structures that require special treatment,
  each contained in a corresponding box:
\begin{enumerate}
  \item [(1)] doubled T-junctions;
  \item [(2)] recursive T-junction-derived structures,
              generated by doubled T-junctions and
              recursively by structures of this same type;
  \item [(3)] doubled glued crosses;
  \item [(4)] doubled free crosses;
\end{enumerate}

For each of these four types,
  we specify a local perturbation gadget.
The gadget modifies the portion of the embedded complex inside the corresponding box,
  while leaving the complex unchanged outside the box.
It also specifies the candidate ribs
  that will be used in the next doubling step.
The perturbation gadgets are designed to satisfy the following requirements:
\begin{enumerate}
  \item They must only be supported in the corresponding box,
        so that the perturbations in different boxes do not affect one another;
  \item They produce only three types of local structures near candidate ribs: T-junctions, glued crosses, and free crosses
          (allowing also for subcomplexes of these types).  This ensures that the next doubling step can be performed correctly;
  \item They ensure that the perturbed structure is supported only on planes satisfying
        \[
          x \equiv 0,2,\ell-2 \pmod{\ell},
          \quad
          y \equiv 0,2,\ell-2 \pmod{\ell},
          \quad
          z \equiv 0,2,\ell-2 \pmod{\ell},
        \]
        so that the copy $X'_{(i),\mrm{copy}}$, shifted by $(1,1,1)$,
        does not overlap with the original copy $X'_{(i)}$.
\end{enumerate}

Before introducing the gadgets that describe the refinement step for each type of local structure,
  we make a few general comments.

For types (1) and (2),
  applying the corresponding gadget result in regions
  where there are still complicated local structures.
These regions are again contained in boxes, but
  importantly they are precisely of type (2).
Thus, they will be perturbed again in the next iteration.

For types (3) and (4),
  applying the corresponding gadget results in a simple local structure,
  with no complicated regions remaining.
See \Cref{fig:refinement-gadget-schematic} for a schematic illustration of
  how these local structures transform under refinement.

A few general remarks apply to all of these gadgets. In each figure, the leftmost panel shows the original local structure in $X_{(i)}$, while the remaining panels show the corresponding subdivided and perturbed structure in $X'_{(i)}$. Only the region inside the enlarged box, indicated by dashed wirelines, is perturbed; outside the box, the structure is simply rescaled.

In some cases, the local structure is too complicated to display clearly in a single panel, so we split it across several panels. We use a consistent color scheme for surfaces and edges to indicate the correspondence between $X_{(i)}$ and $X'_{(i)}$. Although the geometry is perturbed, the subdivided and perturbed complex $X'_{(i)}$ inherits the local code (i.e. sheaf data) from $X_{(i)}$.

In most cases, the complex $X_{(i)}$ can be read off directly from the embedding because the faces lie on lattice unit squares.
There are, however, some exceptions in the local structures arising from T-junctions.
In these cases, some faces are triangles or pentagons; we discuss an example below.
\begin{example} \label{example:non-square-faces}
  Consider the two gadgets shown in \Cref{fig:quantum_example}.
  In each case, we list the faces in the displayed portion of incident
    to path $O - P - Q$, which overlaps with another set of segments.
  We group the faces according to their position relative to the path $O-P-Q$.

  For the panel on the left, corresponding to the doubled T-junction,
  the faces are as follows:
  \begin{itemize}
    \item Above $O - P - Q$: $O - R - Q - P$ and $R - A_0 - A_1 - Q$.
    \item To the right of $O - P - Q$: $O - P - D_0$ and $P - Q - C_1 - C_0 - D_0$.
    \item Below $O - P - Q$: $O - P - D_1 - B_1 - B_0$ and $P - Q - D_1$.
  \end{itemize}

  For the gadget on the right, corresponding to type (1),
    where we zoomed in to the box near the origin
    the faces are as follows:
  \begin{itemize}
    \item Above $O - P - Q$:
          $O - A_0 - A_1 - P$
          and $P - A_1 - R - Q$.
    \item To the right of $O - P - Q$:
          $O - P - C_0$,
          $P - Q - C_1 - C_0$,
          and $C_1 - Q - R - C_3 - C_2$.
    \item Below $O - P - Q$:
          $O - P - B_1 - B_0$,
          $P - Q - B_2 - B_1$,
          $Q - B_4 - B_3 - B_2$,
          and $Q - R - B_5 - B_4$.
  \end{itemize}

  There are a few other cases in which non-square faces appear,
    such as the other box in type (1) and the box in type (2).
  The faces in these cases can be determined by an analogous procedure to that described above.
  In all of these cases, the non-square faces are either triangles or pentagons.
  Under the embedding map,
  some of these faces collapse to segments,
  while the others are mapped to lattice unit squares.

  \begin{figure}[h]
    \centering
    \includegraphics[width=0.99\textwidth]{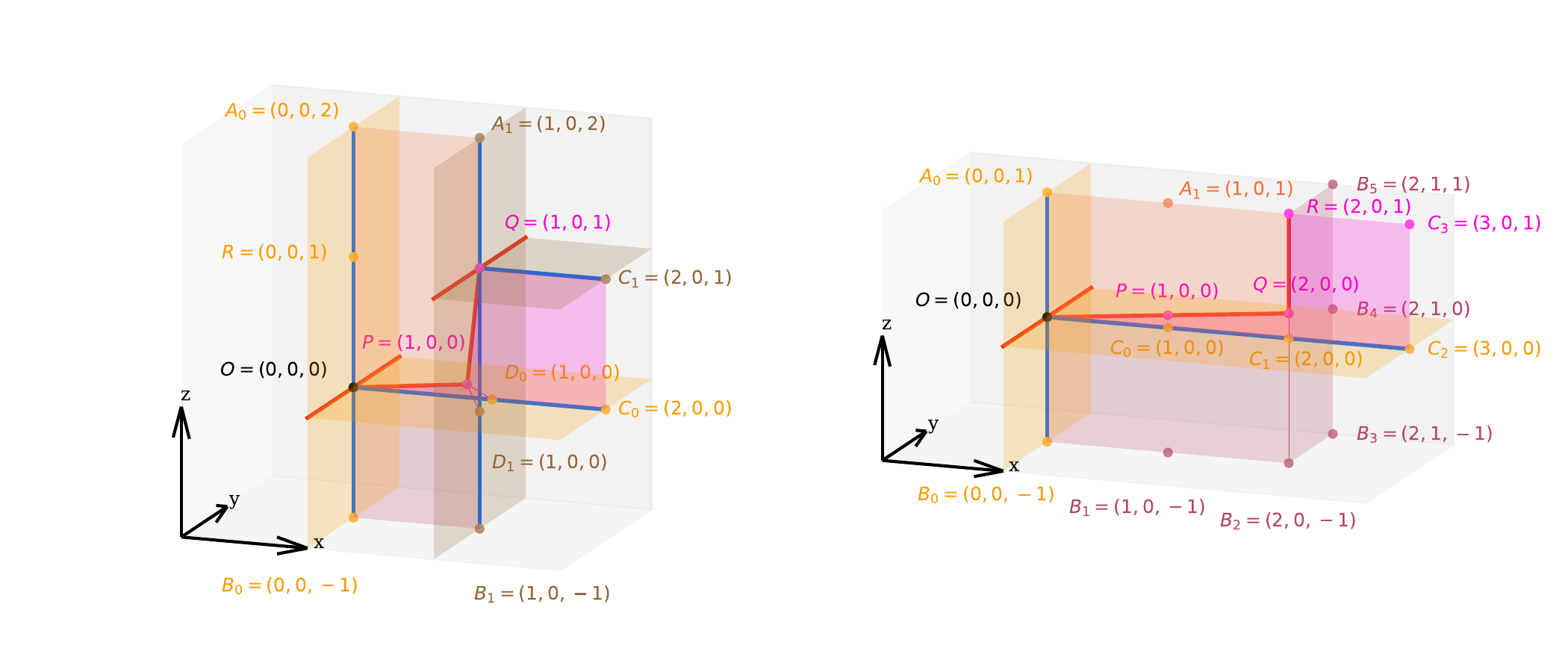}
    \caption{Examples of the non-square faces appearing in the gadgets associated with T-junctions. We slightly distort the position of some vertices to make overlapping vertices visible.}
    \label{fig:quantum_example}
  \end{figure}
\end{example}

Below we display figures showing all of the gadgets.
The accompanying captions describe what local structures they are responsible for
  and the sizes of boxes that surround them after subdivision and perturbation.

\begin{figure}[H]
  \centering
  \includegraphics[width=0.99\textwidth]{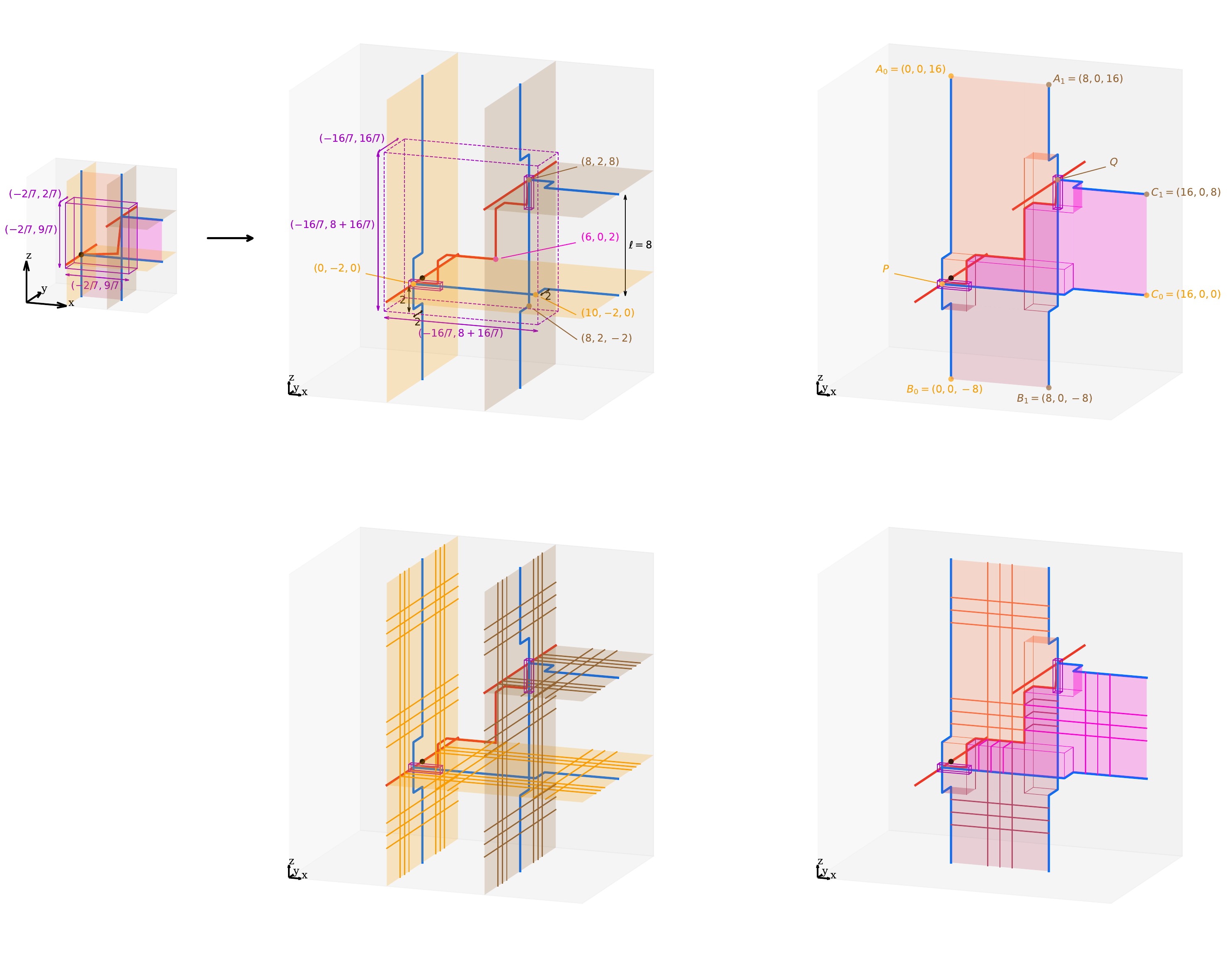}
  \caption{This gadget depicts how a doubled T-junction is perturbed in the first step of one iteration.
            On the right,
            the left panels depict the planes corresponding to faces in $X'_{(i-1)}$ (orange) and $X'_{(i-1),\mrm{copy}}$ (brown),
              while the right panels depict the planes corresponding to faces in $X_{(i-1), \mrm{cyl}}$ (pink).
              Recall that $X_{(i)} = (X'_{(i-1)} \sqcup X'_{(i-1),  \mrm{copy}} \sqcup X_{(i-1), \mrm{cyl}}) / \sim$, as obtained from the previous iteration.
              We use three slightly different hues of pink to visually distinguish
                the three surfaces near this local configuration,
                spanned by the paths
                $A_0-P-Q-A_1$,
                $B_0-P-Q-B_1$,
                and $C_0-P-Q-C_1$.
            The bottom panels further show the candidate ribs. \\
          \hspace*{1em} The key role of the gadget is to separate the double edges (overlapping blue and red edges) and the degenerate faces that appear in $X_{(i)}$.
            Two sets of double edges remain in $X'_{(i)}$.
            These are boxed and will be subdivided and perturbed in the next iteration using the gadget in \Cref{fig:quantum_refinement_box,fig:quantum_refinement_box-2}. \\
          \hspace*{1em} The bottom-left box is
            $(-\frac{2}{\ell-1}, 2+\frac{2}{\ell-1}) \times (-\frac{2}{\ell-1}, \frac{2}{\ell-1}) \times (-\frac{2}{\ell-1}, \frac{2}{\ell-1})$
            shifted by $(0,-2,0)$.
            The upper-right box is
            $(-\frac{2}{\ell-1}, \frac{2}{\ell-1}) \times (-\frac{2}{\ell-1}, \frac{2}{\ell-1}) \times (-2-\frac{2}{\ell-1}, \frac{2}{\ell-1})$
            shifted by $(\ell, 2, \ell)$.
          Since $2 + \frac{2}{\ell-1} \le \frac{2}{\ell-1} \ell = 16/7$,
            the perturbation and the newly induced boxes are contained in the scaled box
            $(-\frac{2}{\ell-1} \ell, (1+\frac{2}{\ell-1}) \ell) \times (-\frac{2}{\ell-1} \ell, \frac{2}{\ell-1} \ell) \times (-\frac{2}{\ell-1} \ell, (1+\frac{2}{\ell-1}) \ell)$.
          }
  \label{fig:quantum_refinement_T_junction}
\end{figure}
\clearpage
\begin{figure}[H]
  \centering
  \vspace*{-2.5em}
  \includegraphics[width=0.8\textwidth]{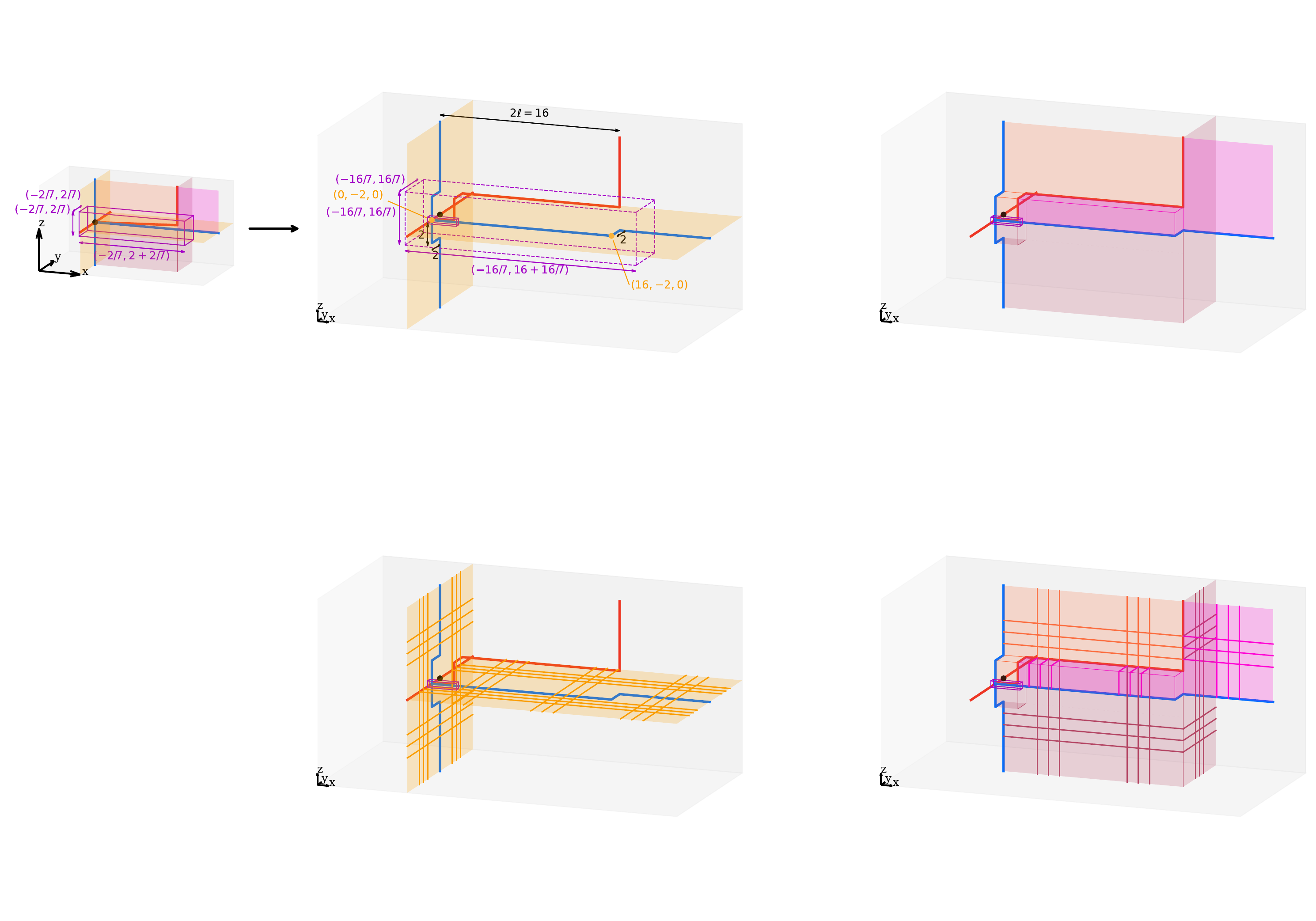}
  \vspace*{-2.5em}
  \caption{This gadget depicts how a doubled T-junction is perturbed in the first step of the iteration procedure, specifically with respect to the lower-left box in \Cref{fig:quantum_refinement_T_junction}.
            On the right,
            the left panels depict the planes induced from $X'_{(i-1)}$ (orange) from the previous iteration,
              while the right panels depict the planes induced from $X_{(i-1), \mrm{cyl}}$ (pink) from the previous iteration.
              We use three slightly different hues of pink to visually distinguish the three surfaces near this local configuration.
            The bottom panels further show the candidate ribs. \\
          \hspace*{1em} The key role of the gadget is to separate the double edges (overlapping blue and red edges) and the degenerate faces that appear in $X_{(i)}$.
          One set of double edges remains in $X'_{(i)}$.
          It is boxed and will be subdivided and perturbed in the next iteration using the same gadget. \\
          \hspace*{1em} The new box is
            $(-\frac{2}{\ell-1}, 2+\frac{2}{\ell-1}) \times (-\frac{2}{\ell-1}, \frac{2}{\ell-1}) \times (-\frac{2}{\ell-1}, \frac{2}{\ell-1})$
            shifted by $(0,-2,0)$.
          Since $2 + \frac{2}{\ell-1} \le \frac{2}{\ell-1} \ell = 16/7$,
            the perturbation and the newly induced boxes are contained in the scaled box
            $(-\frac{2}{\ell-1} \ell, (2+\frac{2}{\ell-1}) \ell) \times (-\frac{2}{\ell-1} \ell, \frac{2}{\ell-1} \ell) \times (-\frac{2}{\ell-1} \ell, \frac{2}{\ell-1} \ell)$.
          }
  \label{fig:quantum_refinement_box}
\end{figure}
\begin{figure}[H]
  \centering
  \vspace*{-2.5em}
  \includegraphics[width=0.8\textwidth]{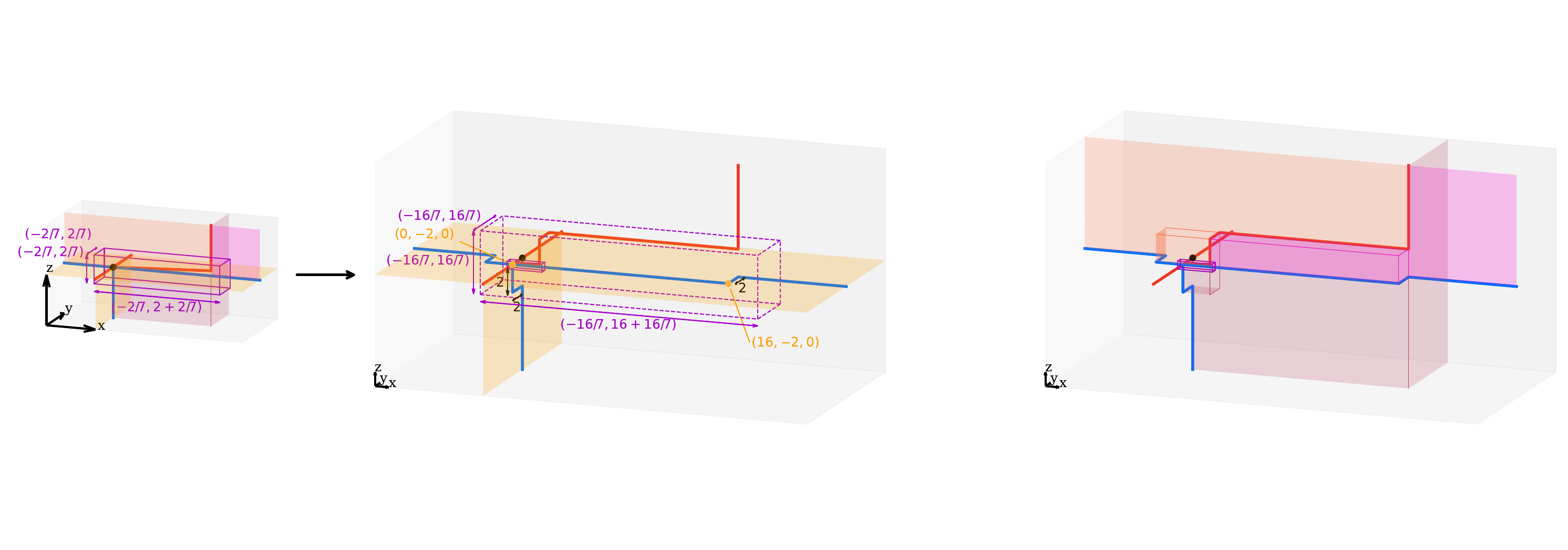}
  \vspace*{-2.5em}
  \caption{This gadget depicts how a doubled T-junction is perturbed in the first step of one iteration, specifically with respect to the upper right box in \Cref{fig:quantum_refinement_T_junction}.
          The structure is rotated so that it appears nearly identical to the previous gadget.
          Candidate ribs are not drawn,
            but they are the same as in the previous figure. \\
          \hspace*{1em} As before, the key role of the gadget is to separate the double edges (overlapping blue and red edges) and the degenerate faces that appear in $X_{(i)}$.
          One set of double edges remains in $X'_{(i)}$.
          It is boxed and will be subdivided and perturbed in the next iteration using the same gadget. \\
          \hspace*{1em} Similarly, the perturbation and the newly induced boxes are contained in the scaled box.
          }
  \label{fig:quantum_refinement_box-2}
\end{figure}
\clearpage
\begin{figure}[H]
  \centering
  \vspace*{-1.5em}
  \includegraphics[width=0.9\textwidth]{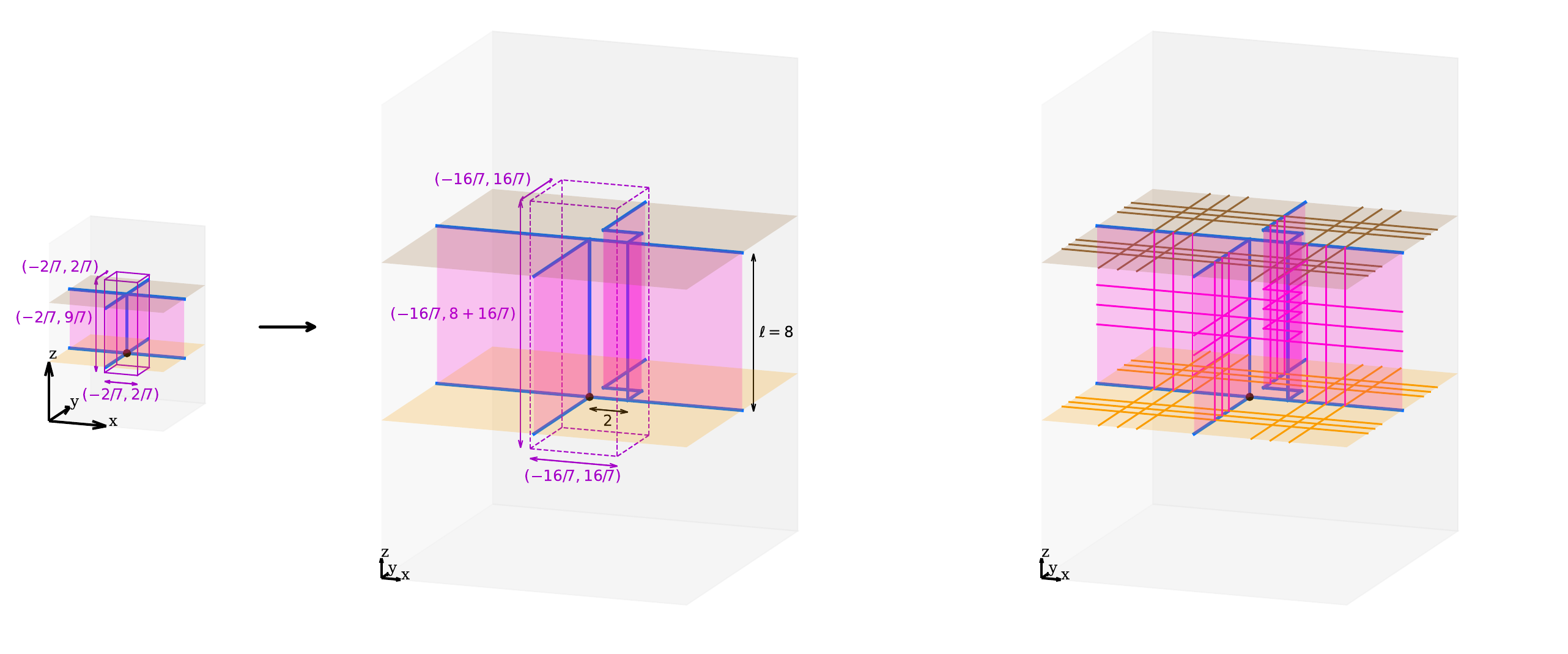}
  \vspace*{-1.5em}
  \caption{This gadget depicts how a doubled glued cross is perturbed in the first step of the iteration procedure.
            The panels on the right depict the structure with and without the candidate ribs.
            The key role of the gadget is to separate the edge incident to four faces into two edges, each incident to three faces.
            Because $2 < \frac{2}{\ell-1} \ell = 16/7$, the perturbation is contained in the scaled box.
          }
  \label{fig:quantum_refinement_glued_cross}
\end{figure}

\begin{figure}[H]
  \centering
  \vspace*{-1.5em}
  \includegraphics[width=0.9\textwidth]{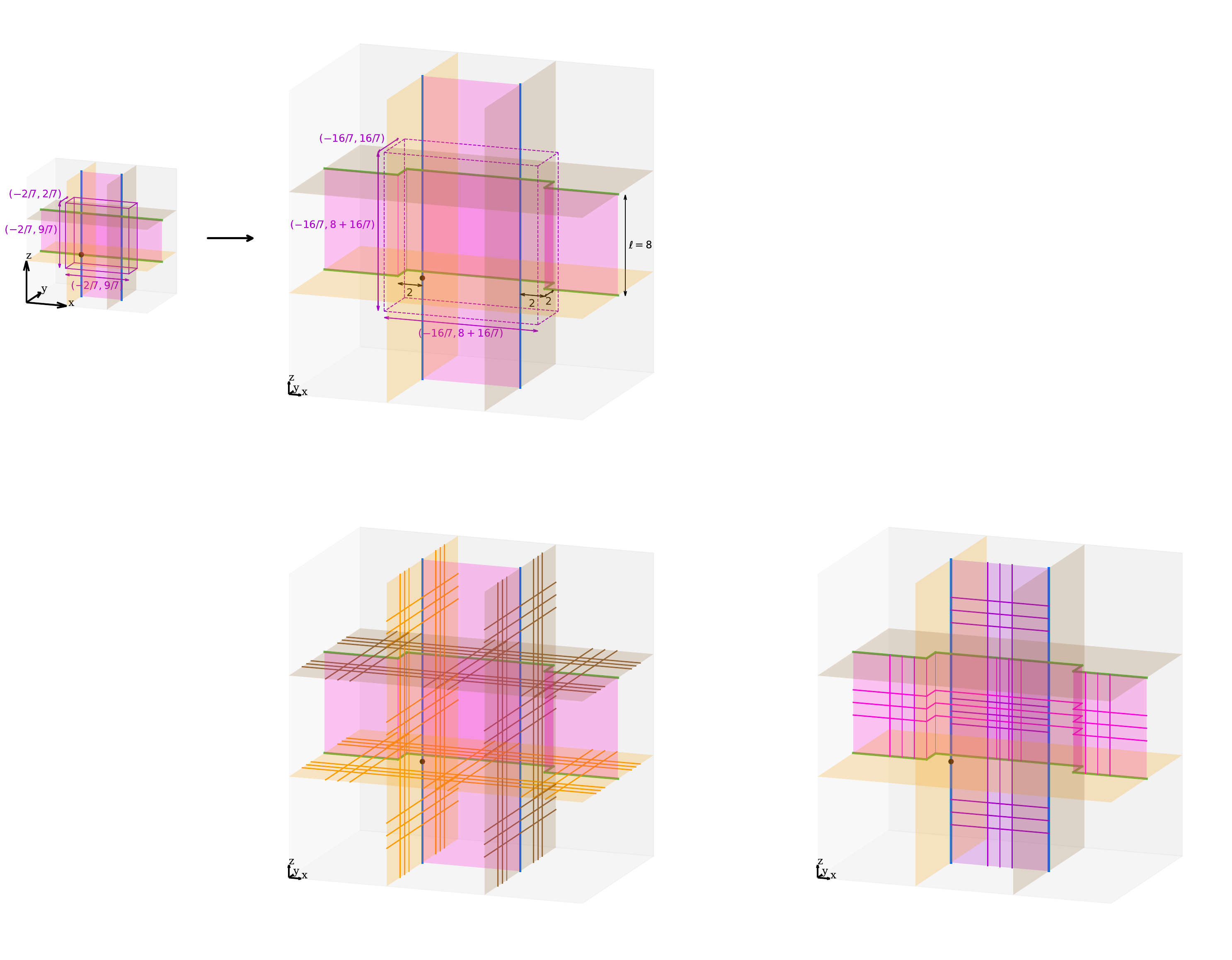}
  \vspace*{-1.5em}
  \caption{This gadget depicts how a doubled free cross is perturbed in the first step of the iteration procedure.
            The panels on the right depict the structure with and without the candidate ribs.
            The candidate ribs are shown in two separate panels to better illustrate the structure.
            The key role of the gadget is to separate double faces that appear in $X_{(i)}$.
            Because $2 < \frac{2}{\ell-1} \ell = 16/7$, the perturbation is contained in the scaled box.
          }
  \label{fig:quantum_refinement_free_cross}
\end{figure}
\clearpage
\subsubsection{Doubling step: $X'_{(i)} \mapsto X_{(i+1)} = (X'_{(i)} \sqcup X'_{(i), \mrm{copy}} \sqcup X_{(i), \mrm{cyl}}) / \sim$}

Starting from the refined complex $X'_{(i)}$,
  we form a copy, denoted $X'_{(i),\mrm{copy}}$,
  by translating $X'_{(i)}$ by the vector $(1,1,1)$.
We then connect the two complexes by the complex $X_{(i),\mrm{cyl}}$,
  which consists locally of unit-width bands running between the two copies.
In general,
  the local configurations can be more intricate,
  and the attaching regions need not be simple straight segments.
These more intricate local configurations are the focus of this subsubsection.

In $X'_{(i)}$ the candidate ribs lie in the planes
\begin{equation*}
  x=3,4,..., \ell-3 \pmod{\ell},
  \quad
  y=3,4,..., \ell-3 \pmod{\ell},
  \quad
  z=3,4,..., \ell-3 \pmod{\ell}.
\end{equation*}
In the translated copy $X'_{(i),\mrm{copy}}$,
  the candidate ribs lie in the planes
\begin{equation*}
  x=4,5,..., \ell-2 \pmod{\ell},
  \quad
  y=4,5,..., \ell-2 \pmod{\ell},
  \quad
  z=4,5,..., \ell-2 \pmod{\ell}.
\end{equation*}
The doubling step promotes compatible candidate ribs to \emph{ribs},
  which serve as the actual attaching regions,
  and connects them using $X_{(i),\mrm{cyl}}$.
By construction,
  $X_{(i),\mrm{cyl}}$ is entirely contained in the planes
\begin{equation*}
  x=4,5,..., \ell-3 \pmod{\ell},
  \quad
  y=4,5,..., \ell-3 \pmod{\ell},
  \quad
  z=4,5,..., \ell-3 \pmod{\ell}.
\end{equation*}

There are three types of local structures near the ribs in $X'_{(i)}$
  whose doubling requires special attention:
\begin{enumerate}
  \item [(a)] T-junctions, located at edges of face degree $3$;
  \item [(b)] glued crosses, located near the center of the scaled square, as seen in the examples in the overview;
  \item [(c)] free crosses, located near the boundary of the scaled square,
      where the routing of the perturbation gadget
      produces two ribs lying in distinct orthogonal planes
      whose images intersect.
\end{enumerate}

For each of these three types,
  we specify a local doubling gadget.
The gadget describes the local structure of $X_{(i), \mrm{cyl}}$
  near the corresponding local structure in $X'_{(i)}$.
It also specifies the boxed region indicating where perturbations
  will occur in the next refinement step.
The doubling gadgets are designed to satisfy the following requirements:

\begin{enumerate}
  \item If the rib is in the plane $x=x_0$
then the support of $X_{(i),\mathrm{cyl}}$ in the gadget is also in same plane $x=x_0$.  The possible planes are
        \begin{equation*}
          x=4,5,..., \ell-3 \pmod{\ell},
          \quad
          y=4,5,..., \ell-3 \pmod{\ell},
          \quad
          z=4,5,..., \ell-3 \pmod{\ell}.
        \end{equation*}
  \item They result in boxes marking any local structures beyond
        T-junctions, glued crosses, and free crosses
        that would arise under simple rescaling without perturbation.  The perturbations in the next refinement step
        are contained within scaled versions of these boxes;
  \item They ensure that the boxes are disjoint,
        so that the perturbations in the next refinement step do not interfere with one another.
\end{enumerate}

Below we display figures showing all of the gadgets.
The accompanying captions describe the local structures handled by each gadget
  and the sizes of the boxes placed around the degenerate features
  that appear after doubling.
\clearpage
\begin{figure}[H]
  \centering
  \vspace*{-2em}
  \includegraphics[width=0.8\textwidth]{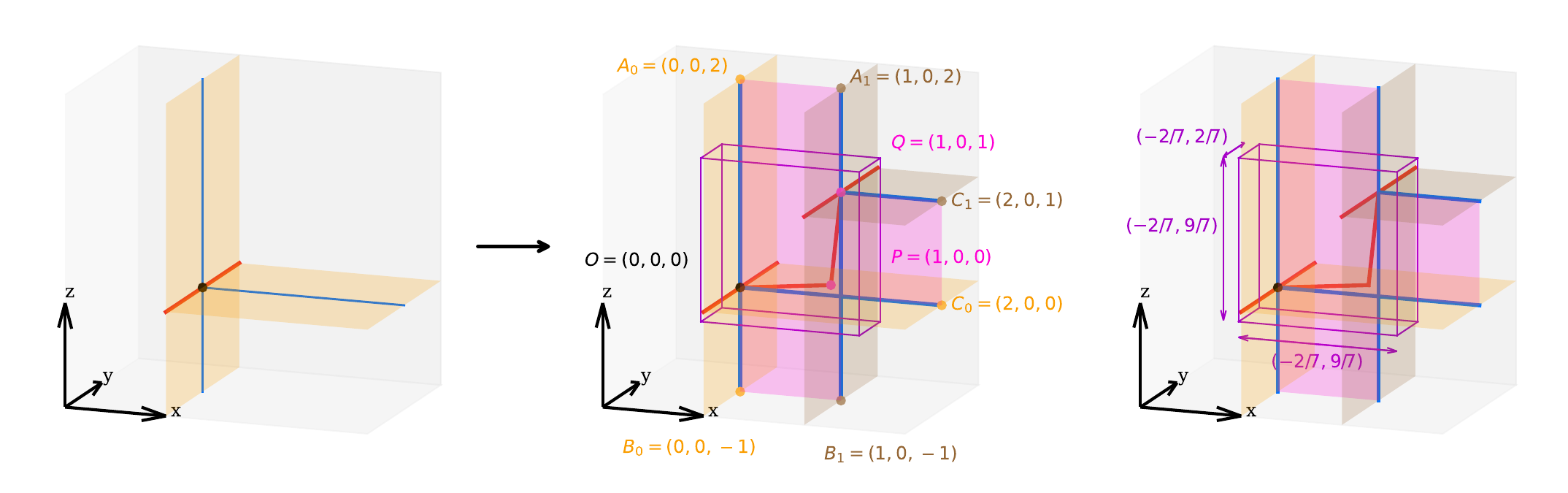}
  \vspace*{-1.5em}
  \caption{The gadgets describing the doubling structure near a T-junction.
            The blue medium-thick segments are the attaching regions for $X_{(i), \mrm{cyl}}$.
            The complexes $X'_{(i)}$, $X'_{(i), \mrm{copy}}$, and $X_{(i), \mrm{cyl}}$
              are shown in orange, brown, and pink, respectively.
            In the figure, we slightly distort the position of $P$ to make the two overlapping edges visible.
            In this neighborhood,
              $X_{(i), \mrm{cyl}}$ consists of three surfaces
              spanned by the vertex chains
              $A_0-O-P-Q-A_1$,
              $B_0-O-P-Q-B_1$,
              and $C_0-O-P-Q-C_1$,
              meeting along the red edge from $O$ to $P$ to $Q$. \\
            \hspace*{1em} After embedding,
              the segments from $(0,0,0)$ to $(0,0,1)$
              and from $(0,0,1)$ to $(1,0,1)$
              each contain two overlapping edges of $X_{(i+1)}$.
            We therefore place a box around these segments:
              $(-\frac{2}{\ell-1}, 1+\frac{2}{\ell-1}) \times (-\frac{2}{\ell-1}, \frac{2}{\ell-1}) \times (-\frac{2}{\ell-1}, 1+\frac{2}{\ell-1})$.
            }
  \label{fig:quantum_doubling_T_junction}
\end{figure}

\begin{figure}[H]
  \centering
  \vspace*{-2em}
  \includegraphics[width=0.8\textwidth]{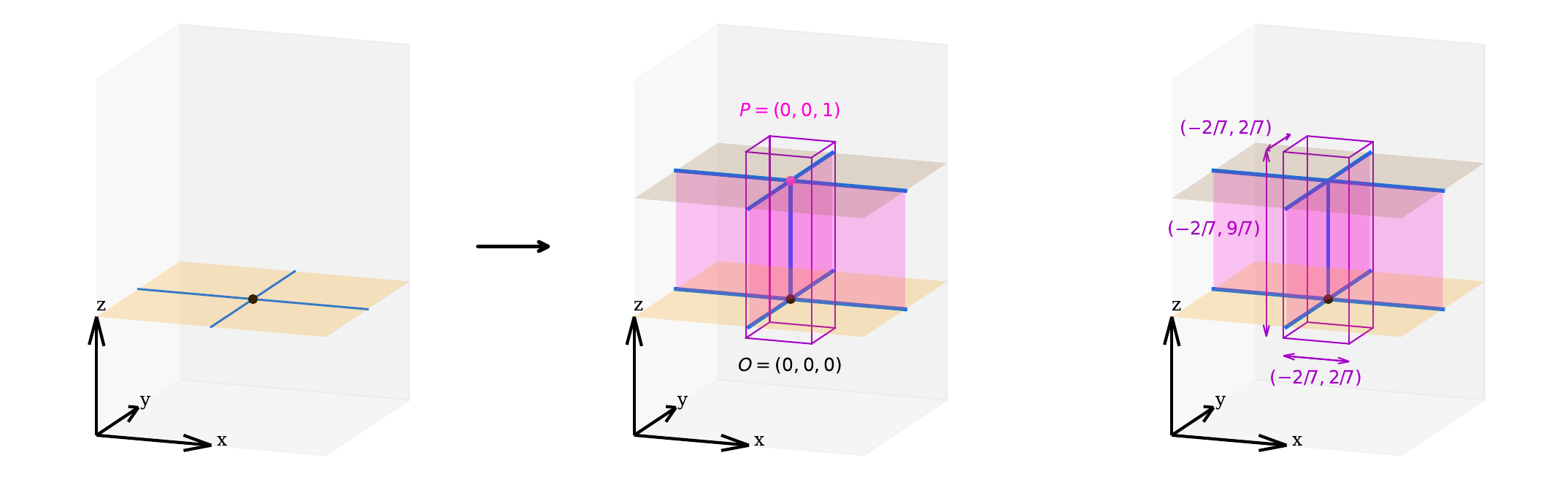}
  \vspace*{-1.5em}
  \caption{The gadgets describing the doubling structure near a glued cross.
            The blue medium-thick segments are the attaching regions for $X_{(i), \mrm{cyl}}$.
            The complexes $X'_{(i)}$, $X'_{(i), \mrm{copy}}$, and $X_{(i), \mrm{cyl}}$
              are shown in orange, brown, and pink, respectively.
            In this neighborhood, $X_{(i), \mrm{cyl}}$ consists of four surfaces meeting along the blue edge from $O$ to $P$. \\
            \hspace*{1em} To ensure that each edge is incident to at most three faces,
              we place a box around the edge $OP$:
              $(-\frac{2}{\ell-1}, \frac{2}{\ell-1}) \times (-\frac{2}{\ell-1}, \frac{2}{\ell-1}) \times (-\frac{2}{\ell-1}, 1+\frac{2}{\ell-1})$.
            }
  \label{fig:quantum_doubling_glued_cross}
\end{figure}

\begin{figure}[H]
  \centering
  \vspace*{-2em}
  \includegraphics[width=0.8\textwidth]{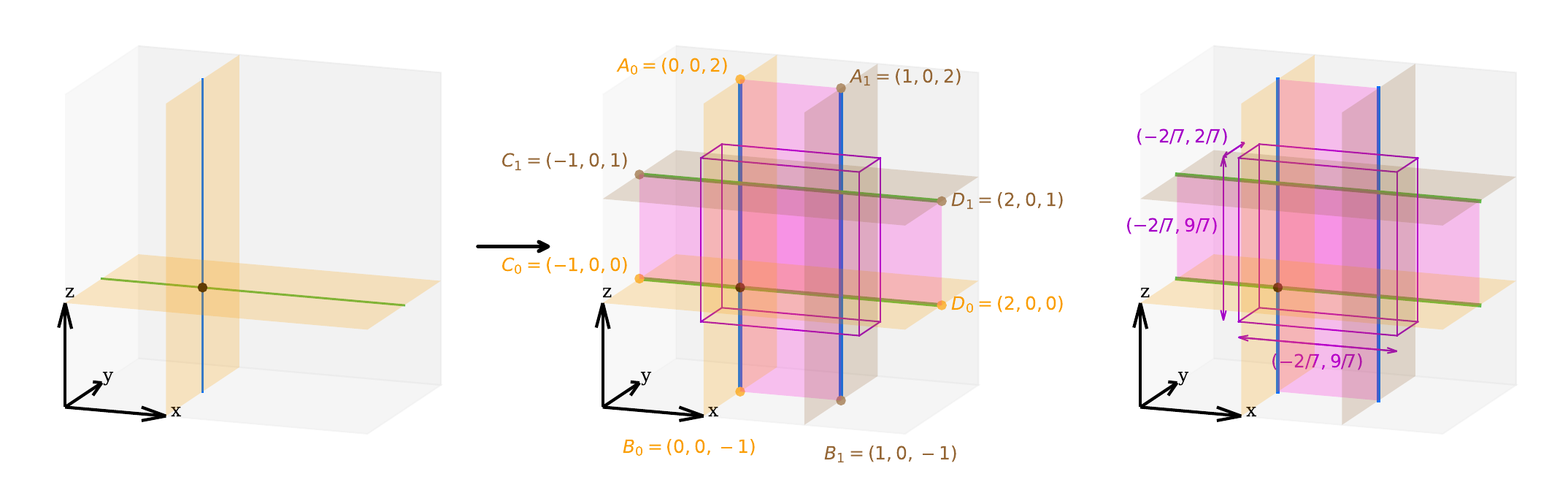}
  \vspace*{-1.5em}
  \caption{The gadgets describing the doubling structure near a free cross.
            The blue and green medium-thick segments are the attaching regions for $X_{(i), \mrm{cyl}}$.
            The complexes $X'_{(i)}$, $X'_{(i), \mrm{copy}}$, and $X_{(i), \mrm{cyl}}$
              are shown in orange, brown, and pink, respectively.
            In this neighborhood, $X_{(i), \mrm{cyl}}$ consists of two disjoint surfaces spanned by the vertex chains
              $A_0-B_0-B_1-A_1$ and
              $C_0-D_0-D_1-C_1$. \\
            \hspace*{1em} After embedding,
              the two surfaces overlap along a unit square.
            We therefore place a box around this square:
              $(-\frac{2}{\ell-1}, 1+\frac{2}{\ell-1}) \times (-\frac{2}{\ell-1}, \frac{2}{\ell-1}) \times (-\frac{2}{\ell-1}, 1+\frac{2}{\ell-1})$.
          }
  \label{fig:quantum_doubling_free_cross}
\end{figure}

Some configurations may appear as subcomplexes of a T-junction,
  such as the examples shown in \Cref{fig:quantum_doubling_plane,fig:quantum_doubling_corner}.
In this case, the same doubling procedure is applied to the relevant subcomplex.

\begin{figure}[H]
  \centering
  \includegraphics[width=0.6\textwidth]{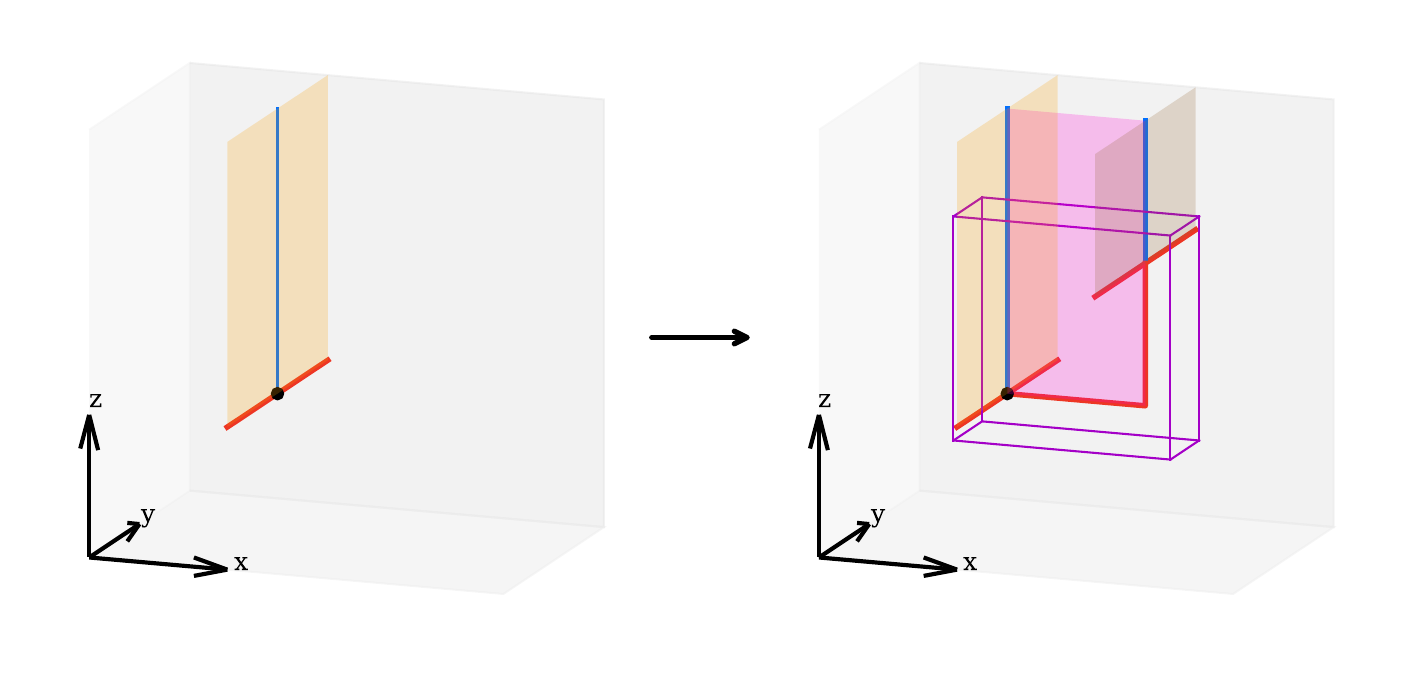}
  \caption{A subcomplex of a T-junction.
            Compared with \Cref{fig:quantum_doubling_T_junction},
              only one of the three surfaces in $X_{(i), \mrm{cyl}}$ is retained,
              namely the surface spanned by the vertex chain
              $A_0-O-P-Q-A_1$.
            The box is technically not necessary since no degeneracy occurs.
              Nevertheless, for uniformity of the proof, and to avoid introducing an additional gadget,
              we still treat it as a subcomplex of the T-junction.
          }
  \label{fig:quantum_doubling_plane}
\end{figure}
\begin{figure}[H]
  \centering
  \includegraphics[width=0.6\textwidth]{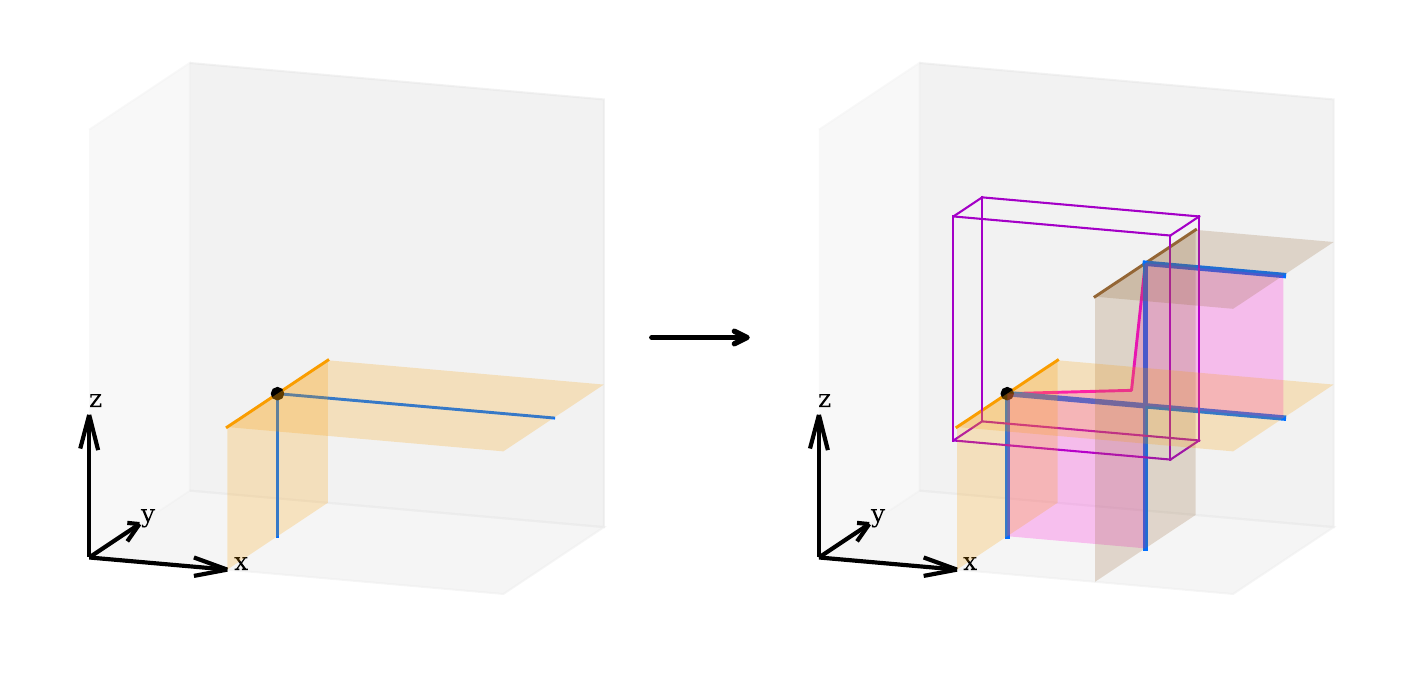}
  \caption{Another subcomplex of a T-junction.
            Compared with \Cref{fig:quantum_doubling_T_junction},
              only two of the three surfaces in $X_{(i), \mrm{cyl}}$ are retained,
              namely those spanned by the vertex chains
              $B_0-O-P-Q-B_1$ and
              $C_0-O-P-Q-C_1$.
            Unlike in  \Cref{fig:quantum_doubling_T_junction},
           the corresponding edge here is not colored red because it is now incident to only two faces.
            In this case,
              the sheaf data is defined implicitly as $\{00, 11\}$.
          }
  \label{fig:quantum_doubling_corner}
\end{figure}

\paragraph{Sheaf data:}

It remains to assign a color to each edge of $X_{(i+1)}$
  whose face degree differs from $2$.
There are two types of such edges.
The first consists of edges in $X'_{(i)}$ and $X'_{(i), \mrm{copy}}$
  that do not lie in the attaching regions.
The second consists of edges in $X_{(i), \mrm{cyl}}$
  introduced during the doubling step.
These include edges in the attaching regions,
  as well as interior edges of $X_{(i),\mrm{cyl}}$,
  which arise in the doubling of T-junctions and glued crosses.

Edges of the first type inherit the color from the corresponding edge of $X_{(i)}$.
Edges of the second type are,
  in most cases,
  colored blue in the iteration $\cR_X$ and red in the iteration $\cR_Z$.
The only exception occurs in the doubling of a T-junction,
  where the edge along $O-P-Q$ inherits the color of the degree-3 edge
  of the original T-junction.

This uniquely specifies the construction of the code $\cQ_{(i+1)}$ from $\cQ_{(i)}$.

\subsubsection{Correctness of the geometric construction and the proof of embedding}

The correctness of the construction for $X_{(i)}, I_{(i)}$
  and the proof of locality and bounded density of the embedding are closely related,
  so we will verify them together.
The main goal of the proof is to ensure that
\begin{itemize}
  \item $X'_{(i)}$ is supported on planes with coordinates $-2, 0, 2 \pmod{\ell}$,
  \item $X'_{(i), \mrm{copy}}$ is supported on planes with coordinates $-1, 1, 3 \pmod{\ell}$,
  \item $X_{(i), \mrm{cyl}}$ is supported on planes with coordinates $4,..., \ell-3 \pmod{\ell}$.
\end{itemize}
These congruence classes keep the three sets of regions separated
  and allow us to rule out unwanted overlaps.

More systematically,
  we decompose $\RR^3$ into four types of regions,
  which will be analyzed separately:
\begin{itemize}
  \item $W_0 = \Big([-2.5,3.5] \times [-2.5,3.5] \times [-2.5,3.5]\Big) + (\ell \zz)^3$,
  \item $W_1 =
    \left(
      \bigcup_{\mrm{perm}}
      [-2.5,3.5] \times [-2.5,3.5] \times [3.5,\ell-2.5]
    \right)
    + (\ell \zz)^3$,
  \item $W_2 =
    \left(
      \bigcup_{\mrm{perm}}
      [-2.5,3.5] \times [3.5,\ell-2.5] \times [3.5,\ell-2.5]
    \right)
    + (\ell\mathbb{Z})^3$,
  \item $W_3 = \Big([3.5,\ell-2.5] \times [3.5,\ell-2.5] \times [3.5,\ell-2.5]\Big) + (\ell \zz)^3$.
\end{itemize}
Here $\bigcup_{\mrm{perm}}$ denotes the union over all permutations
of the three coordinate intervals.
Intuitively, the regions $W_0$ and $W_1$ contain the perturbations performed near the boundaries of the scaled squares,
  while the regions $W_2$ captures the interior regions of the scaled squares,
  which we are already familiar with from the examples in the overview.

After the subdivision and perturbation step,
  inspection of the gadgets in \Cref{fig:quantum_refinement_T_junction,fig:quantum_refinement_box,fig:quantum_refinement_box-2,fig:quantum_refinement_glued_cross,fig:quantum_refinement_free_cross}
  shows that the local structures in $X'_{(i)}$ are as follows (see the beginning of Subsubsec.~\ref{sec:labelinggadgets} for the labeling (1)-(4) of gadgets):
\begin{itemize}
  \item In $W_0$, there are no ribs,
          but there may be boxes arising from gadgets (1) and (2).
  \item In $W_1$, there are no boxes,
          but there may be isolated T-junctions and free crosses.
  \item In $W_2$, there are no boxes,
          but there may be isolated glued crosses.
  \item In $W_3$, there are no local structures.
\end{itemize}

Each of the gadgets (1)-(4) is equipped with a box: we will say that these boxes are of type (1)-(4). After the doubling step,
  inspection of the gadgets in \Cref{fig:quantum_doubling_T_junction,fig:quantum_doubling_glued_cross,fig:quantum_doubling_free_cross}
  shows the possible regions that boxes in $X_{(i+1)}$ are located in:
\begin{itemize}
  \item In $W_0$, there are boxes inherited from $X'_{(i)}$
        and from its translated copy $X'_{(i), \mrm{copy}}$.  These boxes are of type (2).  No additional boxes are induced by $X_{(i), \mrm{cyl}}$.
  \item In $W_1$, boxes are induced by the doubling of T-junctions
        and free crosses, giving boxes of types (1) and (4), respectively.
  \item In $W_2$, boxes are induced by the doubling of glued crosses,
        giving boxes of type (3).
  \item In $W_3$, there are no boxes.
\end{itemize}

\begin{table}[ht]
  \centering
  \begin{tabular}{c|cc|c}
    Region & Boxes & Ribs & New boxes after doubling \\
    \hline
    $W_0$ & From gadgets (1) and (2) & $\emptyset$ & Need gadget (2) \\
    $W_1$ & $\emptyset$ & T-junctions and free crosses & Need gadgets (1) and (4) \\
    $W_2$ & $\emptyset$ & Glued crosses & Need gadget (3) \\
    $W_3$ & $\emptyset$ & $\emptyset$ &  $\emptyset$ \\
  \end{tabular}
  \caption{Summary of the behavior before and after doubling.}
  \label{tab:region-behavior}
\end{table}

We summarize the behavior in the regions $W_0, W_1, W_2, W_3$
  before and after doubling in \Cref{tab:region-behavior}.  We now establish the key lemma for the induction.
The proof proceeds by analyzing the four regions separately.
\begin{lemma} \label{lemma:quantum_embedding_well_defined}
  Suppose that $X_{(i)}$ and $I_{(i)}$, together with the boxed regions,
    satisfy the following properties:
  \begin{itemize}
    \item the boxes are disjoint;
    \item the structure outside the boxes consists only of simple local structures.
  \end{itemize}
  Then the construction describe above is well defined.
  Moreover, $X_{(i+1)}$ and $I_{(i+1)}$, together with the new boxed regions,
    satisfy the same two properties.
\end{lemma}
\begin{proof}
  By the induction hypothesis, the boxes are disjoint.
  Thus, the perturbations in the refinement step can be performed independently,
    so the complex $X'_{(i)}$ is well defined.

  For the doubling step to be well defined,
    we must verify that local structures near the attaching regions of $X'_{(i)}$
    are T-junctions, glued crosses, or free crosses.

  We first check this outside the scaled boxes coming from $X_{(i)}$.
  By the induction hypothesis,
    the structure outside the boxes in $X_{(i)}$
    consists only of simple local structures.
  After scaling,
    these simple local structures produce only
    T-junctions, glued crosses, or free crosses near the attaching regions.

  We now check this inside the scaled boxes.
  By inspecting the perturbation gadgets in \Cref{fig:quantum_refinement_T_junction,fig:quantum_refinement_box,fig:quantum_refinement_box-2,fig:quantum_refinement_glued_cross,fig:quantum_refinement_free_cross},
    we verify that the local structures near the attaching regions
    are again T-junctions, glued crosses, or free crosses.

  Therefore,
    the doubling step can be defined by specifying the doubling gadgets
    for these three types of local structures,
    and hence $X_{(i+1)}$ is well defined.

  We verify the two properties of $X_{(i+1)}$ one by one.
  \begin{claim} \label{claim:quantum_embedding_disjoint_boxes}
    The boxes in $X_{(i+1)}$ are disjoint.
  \end{claim}
  \begin{proof}
    We verify that the boxes are disjoint in all four regions
      $W_0, W_1, W_2, W_3$.
    We use the shorthand $(\{-2,0\}, 0, 0)$
      to denote the set of vectors $(-2,0,0)$ and $(0,0,0)$,
      and similarly for expressions such as
      $(\{-2,0,2\}, \{-2,0,2\}, z)$.
    \begin{itemize}
      \item For each component of $W_0$,
            $\big([-2.5,3.5] \times [-2.5,3.5] \times [-2.5,3.5]\big) + (\ell x', \ell y', \ell z')$,
            where $x', y', z' \in \zz$,
            there may be a box of the form
            \[
              \left(-\frac{2}{\ell-1}, 2+\frac{2}{\ell-1}\right)
              \times
              \left(-\frac{2}{\ell-1}, \frac{2}{\ell-1}\right)
              \times
              \left(-\frac{2}{\ell-1}, \frac{2}{\ell-1}\right)
              + (\{-2,0\}, 0, 0)
            \]
            up to permutation of the coordinates.
            The corresponding box in the translated copy $X'_{(i), \mrm{copy}}$ is obtained by shifting the box by $(1,1,1)$.
            These boxes are disjoint.
      \item In the component
            $[-2.5,3.5] \times [-2.5,3.5] \times [3.5,\ell-2.5] + (\ell \zz)^3$
            of $W_1$,
            the boxes induced by the doubling of T-junction and free crosses have the form
            \[
              \left(-\frac{2}{\ell-1}, 1+\frac{2}{\ell-1}\right)
              \times
              \left(-\frac{2}{\ell-1}, 1+\frac{2}{\ell-1}\right)
              \times
              \left(-\frac{2}{\ell-1}, \frac{2}{\ell-1}\right)
              + (\{-2,0,2\}, \{-2,0,2\}, z)
              + (\ell \zz)^3,
            \]
            where $z \in \{4, 5, ..., \ell-3\}$.
            The other components of $W_1$ are treated in the same way,
              up to permutation of the coordinates.
            These boxes are disjoint.
      \item In the component
            $[-2.5,3.5] \times [3.5,\ell-2.5] \times [3.5,\ell-2.5] + (\ell \zz)^3$
            of $W_2$,
            the boxes induced by the doubling of glued crosses have the form
            \[
              \left(-\frac{2}{\ell-1}, 1+\frac{2}{\ell-1}\right)
              \times
              \left(-\frac{2}{\ell-1}, \frac{2}{\ell-1}\right)
              \times
              \left(-\frac{2}{\ell-1}, \frac{2}{\ell-1}\right)
              + (\{-2,0,2\}, y, z)
              + (\ell \zz)^3,
            \]
            where $y, z \in \{4, 5, ..., \ell-3\}$.
            The other components of $W_2$ are treated in the same way,
              up to permutation of the coordinates.
            These boxes are disjoint.
      \item In $W_3$, there are no boxes.
    \end{itemize}
  \end{proof}

  \begin{claim} \label{claim:quantum_embedding_outside_boxes}
    The structure outside of the boxes in $X_{(i+1)}$
      consists only of simple local structures.
  \end{claim}
  Recall that the simple local structures consist of one square type and two for segment types:
  \begin{itemize}
    \item a lattice unit square that is the image of a single face of $X_{(i+1)}$;
    \item a lattice unit-length segment that is the image of a single edge of $X_{(i+1)}$
            and is incident to at most three faces of $X_{(i+1)}$;
    \item a lattice unit-length segment that is the common image of two edges of $X_{(i+1)}$,
            each incident to at most two faces of $X_{(i+1)}$,
            such that the corresponding faces form locally flat sheets
            meeting along the segment.
  \end{itemize}

  \begin{proof}
    We verify the claim in all four regions
      $W_0, W_1, W_2, W_3$.
    \begin{itemize}
      \item In $W_0$, the structure consists of the portions of
              $X'_{(i)}$ and $X'_{(i),\mrm{copy}}$.
            First consider $X'_{(i)}$ by itself, ignoring the translated copy.
            By inspecting the gadgets in \Cref{fig:quantum_refinement_T_junction,fig:quantum_refinement_box,fig:quantum_refinement_box-2,fig:quantum_refinement_glued_cross,fig:quantum_refinement_free_cross},
              we verify that, outside of the boxed regions,
              each lattice unit square is occupied at most once,
              and that each lattice unit-length segment is one of the two permissible types.
            The same statement holds for the translated copy $X'_{(i),\mrm{copy}}$
              considered by itself.

            \hspace*{1em} It remains to rule out unwanted overlaps between the two copies:
              namely that a lattice unit square is occupied once by $X'_{(i)}$
                and once by $X'_{(i),\mrm{copy}}$,
              or that a lattice unit-length segment is occupied once by each copy
                in a way not covered by the two types.

            \hspace*{1em} The first possibility cannot occur,
              since the unit squares in $X'_{(i)}$ have coordinates in $\{-2,0,2\}$,
              whereas the unit squares in $X'_{(i),\mrm{copy}}$ have coordinates in $\{-1,1,3\}$.

            \hspace*{1em} The only possible overlaps occur along lattice unit-length segments,
              and in this case the corresponding faces form locally flat sheets meeting along the segment.
            Hence these overlaps are of the last type listed above.
      \item In $W_1$, the structure is formed by doubling of T-junction and free crosses.
            By inspecting the gadgets in \Cref{fig:quantum_doubling_T_junction,fig:quantum_doubling_free_cross},
              we verify that, outside of the boxed regions,
              each lattice unit square is occupied at most once,
              and that each lattice unit-length segment is one of the two permissible types.
      \item In $W_2$, the structure is formed by doubling of glued crosses.
            By inspecting the gadget in \Cref{fig:quantum_doubling_glued_cross},
              we verify that, outside of the boxed regions,
              each lattice unit square is occupied at most once,
              and that each lattice unit-length segment is one of the two permissible types.
      \item In $W_3$, there are no faces or edges.
    \end{itemize}
  \end{proof}
This proves the Lemma.
\end{proof}

\begin{corollary}
  The code family we constructed above is well defined.

  Furthermore,
    $I_{(i)}$ is a $(30,1)$-embedding of $X_{(i)}$ in $\RR^3$.
  In particular,
    every open ball of diameter $1$ in $\RR^3$ intersects at most $30$ faces of the image of $X_{(i)}$
    and adjacent $0$-cells are mapped to points with distance at most $1$.
\end{corollary}
\begin{proof}
  It is straightforward to check that the initial code $\cQ_{(0)}$
    contains only simple local structures,
    so the base case of the induction is satisfied.
  The induction step is verified by \Cref{lemma:quantum_embedding_well_defined}.
  Hence the code family is well defined.

  Since the construction is well-defined,
    geometric locality follows immediately:
    each edge of $X_{(i)}$ is mapped to a unit-length segment or collapsed to a point,
    so adjacent $0$-cells are mapped to points at distance at most $1$.

  It remains to verify bounded density.
  Outside the boxed regions,
    \Cref{claim:quantum_embedding_outside_boxes},
    implies that each lattice square is the image of at most one face of $X_{(i)}$.
  Inside the boxed regions,
    the disjointness of the boxes, established in
    \Cref{claim:quantum_embedding_disjoint_boxes},
    allows us to analyze each box independently.
  The gadgets show that each lattice unit square inside a box
    is the image of at most two nondegenerate faces of $X_{(i)}$.
  In addition,
    there is at most one degenerate face supported on each lattice unit-length segment.

  Since every open ball of diameter $1$ intersects at most $12$ lattice unit squares and at most $6$ lattice unit-length segments,
    it follows that every such ball intersects at most $2 \times 12 + 6 = 30$
    faces of the image of $X_{(i)}$.
\end{proof}

\subsubsection{Discussion of self-correction}

We briefly sketch the argument for the exponential scaling of the memory lifetime,
  which follows closely with proof in the construction with random embedding.
As in \Cref{sec:decoder}, we construct two decoders,
  one for $\cR_{X}$ and one for $\cR_{Z}$.
Each decoder is built from coarse-graining maps
  that take a syndrome $\sigma_i$ at level $i$
  to a syndrome $\sigma_{i-1}$ at level $i-1$, together with a correction $f_i$.
These coarse-graining maps again satisfy
  syndrome reduction and local computability.

In particular, for the code $\cQ(X, \cF, I)$, we have:
\begin{itemize}
  \item for $\cR_{Z}$, $|\sigma_{i-1}| \le |\sigma_{i}|$;
  \item for $\cR_{X}$, $|\sigma_{i-1}| \le \frac{1}{2} |\sigma_{i}|$.
\end{itemize}
For the code $\cQ(X, \cF^\perp, I)$, we have:
\begin{itemize}
  \item for $\cR_{Z}$, $|\sigma_{i-1}| \le \frac{1}{2} |\sigma_{i}|$;
  \item for $\cR_{X}$, $|\sigma_{i-1}| \le |\sigma_{i}|$.
\end{itemize}

These coarse-graining maps are then incorporated into
  the decoding graph to define witnesses,
  and the remainder of the analysis (i.e. constructing the decoding graph, witness subgraph, and the Peierls argument) proceeds analogously to that of \Cref{sec:memlifetime}.

There are two main differences.
  The first lies in the explicit description of the coarse-graining maps.
Starting from a syndrome $\sigma_i$,
  we use a different cleaning strategy from those in
  \Cref{sec:coarse-graining-RZ} and \Cref{sec:coarse-graining-RX}.
Here the basic unit is the doubling of a square illustrated in \Cref{fig:quantum_doubling_example}.
This is the reason for choosing $\ell = 8$:
  the resulting unit contains two parallel horizontal ribs
  and two parallel vertical ribs.
Having two parallel ribs allows us to clean the syndrome
  from the interior of the unit without increasing its weight,
  whereas with only one rib in each direction
  such a cleanup could increase the syndrome weight.

After this cleaning, the syndrome is supported on the boundaries of unit,
  precisely on the edges corresponding to those in $X_{i-1}$.
We can then apply the inverse map $\cF^{-1}$
  to obtain the coarse-grained syndrome $\sigma_{i-1}$.

The second difference concerns the symmetry between $X$ and $Z$ errors.
In the earlier discrete description, this symmetry is explicit.
In the sheaf-code formulation,
  we instead use Poincar\'e duality,
  as stated in \Cref{thm:poincare-duality}.
In particular,
  if we show that $C(X,\cF)$ is self-correcting against $X$ errors,
then the same argument shows that
  $C(X,\cF^{\perp})$ is self-correcting against $X$ errors.
By \Cref{thm:poincare-duality},
  this implies that $C(X,\cF)$ is self-correcting against $Z$ errors.

\section{Acknowledgments}
We thank Elia Portnoy for helpful discussions about embedding.  ChatGPT 5.5 was helpful in improving the clarity of our manuscript.   S.B. and M.D. were supported by the Walter Burke Institute for Theoretical Physics at Caltech and by the Institute for Quantum Information and Matter, an NSF Physics Frontiers Center (NSF Grant PHY-2317110).  TCL was supported in part by funds provided by the U.S. Department of Energy (D.O.E.) under the cooperative research agreement DE-SC0009919 and by the Simons Collaboration on Ultra-Quantum Matter, which is a grant from the Simons Foundation (652264 JM).

\appendix
\section*{Appendix}
\addcontentsline{toc}{section}{Appendix}

\section{Explicit verification of the homological perturbation lemma}
\label{app:homological-pert-lemma}

In the homological perturbation lemma, we start with the following data:
\begin{itemize}
  \item a chain complex $(C, d)$,
  \item a smaller chain complex $(H, d_H)$,
  \item chain maps $i: H \to C$, $p: C \to H$, and $h: C \to C[-1]$.
\end{itemize}

These maps are required to satisfy
\begin{enumerate}
  \item $p i = \id_H$,
  \item $i p = \id_C + h d + d h$,
  \item $h i = 0$,
  \item $p h = 0$,
  \item $h^2 = 0$.
\end{enumerate}
Conditions~(1) and~(2) say that
  $(H, d_H)$ is a deformation retract of $(C, d)$,
  while conditions~(3)--(5) are the annihilation conditions on $h$.
Although the annihilation conditions are not always assumed,
  the initial data can always be altered to satisfy them.
  If~(3) and~(4) fail, replace $h$ by $(dh + hd)\, h\, (dh + hd)$;
  if~(5) then still fails, replace $h$ by $h d h$.

We now perturb the differential on $C$ by $\delta$,
  so the new differential is
\begin{equation}
  d' = d + \delta, \qquad (d')^2 = 0.
\end{equation}

The homological perturbation lemma asserts that,
  under a suitable smallness condition on $\delta$,
  this perturbation transfers to $H$.
In other words, there exist new maps
\begin{equation}
  d_H', \quad i', \quad p', \quad h'
\end{equation}
such that $(H, d_H')$ is again a deformation retract of $(C, d')$.

The key intermediate quantity is the geometric-series sum
\begin{equation}
  \Sigma = \sum_{k=0}^{\infty} (\delta h)^k \delta
         = \delta + \delta h \delta + \delta h \delta h \delta + \cdots,
\end{equation}
and the new data are given explicitly by
\begin{equation}
  d_H' = d_H + p \Sigma i, \qquad
  i'   = i   + h \Sigma i, \qquad
  p'   = p   + p \Sigma h, \qquad
  h'   = h   + h \Sigma h.
\end{equation}
We will verify that the new data satisfy the 5 conditions above.

To verify condition~(1),
  note that the annihilation conditions $hi = 0$, $ph = 0$, and $h^2 = 0$ kill three of the four cross terms:
\begin{align}
  p' i'
  &= (p + p \Sigma h)(i + h \Sigma i) \\
  &= p i + p \Sigma h i + p h \Sigma i + p \Sigma h h \Sigma i \\
  &= \id_H
\end{align}

For condition~(2),
  we must show $i'p' = \id_C + d'h' + h'd'$.
We compute each side and show they agree.
For the left-hand side:
\begin{align}
  i' p'
  &= (i + h \Sigma i)(p + p \Sigma h) \\
  &= i p + h \Sigma i p + i p \Sigma h + h \Sigma i p \Sigma h \\
  &= (\id_C + h d + d h) + h' \delta (\id_C + h d + d h) + (\id_C + h d + d h) \delta h' + h' \delta (\id_C + h d + d h) \delta h' \\
  &= (\id_C + h d + d h)
  + h' \delta h d + h' \delta + d h \delta h' + \delta h' \\
  &\quad + h' \delta d h + h d \delta h' + h' \delta \delta h' + h' \delta h d \delta h' + h' \delta d h \delta h',
\end{align}
where we used $ip = \id_C + hd + dh$ and $h\Sigma = h'\delta$, $\Sigma h = \delta h'$.

For the right-hand side:
\begin{align}
  \id_C + h' d' + d' h'
  &= \id_C + (h + h \Sigma h)(d + \delta) + (d + \delta)(h + h \Sigma h) \\
  &= \id_C + h d + h \Sigma h d + h' \delta + d h + d h \Sigma h + \delta h' \\
  &= (\id_C + h d + d h)
  + h' \delta h d + h' \delta + d h \delta h' + \delta h'.
\end{align}

Subtracting, we must show the remainder vanishes:
\begin{align}
  &h' \delta d h + h d \delta h' + h' \delta \delta h' + h' \delta h d \delta h' + h' \delta d h \delta h' \\
  &= (h' \delta d h + h' \delta d h \delta h') + (h d \delta h' + h' \delta h d \delta h') + h' \delta \delta h' \\
  &= h' \delta d h' + h' d \delta h' + h' \delta \delta h' \\
  &= h' (d')^2 h' \\
  &= 0.
\end{align}
In the second line we used $d^2 = 0$ and $d' = d + \delta$.

Conditions~(3)--(5) follow immediately from the annihilation conditions
$hi = 0$, $ph = 0$, $h^2 = 0$,
\begin{align}
  h'i' &= (h + h\Sigma h)(i + h\Sigma i)
        = hi + hh\Sigma i + h\Sigma hi + h\Sigma hh\Sigma i = 0, \\
  p'h' &= (p + p\Sigma h)(h + h\Sigma h)
        = ph + ph\Sigma h + p\Sigma hh + p\Sigma hh\Sigma h = 0, \\
  h'^2 &= (h + h\Sigma h)^2
        = hh + hh\Sigma h + h\Sigma hh + h\Sigma hh\Sigma h = 0.
\end{align}

It remains to verify that $(d_H')^2 = 0$:
\begin{align}
  d_H' d_H'
  &= (d_H + p \Sigma i)(d_H + p \Sigma i) \\
  &= d_H^2 + d_H p \Sigma i + p \Sigma i d_H + p \Sigma i p \Sigma i \\
  &= 0 + p d \Sigma i + p \Sigma d i + p \Sigma i p \Sigma i \\
  &= p d \Sigma i + p \Sigma d i + p \Sigma (\id_C + h d + d h) \Sigma i \\
  &= p' d \Sigma i + p \Sigma d i' + p \Sigma \Sigma i \\
  &= p' d \delta i' + p' \delta d i' + p' \delta \delta i' \\
  &= p' (d')^2 i' \\
  &= 0.
\end{align}

\section{Deferred embedding proofs}\label{app:extraemb}

Here, we will provide a proof of the following claim, which is used to show that there exists a perturbation inducing a $(t',\ell)$-embedding in $\mathbb{R}^3$.  We refer the reader to the description of step 1 in the random construction for the relevant notation.

\begin{claim}
Given a fixed unit ball $B'(1) \subset B(\alpha \lambda)$, the probability that a 2-simplex in $B(\alpha \lambda)$ intersects $B'(1)$ after perturbation is $\leq \Gamma/\lambda$ for some constant $\Gamma > 0$.
\end{claim}
\begin{proof}

To prove this (including the constants), suppose under perturbation, a 2-simplex $\sigma$ with the normal vector $\hat{n}$ intersects $B'(1)$.  Next, consider all triples of points $a,e,f$ such that: $a$ is a vertex of the 2-simplex, $f$ is a point on the opposite edge, $e$ is the closest point in ${B}'(1)$ to $f$ such that $a,e$ and $f$ are collinear.
Define $d$ as the maximum length of line segment $ef$ over all valid triples.  For the vertex $a$ that maximizes $d$, call $C_a$ the corresponding cap and $p_a$ the perturbation vector on $C_a$.  Construct a plane $P$ passing through $p_a$ with normal vector $\hat{n}$.  If we moved $a$ by $d' \hat{n}$ or $-d' \hat{n}$, by similar triangles, $e$ must be moved by $\geq \frac{d}{(2\alpha+1)\lambda} d'$ (in the direction of $\hat{n}$) to be collinear with $a$ and $f$.  We would like to determine how much $a$ would need to be moved such that the 2-simplex no longer intersects $B'(1)$. This is achieved when $a$ is moved a distance $d'$ satisfying $\frac{d}{(2 \alpha + 1)\lambda} d' \geq 1$, or $d' \geq \frac{(2 \alpha + 1)\lambda}{d}$.  The equation describing $P$ is $\hat{n} \cdot(x-p_a) = 0$.  Define the thickened plane $P_t$ by $|\hat{n} \cdot(x-p_a)| \leq \frac{(2 \alpha + 1)\lambda}{d}$. If $p_a$ is replaced with a vector lying outside of $P_t \cap C_a$ (but inside $C_a$), the simplex would avoid $B'(1)$.  The area of $P_t\cap C_a$ is at most $4 \pi \alpha \lambda \frac{(2 \alpha + 1)\lambda}{d}$.

If $p_a$, $p_{b}$ and $p_{c}$ are vectors chosen in caps $C_a$, $C_b$, $C_c$, we want to bound the probability of the corresponding simplex intersecting with $B'(1)$, over a product measure $\mu = \mrm{Unif}(C_a) \times \mrm{Unif}(C_b) \times \mrm{Unif}(C_c)$.  Define $\mathbbm{1}(p_a, p_b, p_c) = 1$ if the corresponding simplex intersects with $B'(1)$ and $\mathbbm{1}(p_a, p_b, p_c) = 0$ otherwise.  Further define $\mrm{lab}(p_a, p_b, p_c) \in \{a,b,c\}$ to be the label of the vertex whose position is varied.  Then
\begin{equation}
\mathbb{P}[\sigma \cap B'(1) \neq \varnothing] = \mathbb{E}_{p_a, p_b, p_c \sim \mu}[\mathbbm{1}(p_a, p_b, p_c)] \leq \frac{1}{Q^3} \sum_{l = a,b,c} \,\, \sum_{\substack{p_a,p_b,p_c: \\\mrm{lab}(p_a, p_b, p_c) = l}} \mathbbm{1}(p_a, p_b, p_c)
\end{equation}
where $Q$ is the cap area.  To compute the inner sum, suppose without loss of generality that $l = a$.  Then, sum over all pairs $(p_b, p_c)$ such that $\exists p_a: \mrm{lab}(p_a, p_b, p_c) = a$.  For such a pair $(p_b, p_c)$, the total measure of $p_a$ satisfying $\mathbbm{1}(p_a, p_b, p_c) = 1$ is $\leq 4 \pi \alpha \lambda \frac{(2 \alpha + 1)\lambda}{d}$ by the above argument.  Thus,
\begin{equation}
\sum_{\substack{p_a,p_b,p_c: \\\mrm{lab}(p_a, p_b, p_c) = l}}\mathbbm{1}(p_a, p_b, p_c) \leq 4 \pi \alpha \lambda Q^2  \frac{(2 \alpha + 1)\lambda}{d}.
\end{equation}
and $\mathbb{P}[\sigma \cap B'(1) \neq \varnothing] \leq \frac{12 \pi \alpha \lambda}{Q} \frac{(2 \alpha + 1)\lambda}{d}$.  The final step is to argue that this quantity scales like $1/\lambda$.  The cap area is $Q = 2 \pi \alpha^2\lambda^2(1-\cos\frac{\alpha_1}{2\alpha})$ where $\alpha_1 \lambda$ is the arc length of the cap.

Next, we need a lower bound on $d$. Any point $y$ inside $\sigma$ can be parametrized as $y = \theta_a a + \theta_b b + \theta_c c$ with $\theta_a + \theta_b + \theta_c = 1$ and $\theta_a, \theta_b, \theta_c \geq 0$. Consider such a point that is also in $B'(1)$, i.e. $ y \in B'(1) \cap \sigma$. Denoting $q$ to be the point where the line through $a$ and $y$  intersects $bc$, we can write $y = \theta_a a+ (1-\theta_a) q$. Then, we have $\norm{y-q} = \theta_a \norm{a- q} \geq \theta_ a \gamma \lambda$ by  the bound on the simplex width in Claim~\ref{claim:simplex-thickness}. By definition of $d$, and denoting $z$ the location of the center of the ball $B'(1)$, we have
\begin{equation}
    d  = \max_{a,z}\min _{y \in B'(1) \cap \sigma, q \in bc} \norm{y-q}.
\end{equation}
Because of the first maximum, we have $\theta_a \geq 1/3$, and $d \geq \gamma \lambda/3$.  Therefore, there indeed exists a constant $$\Gamma(\alpha, \alpha_2, \alpha_1, \delta) \leq \frac{18 (2\alpha+1)}{\alpha \gamma \left(1-\cos\frac{\alpha_1}{2\alpha}\right)} $$ such that $\mathbb{P}[\sigma \cap B'(1) \neq \varnothing] \leq \Gamma/\lambda$.
\end{proof}

\section{A tensor product construction} \label{appB:product-construction}

In this appendix, we briefly discuss an alternative construction of a 3D self-correcting quantum memory which is based on the tensor, or homological, product of two chain complexes associated with two classical codes.  We keep the discussion somewhat informal, but all the details can be made rigorous.

Consider first the classical code introduced in Subsec.~\ref{sec:classicalcode}
\begin{equation}
    \cC_{(2k)} = (\cR^T\circ \cR)^k(\cC_{(0)})
\end{equation}
and call $\cD_{(2k)}:=\cC_{(2k)}^T$.  Let $\cQ_{(2k)}$ be the
quantum CSS code associated with the homological product
\begin{equation}
    \cQ_{(2k)} := \cC_{(2k)}\otimes \cD_{(2k)} = \cC_{(2k)}\otimes \cC_{(2k)}^T.
\end{equation}
We claim that $\cQ_{(2k)}$ has self-correction properties after one specifies an appropriate ``subsystem'' encoding scheme.
The need for this comes from the fact that the construction of $\cC_{(2k)}$ creates codewords at many length scales.  In particular, each $\cR^T$ step introduces many $O(1)$-distance codewords, while each $\cR$ step doubles the distance of existing codewords.  For $\cC_{(2k)}^T$, the roles of $\cR$ and $\cR^T$ are interchanged.  Hence, the quantum code $\cQ_{(2k)}$ inherits small logical operators, and the logical qubits associated with these operators cannot themselves be self-correcting. We therefore define a protected subsystem by treating the small logical operators as ``gauge operators'', borrowing terminology from subsystem codes.  The physical Hamiltonian is still the stabilizer Hamiltonian associated with $\cQ_{(2k)}$, but the encoding and decoding are defined only modulo the small logical operators.  Our claim is that using this encoding scheme yields protected qubits whose $X$- and $Z$-type logical distances are both large and whose memory lifetime is exponential (in a power of the system size) at sufficiently low temperature.

More formally, we express $\cC_{(2k)}$ as a two-term chain complex $C_1 \rightarrow C_0$, where bits label the basis of the degree-1 vector space, and the checks label the basis of the degree-0 vector space. The number of codewords of this code scales as $H_1(\cC_{(2k)}) \cong\ff_2^{\mrm{poly}(n)}$.
We define an auxiliary code $\cC'_{(2k)}$, which is a \emph{3-term} chain complex:
$$C_2' \longrightarrow C_1 \longrightarrow C_0,$$
where $B_1(C'_2)$ captures the space of all the small-distance codewords, which are now treated as gauge degrees of freedom. Thus, we have $H_1(\cC'_{(2k)}) = \ff_2$. The generator of this homology class defines the ``large'' codeword of $\cC_{(2k)}$.
One can then lower-bound the memory lifetime for $\cC_{(2k)}$ with this encoding. In particular, one can show that the modified isoperimetric inequality
\begin{equation}\label{eq:gaugeiso}
  |\partial c| \ge \sigma \min_{b \in B_1(\cC'_{(2k)})} |c + b|^\eta
\end{equation}
for all $c \in A_1$ and constants $\sigma, \eta > 0$ is satisfied.  Since $\cD_{(2k)}=\cC_{(2k)}^T$ has analogous properties, we can define an analogous auxiliary complex $\cD'_{(2k)}$ and show an analogous isoperimetric inequality for $\cD_{(2k)}$.

We now describe the protected subspace of the product code. The $Z$-type logical operators of $\cQ_{(2k)}$ are captured by $H_1(\cQ_{(2k)})$. By the K\"unneth formula:
\begin{equation}
    H_1(\cQ_{(2k)}) \cong H_1(\cC_{(2k)}) \otimes H_0(\cD_{(2k)}) \, \oplus \, H_0(\cC_{(2k)}) \otimes H_1(\cD_{(2k)})
\end{equation}
Consider first the first summand, $H_1(\cC_{(2k)}) \otimes H_0 (\cD_{(2k)})$.  In the protected subsystem, we keep only the large codeword of $\cC_{(2k)}$, replacing $H_1(\cC_{(2k)})$ by $H_1(\cC'_{(2k)})\cong \ff_2$.  We must also isolate the corresponding large class in $H_0(\cD_{(2k)})$.

$H_0(\cD_{(2k)})$ is associated with the relations between the checks of the code $\cD_{(2k)}$; it turns out that this code has a large $\mrm{poly}(n)$ number of relations. To remove the ``small'' relations from the protected subsystem, we extend $\cD_{(2k)}'$ to a \emph{4-term} chain complex via:
$$ \cD''_{(2k)}:  D_2' \rightarrow D_1 \rightarrow D_0 \rightarrow D_{-1}'$$
where the part $D_0 \rightarrow D_{-1}'$ removes all the small zeroth homology classes, leaving $H_0(\cD''_{(2k)}) \cong \ff_2$. We define $\cC''_{(2k)}$ analogously. We also note that $H_1( \cC'_{(2k)}) \cong H_1(\cC''_{(2k)}) \cong \ff_2$, so we can use the 4-level chain complexes in the context of subsystem encoding from now on.

Collecting everything together, the protected $Z$-type logical sectors are
\begin{equation}
    H_1(\cC''_{(2k)})\otimes H_0(\cD''_{(2k)}) \cong \ff_2,
    \qquad
    H_0(\cC''_{(2k)})\otimes H_1(\cD''_{(2k)}) \cong \ff_2.
\end{equation}
The corresponding $X$-type protected logical sectors are defined via the cohomology
\begin{equation}
    H^1(\cC''_{(2k)})\otimes H^0(\cD''_{(2k)}) \cong \ff_2,
    \qquad
    H^0(\cC''_{(2k)})\otimes H^1(\cD''_{(2k)}) \cong \ff_2.
\end{equation}
Using the homology-cohomology pairing, together with the relation $\cD_{(2k)}=\cC_{(2k)}^T$, gives a nondegenerate pairing between these protected $Z$- and $X$-type sectors.  Thus, the protected subsystem contains two logical qubits.  The remaining logical operators of the stabilizer code $\cQ_{(2k)}$ are treated as gauge operators for the purposes of encoding and decoding.

Having defined the subsystem encoding, we can bound the memory lifetime in the $X$ and $Z$ logical sectors separately, as in Sec.~\ref{sec:memlifetime}.  The main difference is that, although the Hamiltonian is the local stabilizer Hamiltonian of $\cQ_{(2k)}$, the decoder treats errors modulo the small logical operators.  Equivalently, the decoder corrects only the protected subsystem.  Using the Peierls-type argument analogous to that in Sec.~\ref{sec:memlifetime} gives self-correction for the protected logical qubits.  Though we were able to formally prove this, the proof is quite lengthy, and thus we do not provide it here.

We also claim that the random embedding technique from the main text applies to this construction.  If $\ell$ is chosen sufficiently large, the code can be embedded in $\RR^3$.  To see this, one can express the construction of $\cQ_{(2k)}$ as a four-step construction geometrically similar to that in Subsec.~\ref{subsec:4-step}.  In particular, Subsec.~\ref{subsec:4-step} explains the four-step version of the main construction, which consists of perturbation, subdivision, thickening, and degree reduction.  For $\cQ_{(2k)}$ the perturbation and subdivision steps remain unchanged; thickening and degree reduction steps are modified. Given the code $\cQ_{(i)}$, we obtain a new code $\cQ_{(i)}'$ after perturbation and subdivision; then, we perform the next steps.
\begin{quote}
\textbf{3. Thickening:} The local patches of the surface code in $\cQ_{(i)}'$ are replaced by thickness-$1$ slabs of the $(1,3)$-type 4D surface code, with the additional boundary surfaces chosen to be ``volume-operator'' condensing.  In the $\cR_X$ iteration, the volume-like operators are $X$-type; in the $\cR_Z$ iteration, they are $Z$-type.  Thus, the thickening step increases the energy cost of either $X$-type or $Z$-type errors, depending on the iteration.

In the $\cR_X$ iteration, this step corresponds to taking a tensor product $\cQ'_{(i)}\otimes K$, where
\begin{equation}
    K:\quad
    K^0 \longrightarrow K^1 \longrightarrow K^2
\end{equation}
is the cochain complex of a single square face, with four edges and four corner vertices, equipped with the natural incidence maps.  We then keep only the two-skeleton of the product complex, i.e. the terms of total degree at most $2$.

In the $\cR_Z$ iteration, we instead take a tensor product $\cQ'_{(i)}\otimes \wt{K}$, where
\begin{equation}
    \wt{K}:\quad
    \wt{K}^{-2} \longrightarrow \wt{K}^{-1}
    \longrightarrow \wt{K}^{0}
\end{equation}
is the dual cochain complex of $K$.  We then keep only the terms of total degree at least $0$.
\end{quote}
The details of the degree reduction step also change, but the rough idea is the same. Recall that the role of the degree reduction step is to restore the maximum edge and vertex degrees to universal constants $z_e$ and $z_v$; this is done by removing some of the redundant checks and some of the bits paired with checks in a way that preserves the total number of logical qubits in the code. In fact, this removal can be simply chosen to remove some of the elements in $\cQ'_{(i)}$ such that this construction identically reproduces the product construction.

Finally, the parameters in the perturbation step can be chosen so that the construction is equipped with a local and bounded-density embedding map $I_{(k)}:\cQ_{(2k)}\longrightarrow \RR^3$.  Within the 4-step framework, the analysis of the perturbation step is nearly identical to the one from the main text, and would imply a $(t, \ell)$-embedding map for the product code for constants $t$ and $\ell$.

As a remark, because this construction contains many small logical operators, its perturbative stability is a subtle question that we leave to future work.  We hope that the protected subsystem is perturbatively stable, in which case, from the physics standpoint, the code Hamiltonian could correspond to an unusual gapless phase of matter. In particular, there are different encoding subspaces at each level $k$, and a Hamiltonian perturbation can split the degeneracy associated with the subspace, opening a gap of size $\exp(-O(k))$ at level $k$.

\bibliographystyle{unsrturl}
\bibliography{ref}

\begin{thebibliography}{10}

\bibitem{dennis2002topological}
Eric Dennis, Alexei Kitaev, Andrew Landahl, and John Preskill.
\newblock Topological quantum memory.
\newblock {\em Journal of Mathematical Physics}, 43(9):4452--4505, 2002.

\bibitem{brown2016quantum}
Benjamin~J Brown, Daniel Loss, Jiannis~K Pachos, Chris~N Self, and James~R Wootton.
\newblock Quantum memories at finite temperature.
\newblock {\em Reviews of Modern Physics}, 88(4):045005, 2016.

\bibitem{alicki2010thermal}
Robert Alicki, Michal Horodecki, Pawel Horodecki, and Ryszard Horodecki.
\newblock {On thermal stability of topological qubit in Kitaev's 4D model}.
\newblock {\em Open Systems \& Information Dynamics}, 17(01):1--20, 2010.

\bibitem{alicki2009thermalization}
Robert Alicki, Mark Fannes, and Michal Horodecki.
\newblock {On thermalization in Kitaev's 2D model}.
\newblock {\em Journal of Physics A: Mathematical and Theoretical}, 42(6):065303, 2009.

\bibitem{bravyi2009no}
Sergey Bravyi and Barbara Terhal.
\newblock A no-go theorem for a two-dimensional self-correcting quantum memory based on stabilizer codes.
\newblock {\em New Journal of Physics}, 11(4):043029, 2009.

\bibitem{yoshida2011feasibility}
Beni Yoshida.
\newblock Feasibility of self-correcting quantum memory and thermal stability of topological order.
\newblock {\em Annals of Physics}, 326(10):2566--2633, 2011.

\bibitem{haah2011local}
Jeongwan Haah.
\newblock Local stabilizer codes in three dimensions without string logical operators.
\newblock {\em Physical Review A—Atomic, Molecular, and Optical Physics}, 83(4):042330, 2011.

\bibitem{bravyi2013quantum}
Sergey Bravyi and Jeongwan Haah.
\newblock Quantum self-correction in the 3d cubic code model.
\newblock {\em Physical review letters}, 111(20):200501, 2013.

\bibitem{siva2017topological}
Karthik Siva and Beni Yoshida.
\newblock Topological order and memory time in marginally-self-correcting quantum memory.
\newblock {\em Physical Review A}, 95(3):032324, 2017.

\bibitem{haah2013commuting}
Jeongwan Haah.
\newblock {Commuting Pauli Hamiltonians as maps between free modules}.
\newblock {\em Communications in Mathematical Physics}, 324(2):351--399, 2013.

\bibitem{michnicki20143d}
Kamil~P Michnicki.
\newblock 3d topological quantum memory with a power-law energy barrier.
\newblock {\em Physical review letters}, 113(13):130501, 2014.

\bibitem{portnoy2023local}
Elia Portnoy.
\newblock Local quantum codes from subdivided manifolds, 2023.
\newblock URL: \url{https://arxiv.org/abs/2303.06755}, \href {https://arxiv.org/abs/2303.06755} {\path{arXiv:2303.06755}}.

\bibitem{lin2023geometrically}
Ting-Chun Lin, Adam Wills, and Min-Hsiu Hsieh.
\newblock Geometrically local quantum and classical codes from subdivision.
\newblock {\em arXiv preprint arXiv:2309.16104}, 2023.

\bibitem{williamson2023layer}
Dominic~J Williamson and Nou{\'e}dyn Baspin.
\newblock Layer codes.
\newblock {\em arXiv preprint arXiv:2309.16503}, 2023.

\bibitem{gu2025layer}
Shouzhen Gu, Libor Caha, Shin~Ho Choe, Zhiyang He, Aleksander Kubica, and Eugene Tang.
\newblock Layer codes as partially self-correcting quantum memories.
\newblock {\em arXiv preprint arXiv:2510.06659}, 2025.

\bibitem{williamson2025partial}
Dominic~J Williamson.
\newblock Partial self-correction in layer codes.
\newblock {\em arXiv preprint arXiv:2510.09218}, 2025.

\bibitem{baspin2025free}
Nou{\'e}dyn Baspin.
\newblock {The Free Energy Barrier: An Eyring-Polanyi bound for stabilizer Hamiltonians, with applications to quantum error correction}.
\newblock {\em arXiv preprint arXiv:2509.17356}, 2025.

\bibitem{brell2016proposal}
Courtney~G Brell.
\newblock A proposal for self-correcting stabilizer quantum memories in 3 dimensions (or slightly less).
\newblock {\em New Journal of Physics}, 18(1):013050, 2016.

\bibitem{vezzani2003spontaneous}
Alessandro Vezzani.
\newblock {Spontaneous magnetization of the Ising model on the Sierpinski carpet fractal, a rigorous result}.
\newblock {\em Journal of Physics A: Mathematical and General}, 36(6):1593, 2003.

\bibitem{lin2024proposals}
Ting-Chun Lin, Hsin-Po Wang, and Min-Hsiu Hsieh.
\newblock Proposals for 3d self-correcting quantum memory.
\newblock {\em arXiv preprint arXiv:2411.03115}, 2024.

\bibitem{gromov2012generalizations}
Misha Gromov and Larry Guth.
\newblock {Generalizations of the Kolmogorov--Barzdin embedding estimates}.
\newblock {\em Duke Mathematical Journal}, 161(13):2549--2603, 2012.
\newblock URL: \url{https://doi.org/10.1215/00127094-1812840}.

\bibitem{chesi2010thermodynamic}
Stefano Chesi, Daniel Loss, Sergey Bravyi, and Barbara~M Terhal.
\newblock Thermodynamic stability criteria for a quantum memory based on stabilizer and subsystem codes.
\newblock {\em New Journal of Physics}, 12(2):025013, 2010.
\newblock URL: \url{https://arxiv.org/abs/0907.2807}.

\bibitem{Bravyi_2010}
Sergey Bravyi, Matthew~B. Hastings, and Spyridon Michalakis.
\newblock Topological quantum order: Stability under local perturbations.
\newblock {\em Journal of Mathematical Physics}, 51(9), 2010.
\newblock URL: \url{http://dx.doi.org/10.1063/1.3490195}, \href {https://doi.org/10.1063/1.3490195} {\path{doi:10.1063/1.3490195}}.

\bibitem{li2024domain}
Yabo Li, Zijian Song, Aleksander Kubica, and Isaac~H Kim.
\newblock {Domain walls from SPT-sewing}.
\newblock {\em arXiv preprint arXiv:2411.11967}, 2024.

\bibitem{bergamaschi2026rapid}
Thiago Bergamaschi, Reza Gheissari, and Yunchao Liu.
\newblock {Rapid mixing for Gibbs states within a logical sector: a dynamical view of self-correcting quantum memories}.
\newblock In {\em Proceedings of the 2026 Annual ACM-SIAM Symposium on Discrete Algorithms (SODA)}, pages 3407--3422. SIAM, 2026.

\bibitem{Anshu_2024}
Anurag Anshu, Nikolas~P. Breuckmann, and Quynh~T. Nguyen.
\newblock Circuit-to-hamiltonian from tensor networks and fault tolerance.
\newblock In {\em Proceedings of the 56th Annual ACM Symposium on Theory of Computing}, STOC ’24, page 585–595. ACM, 2024.
\newblock URL: \url{http://dx.doi.org/10.1145/3618260.3649690}, \href {https://doi.org/10.1145/3618260.3649690} {\path{doi:10.1145/3618260.3649690}}.

\bibitem{li2024transform}
Xingjian Li, Ting-Chun Lin, and Min-Hsiu Hsieh.
\newblock Transform arbitrary good quantum ldpc codes into good geometrically local codes in any dimension.
\newblock {\em arXiv preprint arXiv:2408.01769}, 2024.

\bibitem{yuan2026unified}
Andrew~C Yuan.
\newblock Unified framework for quantum code embedding.
\newblock {\em Physical Review A}, 113(2):022438, 2026.

\bibitem{nlab:homological_perturbation_theory}
{nLab authors}.
\newblock homological perturbation theory.
\newblock \url{https://ncatlab.org/nlab/show/homological+perturbation+theory}, April 2026.
\newblock \href{https://ncatlab.org/nlab/revision/homological+perturbation+theory/13}{Revision 13}.

\bibitem{Aasen_2020}
David Aasen, Daniel Bulmash, Abhinav Prem, Kevin Slagle, and Dominic~J. Williamson.
\newblock Topological defect networks for fractons of all types.
\newblock {\em Physical Review Research}, 2(4), October 2020.
\newblock URL: \url{http://dx.doi.org/10.1103/PhysRevResearch.2.043165}, \href {https://doi.org/10.1103/physrevresearch.2.043165} {\path{doi:10.1103/physrevresearch.2.043165}}.

\bibitem{first2024good2querylocallytestable}
Uriya~A. First and Tali Kaufman.
\newblock On good $2$-query locally testable codes from sheaves on high dimensional expanders, 2024.
\newblock URL: \url{https://arxiv.org/abs/2208.01778}, \href {https://arxiv.org/abs/2208.01778} {\path{arXiv:2208.01778}}.

\bibitem{panteleev2024maximallyextendablesheafcodes}
Pavel Panteleev and Gleb Kalachev.
\newblock Maximally extendable sheaf codes, 2024.
\newblock URL: \url{https://arxiv.org/abs/2403.03651}, \href {https://arxiv.org/abs/2403.03651} {\path{arXiv:2403.03651}}.

\bibitem{lin2024transversalnoncliffordgatesquantum}
Ting-Chun Lin.
\newblock Transversal non-clifford gates for quantum ldpc codes on sheaves, 2024.
\newblock URL: \url{https://arxiv.org/abs/2410.14631}, \href {https://arxiv.org/abs/2410.14631} {\path{arXiv:2410.14631}}.

\bibitem{li2025poincar}
Yiming Li, Zimu Li, Zi-Wen Liu, and Quynh~T Nguyen.
\newblock Poincar$\backslash$'e duality and multiplicative structures on quantum codes.
\newblock {\em arXiv preprint arXiv:2512.21922}, 2025.

\bibitem{moser2010constructive}
Robin~A Moser and G{\'a}bor Tardos.
\newblock {A constructive proof of the general Lov{\'a}sz local lemma}.
\newblock {\em Journal of the ACM (JACM)}, 57(2):1--15, 2010.
\newblock URL: \url{https://dl.acm.org/doi/abs/10.1145/1667053.1667060}.

\bibitem{moser2009constructive}
Robin~A Moser.
\newblock {A constructive proof of the Lov{\'a}sz local lemma}.
\newblock In {\em Proceedings of the forty-first annual ACM symposium on Theory of computing}, pages 343--350, 2009.
\newblock URL: \url{https://dl.acm.org/doi/abs/10.1145/1536414.1536462}.

\bibitem{chandrasekaran2013deterministic}
Karthekeyan Chandrasekaran, Navin Goyal, and Bernhard Haeupler.
\newblock Deterministic algorithms for the lov{\'a}sz local lemma.
\newblock {\em SIAM Journal on Computing}, 42(6):2132--2155, 2013.

\end{thebibliography}
\end{document}